\documentclass{article}
\usepackage[textwidth=18cm,textheight=23cm,centering]{geometry}
\usepackage[all]{xy}
\SelectTips{cm}{}
\usepackage{amsthm, amsmath, bbold, stmaryrd, MnSymbol, amsxtra}
\usepackage{enumerate}
\usepackage{xspace}

\newdir{ >}{{}*!/-5pt/\dir{>}}
\newdir{ (}{{}*!/-10pt/\dir^{(}}

\theoremstyle{plain}
\newtheorem{theorem}{Theorem}[subsection]
\newtheorem{lemma}[theorem]{Lemma}
\newtheorem{corollary}[theorem]{Corollary}
\newtheorem{proposition}[theorem]{Proposition}

\theoremstyle{definition}
\newtheorem{definition}[theorem]{Definition}

\newtheorem{example}[theorem]{Example}

\newtheorem{note}[theorem]{Note}

\newtheorem{notation}[theorem]{Notation}

\def\epsilon{\varepsilon}

\renewcommand{\rho}{\varrho}

\newcommand{\takeout}[1]{\empty}

\newcommand{\Set}{\mathsf{Set}}

\newcommand{\BA}{\mathsf{BA}}

\newcommand{\DL}{\mathsf{DL}}

\newcommand{\Poset}{\mathsf{Poset}}

\newcommand{\xto}[1]{\xrightarrow{#1}}
\newcommand{\hookto}{\hookrightarrow}

\newcommand{\epito}{\twoheadrightarrow}

\newcommand{\up}{\uparrow}

\newcommand{\down}{\downarrow}

\newcommand{\monoto}{\rightarrowtail}

\newcommand{\acc}{{\tt acc}}

\newcommand{\two}{\mathbb{2}}

\newcommand{\three}{\mathbb{3}}

\newcommand{\ang}[1]{\langle #1 \rangle}
\newcommand{\sem}[1]{\llbracket #1 \rrbracket}
\newcommand{\Land}{\bigwedge}
\newcommand{\Lor}{\bigvee}
\newcommand{\To}{\Rightarrow}
\newcommand{\oT}{\Leftarrow}

\usepackage{mathrsfs}
\usepackage[draft]{fixme} 

\usepackage{xcolor}
\usepackage{hyperref}
\hypersetup{
  colorlinks,
  linkcolor={red!50!black},
  citecolor={blue!50!black},
  urlcolor={blue!80!black},
  hypertexnames=false
}

\CompileMatrices

\makeatletter
\newcommand{\xTo}[2][]{\ext@arrow 0359\Rightarrowfill@{#1}{#2}}
\makeatother

\newcommand{\jslTight}[1]{\mathsf{Ti}[#1]}

\newcommand{\Nleq}{\mathtt{Nleq}}

\newcommand{\syncp}{\owedge}
\newcommand{\ttenp}{\otimes_t}

\newcommand{\cov}{\mathscr{C}}

\newcommand{\rG}{\mathcal{G}}
\newcommand{\rH}{\mathcal{H}}
\newcommand{\rI}{\mathcal{I}}

\newcommand{\rR}{\mathcal{R}}
\newcommand{\rS}{\mathcal{S}}
\newcommand{\rT}{\mathcal{T}}

\newcommand{\rM}{\mathcal{M}}
\newcommand{\rN}{\mathcal{N}}
\newcommand{\rO}{\mathcal{O}}

\newcommand{\rD}{\mathcal{D}}
\newcommand{\rE}{\mathcal{E}}
\newcommand{\rF}{\mathcal{F}}

\newcommand{\aR}{\mathbb{R}}
\newcommand{\aS}{\mathbb{S}}
\newcommand{\aT}{\mathbb{T}}
\newcommand{\aU}{\mathbb{U}}

\newcommand{\pP}{\mathtt{P}}
\newcommand{\pQ}{\mathtt{Q}}

\newcommand{\jslER}[1]{\mathbb{ER}(#1)}

\newcommand{\minDfa}[1]{\mathbf{dfa}(#1)}
\newcommand{\minSatDfa}[1]{\mathbf{dfa}_\down(#1)}
\newcommand{\rs}[1]{\mathrm{rs}(#1)}
\newcommand{\rsc}[1]{\mathrm{\bf rsc}(#1)}
\newcommand{\reach}[1]{\mathbf{reach}(#1)}
\newcommand{\rev}[1]{\mathrm{\bf rev}(#1)}
\newcommand{\coreach}[1]{\mathbf{coreach}(#1)}
\newcommand{\simple}[1]{\mathbf{simple}(#1)}
\newcommand{\simpleIrr}[1]{\mathbf{simple}_{\lor}(#1)}
\newcommand{\jslReach}[1]{\mathbb{reach}(#1)}
\newcommand{\jslDfaReach}[1]{\mathfrak{reach}(#1)}
\newcommand{\jslDfaSimple}[1]{\mathfrak{simple}(#1)}
\newcommand{\jslDfaMin}[1]{\mathfrak{dfa}(#1)}
\newcommand{\jslDfaBoolMin}[1]{\mathfrak{dfa}_{\neg}(#1)}
\newcommand{\jslDfaDistMin}[1]{\mathfrak{dfa}_{\land}(#1)}
\newcommand{\jslDfaSynMin}[1]{\mathfrak{dfa}_{\mathsf{Syn}}(#1)}
\newcommand{\jslDfaTs}[1]{\mathfrak{ts}(#1)}
\newcommand{\jslDfaSc}[1]{\mathfrak{sc}(#1)}

\newcommand{\cl}{\mathrm{\bf cl}}
\newcommand{\inte}{\mathrm{\bf in}}

\newcommand{\fN}{\mathfrak{N}}
\newcommand{\fM}{\mathfrak{M}}

\newcommand{\JSL}{{\mathsf{JSL}}}

\newcommand{\BiCliq}{\mathsf{Dep}}

\newcommand{\Dep}{\mathsf{Dep}}
\newcommand{\AutDep}{\mathsf{aut}_{\Dep}}

\newcommand{\Rel}{\mathsf{Rel}}

\newcommand{\Id}{\mathsf{Id}}
\newcommand{\Pow}{\mathcal{P}}
\newcommand{\FPow}{\Pow_f}

\newcommand{\Pirr}{\mathtt{Pirr}}
\newcommand{\Open}{\mathtt{Open}}
\newcommand{\Rev}{\mathtt{Rev}}

\newcommand{\Airr}{\mathtt{Airr}}
\newcommand{\Det}{\mathtt{Det}}

\newcommand{\pOp}{\mathsf{op}}

\newcommand{\jslLQ}[1]{\mathbb{LQ}(#1)}

\newcommand{\jslLP}[1]{\mathbb{LP}(#1)}
\newcommand{\jslLD}[1]{\mathbb{LD}(#1)}
\newcommand{\jslLRP}[1]{\mathbb{LRP}(#1)}
\newcommand{\jslAtz}[1]{\mathbb{At}_{#1}}

\newcommand{\at}[1]{\mathsf{at}_{#1}}
\newcommand{\Atz}[1]{\mathsf{At}_{#1}}

\newcommand{\jslLangs}[1]{\mathbb{langs}(#1)}

\newcommand{\rDR}[1]{\mathcal{DR}_{#1}}

\newcommand{\rMN}[1]{\mathcal{MN}_{#1}}

\newcommand{\dfa}[1]{\mathsf{dfa}_{#1}}

\def\endbox{\hfill$\blacksquare$}

\newcommand{\rep}{\mathrm{rep}}
\newcommand{\red}{\mathrm{red}}

\newcommand{\LW}[1]{\mathsf{LW}(#1)}

\newcommand{\LQ}[1]{\mathsf{LQ}(#1)}
\newcommand{\RQ}[1]{\mathsf{RQ}(#1)}
\newcommand{\ILQ}[1]{\mathsf{ILQ}(#1)}

\newcommand{\LP}[1]{\mathsf{LP}(#1)}
\newcommand{\LD}[1]{\mathsf{LD}(#1)}
\newcommand{\LRP}[1]{\mathsf{LRP}(#1)}
\newcommand{\LRW}[1]{\mathsf{LRW}(#1)}

\newcommand{\iffLabel}[1]{\overset{#1}\iff}
\newcommand{\src}{i}
\newcommand{\snk}{o}

\newcommand{\toJdfa}[1]{{\bf jdfa}(#1)}
\newcommand{\toLext}[1]{{\bf lext}(#1)}

\newcommand{\rLCov}[2]{\ang{#1, #2}}
\newcommand{\revEl}[2]{{rev}_{#1} (#2)}
\newcommand{\dep}[1]{dep(#1)}

\newcommand\longepi[2][]{%
  \mathrel{\ooalign{$\xrightarrow[#1\mkern4mu]{#2\mkern4mu}$\cr%
  \hidewidth$\rightarrow\mkern4mu$}}
}

\newcommand{\SynMon}[1]{{\bf Syn}(#1)}
\newcommand{\jslDfaRqc}[1]{\mathfrak{rqc}(#1)}
\newcommand{\TM}[1]{{\bf TM}(#1)}
\newcommand{\TS}[1]{{\bf TS}(#1)}
\newcommand{\SynSr}[1]{{\bf Syn}_\lor(#1)}

\newcommand{\FMon}[1]{{\bf #1}^*}
\newcommand{\jslDfaSyn}[1]{\mathfrak{syn}(#1)}
\newcommand{\SynSrCong}[1]{\rS_{#1}^\lor}
\newcommand{\SynCong}[1]{\rS_{#1}}
\newcommand{\SynMonDfa}[1]{{\bf dfa}_{{\bf Syn}(#1)}}
\newcommand{\SynMonSatDfa}[1]{({\bf dfa}_{{\bf Syn}(#1)})_\down}
\newcommand{\TMDfa}[1]{{\bf dfa}_{{\bf TM}(#1)}}
\newcommand{\bS}{{\bf S}}
\newcommand{\bT}{{\bf T}}
\newcommand{\bSigma}{{\bf \Sigma}}
\newcommand{\jslDfaSynBoolMin}[1]{\mathfrak{dfa}_{\mathsf{Syn}}^\neg(#1)}
\newcommand{\jslDfaSynDistMin}[1]{\mathfrak{dfa}_{\mathsf{Syn}}^\land(#1)}

\usepackage{array}

\CompileMatrices
\begin{document}

\title{
  Nondeterministic Automata and $\JSL$-dfas
}
\author{Robert S.R. Myers}
\maketitle

\section{Introduction}

Here's a summary of our results.

\begin{itemize}
  \item Section 2 provides background on finite join-semilattices and describes an equivalent category $\Dep$. The latter has the finite relations as objects/morphisms; its self-duality takes the converse of objects/morphisms. This section serves as a succinct version of our paper `Representing Semilattices as Relations'.
  
  \item
  Section 3 introduces the concept of \emph{Dependency Automaton} i.e.\ two nfas with a relation between their states satisfying compatibility conditions. They are essentially deterministic automata interpreted in $\Dep$, or equivalently deterministic finite automata interpreted in join-semilattices. The state-minimal $\JSL$-dfa accepting $L$ amounts to the left quotients of $L$. As a dependency automaton it can be represented as the state-minimal dfas for $L$ and $L^r$ related by the \emph{dependency relation} $\rDR{L}(u^{-1} L, v^{-1} L^r) :\iff uv^r \in L$.
  
  We also go into some detail concerning various canonical $\JSL$-dfas and their corresponding dependency automata. For example, Polak's syntactic semiring is the transition semiring of the state-minimal $\JSL$-dfa. Also, the power semiring of the syntactic monoid dualises the closure of $L$ under left/right quotients and boolean operations.

  \item Section 4 contains many results concerning the Kameda-Weiner algorithm. They lack a unifying thread, although they're all concerned with the same topic. The reader might skip to the final subsection. There it is proved that an nfa $\rN$ is `subatomic' iff the transition monoid of $\rsc{\rev{\rN}}$ is syntactic.

\end{itemize}



\section{Relations and semilattices}

The Kameda-Weiner algorithm is not an easy read \cite{KamedaWeiner1970}. It searches for an edge-covering of a bipartite graph by complete bipartite graphs, where each covering induces a nondeterministic automaton. The best known lower-bound techniques for nondeterministic automata involve such edge-coverings \cite{NDLowBoundsHard2006}. Thus we begin with a structural theory at this underlying level. Our approach is based on the work of Moshier and Jipsen \cite{ContextJipsen2012}. Our category $\Dep$ is a variant of their category ${\bf Ctxt}$, and we denote its composition by $\fatsemi$ as they did.

\subsection{Biclique edge-coverings as morphisms}

\begin{notation}[Relations and graphs]
  \item
  \begin{enumerate}
    \item
    Given a subset $X \subseteq Z$ then its \emph{relative complement} is written $\overline{X} \subseteq Z$. Given $z \in Z$ we may write $\overline{z} := \overline{\{ z \}}$. The collection of all subsets of $Z$ is denoted $\Pow Z$.

    \item
    A \emph{relation} is a subset of a specified cartesian product $\rR \subseteq X \times Y$. We denote its \emph{domain} by $\rR_s := X$ and its \emph{codomain} by $\rR_t := Y$. The \emph{relational composition} $\rR ; \rS \subseteq \rR_s \times \rS_t$ is defined whenever $\rR_t = \rS_s$, as follows:
    \[
      \rR ; \rS(x,z) :\iff \exists y \in \rR_t. \rR(x,y) \text{ $\land$ } \rS(y,z).
    \]
    Each set $Z$ has the \emph{identity relation} $\Delta_Z \subseteq Z \times Z$ defined $\Delta_Z(z_1, z_2) :\iff z_1 = z_2$.

    The \emph{image of $X \subseteq \rR_s$ under $\rR$} is denoted $\rR[X] := \{ y \in \rR_t : \exists x \in X.\rR(x, y) \}$; we may write $\rR[\{ z \}]$ as $\rR[z]$. $\rR |_{X, Y} := \rR \cap X \times Y$ is the \emph{domain-codomain restriction} of the relation $\rR$. The \emph{converse relation} is defined $\breve{\rR}(x, y) :\iff \rR(y, x)$, in particular $\breve{\rR}_s = \rR_t$ and $\breve{\rR}_t = \rR_s$.
    
    \item
    An \emph{undirected graph} (or just \emph{graph}) $(V, \rE)$ is a finite set $V$ and an irreflexive and symmetric relation $\rE \subseteq V \times V$.
    A \emph{bipartition} for a graph $(V, \rE)$ is a pair $(X, Y)$ where $X \cap Y = \emptyset$, $X \cup Y = V$ and $\rE = \rE |_{X, Y} \; \cup \; \rE |_{Y, X}$. A graph is said to be \emph{bipartite} if it has a bipartition.
    \endbox
  \end{enumerate}
\end{notation}

\begin{note}[Bipartitioned graphs as binary relations]
  \label{note:bipartitioned_as_binary}
  A bipartite undirected graph $(V, \rE)$ with bipartition $(X, Y)$ amounts to a relation $\rE |_{X, Y} \subseteq X \times Y$. This completely captures its structure. Every relation between finite sets arises from a bipartitioned graph, modulo bijective relabelling of its domain (or codomain).
  
  From this perspective, complete bipartite graphs (bicliques) are \emph{cartesian products}. Covering the edges of a bipartite graph by bicliques amounts to \emph{factorising} $\rR = \rS ; \rT$. This relationship is well-known \cite{GregoryBiclique1991}.  The number of bicliques is the cardinality of the set $\rS_t = \rT_s$ factorised through. The minimum possible cardinality is the \emph{bipartite dimension} of the respective bipartite graph i.e.\ the minimum number of bicliques needed to cover the edges. \endbox
\end{note}

\begin{definition}[Biclique edge-coverings]
  \item
  \begin{enumerate}
    \item
    A \emph{biclique} of a relation $\rR$ is a cartesian product $X \times Y \subseteq \rR$.
    \item
    A \emph{biclique edge-covering} of a relation $\rR$ is a factorisation $\rR = \rS ; \rT$.
    Its \emph{underlying bicliques}:
    \[
      \cov_{\rS, \rT} := \{ \breve{\rS}[x] \times \rT[x] : x \in \rS_t = \rT_s  \}
    \]
    satisfy the equality $\bigcup \cov_{\rS, \rT} = \rR$.
    \item
    The \emph{bipartite dimension} $dim(\rR)$ is the minimal cardinality $|\cov_{\rS, \rT}|$ over all factorisations $\rR = \rS ; \rT$.
    \endbox
  \end{enumerate}
\end{definition}

\begin{notation}[Lower/upper bipartition]
  When a binary relation $\rR$ is viewed as a bipartitioned graph, we may refer to its domain as the \emph{lower bipartition} and its codomain as the \emph{upper bipartition}. \endbox
\end{notation}

\begin{example}[Biclique edge-coverings]
  \label{ex:biclique_edge_coverings}
  \item
  \begin{enumerate}
    \item 
    Each relation $\rR$ has two canonical biclique edge-coverings i.e.\ $\rR = \Delta_{\rR_s} ; \rR$ and $\rR = \rR ; \Delta_{\rR_t}$. Viewed as a bipartite graph, the stars centered at each vertex of the lower bipartition cover the edges. Alternatively we can take each star centered at a vertex in the upper bipartition. Consequently $dim(\rR) \leq min(|\rR_s|, |\rR_t|)$.
    


    \item
    Each undirected graph $(V, \rE)$ provides an irreflexive symmetric relation $\rE \subseteq V \times V$. From our viewpoint, this relation defines a bipartitioned graph. It is known as the \emph{bipartite double cover of $(V, \rE)$} i.e.\ take two copies of $V$ and connect $e_1(u)$ to $e_2(v)$ iff $\rE(u, v)$. Starting with a complete graph on vertices $V$ yields the relation $\overline{\Delta_V} \subseteq V \times V$. Interestingly, $dim(\overline{\Delta_V}) \approx \lceil log_2(|V|) \rceil $ by applying Sperner's theorem \cite{OnBicliqCov2004}.
    
     \item
     Each finite poset $\pP = (P, \leq_\pP)$ provides an order relation $\leq_\pP \; \subseteq P \times P$. Viewed as a bipartitioned graph, paths amount to alternating relationships $p_1 \leq_\pP p_2 \geq_\pP p_3 \leq_\pP p_4 \; \cdots$. The edge-coverings from (1) above are optimal i.e.\ $dim(\leq_\pP) = |P|$. The lower/upper bipartition's stars correspond to principal up/downsets.

    \item
    Consider $C_{2n} = (\{ 0, ..., 2n - 1 \}, \rE_{2n})$ where $\rE_{2n}(i, j) :\iff |j - i| = 1$ modulo $2n$ i.e.\ the $2n$-cycle. It has precisely two bipartitions if $n \geq 1$. Assuming $0$ is in the lower bipartition, the respective relation $\rE_{2n} |_{X, Y}$ relates evens to odds and has bipartite dimension $|X| = |Y| = n$.
    
    \item
    Consider $P_n = (\{ 0, ..., n \}, \rE_n)$ where $\rE_n(i, j) :\iff |j - i| = 1$ i.e.\ the path of edge-length $n$. It has precisely two bipartitions if $n \geq 1$. Assuming $0$ is in the lower bipartition, the respective relation $\rE_n |_{X, Y}$ relates evens to odds. Moreover $dim(\rE_{2n} |_{X, Y}) = dim(\rE_{2n - 1} |_{X, Y}) = n$. \endbox


  \end{enumerate}
\end{example}

\begin{definition}[The category $\Dep$]
  \label{def:the_cat_dep}
  The objects of the category $\Dep$ are the relations $\rG \subseteq \rG_s \times \rG_t$ between finite sets. A morphism $\rR : \rG \to \rH$ is a relation $\rR \subseteq \rG_s \times \rH_t$ such that the diagram:
  \[
  \xymatrix@=15pt{
    \rG_t \ar@{..>}[rr]^{(\rR_u)\spbreve}  && \rH_t
    \\
    \rG_s \ar[u]^-{\rG} \ar[urr]^-\rR \ar@{..>}[rr]_{\rR_l} && \rH_s \ar[u]_-{\rH}
  }
  \]
  commutes in $\Rel_f$\footnote{$\Rel_f$ is the category whose objects are the finite sets and whose morphisms are the binary relations, composed via relational composition.} for some $\rR_l \subseteq \rG_s \times \rH_s$ and $\rR_u \subseteq \rH_t \times \rG_t$\footnote{The converse symbol $(\rR_u)\spbreve$ is intentional. It provides symmetry later on.}. The identity morphisms are $id_\rG := \rG$ and composition $\rR \fatsemi \rS : \rG \xto{\rR} \rH \xto{\rS} \rI$ is defined:
  \[
  \xymatrix@=15pt{
    \rG_t \ar@{..>}[rr]^{(\rR_u)\spbreve}  && \rH_t \ar@{..>}[rr]^{(\rS_u)\spbreve} && \rI_t
    \\
    \rG_s \ar[u]^-{\rG} \ar[urr]^-\rR \ar@{..>}[rr]_{\rR_l} && \rH_s \ar[u]_-{\rH} \ar[urr]^-\rS \ar@{..>}[rr]_{\rS_l} && \rI_s \ar[u]_-{\rI}
  }
  \]
  That is, $\rR \fatsemi \rS := \rR_l ; \rS_l = \rR_l ; \rS = \rR_l ; \rH ; (\rS_u)\spbreve = \rR ; (\rS_u)\spbreve = \rG ; (\rR_u)\spbreve ;(\rS_u)\spbreve$ is any of the five equivalent relational compositions starting from the bottom left and ending at the top right. \endbox
\end{definition}

Then a $\Dep$-morphism $\rR : \rG \to \rH$ is a relation factorising through $\rG$ on the left and $\rH$ on the right. In view of Note \ref{note:bipartitioned_as_binary}, it amounts to two biclique edge-coverings of $\rR$.


\begin{lemma}
  $\Dep$ is a well-defined category.
\end{lemma}

\begin{proof}
  Concerning identity morphisms, $\Delta_{\rG_s} ; \rG = \rG = \rG ; \Delta_{\rG_t}\spbreve$; graph-theoretically we are using the star-coverings from Example \ref{ex:biclique_edge_coverings}.1. Concerning composition, $\rR \fatsemi \rS$ is well-defined: (i) the commuting rectangle provides witnessing relations $\rR_l ; \rS_l$ and $\rS_u ; \rR_u$, (ii) $\rR \fatsemi \rS$ is independent of the witnesses for $\rR$ and $\rS$ by considering the $5$ relational compositions. We have $\rR \fatsemi id_\rG = \rR ; \Delta_{\rG_t}\spbreve = \rR$ and $id_\rG \fatsemi \rR = \Delta_{\rG_s} ; \rR =  \rR$. Composition is associative because relational composition is.
\end{proof}

\begin{example}[$\Dep$-morphisms]
  \label{ex:dep_morphisms}
  \item
  \begin{enumerate}
    \item \emph{$\Dep$-morphisms are closed under converse and union}.

    Given $\rR : \rG \to \rH$ then $\breve{\rR} : \breve{\rH} \to \breve{\rG}$ by taking the converse of the commutative square, which actually swaps the witnessing relations. We have $\emptyset : \rG \to \rH$ via empty witnessing relations. Given $\rR, \rS: \rG \to \rH$ then $\rR \cup \rS: \rG \to \rH$ by (i) unioning the respective witnessing relations, (ii) the bilinearity of relational composition w.r.t.\ union.

    \item \emph{Bipartite graph isomorphisms $\beta : G_1 \to G_2$ induce $\Dep$-isomorphisms.}
    
    Suppose we have a bipartite graph isomorphism $\beta : G_1 \to G_2$ where each $G_i = (V_i, \rE_i)$, so $\rE_1(x,y) \iff \rE_2(\beta(x),\beta(y))$. Given any bipartition $(X, Y)$ of $G_1$ we obtain a bipartition $(\beta[X], \beta[Y])$ of $G_2$. Setting $\rG_i := \rE_i |_{X \times Y}$ provides the $\Dep$-morphism below left:
    \[
      \begin{tabular}{ccc}
        $\xymatrix@=15pt{
          Y \ar@{->}[rr]^{\beta |_{Y \times \beta[Y]}}  && \beta[Y]
          \\
          X \ar[u]^-{\rG_1} \ar@{..>}[urr]^-\rR \ar@{->}[rr]_{\beta |_{X \times \beta[X]} } && \beta[X] \ar[u]_-{\rG_2}
        }$
        &&
        $\xymatrix@=15pt{
          \beta[Y] \ar@{->}[rr]^{\breve{\beta} |_{\beta[Y] \times Y}}  && Y
          \\
          \beta[X] \ar[u]^-{\rG_2} \ar@{..>}[urr]^-\rS \ar@{->}[rr]_{\breve{\beta} |_{\beta[X] \times X} } && X \ar[u]_-{\rG_1}
        }$
      \end{tabular}
    \]
    The bijective inverse $\beta^{-1} = \breve{\beta}$ provides witnessing relations in the opposite direction
    i.e.\ the $\Dep$-morphism $\rS : \rG_2 \to \rG_1$ above right. These morphisms are mutually inverse: $\rG_1$ is $\Dep$-isomorphic to $\rG_2$.

    \item \emph{The canonical quotient poset of a preorder defines a $\Dep$-isomorphism.}
    
    Let $\rG \subseteq X \times X$ be a transitive and reflexive relation. There is a canonical way to construct a poset $\pP = (X/{\rE}, \leq_{\pP})$ via the equivalence relation $\rE(x_1, x_2) :\iff \rG(x_1, x_2) \text{ $\land$ } \rG(x_2, x_1)$, where $\sem{x_1}_\rE \leq_{\pP} \sem{x_2}_\rE :\iff \rG(x_1, x_2)$.
    
    Consider the $\Rel$-diagram:
    \[
      \xymatrix@=15pt{
        X \ar@{->}[rr]^-{\nin}
        && \{ \overline{\breve{\rG}[x]} : x \in X \} \ar@{=}[r]
        &
        \{ \overline{\bigcup \down_{\pP} \sem{x}_\rE} : x \in X \} \ar[rrr]^-{(\lambda \sem{x}_\rE. \overline{\breve{\rG}[x]})\spbreve}
        &&& X / \rE
        \\
        X \ar[u]^-{\rG} \ar@{->}[rr]_-{\lambda x.\rG[x]}
        && \{ \rG[x] : x \in X \} \ar@{=}[r] \ar[u]_-{\nsubseteq}
        &
        \{ \bigcup \up_{\pP} \sem{x}_\rE : x \in X \}
        \ar@{->}[rrr]_-{(\lambda \sem{x}_\rE. \rG[x])\spbreve}
        &&& X / \rE \ar[u]_{\leq_{\pP}}
      }
    \]
    Note that $\rG[x]$ is the `upwards closure' i.e.\ the union of the upwards closure $\up_{\pP} \sem{x}_{\rE}$, whereas $\breve{\rG}[x]$ is the `downwards closure' in a similar manner. The left square commutes for completely general reasons, defining the $\Dep$-morphism:
    \[
      \rR(x_1, \overline{\breve{\rG}[x_2]})
      :\iff \exists x \in X.[\rG(x_1, x) \text{ $\land$ } \rG(x, x_2) ]
      \iff \rG(x_1, x_2).
    \]
    The right square involves bijections via (i) identifying elements of $\pP$ with principal up/downsets, (ii) the disjointness of equivalence classes. It also commutes:
    \[
      \begin{tabular}{lll}
        $\bigcup \up_{\pP} \sem{x_1}_\rE \nsubseteq \overline{\bigcup \down_{\pP} \sem{x_2}_\rE}$
        &
        $\iff \bigcup \up_{\pP} \sem{x_1}_\rE \;\cap\; \bigcup \down_{\pP} \sem{x_2}_\rE \neq \emptyset$
        \\&
        $\iff \exists x \in X. \sem{x_1}_\rE \leq_\pP \sem{x}_\rE \leq_\pP \sem{x_2}_\rE$
        \\&
        $\iff \sem{x_1}_\rE \leq_{\pP} \sem{x_2}_\rE$.
      \end{tabular}
    \]
    In fact, $\rR : \rG \to \;\nsubseteq$ is an instance of the natural isomorphism $\red_\rG$ from Theorem \ref{thm:dep_equiv_jslf} further below, and the right square defines a $\Dep$-isomorphism by Example 2 above. Thus $\rG \cong \;\leq_\pP$, although whenever $|X| > |X / \rE|$ this isomorphism \emph{cannot arise from a bipartite graph isomorphism}.

    \item \emph{Monotonicity can be characterised by $\Dep$-morphisms.}
    
    Given finite posets $\pP$ and $\pQ$, a function $f: P \to Q$ is monotonic iff the following $\Rel$-diagram commutes:
    \[
    \xymatrix@=15pt{
    P \ar[r]^f & Q \ar[r]^{\leq_{\pQ}} & Q 
    \\
    P \ar[u]^{\leq_\pP} \ar[rr]_f && Q \ar[u]_{\leq_\pQ}
    }
    \]
    as the reader may verify. Actually, $f$ is monotonic iff $f ; \leq_\pQ  \,:\, \leq_\pP \,\to\, \leq_\pQ$ is a $\Dep$-morphism.
    Indeed, given that $f ; \leq_\pQ  \,:\, \leq_\pP \,\to\, \leq_\pQ$ is a $\Dep$-morphism we'll prove that $f$ is monotonic in Example \ref{ex:dep_char_monotonocity} further below.

    \item \emph{Biclique edge-coverings amount to $\Dep$-monos}.

    Generally speaking, $\Dep$-morphisms represent two edge-coverings of a bipartitioned graph. A \emph{single edge-covering} amounts to a $\Dep$-mono of a special kind:
      \[
        \xymatrix@=15pt{
          \rG_t \ar[rr]^{\Delta_{\rG_t}}  && \rG_t
          \\
          \rG_s \ar[u]^-{\rG} \ar[urr]^-\rG \ar@{->}[rr]_{\rG_l} && \rH_s \ar[u]_-{\rH}
        }
      \]
      i.e.\ morphisms $\rG: \rG \to \rH$ with the additional assumption $\rG_t = \rH_t$. It will follow later that any mono $\rR : \rG \to \rI$ induces such a $\rG : \rG \to \rH$ where $|\rH_s| \leq |\rI_s|$ and $|\rH_t| \leq |\rI_t|$ i.e.\ see Theorem \ref{thm:extension_cover_correspondence}.
      
    \item \emph{Biclique edge-coverings amount to $\Dep$-epis}.

    Analogous to the previous example, a single edge-covering can be represented as a $\Dep$-epi $\rG : \rH \to \rG$ where $\rG_s = \rH_s$. This will follow from self-duality i.e.\ epis are precisely the converses of monos.
    
    \endbox

  \end{enumerate}
\end{example}

\subsection{The categorical equivalence}

Some of the examples above are order-theoretic in nature. Indeed, the main result of this section is:
\begin{quote}
  \it
  $\Dep$ is categorically equivalent to the finite join-semilattices equipped with join-preserving morphisms.
\end{quote}

This result will characterise $\Dep$-objects modulo isomorphism. They are the union-free or \emph{reduced} relations. Algebraically they correspond to the finite lattices.

\begin{definition}[Image, preimage, closure, interior]
  \label{def:up_down_cl_inte}
For any binary relation $\rR \subseteq \rR_s \times \rR_t$ define:
\[
\begin{tabular}{llllll}
$\rR^\up : \Pow \rR_s \to \Pow \rR_t$
&&
$\rR^\down : \Pow \rR_t \to \Pow \rR_s$
&&
$\cl_\rR := \rR^\down \circ \rR^\up : (\Pow \rR_s,\subseteq) \to (\Pow \rR_s,\subseteq)$
\\
$\rR^\up (X) := \rR[X]$
&&
$\rR^\down (Y) := \{ x \in \rR_s : \rR[x] \subseteq Y \}$
&&
$\inte_\rR := \rR^\up \circ \rR^\down : (\Pow \rR_t,\subseteq) \to (\Pow \rR_t,\subseteq)$
\end{tabular}
\]
where $\Pow Z$ is the collection of all subsets of $Z$. The fixed points of the closure operator $\cl_\rR$ are $C(\rR) := \{ \rR^\down(Y) : Y \subseteq \rR_t \}$. The fixed-points of the interior operator\footnote{Interior operators are also known as co-closure operators: monotone, idempotent and \emph{co}-extensive i.e.\ $\inte(Y) \subseteq Y$.} $\inte_\rR$ are $O(\rR) := \{ \rR[X] : X \subseteq \rR_s \}$ and are called the \emph{$\rR$-open sets}.
\endbox
\end{definition}

\begin{definition}[Component relations]
  \label{def:maximum_witnesses}
For each $\Dep$-morphism $\rR : \rG \to \rH$ define:
\[
\begin{tabular}{llll}
$\rR_- := \{ (g_s,h_s) \in \rG_s \times \rH_s : h_s \in \rH^\down(\rR[g_s]) \}$
&
$\rR_+ := \{ (h_t,g_t) \in \rH_t \times \rG_t : g_t \in \breve{\rG}^\down(\breve{\rR}[h_t]) \}$
\end{tabular}
\]
called the \emph{lower/upper components} respectively. \endbox
\end{definition}

Importantly, $\rR$'s component relations are witnesses and contain all other witnesses.

\begin{lemma}[Morphisms characterisation and maximum witnesses]
  \label{lem:mor_char_max_witness}
  \item
  Let $\rR \subseteq \rG_s \times \rH_t$ be any relation between finite sets.
  \begin{enumerate}
    \item 
    $
    \rR^\up(X) \subseteq Y
    \iff
    X \subseteq \rR^\down(Y)
    $
    for all subsets $X \subseteq \rG_s$, $Y \subseteq \rH_t$.
    
    \item
    The following labelled equalities hold:
    \[
      \begin{tabular}{ccccc}
      $(\up\down\up)$ & $\rR^\up \circ \rR^\down \circ \rR^\up = \rR^\up$
      &&
      $\rR^\down \circ \rR^\up \circ \rR^\down = \rR^\down$ & $(\down\up\down)$
      \\[1ex]
      $(\neg\up\neg)$ & $\neg_{\rG_t} \circ \rR^\up \circ \neg_{\rG_s} = \breve{\rR}^\down$
      &&
      $\neg_{\rG_s} \circ \rR^\down \circ \neg_{\rG_t} = \breve{\rR}^\up$ & $(\neg\down\neg)$
      \end{tabular}
      \]

    \item
    $\cl_\rR$ (resp.\ $\inte_\rR$) is a well-defined closure (resp.\ interior) operator, in fact $\inte_\rR = \neg_{\rR_t} \circ \cl_{\breve{\rR}} \circ \neg_{\rR_t}$.

    \item
    $\rR$ defines a $\Dep$-morphism $\rG \to \rH$ iff
    $
    \rR^\up \circ \cl_\rG = \rR^\up = \inte_\rH \circ \rR^\up,
    $
    or equivalently $\rR^\up \circ \cl_\rG = \inte_\rH \circ \rR^\up$.
    
    \item
    Each $\rR : \rG \to \rH$ has the maximum witnesses $(\rR_-,\rR_+)$ i.e.\ 
    \begin{itemize}
    \item[--]
    $\rR_- ; \rH = \rR = \rG ; \rR_+\spbreve$.
    \item[--]
    whenever $\rR_l ; \rH = \rR = \rG ; \rR_r\spbreve$ then $\rR_l \subseteq \rR_-$ and $\rR_r \subseteq \rR_+$.
    \end{itemize}

    \item
    For every $\rG$-open $Y \in O(\rG)$ and every $g_t \in \rG_t$.
    \[
      Y \subseteq \inte_\rG(\overline{g_t}) \iff g_t \nin Y.
    \]

  \end{enumerate}
\end{lemma}

\begin{proof}
  See background paper i.e.\
  \begin{itemize}
    \item[--]
    Lemma 4.1.7, \emph{Relating $(-)^\up$ and $(-)^\down$} and also Lemma 4.2.7.
    \item[--]
    Lemma 4.1.10, \emph{Morphism characterisation and maximum witnesses}.
  \end{itemize}
\end{proof}

\begin{example}
\label{ex:dep_char_monotonocity}
\item
\begin{enumerate}
  \item 
  \emph{Characterizing monotonicity}.

  Recalling Example \ref{ex:dep_morphisms}.4, suppose $f : P \to Q$ is a function and $f ; \leq_\pQ \; : \; \leq_\pP \; \to \; \leq_\pQ$ is a $\Dep$-morphism. Since $\cl_{\leq_\pP}$ constructs the upwards closure in $\pP$, by Lemma \ref{lem:mor_char_max_witness}.4 for any $p \in P$:
  \[
    (f ; \leq_\pQ)^\up (\up_\pP p) = (f ; \leq_\pQ)[p]
    \qquad \text{or equivalently} \qquad
    \up_\pQ f[\up_\pP p] = \; \up_\pQ f(p).
  \]
  Thus whenever $p \leq_\pP p'$ we know $\up_\pQ f(p') \subseteq \; \up_\pQ f(p)$ i.e.\ $f(p) \leq_\pQ f(p')$.

  \item 
  \emph{One-sided maximal bicliques}.

  When searching for small edge-coverings by bicliques one can restrict to \emph{maximal ones} i.e.\ $X \times Y \subseteq \rG$ where $(X, Y)$ is pairwise-maximal w.r.t.\ inclusion \cite{OrlinContentment1977}. Then Lemma \ref{lem:mor_char_max_witness}.4 says:
  \begin{quote}
    $\Dep$-morphisms $\rR: \rG \to \rH$ are two \emph{one-sided} maximal edge-coverings, i.e.\ $(\rR_-\spbreve[h_s] \times \rH[h_s])_{h_s \in \rH_s}$ is left-maximal and $(\breve{\rG}[g_t] \times \rR_+\spbreve[g_t])_{g_t \in \rG_t}$ is right-maximal.
  \end{quote}
  Observe we cannot in general pass to \emph{maximal} bicliques without changing the morphism's domain/codomain. \endbox
\end{enumerate}
\end{example}

\begin{definition}[$\JSL_f$]
  A \emph{join-semilattice} is a set with a binary operation $\lor$ and a nullary operation $\bot$, satisfying:
  \[
    \bot \lor x = x = \bot \lor x
    \qquad
    x \lor (y \lor z) = (x \lor y) \lor z
    \qquad
    x \lor y = y \lor x
    \qquad
    x \lor x = x.
  \]
  We write them $\aS = (S, \lor_\aS, \bot_\aS)$ where $S$ is the underlying set, $\lor_\aS : S \times S \to S$ is a function and $\bot_\aS \in S$. A join-preserving morphism $f : \aS \to \aT$ is a function $f : S \to T$ such that $f(s_1 \lor_\aS s_2) = f(s_1) \lor_\aT f(s_2)$ and $f(\bot_\aS) = \bot_\aT$. Finally, $\JSL_f$ is the category of \emph{finite} join-semilattices and join-preserving morphisms.
  \endbox
\end{definition}

\begin{example}[Clarifying join-semilattices]
  \item
  \begin{enumerate}
    \item
    Join-semilattices are precisely the commutative and idempotent monoids. Consequently each $\aS$ has a Cayley-representation as endofunctions $(- \lor_\aS s : S \to S)_{s \in S}$ closed under functional composition.
    
    \item
    More importantly, the join-semilattices are precisely the partially-ordered sets with all finite suprema. The binary operation $\lor_\aS$ is the binary join, $\bot_\aS$ is the bottom element. Inductively, $\Lor_\aS X$ exists for all finite $X \subseteq S$.

    \item
    The finite join-semilattices are precisely the finite bounded lattices: the finite partially-ordered sets with all finite suprema and infima. Indeed, every finite join-semilattice is complete i.e.\ has \emph{all} joins, hence has all meets too. That is, $\Land_\aS X$ exists for any finite subset $X \subseteq S$.
    
    \item
    By (3), each finite $\aS := (S, \lor_\aS, \bot_\aS)$ can be flipped yielding the \emph{order-dual} join-semilattice $\aS^\pOp := (S, \land_\aS, \top_\aS)$.

    \item
    The join-semilattice isomorphisms are precisely the bijective join-semilattice morphisms. They are also precisely the order-isomorphisms between the underlying posets i.e.\ bijections preserving and reflecting the ordering. For finite join-semilattices they are precisely the bounded lattice isomorphisms by (3).

    \item
    Let $\two = (\{ 0, 1 \}, \lor_\two, 0)$ be the two element set with ordering $0 \leq_\two 1$.
    Modulo isomorphism there is only one join-semilattice with two elements.
    \endbox
  \end{enumerate}
\end{example}

Each binary relation $\rG$ induces two isomorphic join-semilattices: the $\inte_\rG$-fixpoints $(O(\rG), \cup, \emptyset)$ and the $\cl_\rG$-fixpoints $(C(\rG), \lor_{\cl_\rG}, \cl_\rG(\emptyset))$. The latter's join constructs the closure of the union, whereas its meet is simply intersection. This situation is well-known within the area of Formal Concept Analysis.

\begin{theorem}[Bounded lattice isomorphism of a bipartitioned graph]
  \label{thm:graph_induces_isom}
  For any $\rG \subseteq \rG_s \times \rG_t$ we have the isomorphism:
  \[
    \alpha_\rG : (C(\rG), \lor_{\cl_\rG}, \cl_\rG(\emptyset)) \to (O(\rG), \cup, \emptyset)
    \qquad
    \alpha_\rG(X) := \rG[X]
    \qquad
    \alpha_\rG^{-1}(Y) := \rG^\down(Y).
  \]
\end{theorem}

\begin{proof}
  See background paper i.e.\ Lemma 4.2.5, \emph{The bounded lattices of $\rG$-open/closed sets and their irreducibles}.
\end{proof}

\begin{note}[More about join-semilattices]
  \label{note:jsl_extras}
  \item
  \begin{enumerate}
    \item \emph{Join and meet-irreducibles}.

    Fix a join-semilattice $(S, \lor_\aS, \bot_\aS)$. An element $s \in S$ is \emph{join-irreducible} if whenever $s = \Lor_\aS X$ for finite $X \subseteq S$ we have $s \in X$. They are denoted $J(\aS) \subseteq S$. Likewise $s \in S$ is \emph{meet-irreducible} if whenever $s$ is a finite meet $\Land_\aS X$ then $s \in X$; they are denoted $M(\aS) \subseteq S$.

    \item \emph{Adjoint morphisms}.

    Each $\JSL_f$-morphism $f : \aS \to \aT$ has an \emph{adjoint} $f_* : \aT^\pOp \to \aS^\pOp$ defined $f_* (t) := \Lor_\aS \{ s \in S : f(s) \leq_\aT t \}$. It is uniquely determined by the adjoint relationship $f(s) \leq_\aT t \iff s \leq_\aS f_*(s)$, and preserves all finite meets in $\aT$. We've already seen examples i.e.\ $\rR^\down = (\rR^\up: (\Pow\rR_s, \cup, \emptyset) \to (\Pow\rR_t, \cup, \emptyset))_*$.
    
    \item \emph{Self-duality of $\JSL_f$}.
    
    Adjoint morphisms actually define an equivalence functor $(-)_* : \JSL_f^{op} \to \JSL_f$ where $\aS_* := \aS^\pOp$ is the order-dual join-semilattice and $f_*$ is the adjoint morphism. It is witnessed by the natural isomorphism $\lambda : Id_{\JSL_f} \To (-)_* \circ ((-)_*)^{op}$ where $\lambda_\aS := id_\aS$.

    \item \emph{Monos and epis}.
    
    The $\JSL_f$-monomorphisms are precisely the injective ones and the epimorphisms are precisely the surjective ones. The latter situation is unlike the case of distributive lattices where $\three \hookto \two \times \two$ is epic. Injective $\JSL_f$-morphisms $f$ are also \emph{order-embeddings} i.e.\ $f(s_1) \leq_\aT f(s_2) \iff s_1 \leq_\aS s_2$. Generally speaking, injective monotone functions needn't be order-embeddings e.g.\ take a bijection from a $2$-antichain to a $2$-chain.
    \endbox

  \end{enumerate}
\end{note}

\begin{note}[Irreducibles]
  \label{note:irreducibles}
  \item
  \begin{enumerate}
    \item
    A bottom element is the empty join and hence never join-irreducible. The top element of a finite join-semilattice is the empty meet, hence never meet-irreducible.

    \item
    The join-irreducibles of $(\Pow X, \cup, \emptyset)$ are the singleton sets $\{ x \}$, the meet-irreducibles their relative complements.

    \item
    Each element of a join-semilattice is the join of those join-irreducibles below it. In fact, $J(\aS) \subseteq S$ is the minimal subset generating $\aS$ under joins. Order dually, $M(\aS) \subseteq \aS$ is the minimal subset generating $\aS$ under meets.

    \item
    Finite distributive lattices $\aS$ are determined by their subposet of join-irreducibles. That is, they are isomorphic to the downwards-closed subsets of $(J(\aS), \leq_\aS)$, equipped with union (binary join) and intersection (binary meet). Every join-irreducible is actually \emph{join-prime} i.e.\ $j \leq_\aS \Lor_\aS S \iff \exists s \in S. j \leq_\aS s$.

    \item
    For finite distributive lattices $\aS$, the subposet of join-irreducibles is order-isomorphic to the subposet of meet-irreducibles via $\tau_\aS : J(\aS) \to M(\aS)$ with action $j \mapsto \Lor_\aS \overline{\up_\aS j}$ and inverse $m \mapsto \Land_\aS \overline{\down_\aS m}$. This fails for non-distributive lattices.
  \endbox
  \end{enumerate}
\end{note}

We now have enough structure to define the functorial translation between relations and algebras.

\begin{definition}[The equivalence functors]
  \label{def:open_and_pirr}
  \item
  \begin{enumerate}
    \item
    $\Open: \Dep \to \JSL_f$ constructs the semilattice of $\rG$-open sets:
    \[
      \Open \rG := (O(\rG), \cup, \emptyset)
      \qquad
      \Open \rR := \lambda Y. \rR_+\spbreve [Y].
    \]

    \item
    $\Pirr : \JSL_f \to \Dep$ constructs Markowsky's poset of irreducibles \cite{MarkowskyLat1975}.
    \[
      \begin{tabular}{rl}
      $\Pirr \aS := \;\nleq_\aS |_{J(\aS) \times M(\aS)}$
      &
      $\Pirr f (j, m) :\iff f(j) \nleq_\aT m$
      \\[1ex]
      $(\Pirr f)_-(j_1, j_2) :\iff j_2 \leq_\aT f(j_1)$
      &
      $(\Pirr f)_+(m_1, m_2) :\iff f_*(m_1) \leq_\aS m_2$,
      \end{tabular}
    \]
    where $\Pirr f$'s components are also described above. \endbox
  \end{enumerate}
\end{definition}

\smallskip

$\Open \rG$ is the inclusion-ordered set of neighbourhoods $\rG[X]$ of the lower bipartition i.e.\ particular subsets of the upper bipartition. In the other direction, $\Pirr \aS$ is the domain/codomain restriction of $\nleq_\aS \; \subseteq S \times S$ to join/meet-irreducibles respectively. We have extended the concept studied by Markowsky to morphisms.

\begin{note}[Concerning $\Open$'s action on morphisms]
  \label{note:open_morphism_alt}
  Given $\rR: \rG \to \rH$ we defined $\Open\rR$ as $\lambda Y. \rR_+\spbreve[Y]$. It may equivalently be defined $\lambda Y. \rR^\up \circ \rG^\down (Y)$ i.e.\ one needn't compute the maximal witness $\rR_+$. It may also be defined $\lambda Y. \rR_u\spbreve[Y]$ where $\rR_u \subseteq \rR_+$ is any upper witness, since $\rR_+\spbreve [\rG[X]] = \rG ; \rR_+\spbreve [X] = \rG ; \rR_u\spbreve [X] = \rR_u\spbreve [\rG[X]]$. \endbox
\end{note}

\begin{example}[Semilattices as binary relations]
  \label{ex:semilattice_as_bipartite}
  \item
  \begin{enumerate}
    \item
    \emph{Boolean lattices correspond to identity relations}.

    Observe $\Open \Delta_X = (\Pow X, \cup, \emptyset)$ for any finite set $X$. Applying $\Pirr$ yields the bijection $\{x\} \mapsto \overline{x}$, which is bipartite isomorphic to $\Delta_X$ and hence $\Dep$-isomorphic.

    \item
    \emph{Distributive lattices correspond to order relations}.
    
    Given any order-relation $\leq_\pP \; \subseteq P \times P$ then $\Open \leq_\pP$ consists of all upwards closed subsets of $P$ ordered by inclusion. Since they are closed under unions and intersections, $\Open \leq_\pP$ is a distributive lattice. Conversely if $\aS$ is distributive one can show $\Pirr \aS ; \tau_\aS\spbreve = \; \leq_{\aS^\pOp} |_{J(\aS) \times J(\aS)}$, using notation from Note \ref{note:irreducibles}.5. See the background paper Lemma 2.2.3.14 `Standard order-theoretic results'. Then $\Pirr\aS$ is bipartite isomorphic to an order-relation and hence $\Dep$-isomorphic too.
    
    \item
    \emph{Partition lattices represented via functional composition}.

    Recall the inclusion-ordered lattice $\jslER{X}$ of equivalence relations on a finite non-empty set $X$. Meets are intersections, whereas joins are constructed by taking the transitive closure of the union. Viewed as a join-semilattice, there is a natural binary relation $\rG$ such that $\Open \rG \cong \jslER{X}$:
    \[
      \rG \subseteq \Set(2, X) \times \Set(X, 2)
      \qquad
      \rG(f, g) :\iff g \circ f = id_2
    \]
    where $2 := \{ 0, 1 \}$ and $\Set(A, B)$ is the set of functions from $A$ to $B$. Notice we have:
    \[
      |X|^2 = |\rG_s| > |J(\jslER{X})| = \binom{|X|}{2}
      \qquad
      2^{|X|} = |\rG_t| = |M(\jslER{X})|.
    \]
    
    \item
    \emph{A concrete example}.
    
    Let $P_6$ be the path of edge-length $6$ with vertices $\{ 0, ..., 6 \}$. One of its bipartitions amounts to $\rG \subseteq \{ 1, 3, 5 \} \times \{ 0, 2, 4, 6 \}$ where $\rG(x, y) :\iff |x - y| = 1$. Applying $\Open$ yields:
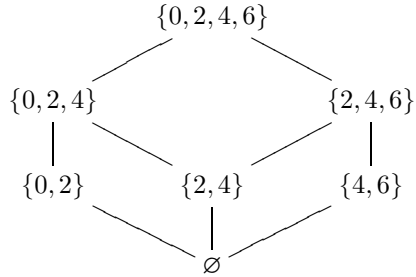
    \[
      \xymatrix@=5pt{
      && \{ 0, 2, 4, 6 \}
      \\
      \\
      \{ 0, 2, 4 \} \ar@{-}[uurr]
      &&
      && \{ 2, 4, 6 \} \ar@{-}[uull]
      \\
      \\
      \{ 0, 2 \} \ar@{-}[uu]
      && \{ 2, 4 \} \ar@{-}[uull] \ar@{-}[uurr]
      && \{ 4, 6 \} \ar@{-}[uu]
      \\
      \\
      && \emptyset \ar@{-}[uull] \ar@{-}[uurr] \ar@{-}[uu]
      }
    \]
    this being the smallest join-semilattice $\aS$ such that $|J(\aS)| < |M(\aS)|$.
    \endbox

  \end{enumerate}
\end{example}

\begin{example}[$\JSL_f$-morphisms as $\Dep$-morphisms]
  \label{ex:jsl_morphisms_in_dep}
  \item
  Given $f : \aS \to \aT$ we have the following $\Dep$-morphism of type $\nleq_\aS \;\to\; \nleq_\aT$,
  \[
    \xymatrix@=15pt{
      S \ar[rr]^{f_*\spbreve} && T
      \\
      S \ar[u]^{\nleq_\aS} \ar[rr]_f && T \ar[u]_{\nleq_\aT}
    }
  \]
  via the adjoint relationship $f(s) \nleq_\aT t  \iff s \nleq_\aS f_*(t)$ in contrapositive form.
  $\Pirr f$ arises by restricting the domain/codomain and passing to the maximum witnesses.
  Importantly, there is an equivalence functor $\Nleq : \JSL_f \to \Dep$ with $\Nleq \aS := \; \nleq_\aS$ and $\Nleq f (s, t) :\iff f(s) \nleq_\aT t$. In a precise sense, $\Pirr$ is the smallest restriction possible.
  \endbox
\end{example}

\smallskip

We now explicitly describe the equivalence between semilattices and graphs, including the relevant component relations. From one perspective we represent each $\aS$ as inclusion-ordered subsets of $M(\aS)$; from another we show each bipartitioned graph is $\Dep$-isomorphic to its \emph{reduction} -- a kind of union-free normal form.

\begin{theorem}[Categorical equivalence]
  \label{thm:dep_equiv_jslf}
  $\Open : \Dep \to \JSL_f$ and $\Pirr : \JSL_f \to \Dep$ define an equivalence of categories via natural isomorphisms:
  \[
  \begin{tabular}{llll}
  $rep : \Id_{\JSL_f} \To \Open \circ \Pirr$
  &&
  $rep_\aS$ & $:= \lambda s \in S.\{ m \in M(\aS) : s \nleq_\aS m \}$
  \\
  && $rep_\aS^{\bf-1}$ & $:= \lambda Y. \Land_\aS M(\aS) \setminus Y$
  \\[1ex]
  $red : \Id_{\Dep} \To \Pirr \circ \Open$
  &&
  $red_\rG$ & $:= \{ (g_s,Y) \in \rG_s \times M(\Open\rG) : \rG[g_s] \nsubseteq Y \}$ 
  \\
  &&
  $red_\rG^{\bf-1}$ & $:= \breve{\in} \; \subseteq J(\Open\rG) \times \rG_t$
  \end{tabular}
  \]
  where $\red_\rG$ and its inverse have associated component relations:
  \[
  \begin{tabular}{ll}
  $(red_\rG)_- := \{ (g_s,X) \in \rG_s \times J(\Open\rG) : X \subseteq \rG[g_s] \}$
  &
  $(red_\rG)_+ := \breve{\nin} \; \subseteq  M(\Open\rG) \times \rG_t$
  \\[2ex]
  $(red_\rG^{\bf-1})_- := \{ (X,g_s) \in J(\Open\rG)\times\rG_s : \rG[g_s] \subseteq X \}$
  &
  $(red_\rG^{\bf-1})_+ := \{ (g_t,Y) \in \rG_t \times M(\Open\rG) : \inte_\rG(\overline{g_t}) \subseteq Y  \}$
  \end{tabular}
  \]
\end{theorem}

\begin{proof}
  See background paper i.e.\ Theorem 4.2.10, \emph{$\Dep$ is equivalent to $\JSL_f$}.
\end{proof}

So $\rep_\aS$ \emph{represents} a join-semilattice as neighbourhoods of the relation $\Pirr\aS$.
Its inverse is relatively clear: every element arises uniquely as the meet of those meet-irreducibles above it. Concerning the other natural isomorphism, 
\begin{quote}
  \it
  $\red_\rG$ {\bf reduces} a bipartitioned graph $\rG$ by discarding vertices whose neighbourhood is a union of other vertices' neighbourhoods in a canonical manner.
\end{quote}

It is worth clarifying the above statement. Firstly,
\[
  J(\Pirr\Open\rG) \subseteq \{ \rG[g_s] : g_s \in \rG_s \}
  \qquad
  M(\Pirr\Open\rG) \subseteq \{ \inte_\rG(\overline{g_t}) : g_t \in \rG_t \}
\]
because the supersets join/meet-generate $\Open\rG$ respectively. The join-irreducible $\rG[g_s]$'s correspond to those $g_s$ whose neighbourhood is not a union of others. Less obviously the meet-irreducible $\inte_\rG(\overline{g_t})$'s correspond to those $g_t$ whose neighbourhood $\breve{\rG}[g_t]$ is not a union of others:
\[
\begin{tabular}{lll}
  $\inte_\rG(\overline{g_t}) = \inte_\rG(\overline{g_t^1}) \land_{\inte_\rG} \inte_\rG(\overline{g_t^2})$
  &
  $\iff \inte_\rG(\overline{g_t}) = \inte_\rG(\inte_\rG(\overline{g_t^1}) \cap \inte_\rG(\overline{g_t^2}) )$
  & (definition of $\land_{\inte_\rG}$)
  \\ &
  $\iff \rG^\down(\overline{g_t}) = \rG^\down(\inte_\rG(\overline{g_t^1}) \cap \inte_\rG(\overline{g_t^2}) )$
  & (by Lemma \ref{thm:graph_induces_isom})
  \\ &
  $\iff \rG^\down(\overline{g_t}) = \rG^\down(\inte_\rG(\overline{g_t^1})) \cap \rG^\down(\inte_\rG(\overline{g_t^2}))$
  & ($\rG^\down$ preserves $\cap$)
  \\ &
  $\iff \rG^\down(\overline{g_t}) = \rG^\down(\overline{g_t^1}) \cap \rG^\down(\overline{g_t^2})$
  & $(\down\up\down)$
  \\ &
  $\iff \rG^\down(\overline{g_t}) = \rG^\down(\overline{g_t^1} \cap \overline{g_t^2})$
  & ($\rG^\down$ preserves $\cap$)
  \\ &
  $\iff \neg_{\rR_s} \circ \rG^\down(\overline{g_t}) = \neg_{\rR_s} \circ \rG^\down(\overline{g_t^1} \cap \overline{g_t^2})$
  & ($\neg_{\rR_s}$ is bijective)
  \\ &
  $\iff \breve{\rG}[g_t] = \breve{\rG}[ g_t^1] \cup \breve{\rG}[g_t^2]$
  & $(\neg \down \neg)$.
\end{tabular}  
\]
Then finally we have:
\[
  \rG[g_s] \nsubseteq \inte_\rG(\overline{g_t})
  \iffLabel{(1)} g_s \nsubseteq \rG^\down \circ \rG^\up \circ \rG^\down(\overline{g_t})
  \iffLabel{(\down\up\down)} g_s \nsubseteq \rG^\down(\overline{g_t})
  \iffLabel{(3)} \rG[g_s] \nsubseteq \overline{g_t}
  \iff \rG(g_s, g_t)
\]
where $(1)$ and $(3)$ follow by the adjoint relationship in Lemma \ref{lem:mor_char_max_witness}.1. So reduction discards `degenerate' vertices and every relation is $\Dep$-isomorphic to its reduction. This is a form of union-freeness. Importantly:

\begin{proposition}[Reduction preserves bipartite dimension]
  $dim(\rG) = dim(\Pirr\Open\rG)$ for any $\rG \subseteq \rG_s \times \rG_t$.
\end{proposition}

\begin{example}[Reduction and bipartite dimension]
  \item
  \begin{enumerate}
    \item
    Isolated points have empty neighbourhoods and so are `discarded' by $\red_\rG$. The bipartite dimension is preserved because it is defined in terms of edges.
    \item
    If two points have the same neighbourhood, only one representative occurs in the reduction $\Pirr\Open\rG$. The square $C_4$ arises as $\rG \subseteq \{ 0, 2 \} \times \{ 1, 3 \}$ where $\rG[0] = \rG[2] = \rG_t$. Its reduction $\{ \emptyset \} \times \{ \rG_t \}$ is bipartite graph isomorphic to a single edge $P_1$. Concerning bipartite dimension, if two vertices have the same neighbourhood we may assume they reside in the same bicliques.
    \item
    A vertex's neighbourhood can be a non-degenerate union of others e.g.\ $\rG[2] = \rG[0] \cup \rG[4]$ below:
    
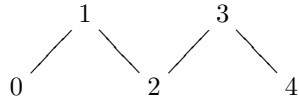
\[
      \xymatrix@=15pt{
        & 1 && 3
        \\
        0 \ar@{-}[ur] && 2 \ar@{-}[ul]  \ar@{-}[ur] && 4 \ar@{-}[ul]
      }
    \]
    Applying $\red_\rG$ we obtain two disjoint edges. This preserves the bipartite dimension because we can add $2$ to each biclique involving $0$ or $4$. This method extends to the general cases $\rG[x] = \rG[X]$ and $\breve{\rG}[y] = \breve{\rG}[Y]$.

    \item
    Suppose $\rG$ is a disjoint union of bicliques i.e.\ $\rG = \bigcup_{i \in I} X_i \times Y_i$ where $X_i \cap X_j = \emptyset = Y_i \cap Y_j$ whenever $i \neq j$. Then reduction is a special case of Example \ref{ex:dep_morphisms}.3 i.e.\ a preorder whose quotient poset is discrete.
    
    \item
    Example \ref{ex:semilattice_as_bipartite}.3 described a natural bipartitioned graph which was not reduced. In the automata-theoretic section we'll see many important examples. \endbox

  \end{enumerate}
\end{example}

\smallskip
We described the self-duality of $\JSL_f$ in Note \ref{note:jsl_extras}.3. In $\Dep$, this self-duality \emph{simply takes the converse relation on both objects and morphisms}. Furthermore, the associated component relations are \emph{simply swapped}.

\begin{theorem}[Self-duality]
  \label{def:bicliq_self_duality}
  We have the \emph{self-duality functor} $(-)\spcheck : \Dep^{op} \to \Dep$:
  \[
  \rG\spcheck := \breve{\rG}
  \qquad
  \dfrac{\rR : \rG \to \rH}{(\rR^{op})\spcheck := \breve{\rR} : \breve{\rH} \to \breve{\rG}}
  \qquad
  (\rG\spcheck)_- := \rG_+
  \qquad
  (\rG\spcheck)_+ := \rG_-
  \]
  with witnessing natural isomorphism $\alpha : \Id_{\BiCliq} \To (-)\spcheck \circ ((-)\spcheck)^{op}$ defined $\alpha_\rG := id_\rG = \rG$.
\end{theorem}

\begin{proof}
  See background paper i.e.\ Theorem 4.1.13, \emph{Self-duality of $\Dep$}.
\end{proof}

There is also an important natural isomorphism connecting the two self-dualities.

\begin{theorem}[Self-duality transfer]
  \label{thm:dep_jsl_self_duality_transfer}
  \item
  \begin{enumerate}
    \item 
    $\partial: (-)_* \circ \Open^{op} \To \Open \circ (-)\spcheck$ with $\partial_\rG := \lambda Y. \breve{\rG}[\overline{Y}]$ is a natural isomorphism with inverse $\partial_\rG^{\bf-1} := \lambda Y.\rG[\overline{Y}]$.

    In fact, if $\rR : \rG \to \rH$ is a $\Dep$-morphism then $(\Open\rR)_* = \partial_\rG^{\bf-1} \circ \Open \breve{\rR} \circ \partial_\rH$ with action $\lambda Y. \rH^\up \circ \rR^\down(Y)$.

    \item
    $\lambda : (-)\spcheck \circ \Pirr^{op} \To \Pirr \circ (-)_*$ where $\lambda_\aS := id_{\Pirr\aS} = \Pirr\aS$ is a self-inverse natural isomorphism.

  \end{enumerate}
\end{theorem}

\begin{proof}
  \item
  \begin{enumerate}
    \item 
    See background paper Theorem 4.6.7 i.e.\ `$\partial$ defines a natural isomorphism'.
    \item
    Given any $\JSL_f$-morphism $f : \aS \to \aT$ we need to establish the following square commutes:
    \[
      \xymatrix@=15pt{
        (\Pirr\aS)\spcheck \ar[rr]^-{id_\aS} && \Pirr(\aS^\pOp)
        \\
        (\Pirr\aT)\spcheck \ar[u]^{(\Pirr \, f^{op})\spbreve} \ar[rr]_-{id_\aT} && \Pirr(\aT^\pOp) \ar[u]_-{\Pirr(f_*)}
      }
    \]
    Indeed for any $s \in S$ and $t \in T$,
    \[
      \Pirr (f_*) (t, s)
      :\iff f_*(t) \nleq_{\aS^\pOp} s
      \iff s \nleq_\aS f_*(t)
      \iff f(s) \nleq_\aS t
      \iff \Pirr f (s,t)
      \iff (\Pirr f)\spbreve (t, s)
    \]
    via the usual adjoint relationship.

  \end{enumerate}
\end{proof}


Note that $\JSL_f$ has enough projectives, using category-theoretic parlance.

\begin{proposition}[$\JSL_f$ and $\Dep$ have enough projectives]
  \label{prop:jsl_dep_enough_proj}
  \item
  Let $Z$ be a finite set and $\aS$ a finite join-semilattice.
  \begin{enumerate}
    \item
    $\Open\Delta_Z = (\Pow Z, \cup, \emptyset)$ is the free join-semilattice on $|Z|$-generators.
    
    \item 
    $\epsilon_\aS : \Open\Delta_{J(\aS)} \epito \aS$ where $\epsilon_\aS (X) := \Lor_\aS X$ is surjective and extends $J(\aS) \hookto S$. Correspondingly, $\Dep$ has epimorphisms $\rG : \Delta_{\rG_s} \to \rG$.

    \item
    Given $f : \Open\Delta_Z \to \aT$ and surjective $q : \aS \epito \aT$ then $f = q \circ g$ where $g : \Open\Delta_Z \to \aS$ extends $\lambda z. q_*(f(\{ z \}))$.

  \end{enumerate}
\end{proposition}

Since the self-duality preserves freeness, $\JSL_f$ has enough injectives too. The witnessing embeddings $\iota_\aS := \iota \circ \rep_\aS : \aS \to \Open \Delta_{M(\aS)}$ first represent and then include into a powerset. In $\Dep$ they amount to monomorphisms $\rG : \rG \to \Delta_{\rG_t}$. Concerning both projectivity and injectivity, there is an important special case involving endomorphisms.

\begin{corollary}[Endomorphism representations]
  \label{cor:jsl_endomorphism_reps}
  \item
  The $\JSL_f$-diagrams below commute for any $\aS$-endomorphism $\delta$,
  \[
    \xymatrix@=15pt{
       \aS \ar[rr]^\delta
       && \aS
       \\
       \Open\Delta_{J(\aS)} \ar@{->>}[u]^{\epsilon_\aS} \ar@{..>}[rr]_{(\Pirr \delta)_-^\up}
       && \Open \Delta_{J(\aS)} \ar@{->>}[u]_{\epsilon_\aS}
    }
    \qquad
    \xymatrix@=15pt{
      \Open\Delta_{M(\aS)} \ar@{..>}[rr]^{((\Pirr \delta)_+\spbreve)^\up} && \Open\Delta_{M(\aS)}
        \\
        \aS \ar@{>->}[u]_{\iota_\aS} \ar[rr]_{\delta}
        && \aS \ar@{>->}[u]_{\iota_\aS}
    }
  \]
  Concerning $\Dep$, $\rR: \rG \to \rG$ induces both $\rR_- : \Delta_{\rG_s} \to \Delta_{\rG_s}$ and $\rR_+\spbreve : \Delta_{\rG_t} \to \Delta_{\rG_t}$.
\end{corollary}

\smallskip
Next, recall $\Pirr\aS$ restricts to the join/meet-irreducibles of $\aS$. It turns out one can instead pass to any join/meet generators. Roughly speaking, one can extend the domain/codomain of $\Pirr\aS$ and $\Pirr f$ in the `obvious' way.

\begin{proposition}[$\Dep$ generator isomorphisms]
  \label{prop:dep_generator_isoms}
  Let $f : \aS \to \aT$ be a join-semilattice morphism, $J_\aS,\, M_\aS \subseteq S$ be join/meet generators for $\aS$, and $J_\aT,\, M_\aT \subseteq T$ be join/meet generators for $\aT$.
  \begin{enumerate}
    \item
    We have the $\Dep$-isomorphism $\rI_\aS : \; \nleq_\aS |_{J_\aS \times M_\aS} \to \Pirr\aS$,
    \[
    \begin{tabular}{llllll}
    && $\rI_\aS := \; \nleq_\aS |_{J_\aS \times M(\aS)}$
    & $(\rI_\aS)_- (x, j) :\iff j \leq_\aS x$
    & $(\rI_\aS)_+ (m, y) :\iff m \leq_\aS y$
    \\[1ex]
    && $\rI_\aS^{\bf-1} := \; \nleq_\aS |_{J(\aS) \times M_\aS}$
    & $(\rI_\aS^{\bf-1})_-(j, x) :\iff x \leq_\aS j$
    & $(\rI_\aS^{\bf-1})_+ (y, m) :\iff y \leq_\aS m$.
    \\ 
    \end{tabular}
    \]
    
    \item
    The following $\Dep$-diagram commutes where $\rI_f (s, t) :\iff f(s) \nleq_\aS t$:
    \[
      \xymatrix@=15pt{
        \nleq_\aS |_{J_\aT \times M_\aT} \ar@{<-}[rrr]^-{\rI_\aT^{\bf-1}} &&& \Pirr\aT
        \\
        \nleq_\aS |_{J_\aS \times M_\aS} \ar[u]^{\rI_f} \ar[rrr]_-{\rI_\aS} &&& \Pirr\aS \ar[u]_{\Pirr f}
      }  
    \]
  \end{enumerate}
\end{proposition}

\begin{example}
  Applying Proposition \ref{prop:dep_generator_isoms} we see that Example \ref{ex:jsl_morphisms_in_dep} is essentially $\Pirr f$. \endbox
\end{example}

So far we've seen that $\Dep$ is well-behaved w.r.t.\ bipartite dimension. However, aside from that, connections with graph theory have been a bit thin on the ground. So before proceeding to the automata-theoretic constructions we mention some additional relationships.

\begin{example}[Further graph-theoretic connections]
  \label{ex:further_graph_theory}
  \item
  \begin{enumerate}
    \item
    \emph{Discarding vertices}.

    Let $\rG$ be a reduced relation.
    Discarding a vertex $g_s \in \rG_s$ in the lower bipartition amounts to generating a sub join-semilattice $\ang{J(\aS) \setminus \{ j \}}_\aS$.
    Discarding $g_t \in \rG_t$ amounts to constructing a quotient $(\ang{M(\aS) \setminus \{ m \}}_{\aS^\pOp})^\pOp$.

    \item
    \emph{Kronecker product over boolean semiring}.

    One can combine binary relations via $\rG \syncp \rH ((g_s, h_s), (g_t, h_t)) :\iff \rG(g_s, g_t) \land \rH(h_s, h_t)$ i.e.\ the Kronecker product over the boolean semiring \cite{WattsBooleanRankKronecker2001}.
    It defines a functor $- \syncp - : \Dep \times \Dep \to \Dep$ whose 
    corresponding join-semilattice functor is the \emph{tight tensor product} $\aS \ttenp \aT := \jslTight{\aS^{\pOp},\aT}$. To explain briefly,
    
    \begin{enumerate}
      \item
      $\jslTight{- , -} : \JSL_f^{op} \times \JSL_f \to \JSL_f$ restricts the usual hom-functor to morphisms which factor through a boolean lattice. The join-semilattice structure on morphisms is defined pointwise.
      \item
      There is a universal property w.r.t.\ bilinearity via a natural isomorphism $ut : \jslTight{- \ttenp -, -} \To \jslTight{-,\jslTight{-,-}}$.
      \item
      The tight tensor product is distinct from the tensor product \cite{GratzerTensorSemilattices2005}; they coincide on distributive lattices.
    \end{enumerate}

    \item
    \emph{Extension to non-bipartite graphs}.

    We've seen that reduced relations correspond to finite join-semilattices. This categorical equivalence can be extended to \emph{reduced undirected graphs} versus \emph{finite De Morgan algebras} i.e.\ bounded lattices with an order-reversing involution where distributivity is not assumed.

    \begin{enumerate}
      \item
      By undirected graph we mean a symmetric relation $\rE = \breve{\rE}$ i.e.\ a standard undirected graph where self-loops are now permitted.

      \item
      The algebras may be axiomatised by extending join-semilattices with a unary operation satisfying $\sigma(x \lor y) \leq \sigma(x)$ and $\sigma(\sigma(x)) = x$. A morphism is a join-semilattice morphism preserving $\sigma$.

      \item
      Given $(V, \rE)$ we construct the De Morgan algebra $\partial_\rE : \Open\rE \to (\Open\rE)^\pOp$ where $\partial_\rE(X) := \rE[\overline{X}]$. Given a De Morgan algebra $\sigma : \aS \to \aS^\pOp$ we construct the undirected graph $(J(\aS), \Pirr \sigma)$. \endbox

    \end{enumerate}

  \end{enumerate}
\end{example}

\section{Dependency Automata}

\subsection{From Nondeterministic to Dependency Automata}

\begin{definition}[Nondeterministic finite automaton]
  \label{def:nfa_basics}
  \item
  \begin{enumerate}
    \item 
    A \emph{nondeterministic finite automaton} (or \emph{nfa}) is a tuple $\rN = (I, Z, \rN_a, F)$ where:
    \begin{itemize}
      \item[--] $Z$ is a finite set,
      \item[--] $I, F \subseteq Z$ are subsets, and
      \item[--] $\rN_a \subseteq Z \times Z$ for each $a \in \Sigma$.
    \end{itemize}

    The elements of $Z$, $I$, $F$ are called \emph{states}, \emph{initial states} and \emph{final states} respectively. Each $\rN_a$ is called the \emph{$a$-transition relation}. We often reuse the symbol denoting the nfa (e.g. $\rN$) to denote the transitions (e.g. $\rN_a$). We may also denote the states, initial states and final states by $Z_\rN$, $I_\rN$ and $F_\rN$  respectively.

    \item
    For $w \in \Sigma^*$ inductively define $\rN_w$ as $\rN_\epsilon := \Delta_Z$, $\rN_{ua} := \rN_u ; \rN_a$. Then we say $\rN$ \emph{accepts the language}:
    \[
      L(\rN) := \{ w \in \Sigma^* : \rN_w [I] \, \cap F \neq \emptyset \}.
    \]
    
    \item
    \emph{Constructions on nondeterministic automata $\rN = (I, Z, \rN_a, F)$}.

    \begin{enumerate}[a.]
      \item
      Given $S \subseteq Z$ then $\rN_{@S} := (S, Z, \rN_a, F)$ is the nfa with its initial states changed to $S$. Notice that $L(\rN_{@S}) = \bigcup_{z \in S} L(\rN_{@z})$.

      \item 
        $\rN$'s \emph{reverse nfa} is:
        \[
          \rev{\rN} := (F, Z, (\rN_a)\spbreve, I)
          \qquad
          \text{and accepts the reverse language $(L(\rN))^r$.}
        \]
      \item
        There are various concepts relating to \emph{reachability}:
        \[
          \begin{tabular}{l}
            $\rs{\rN} := \{ \rN_w[I] : w \in \Sigma^* \}$
            \qquad
            $reach(\rN) := \bigcup \rs{\rN} \subseteq Z$
            \\[1ex]
            $\rsc{\rN} := (\{ I \}, \rs{\rN}, \lambda X. \rN_a[X], \{ X \in \rs{\rN} : X \cap F \neq \emptyset \})$
          \end{tabular}
        \]
        That is, $\rs{\rN}$ consists of $\rN$'s \emph{reachable subsets}, $reach(\rN)$ consists of $\rN$'s \emph{reachable states} and finally $\rsc{\rN}$ is the famous \emph{reachable subset construction}. The latter is a dfa -- see Definition \ref{def:dfas} below.

        \item 
        If $I \subseteq X \subseteq Z$ and $\rN_a[X] \subseteq X$ (for $a \in \Sigma$) the nfa $\rN \cap X := (I, X, \rN_a |_{X \times X}, F \cap X)$ accepts $L(\rN)$. Then:
        \[
          \begin{tabular}{c}
            $\reach{\rN} := reach(\rN) \; \cap \; \rN$
            \qquad is the \emph{reachable part of $\rN$}.
          \end{tabular}
        \]
        
      \item
        The \emph{coreachable part of $\rN$} also accepts $L(\rN)$:
        \[
          \coreach{\rN} := \rev{\reach{\rev{\rN}}}.
        \]

      \item
      An \emph{nfa isomorphism} $f :\rM \to \rN$ is a bijection $f : Z_\rM \to Z_\rN$ which preserves and reflects the initial states, the final states, and also the transitions. That is:
      \[
        z \in I_\rM \iff f(z) \in I_\rN
        \qquad
        \rM_a(z_1, z_2) :\iff \rN_a(f(z_1), f(z_2))
        \qquad
        z \in F_\rM \iff f(z) \in F_\rN
      \]
      for each $z, z_1, z_2 \in Z$ and $a \in \Sigma$. We may also write $\rM \cong \rN$.

      \item
      Each nfa has an associated \emph{join-semilattice of accepted languages} by varying the initial states:
      \[
        \jslLangs{\rN} := (\{ L(\rN_{@S}) : S \subseteq Z \}, \cup, \emptyset).
      \]
      Equivalently, $\jslLangs{\rN} := \jslLangs{\Det(dep(\rN))}$ is the join-semilattice of languages accepted by the full subset construction -- see Definition \ref{def:jsl_reach_simple}.2 and Note \ref{note:full_subset_construction}.

    \end{enumerate}
  
  \item
  We say $\rN$ is \emph{state-minimal} if there is no nfa accepting $L(\rN)$ with strictly fewer states.
  \endbox

  \end{enumerate}
\end{definition}

\begin{example}[Some small nfas]
  \label{ex:some_nfas}
  \item
  
  \begin{enumerate}
  \item
  $L = a(b + c) + b(a + c) + c(a + b)$ from \cite{NoteMinNfa1992} is a language with two state-minimal nfas.\footnote{Here, $\xto{b,c}$ indicates there is one $b$-labelled edge and another parallel $c$-labelled one. Initial states are indicated by {\tt i}, final states by {\tt o}.} 
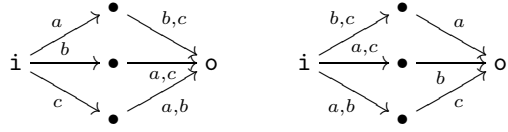
  \[
  \begin{tabular}{lll}
  $\xymatrix@=10pt{
  && \bullet \ar[drr]^{b,c}
  \\
  {\tt i} \ar[urr]^a \ar[rr]^b \ar[drr]_c
  && \bullet \ar[rr]_{a,c}
  && {\tt o}
  \\
  && \bullet \ar[urr]_{a,b}
  }$
  &&
  $\xymatrix@=10pt{
  && \bullet \ar[drr]^a
  \\
  {\tt i} \ar[urr]^{b,c} \ar[rr]^>>>>{a,c} \ar[drr]_{a,b}
  && \bullet \ar[rr]_<<<<b
  && {\tt o}
  \\
  && \bullet \ar[urr]_c
  }$
  \end{tabular}
  \]
  
  \item
  $L = {a + aa}$ from \cite{MinNfaBiRFSA2009} is an example of a language which is not `biresidual' \cite{SomeBiresBisep2010,MinNfaBiRFSA2009}. It has $5$ state-minimal nfas:
  \[
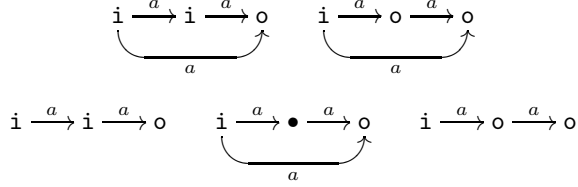
  \begin{tabular}{cc}
  $\xymatrix@=5pt{
  {\tt i} \ar[rr]^a \ar`d[r]`[rrrr]_a[rrrr] && {\tt i} \ar[rr]^a && {\tt o} 
  }$
  &
  $\xymatrix@=5pt{
  {\tt i} \ar[rr]^a \ar`d[r]`[rrrr]_a[rrrr] && {\tt o} \ar[rr]^a && {\tt o} 
  }$
  \end{tabular}
  \]
  \[
  \begin{tabular}{ccc}
  $\xymatrix@=5pt{
  {\tt i} \ar[rr]^a  && {\tt i} \ar[rr]^a && {\tt o} 
  }$
  &
  $\xymatrix@=5pt{
  {\tt i} \ar[rr]^a \ar`d[r]`[rrrr]_a[rrrr]  && \bullet \ar[rr]^a && {\tt o} 
  }$
  &
  $\xymatrix@=5pt{
  {\tt i} \ar[rr]^a  && {\tt o} \ar[rr]^a && {\tt o} 
  }$
  \end{tabular}
  \]
  
  \item
  $L = {(ab)^* + (abc)^*}$ has a unique state-minimal nfa shown below left.
  \[
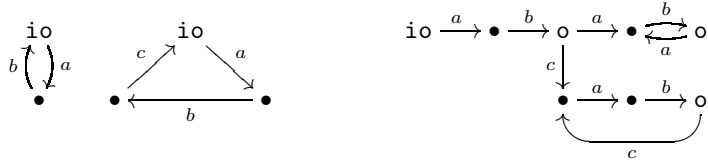
  \begin{tabular}{ccc}
  $\xymatrix@=15pt{
  {\tt io} \ar@/^5pt/[d]^a
  &&
  {\tt io} \ar[dr]^a
  \\
  \bullet \ar@/^5pt/[u]^b
  & \bullet \ar[ur]^c && \bullet \ar[ll]^b
  }$
  &
  \qquad
  \qquad
  &
  $\xymatrix@=15pt{
  {\tt io} \ar[r]^a & \bullet \ar[r]^b &
  {\tt o} \ar[r]^a \ar[d]_c & \bullet \ar@/^3pt/[r]^b & {\tt o} \ar@/^3pt/[l]^a
  \\
  &&
  \bullet \ar[r]^a & \bullet \ar[r]^b & {\tt o} \ar `d[l]`[ll]^c[ll]
  }$
  \end{tabular}
  \]
  Every regular language has a unique state-minimal \emph{partial deterministic} machine, shown for this $L$ above right.
  
  \item
  Consider the language $L_n = {(a + b)^* a (a + b)^n}$ for any $n \geq 0$. If $n = 3$ then the state-minimal nfa with the greatest (resp.\ least) number of transitions is shown below on the left (resp.\ right). 
  
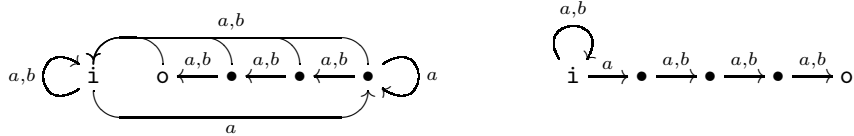
\[
  \xymatrix@=15pt{
  {\tt i} \ar@(dl,ul)^{a,b} \ar`d[r]`[rrrr]_-a[rrrr]
  & {\tt o} \ar@{-<}`u[l]`[l][l]
  & \bullet \ar[l]_{a,b} \ar@{-<}`u[l]`[ll][ll]
  & \bullet  \ar[l]_{a,b} \ar@{-<}`u[l]`[lll][lll]
  & \bullet  \ar[l]_{a,b} \ar@(ur,dr)^a  \ar@{-<}`u[l]`[llll]_-{a,b}[llll]
  }
  \qquad
  \qquad
  \xymatrix@=15pt{
  {\tt i} \ar@(ul,ur)^{a,b} \ar[r]^a & \bullet \ar[r]^{a,b} & \bullet \ar[r]^{a,b} & \bullet \ar[r]^{a,b} & {\tt o}
  }
  \]
  Each state-minimal nfa accepting $L_3$ arises by removing transitions from the left machine. One may remove any edge $\bullet \to \src$, and also the rightmost $a$-loop. There is a similar state-minimal machine for any $L_n$ with $n + 2$ nodes and $2 \cdot (n + 1) + 1$ optional transitions, so there are $2^{2n + 3}$ state-minimal nfas accepting $L_n$. On the other hand, the state-minimal partial deterministic automaton accepting $L_n$ has $2^{n + 1}$ nodes. \endbox
  
  \end{enumerate}
  \end{example}

We now recall deterministic finite automata and their associated canonical construction i.e.\ the state-minimal deterministic machine for a regular language.

\begin{definition}[Deterministic finite automaton]
  \label{def:dfas}
  \item
  \begin{enumerate}
    \item 
    A \emph{deterministic finite automaton} (or \emph{dfa}) is an nfa $(I, Z, \rN_a, F)$ where $|I| = 1$ and each $\rN_a$ is a function. We may write them as $\delta =  (i, Z, \delta_a, F)$ where $i \in Z$. For each $w \in \Sigma^*$ we inductively define the endofunction $\delta_w : Z \to Z$ as follows: $\delta_\epsilon := id_Z$ and $\delta_{ua} := \delta_a \circ \delta_u$ for each $(u,a) \in \Sigma^* \times \Sigma$.
    
    \item
    Given a dfa $\delta = (z_0, Z, \delta_a, F)$ accepting $L$ and $u \in \Sigma^*$,
    \[
      L(\delta_{@\delta_u(i)})
      = u^{-1} L := \{ w \in \Sigma^* : uw \in L \}.
    \]
    In other words, the unique $u$-successor of $z_0$ accepts $u^{-1} L$. The latter set is the \emph{left word quotient of $L$ by $u$} and is also known as the Brzozowski derivative \cite{BrzozowskiDRE1964}.

    \item
    Fix any regular $L \subseteq \Sigma^*$ and let $\LW{L} := \{ u^{-1} L : u \in \Sigma^* \}$ be $L$'s \emph{left word quotients}. Then:
    \[
      \minDfa{L} := (L, \LW{L}, \lambda X. a^{-1} X, \{ X \in \LW{L} : \epsilon \in X \})
    \]
    is the \emph{state-minimal dfa accepting $L$}. It is well-defined because $a^{-1} (u^{-1} L) = (ua)^{-1} L$.
    
    \item
    A dfa morphism $f : (x_0, X, \gamma_a, F_X) \to (y_0, Y, \delta_a, F_Y)$ is a function $f: X \to Y$ such that for $a \in \Sigma$:
    \[
      f \circ \delta_a = \gamma_a \circ f
      \qquad
      f(x_0) = y_0
      \qquad
      f^{-1} (F_Y) = F_X.
    \]
    The final condition asserts that the final states are both preserved and reflected, noting that $f^{-1} : \Pow Y \to \Pow X$ is the preimage function. Importantly, dfa morphisms always preserve the accepted language.

    \item
    Each dfa $\delta := (z_0, Z, \delta_a, F)$ has a \emph{dfa of accepted languages}:
    \[
      \simple{\delta} := (L, langs(\delta), \lambda X. a^{-1} X, \{ X \in langs(\delta) : \epsilon \in X \})
      \quad\text{where}\quad
      langs(\delta) := \{ L(\delta_{@z}) : z \in Z \}.
    \]
    There is a surjective dfa morphism $acc_\delta: \delta \epito \simple{\delta}$ defined $acc_\delta(z) := L(\delta_{@z})$ i.e.\ the \emph{acceptance map}. The word \emph{simple} is non-standard yet well-motivated: every surjective dfa morphism $f : \simple{\delta} \epito \gamma$ is bijective.

    \item
    An \emph{ordered dfa} $(p_0, \pP, \delta_a, F)$ consists of a partially ordered set $\pP = (P, \leq_\pP)$ and a dfa $(p_0, P, \delta_a, F)$ whose deterministic transitions are respectively monotonic $\delta_a : \pP \to \pP$. An \emph{ordered dfa morphism}  is a dfa morphism between ordered dfas which is also monotonic.
    \endbox

  \end{enumerate}
\end{definition}

\begin{note}[Concerning $\minDfa{L}$]
  State-minimal dfas are often introduced via Hopcroft's algorithm. One takes the reachable part of a given dfa, afterwards identifying states accepting the same language. The latter uses Hopcroft's partition refinement, essentially constructing the Myhill-Nerode congruence. There are two `representation independent' ways of defining it: (1) as equivalence classes of the Myhill-Nerode congruence $\rMN{L} \subseteq \Sigma^* \times \Sigma^*$ for $L$, (2) as the left word quotients $u^{-1} L$ also known as Brzozowski derivatives \cite{BrzozowskiDRE1964}. \endbox
\end{note}

\begin{note}[Concerning $\simple{\delta}$]
  The dfa $\simple{\delta}$ has no more states than $\delta$. Each state $z$ of the latter accepts $L(\delta_{@z})$ (by definition), as does the state $L(\delta_{@z})$ in $\simple{\delta}$. This construction is defined for dfas but not nfas. However, later we'll introduce a related construction $\simpleIrr{\rN}$ for each nfa $\rN$ -- see Definition \ref{def:nfa_irr_simplification}.
  \endbox
\end{note}

\smallskip

We now introduce dependency automata i.e.\ two nfas compatible w.r.t.\ a bipartitioned graph.

\begin{definition}[Dependency automaton]
  \label{def:dep_aut}
  A dependency automaton is a triple $(\rN, \rG, \rN')$ where:
  \begin{enumerate}
    \item
    $\rG \subseteq \rG_s \times \rG_t$ is a binary relation (bipartitioned graph).
    \item
    $\rN := (I_\rN, \rG_s, \rN_a, F_\rN)$ is an nfa over the lower bipartition.
    \item
    $\rN' := (I_{\rN'}, \rG_t, \rN'_a, F_{\rN'})$ is an nfa over the upper bipartition.
    \item
    $\rN_a ; \rG = \rG ; (\rN_a')\spbreve$ for each $a \in \Sigma$.
    \item
    $F_{\rN'} = \rG[I_\rN]$ and $F_\rN = \breve{\rG}[I_{\rN'}]$.
  \end{enumerate}

  Condition (4) induces $\Dep$-endomorphisms which we denote by $\rN_a^\dagger : \rG \to \rG$ for each $a \in \Sigma$.
  A dependency automaton $(\rN, \rG, \rN')$ \emph{accepts} the language $L(\rN, \rG, \rN') := L(\rN)$ i.e.\ the language accepted by the lower nfa.
  \endbox
\end{definition}

Each nfa induces a dependency automaton with only linear blowup.

\begin{definition}[Nfa's associated dependency automaton]
  \label{def:nfa_to_dep_aut}
  Given an nfa $\rN$ with states $Z$,
  \[
    dep(\rN) := (\rN, \Delta_Z, \rev{\rN})
  \]
  is its associated dependency automaton. \endbox
\end{definition}

For well-definedness consider Definition \ref{def:dep_aut} when $\rG = \Delta_Z$. Then (4) amounts to taking the converse relation and (5) to swapping the initial/final states. So each nfa can be viewed as a dependency automaton. Very importantly, each regular language has an associated dependency automaton too.

\begin{definition}[Canonical dependency automaton]
  \label{def:canon_dep_aut}
  Given a regular language $L \subseteq \Sigma^*$,
  \[
    dep(L) := (\minDfa{L}, \rDR{L} , \minDfa{L^r})
    \qquad \text{where} \qquad
    \rDR{L} (u^{-1} L, v^{-1} L^r) :\iff uv^r \in L
  \]
  is the respective \emph{canonical dependency automaton}.
  \endbox
\end{definition}

\begin{example}[$\dep{L}$]
  If $L = {\tt a + aa}$ so $L = L^r$ then $\minDfa{L} = \minDfa{L^r}$ is $\xymatrix@=10pt{ \src \ar[r]^a & \snk \ar[r]^a & \snk }$ excluding the sink. The canonical dependency automaton takes the form:
  
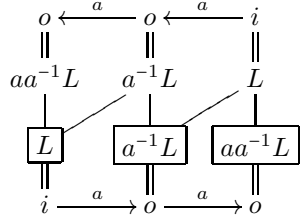
\[
  \xymatrix@=10pt{
  \snk \ar@{<-}[r]^a  & \snk \ar@{<-}[r]^a & \src
  \\
  aa^{-1} L  \ar@{=}[u] & a^{-1} L \ar@{=}[u] & L \ar@{=}[u]
  \\
  *+[F-]{L} \ar@{=}[d] \ar@{-}[u] \ar@{-}[ur] & *+[F-]{a^{-1}L} \ar@{=}[d] \ar@{-}[u] \ar@{-}[ur] & *+[F-]{aa^{-1}L} \ar@{=}[d] \ar@{-}[u]
  \\
  \src \ar[r]^a  & \snk \ar[r]^a & \snk
  }
  \]
  excluding the sink state from the top and bottom (which are isolated in $\rDR{L}$). \endbox
\end{example}

\begin{lemma}
  \label{lem:canon_dep_aut_well_defined}
  $dep(L)$ is a well-defined dependency automaton.
\end{lemma}

\begin{proof}
  Concerning (4),
  \[
    (\lambda X. a^{-1} X) ; \rDR{L} [u^{-1} L]
    = \rDR{L}[(ua)^{-1} L]
    = \{ v^{-1} L^r : v \in \Sigma^* , \; uav^r \in L \}
  \]
  \[
    \begin{tabular}{lll}
      $\rDR{L} ; (\lambda X. a^{-1} X)\spbreve [u^{-1} L]$
      &
      $= (\lambda X. a^{-1} X)\spbreve [\{ v^{-1} L^r : v \in \Sigma^* , \; uv^r \in L  \}]$
      \\ &
      $= \{ v^{-1} L^r : v \in \Sigma^* , \; u(va)^r \in L \}$
      \\ &
      $= \{ v^{-1} L^r : v \in \Sigma^* , \; uav^r \in L \}$.
    \end{tabular}
  \]
  Concerning (5),
  \[
    v^{-1} L^r \in \rDR{L}[L]
    \iff v^r \in L
    \iff v \in L^r
    \iff \epsilon \in v^{-1} L^r
  \]
  \[
    u^{-1} L \in (\rDR{L})\spbreve[L^r]
    \iff u \in L
    \iff \epsilon \in u^{-1} L.
  \]
\end{proof}

In both the classes of examples so far, the upper nfa accepts a word iff the lower nfa accepts its reverse. This situation holds generally for all dependency automata.

\begin{lemma}
  \label{lem:upper_nfa_accepts_reverse}
  If $(\rN, \rG, \rN')$ is a dependency automaton then $L(\rN') = (L(\rN))^r$.
\end{lemma}

\begin{proof}
  Since $\rN_a ; \rG = \rG ; (\rN'_a)\spbreve$ we have $\rN_a^\dagger : \rG \to \rG$ and 
  composing yields $\rN_w ; \rG = \rG ; (\rN'_w)\spbreve$ for $w \in \Sigma^*$. Then:
  \[
    \begin{tabular}{lll}
      $w \in L(\rN')$
      & $\iff \rN'_w [I_{\rN'}] \cap F_{\rN'} \neq \emptyset$
      \\ &
      $\iff \rN'_w [I_{\rN'}] \nsubseteq \overline{F_{\rN'}} = \overline{\rG[I_\rN]}$
      & (by def.)
      \\ &
      $\iff \rN'_w [I_{\rN'}] \nsubseteq \breve{\rG}^\down(I_\rN)$
      & $(\neg\up\neg)$
      \\ &
      $\iff \rN'_w ; \breve{\rG} [I_{\rN'}] \nsubseteq I_\rN$
      & (adjoints)
      \\ &
      $\iff \breve{\rG} ; \rN_w\spbreve [I_{\rN'}] \nsubseteq I_\rN$
      & (see above)
      \\ &
      $\iff \rN_w\spbreve [F_{\rN}] \nsubseteq I_\rN$
      & (by def.)
      \\ &
      $\iff w \in L(\rev{\rN})$
      \\ &
      $\iff w \in (L(\rN))^r$.
    \end{tabular}
  \]
\end{proof}

\begin{definition}[The category $\AutDep$]
  \label{def:cat_aut_dep}
  Its objects are the dependency automata. Take any two of them:
  \[
    \begin{tabular}{lll}
      $(\rM, \rF, \rM')$ & where $\rM = (I_\rM,\rF_s,\rM_a,F_\rM)$ &and $\rM' = (I'_\rM,\rF_t,\rM'_a,F'_\rM)$
      \\
      $(\rN, \rG, \rN')$ & where $\rN = (I_\rN,\rG_s,\rN_a,F_\rN)$ &and $\rM' = (I'_\rN,\rG_t,\rN'_a,F'_\rN)$.
    \end{tabular}
  \]
  An $\AutDep$-morphism $\rR : (\rM, \rF, \rM') \to (\rN, \rG, \rN')$ is a $\Dep$-morphism $\rR : \rF \to \rG$ such that for each $a \in \Sigma$,
  \[
    \rM_a ; \rR = \rR ; (\rN'_a)\spbreve
    \qquad
    \rR[I_\rM] = F_{\rN'}
    \qquad
    \breve{\rR}[I_{\rN'}] = F_\rM.
  \]
  Composition is inherited from $\Dep$.
  The leftmost condition can be written $\rM_a^\dagger \fatsemi \rR = \rR \fatsemi \rN_a^\dagger$.
  \endbox
\end{definition}

\begin{lemma}
  $\AutDep$ is a well-defined category.
\end{lemma}
\begin{proof}
  Identity morphisms $id_{(\rN, \rG, \rN')} := id_\rG = \rG$ are well-defined.
  Indeed, the conditions concerning dependency automata state precisely that $\rG$ is an $\AutDep$-endomorphism of $(\rN, \rG, \rN')$. It remains to verify that compatible $\AutDep$-morphisms are closed under $\Dep$-composition.
  To this end, take a dependency automaton $(\rO, \rH, \rO')$ where $\rO = (I_\rO, \rH_s, \rO_a, F_\rO)$ and
  $\rO' = (I'_\rH, \rH_t, \rO'_a, F'_\rH)$ and also a morphism $\rS : (\rN, \rG, \rN') \to (\rO, \rH, \rO')$.
  The $\AutDep$-morphisms inform us that $\rM_a^\dagger \fatsemi \rR = \rR \fatsemi \rN_a^\dagger$ and $\rN_a^\dagger \fatsemi \rS = \rS \fatsemi \rN_a^\dagger$,
  so that:
  \[
    \rM_a^\dagger \fatsemi (\rR \fatsemi \rS)
    = (\rM_a^\dagger \fatsemi \rR) \fatsemi \rS
    = (\rR \fatsemi \rN_a^\dagger) \fatsemi \rS
    = \rR \fatsemi (\rN_a^\dagger \fatsemi \rS)
    = \rR \fatsemi (\rS \fatsemi \rO_a^\dagger)
    = (\rR \fatsemi \rS) \fatsemi \rO_a^\dagger.
  \]
  Finally,
  \[
    \begin{tabular}{lll}
      $\rR \fatsemi \rS [I_\rM]$
      & $= \rR ; \rS_+\spbreve [I_\rM]
        = \rS_+\spbreve [\rR[I_\rM]]
        = \rS_+\spbreve [F'_\rN]$
      & (def.\ of $\Dep$ and $\AutDep$)
      \\ & $= \rS_+\spbreve [\rG[I_\rN]]
            = \rG; \rS_+\spbreve [I_\rN]
            = \rS [I_\rN] = F'_\rO$.
      & (def.\ of $\Dep$ and $\AutDep$)
      \\
      \\
      $(\rR \fatsemi \rS)\spbreve [I'_\rO]$
      & $= (\rR \fatsemi \rS)\spcheck [I'_\rO]$
      & (def.\ of $(-)\spcheck$)
      \\ & $= \rS\spcheck \fatsemi \rR\spcheck [I'_\rO]$
      & (functoriality)
      \\ & $= \breve{\rS} ; (\rR\spcheck)_+\spbreve [I'_\rO]$
      & ($\Dep$-composition)
      \\ & $= \breve{\rS} ; \rR_-\spbreve [I'_\rO]$
      \\ & $= \rR_-\spbreve [\breve{\rS}[I'_\rO]] = \rR_-\spbreve [F_\rN]$
      & ($\Dep$-composition)
      \\ &
      $= \rR_-\spbreve [\breve{\rG}[I'_\rN]]
      = (\rR_- ; \rG)\spbreve [I'_\rN]
      = \breve{\rR}[I_\rN]
      = F_\rM$
      & (def.\ of $\AutDep$ and $\Dep$).
    \end{tabular}
  \]
\end{proof}

\begin{theorem}[Self-duality of $\AutDep$]
  \label{thm:aut_dep_self_dual}
  We have the self-duality functor $\Rev : \AutDep^{op} \to \AutDep$:
  \[
    \Rev(\rN, \rG, \rN') := (\rN', \breve{\rG}, \rN)
    \qquad
    \Rev \rR := \rR\spcheck.
  \]
  recalling the self-duality of $\Dep$ from Theorem \ref{thm:aut_dep_self_dual}.
\end{theorem}

\begin{proof}
  Its action on objects is well-defined: (4) holds because we know $\rN_a ; \rG = \rG ; (\rN'_a)\spbreve$ and hence $\rN'_a ; \breve{\rG} = \breve{\rG} ; (\rN_a)\spbreve$; (5) holds because it is invariant under swapping the lower/upper nfa. Its action on morphisms is well-defined by a similar argument, recalling that $(\rR : \rF \to \rG)\spcheck := \breve{\rR} : \breve{\rG} \to \breve{\rF}$. Then it is a functor because $(-)\spcheck$ is. It is an equivalence functor for the same reason.
\end{proof}

Next we specify a way in which dependency automata can be isomorphic i.e.\ via distinct pairs of witnessing relations $(\rN_a, \rN'_a)$ of the same $\Dep$-endomorphism $\rN_a ; \rG =: \rN_a^\dagger := \rG ; (\rN'_a)\spbreve$. In other words, there can be \emph{too few transitions} relative to the inclusion-maximal components $(\rN_a^\dagger)_-$ and $(\rN_a^\dagger)_+$. There can also be too few initial/final states (these sets correspond to $\Dep$-morphisms too).

\begin{proposition}[$\AutDep$ transition-based isomorphisms]
  \label{prop:aut_dep_transition_isoms}
  Given $(\rN, \rG, \rN')$ and $(\rM, \rG, \rM')$ such that:
  \[
    \rN_a ; \rG = \rM_a ; \rG \quad (\text{for } a \in \Sigma)
    \qquad
    \rG [I_\rN] = \rG [I_\rM]
    \qquad
    \breve{\rG} [I_{\rN'}] = \breve{\rG} [I_{\rM'}]
  \]
  then $id_\rG = \rG : (\rN, \rG, \rN') \to (\rM, \rG, \rM')$ is an $\AutDep$-isomorphism.
\end{proposition}

\begin{proof}
  $\rG$ certainly defines a $\Dep$-morphism $id_\rG := \rG : \rG \to \rG$. It is an $\AutDep$-morphism $\rG : (\rN, \rG, \rN') \to (\rM, \rG, \rM')$ because:
  \[
    \begin{tabular}{lll}
      $\rM_a ; \rG$
      & $= \rN_a ; \rG$
      & (by assumption)
      \\
      & $=\rG ; \rN_a\spbreve$
      & ($(\rN, \rG, \rN')$ a dependency automaton)
    \end{tabular}
  \]
  and similarly $\rG[I_\rM] = \rG[I_\rM] = F_{\rM'}$, $\breve{\rG}[I_{\rN'}] = {\rG}[I_{\rM'}] = F_{\rM}$. By a symmetric argument we infer $\rG : (\rM, \rG, \rM') \to (\rN, \rG, \rN')$ is well-defined, and also the inverse because $id_\rG \fatsemi id_\rG = id_\rG$ in $\Dep$.
\end{proof}

\begin{proposition}[Polytime canonical dependency automaton]
  \label{prop:polytime_canon_dep_automaton}
  Given dfas $\alpha$, $\beta$ s.t.\ $L(\beta) = (L(\alpha))^r$ one can build $L(\alpha)$'s canonical dependency automaton in polytime.
\end{proposition}

\begin{proof}
  Minimising $\alpha$ in polytime yields $\gamma := (x_0,X,\gamma_a,F_\gamma)$, minimising $\beta$ yields $\delta := (y_0,Y,\delta_a,F_\delta)$. Construct $\rG \subseteq X \times Y$,
  \[
    \rG(\gamma_u(x_0), \delta_v(y_0))
      :\iff \gamma_{v^r}(\gamma_u(x_0)) \in F_\gamma
      \iff uv^r \in L
  \]
  noting that $\gamma$ and $\delta$ are reachable. Then we have the bipartite graph isomorphism:
  \[
    \xymatrix@=15pt{
      Y \ar[rr]^{\beta_2} && \LW{L^r}
      \\
      X \ar[u]^\rG \ar[rr]_{\beta_1} && \LW{L} \ar[u]_{\rDR{L}}
    }
  \]
  where $\beta_1(\gamma_u(x_0)) := u^{-1} L$ and $\beta_2(\delta_v(y_0)) := v^{-1} L^r$. It induces a $\Dep$-isomorphism -- in fact an $\AutDep$-isomorphism.
\end{proof}

\subsection{Deterministic automata over join-semilattices}

Just as each nfa has a reverse, each dependency automaton $(\rN, \rG, \rN')$ has a reverse $\Rev(\rN, \rG, \rN')$. It swaps the lower/upper nfa and takes the converse of $\rG$ (equivalently, swaps the bipartitions). This construction arose from the self-duality of $\Dep$. We now focus on lifting the categorical equivalence $\Dep \cong \JSL_f$ to one between dependency automata and \emph{deterministic finite automata interpreted in join-semilattices}. In the process we'll generalise the subset construction to dependency automata.

\begin{definition}[$\dfa{\JSL}$]
  \label{def:dfa_jsl}
  A $\JSL$-dfa is a 4-tuple $(s, \aS, \gamma_a, F)$ where $\aS = (S, \lor_\aS, \bot_\aS)$ is a finite join-semilattice, $s \in S$ is an element, $\delta_a : \aS \to \aS$ is an join-semilattice morphism for $a \in \Sigma$, and $F := \overline{\down_\aS t} \subseteq S$ for some $t \in S$. It \emph{accepts} the language its underlying dfa does. A $\JSL$-dfa morphism is a dfa morphism which is also a join-semilattice morphism i.e.\ preserves all joins. Given $w \in \Sigma^*$ we inductively define endomorphisms $\delta_\epsilon := id_\aS$ and $\delta_{wa} := \delta_a \circ \delta_w$. The category $\dfa{\JSL}$ consists of the $\JSL$-dfas with their morphisms, where composition is functional.
  \endbox
\end{definition}

Importantly, $\JSL$-dfas are deterministic finite automata interpreted in join-semilattices.\footnote{Other varieties where dfas can be interpreted include pointed sets, distributive lattices, boolean algebras and vector spaces over $\mathbb{F}_2$.}
\[
\begin{tabular}{r|r|r}
  & classical dfa & $\JSL$-dfa
  \\ \hline
  states & $Z$ & $\aS$
  \\
  initial state & $\alpha : \{ * \} \to Z$ & $\alpha : \two \to \aS$
  \\
  transitions & $\delta_a : Z \to Z$ &  $\delta_a : \aS \to \aS$
  \\
  final states & $\omega : Z \to \{ 0, 1 \}$ &  $\omega : \aS \to \two$
  \\
\end{tabular}  
\]
Indeed, viewing sets as algebras for the empty signature then $\{ * \}$ is free $1$-generated, just as $\two$ is the free $1$-generated join-semilattice. Morphisms from such algebras amount to picking a single element. On the other hand, $\{ 0, 1 \}$ and $\two$ are the unique (modulo isomorphism) two-element algebras of their respective varieties. Morphisms to such algebras amount to subsets i.e.\ the elements sent to $1$. Permitting every function $\omega : Z \to \{ 0, 1 \}$ permits any set of final states. Morphisms  $\omega : \aS \to \two$ must have a largest element sent to $0$, so that $\omega^{-1}(\{1\}) = \overline{\down_\aS t}$ for some $t \in S$.

\smallskip
The following Lemma provides further clarification. That is, a join of states accepts the union of the languages accepted by its summands. As a special case, the bottom element accepts the empty language.

\begin{lemma}
  \label{lem:jsl_dfa_joins}
  For any $\JSL$-dfa $\delta = (s_0, \aS, \delta_a, F)$ and $X \subseteq S$,
  \[
    L(\delta_{@\Lor_\aS X}) = \bigcup_{s \in X} L(\delta_{@s}).
  \]
\end{lemma}

\begin{proof}
  Let $t = \Lor_\aS \overline{F}$ be the largest non-final state. Each $\delta_w : \aS \to \aS$ is an endomorphism so $\delta_w (\Lor_\aS X) \nleq_\aS t \iff \Lor_\aS \{ \delta_w(x) : x \in X \} \nleq_\aS t \iff \exists x \in X. \delta_w(x) \nleq_\aS t$.
\end{proof}

Next, the category of $\JSL$-dfas is self-dual. That is, one can take adjoints and exchange the initial state with the largest non-final state.

\begin{theorem}[Self-duality of $\dfa{\JSL}$]
  \label{thm:dfa_jsl_self_dual}
  We have the self-duality $(-)^\pentagram : \dfa{\JSL}^{op} \to \dfa{\JSL}$,
  \[
    (s_0, \aS, \delta_a, F)^\pentagram
    := (\Lor_\aS \overline{F}, \aS^\pOp, (\delta_a)_*, \overline{\up_\aS s_0})
    \qquad
    f^\pentagram := f_*
  \]
  with witnessing natural isomorphism $\lambda : Id_{\dfa{\JSL}} \To (-)^\pentagram \circ ((-)^\pentagram)^{op}$ where $\lambda_{(s_0, \aS, \delta_a, F)} := id_\aS$.
\end{theorem}

Dual machines accept the reversed language.

\begin{lemma}
  \label{lem:dual_dfa_reverse_lang}
  $L(\delta^\pentagram) = (L(\delta))^r$ for any $\JSL$-dfa $\delta$.
\end{lemma}
\begin{proof}
  Let $\delta = (s_0, \aS, \delta_a, F)$ and consider the morphisms:
  \[
    \text{$\alpha : \two \to \aS$ where $\alpha(1) := s_0$}
    \qquad
    \text{$\omega : \aS \to \two$ where $\omega^{-1}(\{ 1 \}) = F$}
    \qquad
    \text{$\delta_w : \aS \to \aS$ for $w \in \Sigma^*$}.
  \]
  Then we calculate:
  \[
  \begin{tabular}{lll}
    $w \in L(\delta)$
    &
    $\iff \omega \circ \delta_w \circ \alpha = id_\two$
    & (consider action on $\top_\two$)
    \\ &
    $\iff \alpha_* \circ (\delta_w)_* \circ \omega_* = (id_\two)_* = id_{\two^\pOp}$
    & (apply $(-)_* : \JSL_f^{op} \to \JSL_f$)
    \\ &
    $\iff \alpha_* \circ (\delta^\pentagram)_{w^r} \circ \omega_* = id_{\two^\pOp}$
    & (by def.\ of $(-)^\pentagram$)
    \\ &
    $\iff \alpha_* \circ (\delta^\pentagram)_{w^r} \circ \omega_*(0) = 0$.
    & (since $\top_{\two^\pOp} = 0$)
    \\ &
    $\iff \alpha_* \circ (\delta^\pentagram)_{w^r} (\Lor_\aS \overline{F}) \nleq_{\two^{\pOp}} 1$
    \\ &
    $\iff (\delta^\pentagram)_{w^r} (\Lor_\aS \overline{F}) \nleq_{\aS^{\pOp}} (\alpha_*)_*(1)$
    & (adjoints)
    \\ &
    $\iff (\delta^\pentagram)_{w^r} (\Lor_\aS \overline{F}) \nleq_{\aS} s_0$
    & (since $(\alpha_*)_* = \alpha$)
    \\ &
    $\iff w^r \in L(\delta^\pentagram)$.
  \end{tabular}  
  \]
\end{proof}

Importantly, dependency automata and dfas interpreted in semilattices are two sides of the same coin.

\begin{definition}[Equivalence functors for automata]
  \label{def:airr_and_det}
  \item
  \begin{enumerate}
    \item
    $\Det : \AutDep \to \dfa{\JSL}$ \emph{determinises} a dependency automaton:
    \[
      \Det(\rN, \rG, \rN') :=
      (F_{\rN'}, \Open\rG, \lambda Y. (\rN'_a)\spbreve[Y §], \{ Y \in O(\rG) : Y \cap I_{\rN'} \neq \emptyset  \})
    \]
    and acts on morphisms as $\Open$ does (see Definition \ref{def:open_and_pirr}.1).

    \item
    $\Airr: \dfa{\JSL} \to \AutDep$ constructs a dependency automaton over the semilattice's irreducibles:
    \[
      \begin{tabular}{rl}
      $\Airr(s, \aS, \delta_a, F)$ & $:= (\rN, \Pirr\aS, \rN')$
      \\
      $\rN$ & $:= ( J(\aS) \,\cap \down_\aS s , J(\aS), (\Pirr\delta_a)_-, J(\aS) \,\cap F)$
      \\
      $\rN'$ & $:= ( M(\aS) \,\cap \up_\aS \Lor_\aS \overline{F}, M(\aS), (\Pirr\delta_a)_+, M(\aS) \cap\, \overline{\up_\aS s})$.
      \end{tabular}
    \]
    It acts on morphisms as $\Pirr$ does (see Definition \ref{def:open_and_pirr}.2). \endbox
  \end{enumerate}
\end{definition}

\begin{theorem}[Automata-theoretic categorical equivalence]
  \label{thm:aut_dep_equiv_jsl_dfa}
  $\Det : \AutDep \to \dfa{\JSL}$ and $\Airr: \dfa{\JSL} \to \AutDep$ define an equivalence of categories with natural isomorphisms inherited from Theorem \ref{thm:dep_equiv_jslf}:
  \[
    \begin{tabular}{ll}
      $\rep  : Id_{\dfa{\JSL}} \To \Det \circ \Airr$
      &
      $\rep_{(s, \aS, \delta_a, F)} := \rep_\aS$
      \\[1ex]
      $\red  : Id_{\AutDep} \To \Airr \circ \Det$
      &
      $\red_{(\rN, \rG, \rN')} := \red_\rG$
    \end{tabular}
  \]
\end{theorem}

\begin{proof}
  \item
  \begin{enumerate}
    \item \emph{$\Det$ is well-defined}.

    Since $\Open$ is a well-defined functor we need only show $\Det$ is well-defined on objects and morphisms. Concerning objects, first recall:
    \[
      \Det(\rN, \rG, \rN')
      := (F_{\rN'}, \Open\rG, \lambda Y. (\rN'_a)\spbreve[Y], \{ Y \in O(\rG) : Y \cap I_{\rN'} \neq \emptyset  \}).
    \]
    Then $F_{\rN'} = \rG[I_\rN] \in O(\rG)$ is an element of the join-semilattice $\Open\rG$, as required. Since $\rN_a ; \rG = \rN_a^\dagger = \rG ; (\rN'_a)\spbreve$ we have  $\rN_a^\dagger : \rG \to \rG$; applying $\Open$ yields an endomorphism of $\Open\rG$ with action $\lambda Y \in O(\rG). (\rN_a^\dagger)_+\spbreve[Y]$. The latter can be rewritten $\lambda Y. (\rN'_a)\spbreve[Y]$ because $\rG ; (\rN_a^\dagger)_+\spbreve = \rG ; (\rN'_a)\spbreve$ and each $Y = \rG[X]$ for some $X$.
    Finally, the non-final states $\{ Y \in O(\rG) : Y \cap I_{\rN'} = \emptyset  \}$ have a largest element $\inte_\rG(\overline{I_\rN})$. Then the $\JSL$-dfa is well-defined.

    To see $\Det$ is well-defined on morphisms we'll show the respective $\JSL_f$-morphisms preserve the additional structure.
    Given $\rR: (\rM, \rG, \rM') \to (\rN, \rH, \rN')$ we have $\Open \rR : \Open\rG \to \Open\rH$. The initial state is preserved:
    \[
      \Open\rG (F_{\rM'})
      = \Open\rG (\rG[I_\rM])
      = \rG ; \rR_+\spbreve [I_\rM]
      = \rR [I_\rM]
      = F_{\rN'}.
    \]
    The transitions are preserved because for any $Y = \rG[X]$,
    \[
      \begin{tabular}{lll}
        $\Open \rR((\rM'_a)\spbreve [Y])$
        & $= \rR_+\spbreve [(\rM'_a)\spbreve[Y]]$
        & (def.\ of $\Open$)
        \\ &
        $= \rG ; (\rM'_a)\spbreve ; \rR_+\spbreve [X]$
        \\ &
        $= \rM_a ; \rG ; \rR_+\spbreve [X]$
        & (def.\ of $\AutDep$)
        \\ &
        $= \rM_a ; \rR [X]$
        \\ &
        $= \rR ; (\rN'_a)\spbreve [X]$
        & (def.\ of $\AutDep$)
        \\ &
        $= \rG ; \rR_+\spbreve ; (\rN'_a)\spbreve [X]$
        \\ &
        $= \rR_+\spbreve ; (\rN'_a)\spbreve [Y]$
        \\ &
        $= (\rN'_a)\spbreve [\Open \rR (Y)]$.
        & (def.\ of $\Open$)
      \end{tabular}
    \]
    To see the final states are preserved, observe $I_{\rM'}$ determines the $\Dep$-morphism:
    \[
      \xymatrix@=15pt{
        \rG_t \ar[rr]^{I_{\rM'} \times \{ 0 \}} && 0
        \\
        \rG_s \ar[u]^\rG \ar[rr]_{F_\rM \times \{ 1 \}} && 1 \ar[u]_{\Pirr\two}
      }
    \]
    Fix $Y = \rG[X] \in O(\rG)$. Then $Y$ is final iff $Y \cap I_{\rM'} \neq \emptyset$ iff $X \cap F_\rM \neq \emptyset$. Moreover $\Open\rR (Y) = \rR[X]$ is final iff $\rR[X] \cap I_{\rN'} \neq \emptyset$. By assumption $\breve{\rR}[I_{\rN'}] = F_\rM$, so $Y$ is final iff $X \cap \breve{\rR}[I_{\rN'}] \neq \emptyset$ iff $\rR[X] \cap I_{\rN'} \neq \emptyset$ iff $\Open\rR(Y)$ is final.

    \item \emph{$\Airr$ is well-defined}.
    
    Since $\Pirr$ is a well-defined functor it suffices to show $\Airr$ is well-defined on objects and morphisms.
    Concerning objects, $\Airr (s, \aS, \delta_a, F) := (\rN, \Pirr\aS, \rN')$ and both $\rN$ and $\rN'$ are well-defined nfas. Condition (4) holds:
    \[
        (\Pirr\delta_a)_- ; \Pirr\aS
        = \Pirr\delta_a \fatsemi id_{\Pirr\aS}
        = \Pirr\delta_a
        = id_{\Pirr\aS} \fatsemi \Pirr\delta_a
        = \Pirr\aS ; (\Pirr\delta_a)_+\spbreve.
    \]
    Finally, condition (5) holds:
    \[
      \begin{tabular}{rl}
        $\Pirr\aS [I_\rN]$
        &
        $= \Pirr\aS [\{ j \in J(\aS) : j \leq_\aS s \}]$
        \\ &
        $= \{ m \in M(\aS) : \exists j \in J(\aS).[ j \leq_\aS s \text{ and } j \nleq_\aS m ] \}$
        \\ &
        $= \{ m \in M(\aS) : s \nleq_\aS m \} = F_{\rN'}$
        \\
        \\
        $(\Pirr\aS)\spbreve[I_{\rN'}]$
        &
        $= (\Pirr\aS)\spbreve[\{ m \in M(\aS) : \Lor_\aS \overline{F} \leq_\aS m \}]$
        \\ &
        $= \{ j \in J(\aS) : \exists m \in M(\aS). [ j \nleq_\aS m \text{ and } \Lor_\aS \overline{F_\aS} \leq_\aS m ] \}$
        \\ &
        $= \{ j \in J(\aS) : j \nleq_\aS \Lor_\aS \overline{F_\aS} \} = F_\rN$.
      \end{tabular}
    \]

    To see $\Airr$ is well-defined on morphisms, take $f : (s_0, \aS, \gamma_a, F_\aS) \to (t_0, \aT, \delta_a, F_\aT)$ so we have $\Pirr f : \Pirr\aS \to \Pirr\aT$. Then let us verify the required identities:
    \[
      \begin{tabular}{lll}
        $(\Pirr \gamma_a)_- ; \Pirr f$
        & $= \Pirr \gamma_a \fatsemi \Pirr f$
        \\ & $= \Pirr (\gamma_a \circ f)$
        \\ & $= \Pirr (f \circ \delta_a)$
        \\ & $= \Pirr f \fatsemi \Pirr\delta_a$
        \\ & $= \Pirr f ; (\Pirr \delta_a)_+\spbreve$
      \end{tabular}
    \]
    Moreover if $\Airr f : (\rM, \Pirr\aS, \rM') \to (\rN, \Pirr\aT, \rN')$ then,
    \[
      \begin{tabular}{ll}
        $\Pirr f[I_\rM]$
        &
        $= \Pirr f[\{ j \in J(\aS) : j \leq_\aS s_0 \}]$
        \\ &
        $= \{ m \in M(\aT) : \exists j \in J(\aS). ( f(j) \nleq_\aT m \text{ and } j \leq_\aS s_0 ) \} $
        \\ &
        $= \{ m \in M(\aT) : \exists j \in J(\aS). ( j \nleq_\aS f_*(m) \text{ and } j \leq_\aS s_0 ) \} $
        \\ &
        $= \{ m \in M(\aT) : s_0 \nleq_\aS f_*(m) \}$.
        \\ &
        $= \{ m \in M(\aT) : f(s_0) \nleq_\aT m \}$.
        \\ &
        $= \{ m \in M(\aT) : t_0 \nleq_\aT m \} = F_{\rN'}$.
      \end{tabular}
    \]
    \[
      \begin{tabular}{ll}
        $(\Pirr f)\spbreve [I_{\rN'}]$
        &
        $= \{ j \in J(\aS) : \exists m \in M(\aT). ( f(j) \nleq_\aT m \text{ and } \Lor_\aS \overline{F_\aT} \leq_\aT m ) \}$
        \\ &
        $= \{ j \in J(\aS) : f(j) \nleq_\aS \Lor_\aT \overline{F_\aT} \}$
        \\ &
        $= \{ j \in J(\aS) : j \nleq_\aS \Lor_\aS \overline{F_\aS} \} = F_\rN$.
      \end{tabular}
    \]

    \item \emph{$\rep$ restricts to a natural isomorphism as claimed}.
    
      Recall the natural isomorphism $\rep : Id_{\JSL_f} \To \Open \circ \Pirr$ where $\rep_\aS := \lambda s \in S. \{ m \in M(\aS) : s \nleq_\aS m \}$. Then given any $\JSL$-dfa $\delta := (s_0, \aS, \delta_a, F)$ it suffices to establish the typing $\rep_\aS : \delta \to \Det\Airr \delta$ where:
      \[
        \Det\Airr \delta
        = (\{ m \in M(\aS) : s_0 \nleq_\aS m \}, \Open\Pirr\aS, \lambda Y. (\Pirr\delta_a)_+\spbreve[Y], \{ Y \in O(\Pirr\aS) : \Lor_\aS \overline{F} \in \;\down_\aS Y \}).
      \]
      The initial state is clearly preserved. Next, the deterministic transitions are preserved:
      \[
        \rep_\aS \circ \delta_a (s)
        = \{ m \in M(\aS) : \delta_a (s) \nleq_\aS m \}
        = \{ m \in M(\aS) : s \nleq_\aS (\delta_a)_*(m) \}
      \]
      \[
        \begin{tabular}{lll}
          $(\Pirr\delta_a)_+\spbreve[\rep_\aS(s)]$
          &
          $= (\Pirr\delta_a)_+\spbreve[\{ m \in M(\aS) : s \nleq_\aS m \}]$
          \\ &
          $= \{ m \in M(\aS) : \exists m' \in M(\aS).( s \nleq_\aS m' \text{ and } (\delta_a)_* \leq_\aS m' ) \}$
          \\ &
          $= \{ m \in M(\aS) : s \nleq_\aS (\delta_a)_*(m) \}$.
        \end{tabular}
      \]
      Concerning the final states we know $\rep_\aS [F] = \{ M(\aS) \,\cap \nleq_\aS [s] :  s \in F \}$,
      so given $s \in F$ let $Y_s := \{ m \in M(\aS) : s \nleq_\aS m \} \in O(\Pirr\aS)$. Then:
      \[
        \begin{tabular}{lll}
          $\Lor_\aS \overline{F} \nin \; \down_\aS Y_s$
          &
          $\iff \forall m \in M(\aS). ( s \nleq_\aS m \To  \Lor_\aS \overline{F} \nleq_\aS m)$
          \\ &
          $\iff \forall m \in M(\aS). ( \Lor_\aS \overline{F} \leq_\aS m \To s \leq_\aS m )$
          \\ &
          $\iff s \leq_\aS \Lor_\aS \overline{F}$
          \\ &
          $\iff s \nin F$.
        \end{tabular}
      \]

    \item \emph{$\red$ restricts to a natural isomorphism as claimed}.

    Recall the natural isomorphism $\red : Id_{\Dep} \To \Pirr \circ \Open$. Take any dependency automaton $\fM := (\rM, \rG, \rM')$. It suffices to establish the typing $\red_\rG : \fM \to \Airr\Det\fM$ whose codomain is $(\rN, \Pirr\Open\rG, \rN')$ where:
    \[
      \begin{tabular}{c}
        $\rN := (\{ j \in J(\Open\rG) : j \subseteq F_{\rM'} \}, J(\Open\rG), \rN_a, \{ Y \in J(\Open\rG) : Y \cap I_{\rM'} \neq \emptyset \})$
        \\[1ex]
        $\rN' := (\{ m \in M(\Open\rG) : \inte_{\rG}(\overline{I_{\rM'}}) \subseteq m \}, M(\aS), \rN'_a, \{ m \in M(\aS) : F_{\rM'} \nsubseteq m \})$.
      \end{tabular}
    \]
    Firstly by $\Dep$-composition and naturality,
    \[
      \rM_a ; \red_\rG
      = \rM_a^\dagger \fatsemi \red_\rG
      = \red_\rG \fatsemi \Pirr\Open \rM_a^\dagger
      = \red_\rG ; (\Pirr\Open\rM_a^\dagger)_+\spbreve
      = \red_\rG ; (\rN'_a)\spbreve.
    \]
    Finally we establish the two remaining conditions:
    \[
      \begin{tabular}{lll}
        $\red_\rG [I_\rM]$
        &
        $= \{ Y \in M(\Open\rG) : \exists g_s \in I_\rM. \rG[g_s] \nsubseteq Y \}$
        \\ &
        $= \{ Y \in M(\Open\rG) : \rG[I_\rM] \nsubseteq Y \}$
        \\ &
        $= \{ Y \in M(\Open\rG) : F_{\rM'} \nsubseteq Y \}$
        & (def. of $\AutDep$)
        \\ &
        $= F_{\rN'}$
        & (see above).
        \\
        \\
        $(\red_\rG)\spbreve[I_{\rN'}]$
        &
        $= \{ g_s \in \rG_s : \exists Y \in M(\Open\rG). (\rG[g_s] \nsubseteq Y \text{ and } \inte_\rG(\overline{I_{\rM'}}) \subseteq Y )\}$
        \\ &
        $= \{ g_s \in \rG_s : \rG[g_s] \nsubseteq \inte_\rG(\overline{I_{\rM'}}) \}$
        \\ &
        $= \{ g_s \in \rG_s : \{ g_s \} \nsubseteq \rG^\down \circ \inte_\rG(\overline{I_{\rM'}}) \}$
        & (adjoints)
        \\ &
        $= \{ g_s \in \rG_s : g_s \nin \rG^\down (\overline{I_{\rM'}}) \}$
        & ($\down\up\down$)
        \\ &
        $= \breve{\rG}[I_{\rM'}]$
        \\ &
        $= F_{\rM}.$
      \end{tabular}
    \]
  \end{enumerate}
\end{proof}

Recall that each nfa $\rN$ induces a dependency automaton $dep(\rN) = \Det(\rN, \Delta_Z, \rev{\rN})$.

\begin{note}[$\Det(\rN, \Delta_z, \rev{\rN})$ is $\rN$'s full subset construction]
  \label{note:full_subset_construction}
  Given any nfa $\rN = (z_0, Z, \rN_a, F)$,
  \[
    \begin{tabular}{ll}
      $\Det(dep(\rN))$
      &
      $= \Det(\rN, \Delta_Z, \rev{\rN})$
      \\ &
      $= (F_{\rev{\rN}}, \Open \Delta_Z, \lambda X. (\rN_a\spbreve)\spbreve[X], \{ X \subseteq Z : X \cap I_{\rev{\rN}} \neq \emptyset \})$
      \\[1ex] &
      $= (I_\rN, (\Pow Z, \cup, \emptyset), \lambda X. \rN_a[X], \{ X \subseteq Z : X \cap F_{\rN} \neq \emptyset \} )$
    \end{tabular}
  \]
  i.e.\ the full subset construction for $\rN$ endowed with inclusion ordering. This explains Definition \ref{def:full_subset_construction} below. \endbox
\end{note}

\begin{definition}[Full subset construction $\jslDfaSc{\rN}$]
  \label{def:full_subset_construction}
  For any nfa $\rN$ define:
  \[
    \jslDfaSc{\rN}
    := \Det(\dep{\rN})
    = (I_\rN, (\Pow Z, \cup, \emptyset), \lambda X. \rN_a[X], \{ X \subseteq Z : X \cap F_{\rN} \neq \emptyset \} ).
  \]
  This is $\rN$'s full subset construction endowed with its $\JSL$-dfa structure.
  \endbox
\end{definition}

\begin{note}[$\Det(\rN, \rG, \rN')$ restricts $\rev{\rN'}$'s full subset construction]
  Generally speaking, $\Det(\rN, \rG, \rN')$ is obtained from $\rev{\rN'}$'s full subset construction by restricting to $\Open\rG \subseteq \Open \Delta_{\rG_t}$. This generalises Note \ref{note:full_subset_construction}.
  \endbox
\end{note}

\begin{note}[$\Det$ and $\Airr$ preserve the accepted language]
  \label{note:det_airr_preserve_acceptance}
  Given any dependency automaton $(\rN, \rG, \rN')$, the classically reachable part of $\Det(\rN, \rG, \rN')$ has a classical description too:
  \[
    \reach{\Det(\rN, \rG, \rN')} = \rsc{\rev{\rN'}}.
  \]
  It follows by Lemma \ref{lem:upper_nfa_accepts_reverse} that $\Det$ preserves the accepted language. The natural isomorphism $\red : Id_{\AutDep} \To \Airr \circ \Det$ informs us that $\Airr$ also preserves the accepted language. \endbox
\end{note}

\begin{corollary}[Language correpondence]
  \label{cor:dfa_nfa_lang_correspondence}
  Let $\delta = (s_0, \aS, \delta_a, F)$ be a $\JSL$-dfa and $\rN$ the lower nfa of $\Airr\delta$. Then:
  \[
    L(\rN_{@Z}) = acc_\delta(\Lor_\aS Z)
    \qquad
    \text{for every $Z \subseteq J(S)$.}
  \]
  In particular each individual state $j \in J(\aS)$ of $\rN$ accepts $acc_\delta(j)$.
\end{corollary}

\begin{proof}
  By Note \ref{note:det_airr_preserve_acceptance} $\Det$ preserves the accepted language. Given $\Airr\delta = (\rN, \rG, \rN')$ and $Z \subseteq J(\aS)$ then:
  \[
    (\rN_{@Z}, \rG, \rM)
    \qquad
    \text{where $\rM := \rev{\rev{\rN'}_{@\rG[Z]}}$}
  \] 
  is a well-defined dependency automaton. Applying $\Det$ yields the $\JSL$-dfa $(\Det\Airr\delta)_{@\rG[Z]}$ which accepts $L(\rN_{@Z})$. Finally $\rep_{\delta'}^{\bf-1}$ provides a language-preserving isomorphism $(\Det\Airr\delta)_{@\rG[Z]} \to \delta_{@\Lor_\aS Z}$.
\end{proof}

\begin{example}[Dualising the full subset construction]
  \label{ex:dual_full_subset_construction}
  Via relative complement we have the $\JSL$-dfa isomorphism:
  \[
    (\jslDfaSc{\rN})^\pentagram \cong \jslDfaSc{\rev{\rN}}
  \]
  which follows by considering
  $
    (\jslDfaSc{\rN})^\pentagram = (\overline{F}, \Open\Delta_Z, \rR_a^\down, \{ X \subseteq Z : X \cap \overline{I} \neq \emptyset \})
  $.
  In other words, \emph{the dual of the full subset construction for $\rN$ is the full subset construction for $\rev{\rN}$}. This isomorphism instantiates the natural isomorphism $\hat{\partial}$ described below.
  \endbox
\end{example}

The self-duality transfers of Theorem \ref{thm:dep_jsl_self_duality_transfer} generalise naturally to the automata-theoretic setting.

\begin{theorem}[Automata-theoretic self-duality transfer]
  \label{thm:aut_dep_jsl_dfa_self_duality_transfer}
  \item
  \begin{enumerate}
    \item 
      $\hat{\partial} : (-)^\pentagram \circ \Det^{op} \To \Det \circ \Rev$ restricts $\partial$ from Theorem \ref{thm:dep_jsl_self_duality_transfer}.1.
    \item
      $\hat{\lambda} : \Rev \circ \Airr^{op} \To \Airr \circ (-)^\pentagram$ restricts $\lambda$ from Theorem \ref{thm:dep_jsl_self_duality_transfer}.2.
  \end{enumerate}
\end{theorem}

\begin{proof}
  \item
  \begin{enumerate}
    \item
    Given a dependency automaton $(\rN, \rG, \rN')$ it suffices to show $\partial_\rG : (\Open\rG)^\pOp \to \Open \breve{\rG}$ defines an $\dfa{\JSL}$-morphism of type $(\Det(\rN, \rG, \rN'))^\pentagram \to \Det\Rev(\rN, \rG, \rN')$.

    Concerning preservation of the initial state,
    \[
      \begin{tabular}{lll}
        $\partial_\rG (\bigcup \overline{\{ Y \in O(\rG) : Y \cap I_{\rN'} \neq \emptyset \}})$
        & $= \partial_\rG (\bigcup \{ Y \in O(\rG) : Y \cap I_{\rN'} = \emptyset \})$
        \\ & $= \partial_\rG (\inte_\rG(\overline{I_{\rN'}}))$
        & (see below)
        \\ & $= \breve{\rG}[\overline{\inte_\rG(\overline{I_{\rN'}})}]$
        & (def.\ of $\partial_\rG$)
        \\ & $= \overline{\rG^\down(\inte_\rG(\overline{I_{\rN'}}))}$
        & (De Morgan duality)
        \\ & $= \overline{\rG^\down(\overline{I_{\rN'}})}$
        & ($\down\up\down$)
        \\ & $= \breve{\rG}[I_{\rN'}]$
        & (De Morgan duality)
        \\ & $= F_\rN$
        & (def.\ of $(\rN, \rG, \rN')$).
      \end{tabular}
    \]
    The marked equality holds because $\inte_\rG(\overline{I_{\rN'}})$ is the largest $\rG$-open in $\overline{I_{\rN'}}$. Next, the final states are preserved and reflected:
    \[
      \begin{tabular}{lll}
        $X \in \partial_\rG^{-1}(\{ Y \in O(\breve{\rG}) : Y \cap I_\rN \neq \emptyset \})$
        &    $\iff \breve{\rG}[\overline{X}] \cap I_\rN \neq \emptyset$
        \\ & $\iff \overline{\rG^\down(X)} \cap I_\rN \neq \emptyset$
        & (De Morgan duality)
        \\ & $\iff I_\rN \nsubseteq \rG^\down(X)$
        & (def.\ of $(\rN, \rG, \rN')$
        \\ & $\iff \rG[I_\rN] \nsubseteq X$
        & (adjoints)
        \\ & $\iff F_{\rN'} \nsubseteq X$
        & (def.\ of $(\rN, \rG, \rN')$).
        \\ & $\iff X \in \overline{\up_{\Open\rG} F_{\rN'}}$.
      \end{tabular}
    \]
    Finally, preservation of the deterministic transitions follows by the naturality of $\partial$.

    \item
    Given a $\JSL$-dfa $\delta = (s_0, \aS, \delta_a, F)$ it suffices to show $\lambda_\aS : (\Pirr\aS)\spbreve \to \Pirr (\aS^\pOp)$ defines an $\AutDep$-morphism of type $\Rev\Airr\delta \to \Airr(\delta^\pentagram)$. It types correctly, and since $\lambda$'s components are identity morphisms we are done.
  \end{enumerate}
\end{proof}

\smallskip
 $\JSL_f$ and $\Dep$ have enough projectives/injectives by Proposition \ref{prop:jsl_dep_enough_proj}, as do the automata-theoretic categories.

\begin{proposition}[$\dfa{\JSL}$ and $\AutDep$ have enough projectives]
  \label{prop:aut_dep_enough_proj}
  \item
  Let $Z$ be a finite set and $\gamma := (s_0, \aS, \gamma_a, F_\gamma)$ be a $\JSL$-dfa.
  \begin{enumerate}
    \item
    Given any nfa $(I, Z, \rR_a, F)$ we have the $\JSL$-dfa $(I, \Open\Delta_Z, \rR_a^\up, F_\gamma)$.
    \item
    $\epsilon_\aS : (J(\aS) \,\cap \down_\aS s_0, \aS, (\Pirr\gamma_a)_-^\up, J(\aS) \cap F_\gamma) \epito \gamma$ where $\epsilon_\aS(X) := \Lor_\aS X$ is a surjective $\JSL$-dfa morphism.
    \item
    Given $(I, \Open\Delta_Z, \rR_a^\up, F)
      \xto{f} (t_0, \aT, \delta_a, F_\delta)
      \overset{q}{\twoheadleftarrow} \gamma$
    then $f = q \circ g$ where $g$ uniquely extends $\lambda z.q_*(f(\{ z \}))$.
  \end{enumerate}
\end{proposition}

The self-duality of $\dfa{\JSL}$ preserves the freeness of the join-semilattice, so we immediately deduce:

\begin{corollary}
  \label{cor:aut_dep_enough_inj}
  $\dfa{\JSL}$ and $\AutDep$ have enough injectives.
\end{corollary}

Recall Proposition \ref{prop:dep_generator_isoms}. Choosing any join/meet-generators for $\aS$ we can construct $\nleq_\aS |_{J \times M} \cong \Pirr\aS$. Likewise, given any $\JSL$-dfa $\delta$ over $\aS$, choosing such generators yields a dependency automaton $\AutDep$-isomorphic to $\Airr \delta$.


\begin{proposition}[$\AutDep$ generator-based isomorphisms]
  \label{prop:aut_dep_generator_isoms}
  \item
  In the notation of Proposition \ref{prop:dep_generator_isoms},
  
  \begin{enumerate}
    \item
    $\rI_\aS : \fN_\aS \to \Airr (s_0, \aS, \gamma_a, F_\aS)$ defines an $\AutDep$-isomorphism where $\fN_\aS := (\rN, \nleq_\aS |_{J_\aS \times M_\aS}, \rN')$,
    \[
      \begin{tabular}{ll}
        $\rN := ( J_\aS \,\cap \down_\aS s_0 , J_\aS, \rN_a, J_\aS \,\cap F_\aS)$
        &
        $\rN_a(x_1, x_2) :\iff x_2 \leq_\aS \delta_a(x_1)$
        \\[1ex]
        $\rN' := ( M_\aS \,\cap \up_\aS \Lor_\aS \overline{F_\aS}, M_\aS, \rN'_a, M_\aS \cap\, \overline{\up_\aS s_0})$
        &
        $\rN'_a (y_1, y_2) :\iff (\delta_a)_*(y_1) \leq_\aS y_2$.
      \end{tabular}
    \]
    \item
    Suppose $f : (s_0, \aS, \gamma_a, F_\aS) \to (t_0, \aT, \delta_a, F_\aT)$ is a $\JSL$-dfa morphism. Then $\rI_f : \fN_\aS \to \fN_\aT$ defines an $\AutDep$-isomorphism.
  \end{enumerate}
\end{proposition}

The full powerset construction $\jslDfaSc{\rN} := \Det(\dep{\rN})$ also defines a free construction.

\begin{theorem}[Free $\JSL$-dfa on a dfa]
  \label{thm:free_jsl_dfa}
  If $\gamma$ is a dfa, $\delta = (t_0, \aT, \delta_a, F_\delta)$ is a $\JSL$-dfa and $f : \gamma \to \delta$ is a dfa morphism,
  \[
    \lambda S. \Lor_\aT f[S] : \jslDfaSc{\gamma} \to \delta
    \qquad
    \text{is a well-defined $\JSL$-dfa morphism}.
  \]
\end{theorem}

\begin{proof}
  Let $\gamma = (z_0, Z, \gamma_a, F_\gamma)$ and denote the candidate $\JSL$-dfa morphism by $\hat{f}$. Concerning the initial state, $\hat{f}( \{ z_0 \}) = f(z_0) = t_0$ because $f$ is a dfa morphism by assumption. Concerning final states $\Lor_\aT f[S] \in F_\delta \iff \exists z \in S. f(s) \in F_\delta$ (since there is a largest non-final state) iff $S \cap F_\gamma \neq \emptyset$ (since $f$ is a dfa morphism). Finally for each $S \subseteq Z$,
  \[
    \begin{tabular}{lll}
      $\hat{f} \circ \gamma_a^\up (S)$
      &
      $= \Lor_\aT f[\gamma_a[S]]$
      \\ &
      $= \Lor_\aT \{ f \circ \gamma_a(s) : s \in S \}$
      \\ &
      $= \Lor_\aT \{ \delta_a \circ f(s) : s \in S \}$
      & ($f$ a dfa morphism)
      \\ &
      $= \delta_a(\Lor_\aT \{ f(s) : s \in S \})$
      & ($\delta_a$ preserves $\aT$-joins)
      \\ &
      $= \delta_a \circ \hat{f}(S)$.
    \end{tabular}
  \]
\end{proof}

There is also a free construction for ordered dfas, recalling Definition \ref{def:dfas}.6.

\begin{theorem}[Free $\JSL$-dfa on an ordered dfa]
  \label{thm:free_jsl_dfa_ordered}
  Let $\gamma = (p_0, \pP, \gamma_a, F_\gamma)$ be an ordered dfa, $\delta = (t_0, \aT, \delta_a, F_\delta)$ a $\JSL$-dfa and $f : \gamma \to \delta$ an ordered dfa morphism. Then we have the well-defined $\JSL$-dfa morphism,
  \[
    \lambda S. \Lor_\aT f[S] : \Det(\gamma_\down, \geq_\pP, \rev{\gamma_\down}) \to \delta
    \qquad
    \text{where $\gamma_\down := (\down_\pP p_0, P, \gamma_a ; \geq_\pP, F_\gamma)$ is an nfa}.
  \]
\end{theorem}

\begin{proof}
  We first verify $(\gamma_\down, \geq_\pP, \rev{\gamma_\down})$ is a dependency automaton. Concerning transitions, $(\gamma_a ; \geq_\pP) ; \geq_\pP = \gamma_a ; \geq_\pP$ by transitivity and $\geq_\pP ; (\gamma_a ; \geq_\pP) = \gamma_a ; \geq_\pP$ by Example \ref{ex:dep_morphisms}.4 (via $\gamma_a : \pP^\pOp \to \pP^\pOp$). Concerning the remaining conditions:
  \[
    \geq_\pP[I_{\minSatDfa{L^r}}]
    = \; \geq_\pP[\down_\pP p_0]
    = \; \down_\pP p_0
    = F_{\rev{\minSatDfa{L^r}}},
  \]
  \[
    \geq_\pP\spbreve [I_{\rev{\minSatDfa{L^r}}}]
    = \; \leq_\pP [\{ Y \in \LW{L^r} : \epsilon \in Y \}]
    = \{ Y \in \LW{L^r} : \epsilon \in Y \}
    = F_{\minSatDfa{L^r}}.
  \]
  Denote the candidate $\JSL$-dfa morphism by $\hat{f}$.
  Recall $\Det(\gamma_\down, \geq_\pP, \rev{\gamma_\down})$ is $\jslDfaSc{\gamma_\down}$ restricted to $\Open \geq_\pP$. Concerning the initial state, $\hat{f}( \down_\pP p_0 ) = f(p_0) = t_0$ by monotonicity and the fact that $f$ is a dfa morphism. Concerning final states, $\Lor_\aT f[S] \in F_\delta \iff \exists z \in S. f(s) \in F_\delta$ iff $S \cap F_\gamma \neq \emptyset$ (since $f$ is a dfa morphism). Finally for each down-closed $S \subseteq P$,
  \[
    \begin{tabular}{lll}
      $\hat{f} \circ (\gamma_a ; \geq_\pP)^\up (S)$
      &
      $= \Lor_\aT f[\gamma_a ; \geq_\pP[S]]$
      \\ &
      $= \Lor_\aT \{ f \circ \gamma_a(s) : s \in S \}$
      & ($f \circ \gamma_a$ monotonic)
      \\ &
      $= \Lor_\aT \{ \delta_a \circ f(s) : s \in S \}$
      & ($f$ a dfa morphism)
      \\ &
      $= \delta_a(\Lor_\aT \{ f(s) : s \in S \})$
      & ($\delta_a$ preserves $\aT$-joins)
      \\ &
      $= \delta_a \circ \hat{f}(S)$.
    \end{tabular}
  \]
\end{proof}

\subsection{Canonical dependency automata}

Previously we described the canonical dependency automaton for $L \subseteq \Sigma^*$ in Definition \ref{def:canon_dep_aut}. We now describe the state-minimal $\JSL$-dfa for $L$. These two machines are actually the same object modulo categorical equivalence.

\begin{definition}[Left and right quotients]
  \label{def:lang_quos}
  Fix any regular language $L \subseteq \Sigma^*$ and recall the left word quotients $\LW{L}$ from Definition \ref{def:dfas}.
  \begin{enumerate}
    \item 
      For $U, V \subseteq \Sigma^*$, $U^{-1} L := \{ w \in \Sigma^* : \exists u \in U. uw \in L \}$ is a \emph{left quotient}, $L V^{-1} := \{ w \in \Sigma^* : \exists v \in V. wv \in L \}$ is a  \emph{right quotient}.
      \item
      Let $\LQ{L} := \{ U^{-1} L : U \subseteq \Sigma^* \}$ and $\RQ{L} := \{ L V^{-1} : V \subseteq \Sigma^*  \}$. Then $\LW{L} \subseteq \LQ{L}$ and likewise we may write $L v^{-1}$ instead of $L \{ v \}^{-1}$.

      \item 
      We have finite join-semilattice $\jslLQ{L} := (\LQ{L}, \cup, \emptyset)$.
      \item
      $J(\jslLQ{L}) \subseteq \LW{L}$ because the latter generate $\LQ{L}$ under unions. Thus $J(\jslLQ{L})$ consists of those left word quotients which are not unions of others, so in particular are non-empty. \endbox

  \end{enumerate}
\end{definition}

\begin{definition}[State-minimal $\JSL$-dfa]
  \label{def:state_min_jsl_dfa}
    Let $\jslDfaMin{L} := (L, \jslLQ{L}, \lambda X. a^{-1} X, \{ X \in \LQ{L} : \epsilon \in X \})$. \endbox
\end{definition}

Observe that the state-minimal $\minDfa{L}$ is obtained from $\jslDfaMin{L}$ by restricting to left \emph{word} quotients $u^{-1} L$. Conversely, every left quotient $U^{-1} L$ is a finite union of left word quotients and $a^{-1}(-)$ preserves unions.

\begin{lemma}
  \label{lem:state_min_jsl_dfa}
  $\jslDfaMin{L}$ is the state-minimal $\JSL$-dfa accepting $L$.
\end{lemma}

\begin{proof}
  It accepts $L$ -- the reachable part of its underlying dfa is precisely the state-minimal dfa $\minDfa{L}$.
  Concerning the state-minimality of $\jslDfaMin{L}$, take any $\JSL$-dfa $\delta = (s_0, \aS, \delta, F)$ accepting $L$ and consider the languages accepted by varying the initial state, noting $|langs(\delta)| \leq |S|$. By Lemma we know \ref{lem:jsl_dfa_joins} $\LQ{L} \subseteq langs(\delta)$, hence $|\LQ{L}| \leq |S|$.
\end{proof}

\begin{example}[$\minSatDfa{L}$]
  \label{ex:saturated_min_dfa}
  Applying Proposition \ref{prop:aut_dep_generator_isoms} to $\jslDfaMin{L}$ with join-generating subset $J_{\jslLQ{L}} := \LW{L}$ yields a dependency automaton, whose lower nfa takes the following form:
  \[
    \begin{tabular}{c}
      $\minSatDfa{L}
      := (\{ X \in \LW{L} : X \subseteq L \}, \LW{L}, \rN_a, \{ X \in \LW{L} : \epsilon \in X \})$
      \\[1ex]
      $\rN_a(X_1, X_2) :\iff X_2 \subseteq a^{-1} X_1$.
    \end{tabular}
  \]
  It accepts $L$ by the witnessing $\AutDep$-isomorphism (see Note \ref{note:det_airr_preserve_acceptance}). Importantly, we'll use it to represent the canonical distributive $\JSL$-dfa further below.
  \endbox
\end{example}

Recall the self-duality $(-)^\pentagram : \dfa{\JSL}^{op} \to \dfa{\JSL}$ of Theorem \ref{thm:dfa_jsl_self_dual}, itself arising from the self-duality of $\JSL_f$ in Note \ref{note:jsl_extras}.3. We now describe an important representation of $\jslDfaMin{L}$ which explains its meet structure. 

\begin{lemma}
  \label{lem:de_morgan_right_quo}
  $\overline{ [\overline{L X^{-1}}]^{-1} L }
    = \{ w \in \Sigma^* : Lw^{-1} \subseteq LX^{-1} \}$
  for any subsets $X, L \subseteq \Sigma^*$.
\end{lemma}
\begin{proof}
  \[
    \begin{tabular}{lll}
    $\overline{ [\overline{L X^{-1}}]^{-1} L }$
    &
    $= \{ w \in \Sigma^* : \forall v \in \Sigma^*.( v \in \overline{LX^{-1}} \To vw \nin L )  \}$
    \\&
    $= \{ w \in \Sigma^* : \forall v \in \Sigma^*.( vw \in L \To v \in LX^{-1} )  \}$
    \\&
    $= \{ w \in \Sigma^* : \forall v \in \Sigma^*.( v \in Lw^{-1} \To v \in LX^{-1} )  \}$
    \\&
    $= \{ w \in \Sigma^* : Lw^{-1} \subseteq LX^{-1} \}$.
    \end{tabular}
  \]
\end{proof}

\begin{theorem}[Fundamental dualising isomorphism $dr_L$]
  \label{thm:dr_L}
  For each regular $L \subseteq \Sigma^*$ we have the $\JSL$-dfa isomorphism:
  \[
  (\jslDfaMin{L^r})^\pentagram \xto{dr_L} \jslDfaMin{L}
  \qquad
  dr_L(X) := [\overline{X}^r]^{-1} L
  \qquad
  dr_L^{\bf -1} := dr_{L^r},
  \]
  noting that reversal/complement of languages commute. There is also an alternative description:
  \[
  dr_L (U^{-1} L^r) :=
  \bigcup \{ X \in \LW{L} : X \,\cap\, U^r = \emptyset  \}.
  \]
\end{theorem}

\begin{proof}
  \item
  \begin{enumerate}
    \item
    We first establish the underlying join-semilattice isomorphism $dr_L : (\jslLQ{L^r})^\pOp \to \jslLQ{L}$.
    Now, $dr_L$ is certainly a well-defined function. It is monotone because $X \subseteq Y \in \LW{L^r}$ implies $\overline{Y}^r \subseteq \overline{X}^r$ and hence $[\overline{Y}^r]^{-1} L \subseteq [\overline{X}^r]^{-1} L$. Likewise $dr_L^{-1}$ is a well-defined monotone function. Next, given any $X \in \LW{L^r}$ ,
    \[
    \begin{tabular}{lll}
    $dr_L^{-1} \circ dr_L (X)$
    &
    $= dr_L^{-1} ([\overline{X}^r]^{-1} L)$
    \\&
    $= [\overline{[\overline{X}^r]^{-1} L}^r]^{-1} L^r$
    \\&
    $= [\overline{L^r \overline{X}^{-1}}]^{-1} L^r$
    \\&
    $= \{ w \in \Sigma^* : L^r w^{-1} \nsubseteq L^r \overline{X}^{-1} \}$
    & (by Lemma \ref{lem:de_morgan_right_quo}).
    \\&
    $= \{ w \in \Sigma^* : \exists u \in \Sigma^*. [ uw \in L^r \text{ and } \forall v \in \Sigma^*.[ uv \in L^r \To v \nin \overline{X }  ]] \}$
    \\&
    $= \{ w \in \Sigma^* : \exists u \in \Sigma^*. [ w \in u^{-1} L^r \text{ and } \forall v \in \Sigma^*. [ v \in u^{-1} L^r \To v \in X ] ] \}$
    \\&
    $= \{ w \in \Sigma^* : \exists u \in \Sigma^*. [ w \in u^{-1} L^r \text{ and } u^{-1} L^r \subseteq X ] \}$
    \\&
    $= X$
    & (since $X \in \LQ{L}$).
    \end{tabular}
    \]
    It immediately follows that $dr_L \circ dr_L^{-1} (Y) = Y$ by substituting $L \mapsto L^r$.
    Thus $dr_L$ is a bijective order-preserving and order-reflecting function, hence a bounded-lattice isomorphism and in particular a $\JSL_f$-morphism. Finally we establish the alternative action:
    \[
    \begin{tabular}{lll}
    $dr_L (U^{-1} L^r)$
    &
    $= [(\overline{U^{-1} L^r})^r]^{-1} L$
    \\&
    $= [\overline{L (U^r)^{-1}}]^{-1} L$
    \\&
    $= \{ w \in \Sigma^* : Lw^{-1} \nsubseteq L(U^r)^{-1} \}$
    & (by Lemma \ref{lem:de_morgan_right_quo}).
    \\&
    $= \{ w \in \Sigma^* : \exists x \in \Sigma^*. [ xw \in L \text{ and } \forall y \in \Sigma^*. [ xy \in L \To y \nin U^r  ] ] \}$
    \\&
    $= \{ w \in \Sigma^* : \exists x \in \Sigma^*. [ w \in x^{-1} L \text{ and } \forall y \in \Sigma^*. [ y \in x^{-1} L \To y \in \overline{U}^r ]] \} $
    \\&
    $= \{ w \in \Sigma^* : \exists x \in \Sigma^*. [ w \in x^{-1} L \text{ and } x^{-1} L \subseteq \overline{U}^r ] \} $
    \\&
    $= \{ w \in \Sigma^* : \exists x \in \Sigma^*. [ w \in x^{-1} L \text{ and } x^{-1} L \cap U^r = \emptyset ] \} $
    \\&
    $= \bigcup \{ X \in \LW{L} : X \,\cap\, U^r = \emptyset  \}$.
    \end{tabular}
    \]

    \item
    It remains to establish that the join-semilattice isomorphism $dr_L$ defines a dfa morphism.
    Concerning preservation of the initial state:
    \[
      \begin{tabular}{lll}
        $dr_L( i_{(\jslDfaMin{L^r})_*} )$
        &
        $= dr_L( \Lor_{\jslLQ{L^r}} \overline{\{ X \in \LW{L^r} : \epsilon \in X \}} )$
        \\ &
        $= dr_L (\bigcup \{ X \in \LW{L^r} : \epsilon \nin X \})$
        \\ &
        $= dr_L (\overline{L}^{-1} L^r)$
        \\ &
        $= \bigcup \{ X \in \LW{L} : X \,\cap\, \overline{L} = \emptyset \}$
        \\ &
        $= \bigcup \{ X \in \LW{L} : X \subseteq L \}$
        \\ &
        $= L = i_{\jslDfaMin{L}}$.
      \end{tabular}
    \]
    Concerning transitions, let $\gamma_a : \jslLQ{L^r} \to \jslLQ{L^r}$ and $\delta_a : \jslLQ{L} \to \jslLQ{L}$ be the deterministic $a$-transitions for the respective machines (both have action $\lambda X. a^{-1} X$). It suffices to show $(\gamma_a)_* = dr_L^{\bf -1} \circ \delta_a \circ dr_L$:
    \[
      \begin{tabular}{lll}
        $dr_L^{\bf -1} \circ \delta_a \circ dr_L(X)$
        &
        $= dr_L^{\bf -1} (a^{-1}([\overline{X}^r]^{-1} L))$
        \\ &
        $= dr_L^{\bf -1} ([\overline{X}^r a]^{-1} L)$
        \\ &
        $= \bigcup \{ X \in \LW{L^r} : Y \cap a \overline{X} = \emptyset \}$
        \\ &
        $= \bigcup \{ X \in \LW{L^r} : a^{-1} Y \subseteq X \}$
        \\ &
        $= (\gamma_a)_*(X)$.
      \end{tabular}
    \]
    The final states are preserved and reflected:
    \[
      \begin{tabular}{lll}
        $X \in dr_L^{-1} (F_{\jslDfaMin{L}})$
        &
        $\iff X \in dr_L^{-1}(\{ Y \in \LQ{L} : \epsilon \in Y \})$
        \\ &
        $\iff \epsilon \in dr_L(X)$
        \\ &
        $\iff \epsilon \in [\overline{X}^r]^{-1} L$
        \\ &
        $\iff \overline{X}^r \cap L \neq \emptyset$
        \\ &
        $\iff \overline{X} \cap L^r \neq \emptyset$
        \\ &
        $\iff L^r \nsubseteq X$
        \\ &
        $\iff X \in F_{(\jslDfaMin{L^r})_*}$.
      \end{tabular}
    \]
  \end{enumerate}
\end{proof}

\begin{note}
  \item
  \begin{enumerate}
    \item
    $dr_L$ provides a bijection between $L$'s left word quotients $U^{-1} L$ and right word quotients $L V^{-1} = ((V^r)^{-1} L^r)^r$.
    \item 
    If $L = L^r$ we have an order-reversing involutive isomorphism $dr_L: (\jslLQ{L})^\pOp \to \jslLQ{L}$.
    Thus $\jslLQ{L}$ is a De Morgan algebra whose bounded lattice structure needn't be distributive.
    This holds for any unary language.
    \endbox
  \end{enumerate}
\end{note}

\begin{corollary}[Meet-generating $\jslLQ{L}$]
  \label{cor:mirr_in_lq}
  \item
  \begin{enumerate}
    \item 
    $\jslLQ{L}$ is join-generated by $\LW{L}$ and meet-generated by $dr_L[\LW{L^r}] = \{ [\overline{L v^{-1}}]^{-1} L : v \in \Sigma^* \}$.
    
    \item
    $Y \subseteq dr_L(v^{-1} L^r) \iff v^r \nin Y$ for each $Y \in \LQ{L}$.

    \item
    Each $Y \in \LQ{L}$ arises as an intersection:
    \[
      \begin{tabular}{lll}
        $Y$
        &
        $= \Land_{\jslLQ{L}} \{ dr_L(v^{-1} L^r) : v^r \nin Y  \}$
        \\[0.5ex] &
        $= \bigcap \{ dr_L(v^{-1} L^r) : v^r \nin Y \}$.
      \end{tabular}
    \]
  \end{enumerate}

\end{corollary}

\begin{proof}
  \item
  \begin{enumerate}
    \item 
    That $\LQ{L}$ is generated by $\LW{L}$ under finite unions follows via Definition \ref{def:lang_quos}.
    Concerning the new claim, the order isomorphism $dr_L$ from Theorem \ref{thm:dr_L} preserves/reflects meet-irreducibles. Then recalling Definition \ref{def:lang_quos} we have $M(\jslLQ{L^r}^\pOp) = J(\jslLQ{L^r}) \subseteq \LW{L^r}$, so applying $dr_L$ we obtain a meet-generating set.

    \item
    Given $Y \subseteq dr_L(v^{-1} L^r)$ then since $dr_L(v^{-1} L^r) = \bigcup \{ X \in \LW{L} : v^r \nin X \}$ doesn't contain $v^r$ we immediately deduce $v^r \nin Y$. Conversely if $v^r \nin Y$ then for every $Y \supseteq X \in \LW{L}$ we have $v^r \nin X$ hence $X \subseteq dr_L(v^{-1} L^r)$, so that $Y \subseteq dr_L(v^{-1} L^r)$ too.

    \item
    By (1) each $Y \in \LW{L}$ is the meet of those $K \in dr_L[\LW{L^r}] \subseteq \LQ{L}$ above it. By (2) $Y \subseteq dr_L(v^{-1} L^r)$ iff $v^r \nin Y$, which implies the first equality. Finally this meet is actually an intersection: given $w \in \bigcap \{ dr_L(v^{-1} L^r) : v^r \nin Y \}$ then if $w \nin Y$ we obtain the contradiction $w \in dr_L((w^r)^{-1} L^r)$.
  \end{enumerate}
\end{proof}

\smallskip
With reference to the Corollary \ref{cor:mirr_in_lq}, the next Lemma explains the strong connection between the state-minimal $\JSL$-dfa and the canonical dependency automaton $(\minDfa{L}, \rDR{L}, \minDfa{L^r})$ where $\rDR{L}(u^{-1} L, v^{-1} L^r) :\iff uv^r \in L$.

\begin{lemma}[Dependency Lemma]
  \label{lem:dependency_lemma}
  For any regular $L \subseteq \Sigma^*$ and words $u, v \in \Sigma^*$,
  \[
    u^{-1} L \nsubseteq [\overline{L v^{-1}}]^{-1} L \iff uv \in L
    \qquad\text{or equivalently}\qquad
    u^{-1} L \nsubseteq dr_L(v^{-1} L^r) \iff uv^r \in L \iff \rDR{L}(u, v).
  \]
\end{lemma}

\begin{proof}
  We calculate:
  \[
  \begin{tabular}{lll}
    $u^{-1} L \nsubseteq [\overline{L v^{-1}}]^{-1} L$
    &
    $\iff \exists y \in \Sigma^*. [ uy \in L \text{ and } y \nin [\overline{L v^{-1}}]^{-1} L]$
    \\&
    $\iff \exists y \in \Sigma^*. [ uy \in L \text{ and } \forall x \in \Sigma^*.[ xy \in L \To x \nin \overline{L v^{-1}}]]$
    \\&
    $\iff \exists y \in \Sigma^*. [ uy \in L \text{ and } \forall x \in \Sigma^*.[ x \in Ly^{-1} \To x \in L v^{-1}]]$
    \\&
    $\iff \exists y \in \Sigma^*. [ u \in Ly^{-1} \text{ and } Ly^{-1} \subseteq L v^{-1}]$
    \\&
    $\iff u \in L v^{-1}$
    \\&
    $\iff uv \in L$.
    \end{tabular}
  \]
  If $L$ is regular we may rewrite this in terms of $dr_L$ and $\rDR{L}$ as above, see Theorem \ref{thm:dr_L} and Definition \ref{def:canon_dep_aut}.
\end{proof}

We are now ready for the main result of this subsection.

\begin{theorem}[Dependency Theorem]
  \label{thm:dependency_thm}
  The state-minimal $\JSL$-dfa is isomorphic to the determinisation of the canonical dependency automaton.
  \[
    \alpha : \jslDfaMin{L} \to \Det (\minDfa{L}, \rDR{L}, \minDfa{L^r})
    \]
  \[
    \alpha (X) := \{ v^{-1} L^r : v \in X^r \}
    \qquad    
    \alpha^{\bf-1}(Y) := [(\bigcap Y)^r]^{-1} L
  \]
\end{theorem}

\begin{proof}
  \item
  \begin{enumerate}
    \item
    We first establish the underlying isomorphism $\alpha: \aS \to \Open\rDR{L}$ where $\aS := \jslLQ{L}$.
    
    By Theorem \ref{thm:dep_equiv_jslf} we have the isomorphism $rep_\aS : \aS \to \Open\Pirr\aS$ where $rep_\aS(X) := \{ Y \in M(\aS) : X \nsubseteq Y \}$. Proposition \ref{prop:dep_generator_isoms} permits one to extend the domain/codomain of $\Pirr\aS$ to join/meet generators, which by Corollary \ref{cor:mirr_in_lq} can be $J_\aS := \LW{L}$ and $M_\aS := \{ [\overline{L v^{-1}}]^{-1} L : v \in \Sigma^* \}$. Then we obtain the $\Dep$-isomorphism:
    \[
      \rI_\aS^{\bf -1} := \; \nleq_\aS |_{J(\aS) \times M_\aS} : \Pirr\aS \to \;\nleq_\aS |_{J_\aS \times M_\aS}
      \qquad
      \text{where $(\rI_\aS^{\bf -1})_+ (y, m) :\iff y \leq_\aS m$}
    \]
    and thus $\Open\rI_\aS^{\bf -1} : \Open\Pirr\aS \to \Open \nleq_\aS |_{J_\aS \times M_\aS}$ is an isomorphism with action $\lambda X. \LW{L^r} \;\cap \down_\aS X$. Next,
     we'll establish the bipartite graph isomorphism:
    \[
      \xymatrix@=15pt{
        M_\aS \ar[rrr]^-{\lambda X. dr_L^{\bf -1}(X)} &&& \LW{L^r}
        \\
        J_\aS \ar[u]^{\nleq_\aS |_{J_\aS \times M_\aS}} \ar[rrr]_{id_{\LW{L}}} &&& \LW{L} \ar[u]_{\rDR{L}}
      }
      \qquad
      \text{}
    \]
    The upper bijective witness is well-defined because $M_\aS = dr_L[\LW{L^r}]$; this $\Rel$-diagram  commutes by the Dependency Lemma \ref{lem:dependency_lemma}. It defines a $\Dep$-isomorphism $\rDR{L} : \;\nleq_\aS |_{J_\aS \times M_\aS} \to \rDR{L}$ and hence a $\JSL_f$-isomorphism $\Open\rDR{L} := \lambda X. dr_L^{-1}[X]$ recalling that any upper witness can be used by Note \ref{note:open_morphism_alt}. Then we have the composite isomorphism:
    \[
      \aS
      \xto{rep_\aS} \Open\Pirr\aS
      \xto{\rI_\aS^{\bf -1}} \Open\nleq_\aS |_{J_\aS \times M_\aS}
      \xto{\Open\rDR{L}}
      \Open \rDR{L}
    \]
    with action and inverse action:
    \[
      \begin{tabular}{llll}
        $X \in \LQ{L}$
        &
        $\mapsto$ & $\{ Y \in M(\aS) : X \nsubseteq Y \}$
        & (apply $\rep_\aS$)
        \\ &
        $\mapsto$ & $\{ Z \in \LW{L^r} : \exists Y \in M(\aS). ( X \nsubseteq Y \text{ and } Z \subseteq Y ) \}$
        & (apply $\Open\rI_\aS^{\bf -1}$)
        \\ &
        $=$ & $\{ Z \in \LW{L^r} : X \nsubseteq Z \}$
        \\ &
        $\mapsto$ & $dr_L^{\bf -1}[\{ Z \in \LW{L^r} : X \nsubseteq Z \}]$
        & (apply $\Open\rDR{L}$)
        \\ &
        $=$ & $\{ dr_L^{\bf -1}(Z) : X \nsubseteq Z \in \LW{L^r} \}$
        \\ &
        $=$ & $\{ Y \in \LW{L^r} : X \nsubseteq dr_L(Y) \}$
        & (substitute $Z := dr_L(Y)$)
        \\ &
        $=$ & $\{ v^{-1} L^r : v \in X^r \}$
        & (alternative action of $dr_L$).
        \\
        \\
        $S \in O(\rDR{L})$
        & $\mapsto$ & $dr_L[S]$
        & (apply $(\Open\rDR{L})^{\bf -1}$)
        \\
        & $\mapsto$ & $M(\aS) \;\cap \down_\aS dr_L[S]$
        & (apply $\Open\rI_\aS$)
        \\
        & $=$ & $M(\aS) \cap dr_L[S]$
        & (by 1)
        \\
        & $=$ & $dr_L[J(\jslLQ{L^r}) \;\cap  S]$
        & ($M(\aS) = dr_L[J(\jslLQ{L^r})]$)
        \\
        & $\mapsto$ & $\Land_\aS M(\aS) \setminus dr_L[J(\jslLQ{L^r}) \;\cap S]$
        & (apply $\rep_\aS^{\bf -1}$)
        \\
        & $=$ & $\Land_\aS \{ dr_L(X) : X \in  J(\jslLQ{L^r}) \setminus S \}$
        &
        \\ &
        $=$ & $dr_L [\bigcup \{ X \in  J(\jslLQ{L^r}) : X \nin S \}]$
        & ($dr_L$ an isomorphism)
        \\ &
        $=$ & $[(\bigcap \{ X \in J(\jslLQ{L^r}) : X \in S \})^r]^{-1} L$
        \\ &
        $=$ & $[(\bigcap \{ X \in \LW{L^r} : X \in S \})^r]^{-1} L$
        & (by 2)
        \\ &
        $=$ & $[(\bigcap S)^r]^{-1} L$.
      \end{tabular}
    \]
    Concerning (1), each $\rDR{L}$-open set $S$ is up-closed in $(\LW{L^r}, \subseteq)$ so $dr_L[S]$ is down-closed in $(M_\aS, \subseteq)$ recalling $M(\aS) \subseteq M_\aS$. Concerning (2), if $w^{-1} L^r \in S$ some join-irreducible $v^{-1} L^r \subseteq w^{-1} L^r$ must also lie in $S$.
    
    \item
    It remains to establish that $\alpha$ is a dfa morphism:
    \[
      \alpha : \jslDfaMin{L} \to (F_{\minDfa{L^r}}, \Open\rDR{L}, \delta_a, \{ L^r  \})
      \qquad
      \delta_a := \lambda S. \{ Y \in \LW{L^r} : \exists X \in S. a^{-1} Y = X  \}
    \]
    The initial state is preserved because $\alpha(L) = \{ v^{-1} L^r : v \in L^r \} = \{ X \in \LW{L^r} : \epsilon \in X \} = F_{\minDfa{L^r}}$. Next we show the transitions are preserved, denoting the domain dfa's transitions by $\gamma_a := \lambda X. a^{-1} X$. Given any $X \in \LW{L}$,
    \[
      \tag{a}
      \alpha \circ \gamma_a (X)
      = \alpha(a^{-1} X)
      = \{ v^{-1} L^r : v \in (a^{-1} X)^r \}
      = \{ v^{-1} L^r : v \in X^r a^{-1} \}.
    \]
    \[
      \tag{b}
      \begin{tabular}{lll}
        $\delta_a \circ \alpha (X)$
        &
        $= \delta_a (\{ w^{-1} L^r : w \in X^r \})$
        \\ &
        $= \{ v^{-1} L^r : \exists w \in X^r. a^{-1}(v^{-1} L^r) = w^{-1} L^r \}$
        \\ &
        $= \{ v^{-1} L^r : \exists w \in X^r. (va)^{-1} L^r = w^{-1} L^r \}$.
      \end{tabular}
    \]
    Let us establish $(a) = (b)$ via mutual inclusions.

    \begin{itemize}
      \item[--]
      $(a) \subseteq (b)$: Given $v \in X^r a^{-1}$ we deduce $w := va \in X^r$, hence $v^{-1} L^r$ resides in (b).
      \item[--]
      $(b) \subseteq (a)$: We may assume $X := (u^r)^{-1} L$; we know there exists $w \in X^r = L^r u^{-1}$ such that $(va)^{-1} L^r = w^{-1} L^r$. Then $wu \in L^r$ hence $u \in w^{-1} L^r$ so that $vau \in L^r$. Thus $va \in L^r u^{-1}$ i.e.\ $v \in X^r a^{-1}$ so $v^{-1} L^r$ resides in (a).
    \end{itemize}
  
    Lastly the final states are preserved and reflected:
    \[
      X \in \alpha^{-1} (\{ L^r \})
      \iff \alpha(X) = L^r
      \iff \epsilon \in X^r
      \iff \epsilon \in X.
    \]
  \end{enumerate}
\end{proof}

\begin{note}[Canonical dependency automaton as canonical residual automata]
  \item
  By the Dependency Theorem \ref{thm:dependency_thm}, the canonical dependency automaton corresponds to the state-minimal dfa interpreted in join-semilattices. On the other hand, the categorical equivalence $\dfa{\JSL} \cong \AutDep$ of Theorem \ref{thm:aut_dep_equiv_jsl_dfa} already provides the component isomorphism:
  \[
    \rep_{\jslDfaMin{L}} : \jslDfaMin{L} \to \Det(\Airr(\jslDfaMin{L})) = \Det(\rN_L, \Pirr\jslLQ{L}, \rN').
  \]
  The lower nfa $\rN_L$ is precisely the canonical residual automaton of \cite{ResidFSA2001}. That is, let $\ILQ{L} := J(\jslLQ{L}) \subseteq \LW{L}$ be the \emph{irreducible left quotients} i.e.\ those left word quotients not arising as the union of others (so, non-empty). Then:
  \[
    \rN_L = (\ILQ{L} \;\cap \down_{\jslLQ{L}} L, \ILQ{L}, \rN_a, \{ X \in \ILQ{L} : \epsilon \in X \})
    \qquad
    \rN_a(X_1, X_2) :\iff X_2 \subseteq a^{-1} X_1.
  \]
  Relabelling the upper bipartition we obtain $(\rN_L, \rDR{L} |_{\,\ILQ{L} \times \ILQ{L^r}}, \rN_{L^r})$. The upper nfa is the canonical residual nfa for $L^r$. The bipartitioned graph is obtained by restricting the dependency relation to irreducibles. It is necessarily $\AutDep$-isomorphic to the canonical dependency automaton, and actually constructable from it in polytime. It is never larger than our chosen description and potentionally far smaller.
  \endbox 
\end{note}

\subsection{Explaining Brzozowski's algorithm}

Recalling Definition \ref{def:dfas}, the minimisation of a classical dfa $\delta := (z_0, Z, \delta_a, F)$ can be understood as follows:
\[
  \xymatrix@=10pt{
    \delta \ar@{->>}[rrrr]^-{acc_\delta}
    &&&& \simple{\delta}
    \\
    &&&&
    \reach{\simple{\delta}}
    \ar@{_{(}->}[u]_{\iota_2}
    \\
    \reach{\delta} \ar@{_{(}->}[uu]^{\iota_1} \ar@{->>}[rrr]_-{acc_{\reach{\delta}}}
    &&&
    \simple{\reach{\delta}}
    \ar@{=}[r]
    &
    \minDfa{L}
    \ar@{=}[u]
    }
\]
Traditionally one first takes $\reach{\delta}$ by restricting to states reachable from $z_0$ via the underlying directed graph $\bigcup_{a \in \Sigma} \delta_a \subseteq Z \times Z$. From the perspective of dfa morphisms \emph{we construct the minimal sub-dfa of $\delta$} (i.e.\ the inclusion $\iota_1$ above). Secondly one can apply Hopcroft's algorithm to compute a partition of the states i.e.\ the equivalence classes over which the state-minimal dfa can then be defined. From the perspective of dfa morphisms \emph{we construct the largest quotient-dfa of $\delta$} (i.e.\ the surjection $acc_{\reach{\delta}}$ above)\footnote{By \emph{largest quotient} we mean the respective equivalence relation is the largest w.r.t.\ inclusion. The respective quotient-dfa actually has the least possible number of states amongst other such quotients.}. The latter sends a state to the language it accepts, yielding precisely the state-minimal machine $\minDfa{L}$. Notice the other way to minimise $\delta$: quotient first; restrict to reachable second.

\smallskip
We expressed minimisation in terms of dfa morphisms because one has exactly the same situation in $\dfa{\JSL}$, whose morphisms must also preserve the join-semilattice structure. For any $\JSL$-dfa $\delta := (s_0, \aS, \delta_a, F)$,
\[
  \xymatrix@=10pt{
    \delta \ar@{->>}[rrrr]^-{acc_\delta}
    &&&& \jslDfaSimple{\delta}
    \\
    &&&&
    \jslDfaReach{\jslDfaSimple{\delta}}
    \ar@{_{(}->}[u]_{\iota_2}
    \\
    \jslDfaReach{\delta} \ar@{_{(}->}[uu]^{\iota_1} \ar@{->>}[rrr]_-{acc_{\jslDfaReach{\delta}}}
    &&&
    \jslDfaSimple{\jslDfaReach{\delta}}
    \ar@{=}[r]
    &
    \jslDfaMin{L}
    \ar@{=}[u]
    }
\]
We've already seen the state-minimal $\jslDfaMin{L}$ and its close connection to the canonical dependency automaton. We now introduce the corresponding concepts of reachability and simplicity, recalling the notation of Definition \ref{def:dfas}.

\begin{definition}[$\JSL$-reachability and simplicity]
\label{def:jsl_reach_simple}
Let $\delta := (s_0, \aS, \delta_a, F)$ be a $\JSL$-dfa. 
\begin{enumerate}
  \item
  $\delta$ is \emph{$\JSL$-reachable} if it has no proper sub $\JSL$-dfas: every injective $\JSL$-dfa morphism $f : \gamma \monoto \delta$ is an isomorphism.
  Given any $\aR \subseteq \aS$ with $s_0 \in R$ and $\delta_a(R) \subseteq R$ for $a \in \Sigma$, then
  $
    \aR \cap \delta := (s_0, \aR, \delta_a |_{R \times R}, F \cap R)
  $
  is a $\JSL$-dfa accepting $L(\delta)$. In particular,
  \[
    \jslDfaReach{\delta} := \jslReach{\delta} \cap \delta
    \qquad
    \text{where $\jslReach{\delta} := \ang{reach(\delta)}_\aS$}
    \footnote{By $\ang{reach(\delta)}_\aS$ we mean the sub join-semilattice of $\aS$ generated by $reach(\delta)$.}
  \]
  is \emph{the reachable sub $\JSL$-dfa of $\delta$}.
  
  \item
  $\delta$ is \emph{simple} if it has no proper quotient $\JSL$-dfas: every surjective $\JSL$-dfa morphism $f : \delta \epito \gamma$ is an isomorphism.
   We have the join-semilattice of accepted languages $\jslLangs{\delta} := (langs(\delta), \cup, \emptyset)$ by Definition \ref{def:dfas}.5 and Lemma \ref{lem:jsl_dfa_joins}. Then the simple $\JSL$-dfa:
  \[
    \jslDfaSimple{\delta} := (L, \jslLangs{\delta}, \lambda X.a^{-1}X, \{ X \in langs(\delta) : \epsilon \in X \})
    \]
  is the \emph{largest quotient $\JSL$-dfa} of $\delta$ via $acc_\delta$. Finally, a $\JSL$-dfa $\delta$ is \emph{simplified} if $\jslDfaSimple{\delta} = \delta$.
   \endbox
\end{enumerate}  
\end{definition}

\begin{lemma}[Well-definedness of $\jslDfaReach{-}$ and $\jslDfaSimple{-}$]
  \label{lem:jsl_reach_simple_constructions}
  \item
  \begin{enumerate}
    \item
    $\jslDfaReach{\delta}$ is the $\JSL$-reachable sub-dfa of $\delta$.
    \item
    $\jslDfaSimple{\delta}$ is the simple quotient dfa of $\delta$.
    \item
    $L(\jslDfaSimple{\delta}_{@X}) = X$ for each $X \in simple(\delta)$.
  \end{enumerate}
\end{lemma}
\begin{proof}
  \item
  \begin{enumerate}
    \item
    $\aR \cap \delta$ is a well-defined $\JSL$-dfa: (a) the conditions ensure each $\delta_a : \aS \to \aS$ restricts to an $\aR$-endomorphism, (b) just as $F = h^{-1}(\{ 1 \})$ for some $h : \aS \to \two$, $F \cap R = (h \circ \iota)^{-1}(\{ 1 \})$ where $\iota : \aR \hookto \aS$.
    
    Concerning well-definedness of $\jslDfaReach{\delta}$, $\aR: = \jslReach{\delta}$ is the reachable part of the underlying classical dfa closed under all $\aS$-joins. Certainly $\aR \subseteq \aS$ and $s_0 \in R$. Next, $\delta_a[R] \subseteq R$ because $\delta_a$ preserves all joins, so applying $\delta_a$ to joins of classically reachable states is the same as taking the join of $a$-successors of classically reachable states. It accepts $L$ because the reachable part of its underlying classical dfa is precisely $\reach{\delta}$.

    Finally, $\jslDfaReach{\delta}$ is $\JSL$-reachable because any sub $\JSL$-dfa must at least contain the underlying reachable part and be closed under the algebraic structure.

    \item
    Concerning well-definedness of $\jslDfaSimple{\delta}$, $langs(\delta)$ is closed under arbitrary unions by Lemma \ref{lem:jsl_dfa_joins}. Certainly $L \in langs(\delta)$ and the transitions are well-defined by Definition \ref{def:dfas}.2. The final states are well-defined because the union of all languages sans $\epsilon$ does not contain it either.
    
    It accepts $L$ because the reachable part of its underlying classical dfa is precisely $\minDfa{L}$. Finally, $acc_\delta : \delta \to \jslLangs{\delta}$ is additionally a join-semilattice morphism by Lemma \ref{lem:jsl_dfa_joins}. It is simple because each state $X \in lang(\delta)$ accepts $X$ i.e.\ distinct states accept distinct languages, so there can be no quotient dfa and thus also no quotient $\JSL$-dfa.

    \item
    Follows via Definition \ref{def:canon_dep_aut}.2.
  \end{enumerate}
\end{proof}

However, the self-duality of $\dfa{\JSL}$ provides an additional relationship.

\begin{theorem}[$\JSL$-reachability is dual to simplicity]
  \label{thm:jsl_reach_dual_simple}
  Let $\delta := (s_0, \aS, \delta_a, F)$ be a $\JSL$-dfa.
  \begin{enumerate}
    \item
    $\delta$ is $\JSL$-reachable iff every join-irreducible $j \in J(\aS)$ is classically reachable.
    \item
    $\delta$ is simple iff distinct states accept distinct languages.
    \item
    $\delta$ is $\JSL$-reachable iff its dual $\delta^\pentagram$ is simple.
  \end{enumerate}
\end{theorem}

\begin{proof}
  $\delta$ is $\JSL$-reachable iff $\jslReach{\delta} = \delta$ iff every state is a join of classically reachable states. Since $J(\aS)$ is the minimal join-generating set we infer (1). Concerning (2), $\delta$ is simple iff $acc_\delta : \delta \to \jslLangs{\delta}$ is bijective iff distinct states accept distinct languages. Finally, the concepts of $\JSL$-reachable and simple are categoricially dual, recalling $\JSL_f$-monos are precisely the injective morphisms and $\JSL_f$-epis are precisely the surjective ones (see Note \ref{note:jsl_extras}.4).
\end{proof}

We also mention a basic characterisation of simplified $\JSL$-dfas.

\begin{lemma}[Simplified $\JSL$-dfas]
  \label{lem:simplified_jsl_dfa_char}
  For any $\JSL$-dfa $\gamma$ t.f.a.e.\
  \begin{enumerate}
    \item 
    $\gamma$ is simplified i.e.\ $\jslDfaSimple{\gamma} = \gamma$.

    \item
    There exists a finite set of regular languages $S \ni L(\gamma)$, closed under unions and left-letter quotients s.t.\
    \[
      \gamma
      = (L(\gamma), (S, \cup, \emptyset), \lambda X. a^{-1} X, \{ K \in S : \epsilon \in K \}).
    \]
  \end{enumerate}
\end{lemma}

\begin{proof}
  Given (1) then (2) follows by choosing $S := langs(\gamma)$, recalling $L(\gamma_a(z)) = a^{-1} L(\gamma_a)$ by Lemma \ref{def:dfas}.5. Given (2), the specified quadruple is a well-defined $\JSL$-dfa because $a^{-1}(-)$ preserves unions and there is a largest non-final state $\bigcup \{ K \in S : \epsilon \in S \}$.
\end{proof}

\begin{corollary}[$\jslDfaSimple{-}$ is the De Morgan dual of $\jslDfaReach{-}$]
  \label{cor:reach_de_morgan_dual_simple}
  \[
      acc_{(\jslDfaReach{\delta^\pentagram})^\pentagram} : (\jslDfaReach{\delta^\pentagram})^\pentagram \to \jslDfaSimple{\delta}
  \]
  is an isomorphism for any $\JSL$-dfa $\delta$.
\end{corollary}
\begin{proof}
  Given $\delta := (s_0, \aS, \delta_a, F)$ there is an injective $\JSL$-dfa morphism $\iota : \jslReach{\delta^\pentagram} \hookto \delta^\pentagram$ by Lemma \ref{lem:jsl_reach_simple_constructions}. By Theorem \ref{thm:dfa_jsl_self_dual} we have:
  \[
    \delta
    \xto{\lambda_\delta}
    (\delta^\pentagram)^\pentagram
    \overset{\iota_*}{\epito}
    (\jslReach{\delta^\pentagram})^\pentagram
  \]
  where the identity function $\lambda_\delta := id_\aS$ is a component of the natural isomorphism witnessing self-duality, and $\iota_*$ is surjective by Note \ref{note:jsl_extras}.4. By Theorem \ref{thm:jsl_reach_dual_simple}.3 the codomain is simple, so the surjective morphism $acc_{(\jslReach{\delta^\pentagram})^\pentagram}$ is an isomorphism. $langs((\jslReach{\delta^\pentagram})^\pentagram) = langs(\delta)$ because $\iota_*$ is surjective, hence  $acc_{(\jslReach{\delta^\pentagram})^\pentagram}$ has codomain $\jslLangs{\delta}$.
\end{proof}

\begin{example}[Dualising the reachable subset construction]
  \label{ex:dualising_reachable_subset_construction}
  \item
  In Example \ref{ex:dual_full_subset_construction} we described the dual of the full subset construction. Again letting $\delta = \Det(\rN, \Delta_Z, \rev{\rN})$, we now provide a description of $\gamma^\pentagram$ where $\gamma := \jslDfaReach{\delta}$. By Corollary \ref{cor:reach_de_morgan_dual_simple} we have the isomorphism:
  \[
    acc_{\gamma^\pentagram} : \gamma^\pentagram \to \jslDfaSimple{\delta^\pentagram}
  \]
  sending unions of reachable subsets to their accepted language via the $\JSL$-dfa  $\gamma^\pentagram$. We now describe this isomorphism in more detail. First we write $\gamma = (I, \aS, \lambda X. \rN_a[X], \{ X \in S : X \cap F \neq \emptyset \})$, so that:
  \[
    \gamma^\pentagram
    = (reach(\rN) \cap \overline{F}, \aS^\pOp, \beta_a, \{ X \in S : I \nsubseteq X \})
    \qquad
    \text{where $\beta_a := (\gamma_a)_*$.}
  \]
  Then $\beta_w = (\gamma_{w^r})_* = \lambda Y. \bigcup \{ X \in \rs{\rN} : \rN_{w^r}[X] \subseteq Y \}$ since the reachable subsets $\rs{\rN}$ join-generate $\gamma$. Next,
  \[
  \begin{tabular}{lll}
    $w \in acc_{\gamma^\pentagram}(Y)$
    &
    $\iff \beta_w(Y) \in F_{\gamma^\pentagram}$
    & (by def.)
    \\ &
    $\iff I \nsubseteq \beta_w(Y)$
    \\ &
    $\iff \neg (I \subseteq \bigcup \{ X \in \rs{\rN} : \rN_{w^r}[X] \subseteq Y \})$
    & (see above)
    \\ &
    $\iff \neg (\rN_{w^r}[I] \subseteq Y)$
    & (see below)
    \\ &
    $\iff \rN_{w^r}[I] \nsubseteq Y$.
  \end{tabular}  
  \]
  Concerning the marked equivalence, $(\oT)$ follows immediately because $I \in \rs{\rN}$. Conversely if for each $z \in I$ we have $z \in X_z \in \rs{\rN}$ with $\rN_{w^r}[X_z] \subseteq Y$ then $\rN_{w^r}[I] \subseteq \rN_{w^r}[\bigcup_{z \in Z} X_z] \subseteq Y$ too. Thus we obtain a more explicit description of the isomorphism i.e.\ $acc_{\gamma^\pentagram}(Y) = \{ w \in \Sigma^* : \rN_{w^r}[I] \nsubseteq Y \}$.
  \endbox
\end{example}

\begin{corollary}
  \label{cor:reach_simple_mutual_preservation}
  $\jslDfaReach{-}$ preserves simplicity and $\jslDfaSimple{-}$ preserves $\JSL$-reachability.
\end{corollary}

\begin{proof}
  If $\delta$ is simple then it is isomorphic to $\gamma := \jslDfaSimple{\delta}$ so distinct states accept distinct languages by Theorem \ref{thm:jsl_reach_dual_simple}.2. Ignoring the join-structure, $\jslDfaReach{\gamma}$ is a sub-dfa of $\gamma$, so distinct states continue to accept distinct languages and reapplying Theorem \ref{thm:jsl_reach_dual_simple}.2 we deduce simplicity. The second statement follows by duality i.e.\ Theorem \ref{thm:jsl_reach_dual_simple}.3.
\end{proof}

\begin{corollary}[Characterisation of $\dfa{\JSL}$-minimality]
  \label{cor:min_jsl_dfa_characterisation}
  A $\JSL$-dfa $\delta$ is $\JSL$-reachable and simple iff $acc_\delta : \delta \to \jslDfaMin{L(\delta)}$ is a well-defined $\JSL$-dfa isomorphism.
\end{corollary}

\begin{proof}
  Let $L := L(\delta)$. Suppose $acc_\delta : \delta \to \jslDfaMin{L}$ has correct typing and is an isomorphism. The notions of `$\JSL$-reachable' and `simple' are invariant under isomorphism, so we show $\gamma := \jslDfaMin{L}$ is $\JSL$-reachable and simple. Firstly, $\jslDfaReach{\gamma} = \gamma$ because the left quotients $\LQ{L}$ arise from $\LW{L}$ via finite unions. Finally, $\gamma$ is simple because $\gamma^\pentagram \cong \jslDfaMin{L^r}$ by Theorem \ref{thm:dr_L} which is $\JSL$-reachable by the preceding argument, so $\gamma$ is simple by Theorem \ref{thm:jsl_reach_dual_simple}.3.
   
  Conversely let $\delta$ be $\JSL$-reachable and simple. Since $\delta$ is simple, the surjection $acc_\delta : \delta \to \jslDfaSimple{\delta}$ is an isomorphism. Since $\delta$ is $\JSL$-reachable, by Theorem \ref{thm:jsl_reach_dual_simple}.1 and Lemma \ref{lem:jsl_dfa_joins} we have $langs(\delta) = \LQ{L}$, hence $\jslDfaSimple{\delta} = \jslDfaMin{L(\delta)}$.
\end{proof}

\begin{corollary}[Meet-generators for $\JSL$-dfas]
  \label{cor:meet_gen_jsl_dfas}
  Let $\gamma = (s_0, \aS, \gamma_a, F)$ be a $\JSL$-dfa.
  \begin{enumerate}
    \item
    If $\gamma$ is simplified it is meet-generated by $\{ \bigcup \{ j \in J(\jslLangs{\gamma}) : w \nin j \} : w \in \Sigma^* \}$.    
    \item
    If $\gamma$ is $\JSL$-reachable it is meet-generated by $\{ \Lor_\aS \{ \gamma_w (s_0) : w \nin j \} :  j \in J(\jslLangs{\gamma^\pentagram}) \}$.
  \end{enumerate}
\end{corollary}

\begin{proof}
\item
\begin{enumerate}
  \item
  By Lemma \ref{lem:simplified_jsl_dfa_char} there exists a set $S$ of union and left-letter-quotient closed languages s.t.\
  \[
    \gamma = (L, \overbrace{(S, \cup, \emptyset)}^\aS, \lambda X. a^{-1} X, \{ K \in S : \epsilon \in K \})
    \qquad
    \text{where $L := L(\gamma)$.}
  \]
  By Theorem \ref{thm:jsl_reach_dual_simple} we know $\gamma^\pentagram$ is $\JSL$-reachable, hence join-generated by elements $(\gamma_w)_*(K_1)$ where $K_1 := \bigcup \{ K \in S : \epsilon \nin K \}$. Finally observe that:
  \[
    \begin{tabular}{lll}
      $(\gamma_w)_*(K_0)$
      &
      $= \bigcup \{ j \in J(\aS) : \gamma_w(j) \subseteq K_0 \}$
      \\ &
      $= \bigcup \{ j \in J(\aS) : \epsilon \nin \gamma_w(j) \}$
      \\ &
      $= \bigcup \{ j \in J(\aS) : w \nin j \}$
    \end{tabular}
  \]
  so $\aS$ is meet-generated by these elements.

  \item
  By Theorem \ref{thm:jsl_reach_dual_simple} we may assume (modulo isomorphism) that $\gamma = \delta^\pentagram$ where $\delta = (t_0, \aT, \lambda X. a^{-1} X, \{ K \in T : \epsilon \in K \})$ is simplified. Consequently $s_0 = \bigcup \{ K \in T : \epsilon \nin K \}$ and we calculate:
  \[
    \begin{tabular}{lll}
      $\Lor_\aS \{ \gamma_{w} (s_0) : w \nin j \}$
      &
      $= \Land_\aT \{ (\delta_{w^r})_* (s_0) : w \nin j \}$
      \\ &
      $= \Land_\aT \{ (\bigcup \{ j' \in J(\aT) : w^r \nin j' \} : w \nin j \}$
      & (see proof of (1))
      \\ &
      $= \Land_\aT \{ (\bigcup \{ j' \in J(\aT) : w \nin j' \} : w \nin j \}$
      \\ &
      $= j$
      & (see below).
    \end{tabular}
  \]
  Concerning the marked equality: $\subseteq$ follows because if $w \nin j$ then $w \nin \{ j' \in J(\aT) : w \nin j' \}$; $\supseteq$ follows because whenever $w \nin j$ we know $j \subseteq \{ j' \in J(\aT) : w \nin j' \}$. Finally, $J(\aT)$ join-generates $\aT$ and thus meet-generates $\aS$.
\end{enumerate}
\end{proof}

The self-duality of $\dfa{\JSL}$ (Theorem \ref{thm:dfa_jsl_self_dual}) corresponds to the self-duality of $\AutDep$ (Theorem \ref{thm:aut_dep_self_dual}). But what does $\JSL$-reachability correspond to at the level of dependency automata? Our next result shows it is a combination of the classical reachable subset construction and the classical reachable nfa construction.

\begin{theorem}[$\AutDep$-reachability]
  \label{thm:reachable_aut_dep}
  We have the $\AutDep$-isomorphism:
  \[
    \breve{\in} :
    \Airr(\jslDfaReach{\dep{\rN}})
    \to
    (\rsc{\rN}, \breve{\in},\rev{\reach{\rN}}).
  \]
  for each nfa $\rN = (z_0, Z, \rN_a, F)$.
\end{theorem}

\begin{proof}
 By Note \ref{note:full_subset_construction} the determinisation $\delta := \Det(\dep{\rN})$ is $\rN$'s full subset construction endowed with its join-semilattice structure $\Open\Delta_Z = (\Pow Z, \cup, \emptyset)$. Consider:
 \[
   \jslDfaReach{\delta} = (I, \aS, \gamma_a, F_\gamma)
   \qquad
   \aS := \jslReach{\delta} = \ang{\rs{\rN}}_{\Open\Delta_Z}
   \qquad
   \iota : \aS \hookto \Open\Delta_Z.
 \]
 Then $\aS$ is join-generated by $J_\aS := \rs{\rN}$ but what about a meet-generating set? The surjective adjoint $\iota_* : (\Pow Z, \cap, Z) \epito \aS^\pOp$ provides one:
 \[
  M(\aS) = J(\aS^\pOp) \subseteq \iota_*[J(\Pow Z, \cap, Z)] = \{ M_z : z \in Z \}
  \qquad
  \text{where $M_z := \iota_*(\overline{z})$.}
 \]
 Since $\rs{\rN}$ join-generates $\aS$ we know $M_z = \bigcup \{ X \in \rs{\rN} : z \nin X \}$ is the union of reachable subsets without $z$. It follows that $M_\aS := \{ M_z : z \in reach(\rN) \} \supseteq M(\aS)$ because if $z$ is unreachable then $M_z = reach(\rN) = \top_\aS \nin M(\aS)$ cannot contribute. Now, $\Airr\delta$ is isomorphic to $(\rM, \nsubseteq, \rM')$ by Proposition \ref{prop:aut_dep_generator_isoms} where:
\[
  \begin{tabular}{rl}
    $\overbrace{(\{ X \in \rs{\rN} : X \subseteq I \}, \rs{\rN}, \rM_a, \{ X \in \rs{\rN} : X \cap F \neq \emptyset \})}^\rM$
    &
    $\rM_a(X_1, X_2) :\iff  X_2 \subseteq \rN_a[X_1]$
    \\[1ex]
    $\underbrace{(\{ M_z : z \in Z,\, \overline{F} \cap S \subseteq M_z \}, M_\aS, \rM'_a, \{ M_z : z \in Z,\, I \nsubseteq M_z \} )}_{\rM'}$
    &
    $\rM'_a(M_{z_1}, M_{z_2}) :\iff (\delta_a)_*(M_{z_1}) \subseteq M_{z_2}$.
  \end{tabular}
\]
The lower nfa $\rM$ turns out to be $\rsc{\rN}$ with some additional degenerate structure, we'll come back to this point. Concerning the upper nfa, the calculations:
\[
  \begin{tabular}{c}
    $\overline{F} \cap S \subseteq M_z
    \overset{(\mathrm{adjoints})}{\iff} \iota(\overline{F} \cap S) \subseteq \overline{z}
    \iff z \nin \overline{F} \cap S
    \iff z \in F \cap reach(\rN)$
    \\[1ex]
    $I \nsubseteq M_z
    \overset{(\mathrm{adjoints})}{\iff} \iota(I) \nsubseteq \overline{z}
    \iff z \in I$.
  \end{tabular}
\]
\smallskip
\[
  \begin{tabular}{lll}
    $(\delta_a)_*(M_{z_1}) \subseteq M_{z_2}$
    &
    $\iff \iota((\delta_a)_*(M_{z_1})) \subseteq \overline{z_2}$
    & (adjoints)
    \\ &
    $\iff z_2 \nin (\delta_a)_*(M_{z_1})$
    \\ &
    $\iff \forall X \in \rs{\rN}.[ \gamma_a(X) \subseteq M_{z_1} \To z_2 \nin X ]$
    \\ &
    $\iff \forall X \in \rs{\rN}.[ z_2 \in X \To \gamma_a(X) \nsubseteq M_{z_1}]$
    \\ &
    $\iff \forall X \in \rs{\rN}.[ z_2 \in X \To z_1 \in \gamma_a(X)]$
    & (via adjoints)
    \\ &
    $\iff \rN_a(z_2, z_1)$
    & (since $z_2$ reachable)
  \end{tabular}
\]
show that it is essentially $\rev{\reach{\rN}}$. More precisely we have the bipartite graph isomorphism:
\[
  \xymatrix@=15pt{
    M_\aS \ar[rr]^{\beta} && reach(\rN)
    \\
    \rs{\rN} \ar[u]^{\nsubseteq} \ar[rr]_{id_{\rs{\rN}}} && \rs{\rN} \ar[u]_{\breve{\in}}
  }
\]
where the bijection $\beta$ has action $M_z \mapsto z$. Indeed $X \nsubseteq M_z \iff X \nsubseteq \overline{z} \iff z \in X \iff \breve{\in}(X, z)$. Then it follows from the earlier calculations that we have the $\AutDep$-isomorphism
$
  \breve{\in} : \Airr\delta \to (\rM, \breve{\in}, \rev{\reach{\rN}})
$. Instantiating Proposition \ref{prop:aut_dep_transition_isoms} provides the isomorphism 
$
  \breve{\in} : (\rM, \breve{\in}, \rev{\reach{\rN}}) \to (\rsc{\rN}, \breve{\in}, \rev{\reach{\rN}}) 
$. This follows by the calculations 
$
  \rM_a ; \breve{\in} (X, z)
  \iff \exists X' \in \rs{\rN}.[ X' \subseteq \rN_a[X] \;\land\; z \in X' ]
  \iff z \in \rN_a[X]
  \iff (\lambda Y.\rN_a[Y]) ; \breve{\in} (X, z)
$ and 
$
  \breve{\in}[\{ X \in \rs{\rN} : X \subseteq I \}]
  = I
  = \breve{\in}[\{ I \}]
$. The third requirement in Proposition \ref{prop:aut_dep_transition_isoms} is trivial because both dependency automata have the same upper nfa. Composing these two $\AutDep$-isomorphisms yields:
\[
  \breve{\in} \fatsemi \breve{\in}
  = id_{\rs{\rN}} ; \breve{\in}
  = \breve{\in}
  : \Airr(\jslDfaReach{\Det (\rN, \Delta_Z, \rev{\rN})}) \to (\rsc{\rN}, \breve{\in}, \rev{\reach{\rN}})
\]
i.e.\ relate a \emph{join-irreducible reachable subset} $Y$ to its elements $z \in Y$ -- all classically reachable in the nfa $\rN$.
\end{proof}

\smallskip

\begin{note}[Reachability in $\AutDep$]
  Given the full subset construction $\delta = \Det(\rN, \Delta_Z, \rev{\rN)}$, Theorem \ref{thm:reachable_aut_dep} describes $\jslDfaReach{\delta}$ as a dependency automaton. What about for arbitrary $\JSL$-dfas? In a sense we've already covered the general case via Corollary \ref{cor:aut_dep_enough_inj}. The $\JSL$-dfas with carrier $\Open\Delta_Z = (\Pow Z, \cup, \emptyset)$ are injective objects and every $\JSL$-dfa embeds into one.
  \endbox
\end{note}

\begin{theorem}[$\AutDep$-simplicity]
  \label{thm:simple_aut_dep}
  We have the $\AutDep$-isomorphism:
  \[
    \begin{tabular}{c}
      $\rI : (\coreach{\rN}, \in, \rsc{\rev{\rN}})
      \to
      \Airr(\jslDfaSimple{\dep{\rN}})$
      \\[1ex]
      $\rI (z, Y) :\iff L(\rN_{@z}) \cap \overline{Y}^r \neq \emptyset$.
    \end{tabular}
  \]
  for any nfa $\rN = (I, Z, \rN_a, F)$.
\end{theorem}

\begin{proof}
  Let $\delta := \Det(\dep{\rN})$ and apply the duality of Theorem \ref{thm:aut_dep_self_dual} to the isomorphism of Theorem \ref{thm:reachable_aut_dep}:
  \[
    \rR := \; \in \;:
    (\rev{\reach{\rN}}, \in, \rsc{\rN})
    \to
    \Rev(\Airr(\jslDfaReach{\delta})).
  \]  
  Observe $\rR$ has bijective lower witness $\alpha := \lambda z. \bigcup \{ X \in \rs{\rN} : z \nin X \}$ by inspecting the proof of Theorem \ref{thm:reachable_aut_dep}.
  By Theorem \ref{thm:aut_dep_jsl_dfa_self_duality_transfer}.2 we have $\hat{\lambda} : \Rev \circ \Airr^{op} \To \Airr \circ (-)^\pentagram$ and hence the component:
  \[
    \hat{\lambda}_{\jslDfaReach{\delta}}
    = id_{\Airr(\jslDfaReach{\delta})}
    = \Pirr\ang{\rs{\rN}}_{\Open\Delta_Z}
  \]
  whose domain is the codomain of $\rR$ and whose codomain is $\Airr(\jslDfaReach{\delta})^\pentagram$. Corollary \ref{cor:reach_de_morgan_dual_simple} provides the isomorphism:
  \[
    f := acc_{(\jslDfaReach{\delta})^\pentagram} : 
    (\jslDfaReach{\delta})^\pentagram
    \to
    \jslDfaSimple{\delta^\pentagram}
  \]
  and thus the $\AutDep$-isomorphism $\Airr f$. By Example \ref{ex:dual_full_subset_construction} we know $\delta^\pentagram \cong \Det(\Rev(\dep{\rN}))$ hence $\jslDfaSimple{\delta^\pentagram}$ exactly equals $\jslDfaSimple{\Det(\Rev(\dep{\rN}))}$. Composing these three $\Dep$-isomorphisms yields:
  \[
    \begin{tabular}{lll}
      $\rR  \fatsemi \Airr f (z, Y)$
      &
      $\iff \alpha ; \Airr f (z, Y)$
      & (using $\rR$'s lower witness)
      \\ &
      $\iff f(\bigcup \{ X \in \rs{\rN} : z \nin X \}) \nleq_{\jslLangs{\delta^\pentagram}} Y$
      \\ &
      $\iff f(\bigcup \{ X \in \rs{\rN} : z \nin X \}) \nsubseteq Y$
      \\ &
      $\iff \{ w \in \Sigma^* : \rN_{w^r}[I] \nsubseteq \bigcup \{ X \in \rs{\rN} : z \nin X \} \} \nsubseteq Y$
      & (by Example \ref{ex:dualising_reachable_subset_construction})
      \\ &
      $\iff \{ w \in \Sigma^* : z \in \rN_{w^r}[I] \} \nsubseteq Y$
      \\ &
      $\iff \{ w \in \Sigma^* : w^r \in L(\rev{\rN}_{@z}) \} \nsubseteq Y$
      \\ &
      $\iff L(\rev{\rN}_{@z}) \cap \overline{Y}^r \neq \emptyset$.
    \end{tabular}
  \]
  Finally we reparameterise via $\rN \mapsto \rev{\rN}$ recalling that $\coreach{\rN} = \rev{\reach{\rev{\rN}}}$ by definition.
\end{proof}

We can now explain the original motivation for the above results.

\begin{theorem}[Brzozowski construction of state-minimal dfa]
  We have the dfa-isomorphism:
  \[
    acc_{\rsc{\rev{\rsc{\rev{\rN}}}}} : \rsc{\rev{\rsc{\rev{\rN}}}} \to \minDfa{L(\rN)}
  \]
  for any nfa $\rN = (I, Z, \rN_a, F)$.
\end{theorem}

\begin{proof}
  Consider the dependency automaton $(\rN, \Delta_Z, \rev{\rN})$. By Theorem \ref{thm:simple_aut_dep} its simplification amounts to $\fN := (\coreach{{\rN}}, \in, \rsc{\rev{\rN}})$. Then $\Det\fN$ is a simple $\JSL$-dfa. By Corollary \ref{cor:reach_simple_mutual_preservation}, $\jslDfaReach{\Det\fN}$ is both simple and $\JSL$-reachable, hence isomorphic to $\jslDfaMin{L(\rN)}$ by Corollary \ref{cor:min_jsl_dfa_characterisation}. Thus the classically reachable part $\reach{\Det\fN}$ is isomorphic to $\minDfa{L(\rN)}$. Finally by Note \ref{note:det_airr_preserve_acceptance} we have $\reach{\Det\fN} = \rsc{\rev{\rsc{\rev{\rN}}}}$.
\end{proof}

\subsection{Minimal boolean and distributive machines}
\label{subsec:more_canonical_jsl_dfas}

Recall the state-minimal machine $\minDfa{L}$ from Definition \ref{def:dfas}. Its states $\LW{L}$ are the left word quotients $u^{-1} L$, also known as Brzozowski derivatives \cite{BrzozowskiDRE1964,Conway71}.

\begin{definition}[Minimal boolean/distributive $\JSL$-dfa]
  \label{def:canon_bool_jsl_dfa}
  Fix a regular language $L \subseteq \Sigma^*$.
  \begin{enumerate}
    \item 
    \emph{$L$'s left predicates and state-minimal boolean $\JSL$-dfa}.
    
    $\LP{L}$ are all set-theoretic boolean combinations of $L$'s left word quotients $\LW{L}$. They admit a boolean algebra structure, with underlying join-semilattice $\jslLP{L} := (\LP{L}, \cup, \emptyset)$. Then $J(\jslLP{L})$ are its atoms and $M(\jslLP{L})$ its co-atoms. The \emph{canonical boolean $\JSL$-dfa for $L$} is defined:
    \[
      \jslDfaBoolMin{L} := (L, \jslLP{L}, \lambda X. a^{-1} X, \{ K \in \LP{L} : \epsilon \in K  \}).
    \]
    
    \item
    \emph{$L$'s positive left predicates and state-minimal distributive $\JSL$-dfa}.

    Let $\LD{L}$ be the closure of $\LW{L}$ under all intersections and unions. The subsets define a distributive lattice with underlying join-semilattice $\jslLD{L} := (\LD{L}, \cup, \emptyset)$. Meet is intersection and its top element is $\Sigma^*$. The \emph{canonical distributive $\JSL$-dfa for $L$} is defined:
    \[
      \jslDfaDistMin{L} := (L, \jslLD{L}, \lambda X. a^{-1} X, \{ K \in \LD{L} : \epsilon \in K  \}).
    \]
    
    \endbox
  \end{enumerate}
\end{definition}

\begin{note}[Canonicity of $\JSL$-dfas]
  \item
  We briefly explain the sense in which these $\JSL$-dfas are canonical, see \cite{CanonNA2014}.
  \begin{itemize}
    \item[--]
    $\jslDfaBoolMin{L}$ is the underlying $\JSL$-dfa of the state-minimal $\BA$-dfa.
    \item[--]
    $\jslDfaDistMin{L}$ is the underlying $\JSL$-dfa of the state-minimal $\DL$-dfa.
    \endbox
  \end{itemize}
  \noindent
\end{note}

In the remainder of this subsection, we'll describe the canonical boolean/distributive $\JSL$-dfas as dependency automata. This immediately provides representations of their dual $\JSL$-dfas. The next subsection is dedicated to the transition-semiring of an nfa. These admit a $\JSL$-dfa structure. In particular, the canonical syntactic $\JSL$-dfa $\jslDfaSynMin{L}$ is the dual of syntactic semiring for $L^r$ \cite{Polak2001}.


\begin{lemma}[Concerning atoms and finality]
  \label{lem:left_preds_basics}
  \item
  \begin{enumerate}
    \item
    The atoms $J(\jslLP{L})$ are pairwise disjoint and their union is $\Sigma^*$.
    \item
    Given $u \in \alpha \in J(\jslLP{L})$ and $Y \in \LP{L}$ then $u \in Y \iff \alpha \subseteq Y$.
    \item
    For any $Y \in \LP{L}$ we have $\epsilon \in Y \iff Y \nsubseteq dr_L(L^r)$.
  \end{enumerate}
\end{lemma}

\begin{proof}
  \item
  \begin{enumerate}
    \item 
    The atoms $J(\jslLP{L})$ are pairwise-disjoint because their meet (intersection) is the bottom element $\emptyset$. The union of all atoms is the top element i.e.\ the empty intersection $\Sigma^*$.
    \item  
    Fix any $Y \in \LP{L}$ and $u \in \alpha \in J(\jslLP{L})$. Given $\alpha \subseteq Y$ then certainly $u \in Y$. Conversely if $u \in Y$ then it must lie in some atom, which is unique by disjointness, hence $\alpha \subseteq Y$.
    \item 
    If $\epsilon \in Y$ then certainly $Y \nsubseteq dr_L(L^r)$ because the latter does not contain $\epsilon$ (see Theorem \ref{thm:dr_L}). Conversely if $Y \nsubseteq dr_L(L^r)$ there exists $u \in Y$ such that $\forall X \in \LW{L}.( u \in X \iff \epsilon \in X)$ i.e.\ $\alpha$ and $\epsilon$ reside in the same atom, so $\epsilon \in Y$ by (2).
  \end{enumerate}
\end{proof}

\begin{lemma}[Well-definedness of canonical boolean/distributive $\JSL$-dfa]
  \label{lem:canon_bool_jsl_dfa}
  \item
  $\jslDfaBoolMin{L}$ and $\jslDfaDistMin{L}$ are well-defined fixpoints of $\jslDfaSimple{-}$ which accept $L$ and have $\jslDfaMin{L}$ as a sub $\JSL$-dfa.
\end{lemma}

\begin{proof}
  The join-semilattice $\jslLP{L}$ is closed under unions, thus well-defined. We have $L \in \LQ{L} \subseteq \LP{L}$ and $a^{-1}(-)$ preserves unions. The final states are well-defined by Lemma \ref{lem:left_preds_basics}.3. Each state $K$ accepts $K$ i.e.\ $w \in K \iff \epsilon \in w^{-1} K \iff w^{-1} K \nsubseteq dr_L(L^r)$ where the latter corresponds to $\JSL$-dfa acceptance. Then it is a fixpoint of $\jslDfaSimple{-}$ as claimed and clearly has the sub $\JSL$-dfa $\jslDfaMin{L}$. Finally $\jslDfaDistMin{L}$ is sandwiched between them via $\JSL$-dfa inclusion morphisms, with well-defined final states by Lemma \ref{lem:left_preds_basics}.3.
\end{proof}

Generally speaking, $L^r$'s left word quotients biject with $\jslLP{L}$'s atoms.

\begin{theorem}[Quotient-atom bijection]
  \label{thm:atom_quotient_bijection}
  Each regular $L$ has the canonical bijection:
  \[
    \kappa_L : \LW{L^r} \to J(\jslLP{L})
    \qquad
    \kappa_L (v^{-1} L^r) := \sem{v^r}_{\rE_L}
    \qquad
    \kappa_L^{\bf-1} (X) := [X^r]^{-1} L^r,
  \]
  and respective relationship:
  \[
    \kappa_L(x^{-1} L^r) \subseteq a^{-1} \kappa_L(y^{-1} L^r)
    \iff (xa)^{-1} L^r = y^{-1} L^r
    \qquad
    \text{for any $x, y \in \Sigma^*$.}
  \]
\end{theorem}

\begin{proof}
  \item
  \begin{enumerate}
    \item 
    We first verify $\kappa_L$ is a well-defined function:
    \[
      \begin{tabular}{lll}
        $v_1^{-1} L^r = v_2^{-1} L^r$
        &
        $\iff \forall w \in \Sigma^*. [ v_1 w \in L^r \iff v_2 w \in L^r ]$
        \\ &
        $\iff \forall w \in \Sigma^*. [ w v_1^r \in L \iff w v_2^r \in L ]$
        \\ &
        $\iff \forall w \in \Sigma^*. [ v_1^r \in w^{-1} L \iff v_2^r \in w^{-1} L ]$
        \\ &
        $\iff \sem{v_1^r}_{\rE_L} = \sem{v_2^r}_{\rE_L}$
        & (by Lemma \ref{lem:cl_L_basics}.6).
      \end{tabular}
    \]
    It is clearly surjective and also injective by reversing the argument above.
    The action of $\kappa_L^{\bf-1}$ is well-defined because $\kappa_L$ is injective. 

    \item
    Suppose $(xa)^{-1} L^r = y^{-1} L^r$ so that $\sem{ax^r}_{\rE_L} = \sem{y^r}_{\rE_L}$ by applying $\kappa_L$. Since $x^r \in a^{-1} \sem{ax^r}_{\rE_L}$ we deduce $\sem{x^r}_{\rE_L} \subseteq a^{-1} \sem{ax^r}_{\rE_L} = a^{-1} \sem{y^r}_{\rE_L}$. Conversely suppose the inclusion $\sem{x^r}_{\rE_L} \subseteq a^{-1} \sem{y^r}_{\rE_L}$ holds. Then $ax^r \in \sem{y^r}_{\rE_L}$ and consequently $\sem{ax^r}_{\rE_L} = \sem{y^r}_{\rE_L}$, so applying $\kappa_L^{\bf-1}$ we infer $(xa)^{-1} L^r = y^{-1} L^r$.
  \end{enumerate}
\end{proof}

\begin{note}[Canonicity of $\kappa_L$]
  $\kappa_L$ arises from the duality between $\Set$-dfas (classical dfas) and $\BA$-dfas i.e.\ finite deterministic automata interpreted in boolean algebras \cite{CanonNA2014}. In particular, the dual of the state-minimal $\BA$-dfa for $L$ is isomorphic to the state-minimal $\Set$-dfa for $L^r$.
  \endbox
\end{note}

\begin{theorem}[Canonical boolean dependency automaton]
  \label{thm:canonical_bool_dep_aut}
  We have the $\AutDep$-isomorphism:
  \[
    \neg_{\jslLP{L}} \circ \kappa_L
    : \dep{\rev{\minDfa{L^r}}} \to \Airr(\jslDfaBoolMin{L})
  \]
  with action $\lambda v^{-1} L^r. \overline{\sem{v^r}_{\rE_L}}$ and inverse $\kappa_L^{\bf-1}$.
\end{theorem}

\begin{proof}
  Consider the dependency automaton of irreducibles:
  \[
    \Airr(\jslDfaBoolMin{L}) = (\rN, \neg_{\jslDfaBoolMin{L}}, \rM)
  \]
  \[
    \begin{tabular}{ll}
      $\rN
      = (\{ \sem{x}_{\rE_L} : x  \in L \}, J(\jslDfaBoolMin{L}), \rN_a, \{ \sem{\epsilon}_{\rE_L} \})$
      &
      $\rN_a(\sem{x_1}_{\rE_L}, \sem{x_2}_{\rE_L})
      \iff (x_2^r a)^{-1} L^r = (x_1^r)^{-1} L^r$
      \\[1ex]
      $\rM = (\{ \overline{\sem{\epsilon}_{\rE_L}} \}, M(\jslDfaBoolMin{L}), \rM_a, \{ \overline{\sem{x}_{\rE_L}} : x \in L \})$
      &
      $\rM_a(\overline{\sem{x_1}_{\rE_L}}, \overline{\sem{x_2}_{\rE_L}})
      \iff \rN_a(\sem{x_2}_{\rE_L}, \sem{x_1}_{\rE_L})$.
    \end{tabular}
  \]
  To explain, $\rN$'s description follows by unwinding the definitions and the relationship $\sem{x_2}_{\rE_L} \subseteq a^{-1} \sem{x_1}_{\rE_L} \iff (x_2^r a)^{-1} L^r = (x_1^r)^{-1} L^r$ from Theorem \ref{thm:atom_quotient_bijection}. Likewise $\rM$ follows via the definitions and the following calculation, where $\gamma_a := \lambda X. a^{-1} X : \jslLP{L} \to \jslLP{L}$:
  \[
    \begin{tabular}{lll}
      $\rM_a(\overline{\sem{x_1}_{\rE_L}}, \overline{\sem{x_2}_{\rE_L}})$
      &
      $\iff (\gamma_a)_*(\overline{\sem{x_1}_{\rE_L}}) \subseteq \overline{\sem{x_2}_{\rE_L}}$
      & (by definition)
      \\ &
      $\iff \bigcup \{ \sem{x}_{\rE_L} : a^{-1} \sem{x}_{\rE_L} \subseteq \overline{\sem{x_1}_{\rE_L}} \} \subseteq \overline{\sem{x_2}_{\rE_L}}$
      \\ &
      $\iff x_2 \nin \bigcup \{ \sem{x}_{\rE_L} : a^{-1} \sem{x}_{\rE_L} \subseteq \overline{\sem{x_1}_{\rE_L}} \}$
      \\ &
      $\iff a^{-1} \sem{x_2}_{\rE_L} \nsubseteq \overline{\sem{x_1}_{\rE_L}}$
      \\ &
      $\iff \sem{x_1}_{\rE_L} \subseteq a^{-1} \sem{x_2}_{\rE_L}$.
    \end{tabular}
  \]
  We now verify the claimed dependency automaton isomorphism using the canonical bijection $\kappa_L$ from Theorem \ref{thm:atom_quotient_bijection}. First of all, $\alpha := \neg_{\jslLP{L}} \circ \kappa_L$ defines a $\Dep$-isomorphism via the bijective witnesses:
  \[
    \xymatrix@=20pt{
      \LW{L^r} \ar[rr]^{\neg_{\jslLP{L}} \circ \kappa_L} && M(\jslLP{L})
      \\
      \LW{L^r} \ar[u]^{\Delta_{\LW{L^r}}} \ar[rr]_{\kappa_L} && J(\jslLP{L}) \ar[u]_{\neg_{\jslLP{L}}}
    }
  \]
  It remains to verify the constaints from Definition \ref{def:cat_aut_dep}. Let $\delta := \minDfa{L^r}$ be the classical state-minimal dfa so that $\delta_a(Y_1, Y_2) \iff Y_2 = a^{-1} Y_1$. Then we calculate:
  \[
    \begin{tabular}{lll}
      $\delta_a\spbreve ; \alpha (v^{-1} L^r, \overline{\sem{y}_{\rE_L}})$
      &
      $\iff \exists w \in \Sigma^*. [ (wa)^{-1} L^r = v^{-1} L^r \,\land\, \alpha(w^{-1} L^r) = \overline{\sem{y}_{\rE_L}}]$
      \\ &
      $\iff \exists w \in \Sigma^*. [ (wa)^{-1} L^r = v^{-1} L^r \,\land\, \sem{w^r}_{\rE_L} = \sem{y}_{\rE_L} ]$
      \\ &
      $\iff \exists w \in \Sigma^*. [ \sem{w^r}_{\rE_L} \subseteq a^{-1} \sem{v^r}_{\rE_L} \,\land\, \sem{w^r}_{\rE_L} = \sem{y}_{\rE_L} ]$
      & (by Theorem \ref{thm:atom_quotient_bijection})
      \\ &
      $\iff \sem{y}_{\rE_L} \subseteq a^{-1} \sem{v^r}_{\rE_L}$
      \\ &
      $\iff \rM_a(\overline{\sem{y}_{\rE_L}}, \overline{\sem{v^r}_{\rE_L}})$
      & (see above)
      \\ &
      $\iff \alpha ; \rM_a\spbreve(v^{-1} L^r, \overline{\sem{y}_{\rE_L}})$
    \end{tabular}
  \]
  Concerning the remaining conditions, $\breve{\alpha}[I_\rM] = \{ \kappa_L^{\bf-1}(\sem{\epsilon}_{\rE_L}) \} = \{ L^r \} = F_{\rev{\delta}}$ and finally:
  \[
    \begin{tabular}{lll}
      $\alpha[I_{\rev{\delta}}]$
      &
      $= \alpha[\{ Y \in \LW{L^r} : \epsilon \in Y \}]$
      \\ &
      $= \alpha[\{ v^{-1} L^r : v \in L^r \}]$
      \\ &
      $= \{ \overline{\sem{v^r}_{\rE_L}} : v^r \in L \}$
      \\ &
      $= F_\rM$.
    \end{tabular}
  \]
\end{proof}

Importantly this provides a dual representation.

\begin{corollary}[Dualising $\jslDfaBoolMin{L}$]
  \label{cor:dual_jsl_dfa_of_left_preds}
  We have the $\JSL$-dfa isomorphism:
  \[
    \theta : \jslDfaSc{\minDfa{L^r}} \to (\jslDfaBoolMin{L})^\pentagram
    \qquad
    \theta(S) := \bigcup \{ \sem{v^r}_{\rE_L} : v^{-1} L^r \nin S  \}
  \]
\end{corollary}

\begin{proof}
  First recall the inverse of the isomorphism from Theorem \ref{thm:canonical_bool_dep_aut},
  \[
    f := \lambda X. [X^r]^{-1} L^r : \Airr(\jslDfaBoolMin{L}) \to \fM
    \qquad
    \fM := \dep{\rev{\minDfa{L^r}}}.
  \]
  It has upper witness $\beta := \kappa_L^{\bf-1} \circ \neg_{\jslDfaBoolMin{L}}$ so that $\Det f = \lambda S. \beta[S]$ and also $(\Det f)^\pentagram = \lambda S. \beta^{\bf-1}[S]$, since the adjoint of an isomorphism acts like the inverse. Recall the natural isomorphism $\hat{\partial}_{-}$ from Theorem \ref{thm:aut_dep_jsl_dfa_self_duality_transfer} and $\rep_{-}$ from Theorem \ref{thm:aut_dep_equiv_jsl_dfa}. Then we have the composite join-semilattice isomorphism:
  \[
    \Det (\dep{\minDfa{L^r}})
    \xto{\hat{\partial}_{\fM}^{\bf-1}} (\Det\fM)^\pentagram
    \xto{(\Det f)^\pentagram} (\Det\Airr(\jslDfaBoolMin{L}))^\pentagram
    \xto{\rep_{\jslDfaBoolMin{L}}^\pentagram} (\jslDfaBoolMin{L})^\pentagram
  \]
  which acts on $S \subseteq \LW{L^r}$ as follows:
  \[
    \begin{tabular}{llll}
      $S$
      &
      $\mapsto$ & $\hat{\partial}_{\dep{\rev{\minDfa{L^r}}}}^{\bf-1}(S)$
      \\ &
      $=$ & $\Delta_{\LW{L^r}}[\overline{S}]$
      \\ &
      $\mapsto$ & $(\Det f)^\pentagram(\overline{S})$
      \\ &
      $=$ & $\{ \overline{\sem{v^r}_{\rE_L}} : v^{-1} L^r \in \overline{S} \}$
      & (see above)
      \\ &
      $\mapsto$ & $\rep_{\jslDfaBoolMin{L}}^\pentagram(\{ \overline{\sem{v^r}_{\rE_L}} : v^{-1} L^r \nin S \})$
      \\ &
      $=$ & $\bigcap \{ \overline{\sem{v^r}_{\rE_L}} : v^{-1} L^r \in S  \}$
      & (adjoint acts as $\rep_{\jslDfaBoolMin{L}}^{\bf-1}$)
      \\ &
      $=$ & $\bigcup \{ \sem{v^r}_{\rE_L} : v^{-1} L^r \nin S \}$.
    \end{tabular}
  \]
\end{proof}

We now turn our attention to positive predicates, recalling $dr_L$ from Theorem \ref{thm:dr_L}.

\begin{lemma}[Concerning irreducibles in $\jslLD{L}$]
  \label{lem:pos_left_preds_basics}
  \item
  \begin{enumerate}
    \item
    $(\jslLD{\overline{L}})^\pOp \cong \jslLD{L}$ via relative complement.

    \item
    $\overline{dr_{\overline{L}}(\overline{v^{-1} L^r})} = \bigcap \{ X \in \LW{L} : v^r \in X \}$.

    \item
    $|J(\jslLD{L})| = |M(\jslLD{L})| = |\LW{L^r}|$ where:
    \[
      \begin{tabular}{c}
        $J(\jslLD{L}) = \{ \overline{dr_{\overline{L}}(\overline{v^{-1} L^r})} : v^{-1} L^r \in \LW{L^r} \}$
        \qquad
        $M(\jslLD{L}) = \{ dr_L(v^{-1} L^r) : v^{-1} L^r \in \LW{L^r} \}$.
        \\[1ex]
      \end{tabular}
    \]
    \item
    $S \subseteq dr_L(v^{-1} L^r) \iff v^r \nin S$, for each $S \in \LD{L}$.

    \item
    The canonical bijection $\tau_{\jslLD{L}} : J(\jslLD{L}) \to M(\jslLD{L})$ has action:
    \[
      \tau_{\jslLD{L}}(\overline{dr_{\overline{L}}(\overline{v^{-1} L^r})})
      := dr_L(v^{-1} L^r)
      \qquad
      \text{see Note \ref{note:irreducibles}.}
    \]
  \end{enumerate}
\end{lemma}

\begin{proof}
  \item
  \begin{enumerate}
    \item
    Consider $\theta := \lambda X. \overline{X} : (\jslLD{\overline{L}})^\pOp \to \jslLD{L}$. It is a well-defined bijection by the set-theoretic De Morgan laws and $\overline{u^{-1} \overline{L}} = u^{-1} L$. It is an order-isomorphism because $X \subseteq Y \iff \overline{Y} \subseteq \overline{X}$, hence a join-semilattice isomorphism too.
    \item
    We calculate:
    \[
      \begin{tabular}{lll}
        $\overline{dr_{\overline{L}}(\overline{v^{-1} L^r})}$
        &
        $= \overline{dr_{\overline{L}}(v^{-1} \overline{L^r})}$
        & ($v^{-1}(-)$ preserves complement)
        \\ &
        $= \overline{\bigcup \{ X \in \LW{\overline{L^r}} : v^r \nin X \}}$
        & (see Theorem \ref{thm:dr_L})
        \\ &
        $= \bigcap \{ \overline{X} : v^r \in X \in \LW{\overline{L^r}} \}$
        \\ &
        $= \bigcap \{ X \in \LW{L^r} : v^r \in X \}$.
      \end{tabular}
    \]

    \item 
    We first show $M(\jslLD{L})$ has the claimed description. By Corollary \ref{cor:mirr_in_lq} each $X \in \LQ{L}$ is an intersection of $dr_L(v^{-1} L^r)$'s, so every $S \in \LD{L}$ is an intersection of them too. Then these elements meet-generate $\jslLD{L}$. To see they are all meet-irreducible, fix $v_0 \in \Sigma^*$. We'll show $dr_L(v_0^{-1} L^r)$ has the following unique cover in $\jslLD{L}$:
    \[
      K_{v_0} :=
      \bigcap \{ dr_L(Y) : dr_L(v_0^{-1} L^r) \subset dr_L(Y), \, Y \in \LW{L} \}.
    \]
    Certainly $dr_L(v_0^{-1} L^r) \subseteq K_{v_0}$. Crucially if $dr_L(v_0^{-1} L^r) \subset dr_L(Y)$ then by \emph{strictness} we know $dr_L(Y) \nsubseteq dr_L(v_0^{-1} L^r)$, hence $v_0^r \in dr_L(Y)$ by Corollary \ref{cor:mirr_in_lq}. Then we have the strict inclusion $dr_L(v_0^{-1} L^r) \subset K_{v_0}$. Since $K_v$ is the meet of all meet-irreducibles strictly greater than $dr_L(v^{-1} L^r)$ it is also the unique cover of the latter.

    The description of $J(\jslLD{L})$ follows by (1) i.e.\ they are the relative complements of the meet-irreducibles in $\jslLD{\overline{L}}$. Finally, both sets have cardinality $|\LW{L^r}|$.

    \item
    For any $v_0 \in \Sigma^*$ we first establish:
    \[
      \begin{tabular}{lll}
        $\bigcap \{ X \in \LW{L^r} : v_0^r \in X \} \nsubseteq dr_L(v^{-1} L^r)$
        &
        $\iff \forall X \in \LW{L}.[ v_0^r \in X \To X \nsubseteq dr_L(v^{-1} L^r) ]$
        & (A)
        \\ &
        $\iff \forall X \in \LW{L}.[ v_0^r \in X \To v^r \in X) ]$
        & (Corollary \ref{cor:mirr_in_lq})
        \\ &
        $\iff v^r \in \bigcap \{ X \in \LW{L^r} : v_0^r \in X \}$.
      \end{tabular}
    \]
    Concerning (A), the implication $(\To)$ follows because $X \subseteq dr_L(v^{-1} L^r)$ would yield a contradiction, whereas $(\oT)$ holds because $X \nsubseteq dr_L(v^{-1} L^r)$ implies $v^r \in X$ by Corollary \ref{cor:mirr_in_lq}, so the intersection contains $v^r$ too. Then, invoking (2) and (3), we've established the original claim whenever $S \in J(\jslLD{L})$. In the general case $S = \bigcup J$ where $J \subseteq J(\jslLD{L})$,
    \[
      \begin{tabular}{lll}
        $S \subseteq dr_L(v^{-1} L^r)$
        &
        $\iff \forall K \in J. K \subseteq dr_L(v^{-1} L^r)$
        \\ &
        $\iff \forall K \in J. v^r \nin K$
        \\ &
        $\iff v^r \nin S$.
      \end{tabular}
    \]

    \item
    Each join-irreducible takes the form $j_v := \overline{dr_{\overline{L}}(\overline{v^{-1} L^r})}$ where $v \in \Sigma^*$. By Note \ref{note:irreducibles} we know $j_v$ is join-prime, so for any $J \subseteq J(\jslLD{L})$ we have $j_v \nsubseteq \bigcup J \iff \forall j \in J. j_v \nsubseteq j$. Then we calculate:
    \[
      \begin{tabular}{lll}
        $\tau_{\jslLD{L}}(j_{v_0})$
        &
        $= \bigcup \{ S \in \LD{L} : j_{v_0} \nsubseteq S \}$
        \\ &
        $= \bigcup \{ j_v \in J(\jslLD{L}) : j_{v_0} \nsubseteq j_v \}$
        & ($j_{v_0}$ join-prime)
        \\ &
        $= \bigcup \{ j_v : dr_{\overline{L}}(\overline{v^{-1} L^r}) \nsubseteq dr_{\overline{L}}(\overline{v_0^{-1} L^r}) \}$
        \\ &
        $= \bigcup \{ j_v : v_0^r \in dr_{\overline{L}}(\overline{v^{-1} L^r}) \}$
        & (Corollary \ref{cor:mirr_in_lq})
        \\ &
        $= \bigcup \{ j \in J(\jslLD{L}) : v_0^r \nin j \}$
        \\ &
        $= dr_L(v_0^{-1} L^r)$
        & (by (4)).
      \end{tabular}
    \]
  \end{enumerate}
\end{proof}

Lemma \ref{lem:pos_left_preds_basics} provides a natural bijection $\LW{L^r} \cong J(\jslLD{L})$ akin to the quotient-atom bijection.

\begin{theorem}[Quotient-intersection bijection]
  Each regular $L$ has the canonical bijection,
  \[
    \lambda_L : \LW{L^r} \to J(\jslLD{L})
    \qquad
    \lambda_L (Y) := \overline{dr_{\overline{L}}(\overline{Y})}
    \qquad
    \lambda_L^{\bf-1} := \lambda_{L^r},
  \]
  and respective relationship:
  \[
    \lambda_L(x^{-1} L^r) \subseteq a^{-1} \lambda_L(y^{-1} L^r)
    \iff y^{-1} L^r \subseteq (xa)^{-1} L^r
    \qquad
    \text{for any $x, y \in \Sigma^*$.}
  \]
\end{theorem}

\begin{proof}
  The bijection follows by Lemma \ref{lem:pos_left_preds_basics}.
  Concerning the relationship,
  \[
    \begin{tabular}{lll}
      $y^{-1} L^r \subseteq (xa)^{-1} L^r$
      &
      $\iff \forall v \in \Sigma^*. [ yv \in L^r \To xav \in L^r ]$
      \\ &
      $\iff \forall v \in \Sigma^*. [ v^ry^r \in L \To v^rax^r \in L ]$
      \\ &
      $\iff \forall v \in \Sigma^*. [ y^r \in [v^r]^{-1} L \To x^r \in [(va)^r]^{-1} L ]$
      \\ &
      $\iff \forall X \in \LW{L}. [ y^r \in X \To x^r \in a^{-1} X ]$
      \\ &
      $\iff \bigcap \{ X \in \LW{L} : x^r \in X \} \subseteq a^{-1} \bigcap \{ X \in \LW{L} : y^r \in X \}$.
    \end{tabular}
  \]
\end{proof}

\begin{note}[Canonicity of $\lambda_L$]
  It arises from the duality between $\Poset$-dfas and $\DL$-dfas, see \cite{CanonNA2014}.
  \endbox
\end{note}

Recall the nfa $\minSatDfa{L}$ from Example \ref{ex:saturated_min_dfa}. It arises from the state-minimal deterministic machine $\minDfa{L}$ by extending the initial states and transitions.

\begin{theorem}[Canonical distributive dependency automaton]
  \label{thm:canonical_dist_dep_aut}
  We have the $\AutDep$-isomorphism:
  \[
    \rD : (\rev{\minSatDfa{L^r}}, \subseteq, \minSatDfa{L^r}) \to \Airr(\jslDfaDistMin{L})
    \qquad
    \rD (v_1^{-1} L^r, dr_L(v_2^{-1} L^r)) :\iff v_1^{-1} L^r \subseteq v_2^{-1} L^r
  \]
  with inverse $\rE (S, v^{-1} L^r) :\iff \lambda_L^{\bf-1}(S) \subseteq v^{-1} L^r$.
\end{theorem}

\begin{proof}
  To see $\rD$'s domain is a well-defined dependency automaton, observe that its dual $(\minSatDfa{L^r}, \supseteq, \rev{\minSatDfa{L^r}})$ is well-defined by Theorem \ref{thm:free_jsl_dfa_ordered}. Next we establish the commuting relations:
  \[
    \xymatrix@=20pt{
      \LW{L^r} \ar[rr]^{\tau_{\jslLD{L} \circ \lambda_L}} && M(\jslLD{L})
      \\
      \LW{L^r} \ar[u]^{\subseteq} \ar[rr]_{\lambda_L} && J(\jslLD{L}) \ar[u]_{\nsubseteq}
    }
  \]
  via the following calculation:
  \[
    \begin{tabular}{lll}
      $\lambda_L(v_1^{-1} L^r) \nsubseteq dr_L(v_2^{-1} L^r)$
      &
      $\iff \overline{dr_{\overline{L}}(\overline{v_1^{-1} L^r})} \nsubseteq dr_L(v_2^{-1} L^r)$
      & (def.\ of $\lambda$)
      \\ &
      $\iff v_2^r \in \overline{dr_{\overline{L}}(\overline{v_1^{-1} L^r})}$
      & (by Lemma \ref{lem:pos_left_preds_basics}.4)
      \\ &
      $\iff v_2^r \nin dr_{\overline{L}}(\overline{v_1^{-1} L^r})$
      \\ &
      $\iff v_2^r \nin dr_{\overline{L}}(v_1^{-1} \overline{L^r})$
      \\ &
      $\iff dr_{\overline{L}}(\overline{v_1^{-1} L^r}) \subseteq dr_{\overline{L}}(v_2^{-1} \overline{L^r})$
      & (by Corollary \ref{cor:mirr_in_lq})
      \\ &
      $\iff v_2^{-1} \overline{L^r} \subseteq v_1^{-1} \overline{L^r}$
      & ($dr_L$ an order-iso)
      \\ &
      $\iff v_1^{-1} L^r \subseteq v_2^{-1} L^r$,
      \\ &
      $\iff v_1^{-1} L^r \subseteq \tau_{\jslLD{l}} \circ \lambda_L(v_2^{-1} L^r)$.
    \end{tabular}
  \]
  Since the witnesses are bijections we've established that $\rD$ underlying $\Dep$-morphism is an isomorphism. Concerning the remaining conditions, $\rD$'s domain has lower nfa $\rN := \rev{\minSatDfa{L^r}}$ with transitions $\rN_a(Y_1, Y_2) :\iff Y_1 \subseteq a^{-1} Y_2$. Furthermore $\Airr(\jslDfaDistMin{L})$'s upper nfa $\rM$ has transitions:
  \[
    \begin{tabular}{lll}
      $\rM_a(dr_L(v_1^{-1} L^r), dr_L(v_2^{-1} L^r))$
      \\
      $\iff (\gamma_a)_*( dr_L(v_1^{-1} L^r)) \subseteq dr_L(v_2^{-1} L^r)$
      \\
      $\iff \bigcup \{ j \in J(\jslLD{L}) : a^{-1} j \subseteq dr_L(v_1^{-1} L^r) \} \subseteq dr_L(v_2^{-1} L^r)$
      \\
      $\iff \bigcup \{ j \in J(\jslLD{L}) : v_1^r \nin a^{-1} j  \} \subseteq dr_L(v_2^{-1} L^r)$
      & (by Lemma \ref{lem:pos_left_preds_basics}.4)
      \\
      $\iff \bigcup \{ j \in J(\jslLD{L}) : av_1^r \nin j  \} \subseteq dr_L(v_2^{-1} L^r)$
      \\
      $\iff dr_L((v_1a)^{-1} L^r) \subseteq dr_L(v_2^{-1} L^r)$
      & (by Lemma \ref{lem:pos_left_preds_basics}.4)
      \\
      $\iff v_2^{-1} L^r \subseteq (v_1a)^{-1} L^r$
      & ($dr_L$ an order-iso)
    \end{tabular}
  \]
  where $\gamma_a := \lambda X. a^{-1} X : \jslLD{L} \to \jslLD{L}$. Then we verify:
  \[
    \begin{tabular}{lll}
      $\rN_a ; \rD ( v_1^{-1} L^r, dr_L(v_2^{-1} L^r))$
      &
      $\iff \exists v \in \Sigma^*. [ v_1^{-1} L^r \subseteq (va)^{-1} L^r \,\land\, \rD (v^{-1} L^r, dr_L(v_2^{-1} L^r) ]$
      \\ &
      $\iff \exists v \in \Sigma^*. [ v_1^{-1} L^r \subseteq (va)^{-1} L^r \,\land\, v^{-1} L^r \subseteq v_2^{-1} L^r ]$
      \\ &
      $\iff \exists v \in \Sigma^*. [ v_1^{-1} L^r \subseteq v^{-1} L^r \,\land\, v^{-1} L^r \subseteq (v_2 a)^{-1} L^r ]$
      & (A)
      \\ &
      $\iff \exists v \in \Sigma^*. [ \rD(v_1^{-1} L^r, dr_L(v^{-1} L^r)) \,\land\, \rM_a\spbreve(dr_L(v^{-1} L^r), dr_L( v_2^{-1} L^r)) ]$
      \\ &
      $\iff \rD ; \rM_a\spbreve ( v_1^{-1} L^r, dr_L(v_2^{-1} L^r))$.
    \end{tabular}
  \]
  Concerning (A), $(\To)$ follows because $a^{-1}(-)$ preserves inclusions so we can choose $v := v_1$; $(\oT)$ follows analogously, choosing $v := v_2$. Finally we verify:
  \[
    \begin{tabular}{lll}
      $\rD[I_\rN]$
      &
      $= \rD[\{ Y : \epsilon \in Y \in \LW{L^r} \}]$
      \\ &
      $= \{ dr_L(Y) : \epsilon \in Y \in \LW{L^r} \}$
      \\ &
      $= \{ dr_L(v^{-1} L^r) : v \in L^r \}$
      \\ &
      $= \{ dr_L(v^{-1} L^r) : L \nsubseteq dr_L(v^{-1} L^r) \}$
      & (by Corollary \ref{cor:mirr_in_lq})
      \\ &
      $= F_\rM$
      & (see Definition \ref{def:airr_and_det}).
      \\
      \\
      $\breve{\rD}[I_\rM]$
      &
      $= \breve{\rD}[\{ dr_L(v^{-1} L^r) : dr_L(L^r) \subseteq dr_L(v^{-1} L^r) \}]$
      & (see Definition \ref{def:airr_and_det})
      \\ &
      $= \breve{\rD}[\{ dr_L(v^{-1} L^r) : v^{-1} L^r \subseteq L^r \}]$
      \\ &
      $= \{ v_1^{-1} L^r : \exists v \in \Sigma^*. [ v_1^{-1} L^r \subseteq v^{-1} L^r \,\land\, v^{-1} L^r \subseteq L^r  ] \}$
      \\ &
      $= \{ v^{-1} L^r : v^{-1} L^r \subseteq L^r \}$
      \\ &
      $= F_\rN$.
    \end{tabular}
  \]
\end{proof}

\begin{corollary}[Dualising $\jslDfaDistMin{L}$]
  \label{cor:dual_jsl_dfa_of_left_dist_preds}
  We have the $\JSL$-dfa isomorphism,
  \[
    \begin{tabular}{c}
      $\rho_L : \Det(\minSatDfa{L^r}, \supseteq ,\rev{\minSatDfa{L^r}}) \to (\jslDfaDistMin{L})^\pentagram$
      \\[1.5ex]
      $\rho_L := \lambda S. \bigcap \{ dr_L(Y) : Y \in S \}$
      \qquad
      $\rho_L^{-1} := \lambda K. \{ Y \in \LW{L^r} : K \subseteq dr_L(Y) \}$.
    \end{tabular}
  \]
\end{corollary}

\begin{proof}
  First recall the isomorphism from Theorem \ref{thm:canonical_dist_dep_aut},
  \[
    \rE : \Airr(\jslDfaDistMin{L}) \to \Rev\fM
    \qquad
    \text{where }
    \fM := (\minSatDfa{L^r}, \supseteq, \rev{\minSatDfa{L^r}}).
  \]
  Since $\rE$ has bijective upper witness $(\tau_{\jslLD{L}} \circ \lambda_L)^{\bf-1}$ and join-semilattice adjoints act as the inverse, it follows that $(\Det \rE)^\pentagram = \lambda X. \tau_{\jslLD{L}} \circ \lambda_L [X]$. Further recall the natural isomorphism $\hat{\partial}_{-}$ (Theorem \ref{thm:aut_dep_jsl_dfa_self_duality_transfer}) and $\rep_{-}$ (Theorem \ref{thm:aut_dep_equiv_jsl_dfa}). Then we have the composite join-semilattice isomorphism:
  \[
    \Det \fM
    \xto{\hat{\partial}_{\fM}^{\bf-1}} (\Det\Rev\fM)^\pentagram
    \xto{(\Det \rE)^\pentagram} (\Det\Airr(\jslDfaDistMin{L}))^\pentagram
    \xto{\rep_{\jslDfaDistMin{L}}^\pentagram} (\jslDfaDistMin{L})^\pentagram.
  \]
  Given any subset $S \subseteq \LW{L^r}$ upwards-closed w.r.t.\ inclusion,
  \[
    \begin{tabular}{llll}
      $S$
      & $\mapsto$ & $\hat{\partial}_{\fM}^{\bf-1}(S)$
      \\
      & $=$ & $\supseteq[\overline{S}]$
      \\
      & $=$ & $\overline{S}$
      & ($\overline{S}$ down-closed)
      \\
      & $\mapsto$ & $(\Det \rE)^\pentagram(\overline{S})$
      \\
      & $=$ & $\tau_{\jslLD{L}} \circ \lambda_L [\overline{S}]$
      \\
      & $=$ & $\{ dr_L(v^{-1} L^r) : v^{-1} L^r \nin S \}$
      \\
      & $\mapsto$ & $\rep_{\jslDfaDistMin{L}}^\pentagram (\{ dr_L(v^{-1} L^r) : v^{-1} L^r \nin S \})$
      \\
      & $=$ & $\bigcap \{ dr_L(v^{-1} L^r) : v^{-1} L^r \in S \}$.
    \end{tabular}
  \]
  Finally the action of the inverse follows by the bijectivity of $\rho_L$.
\end{proof}

\subsection{Minimal boolean syntactic machine}

We start by recalling the syntactic monoid of a regular language and the transition monoid of a classical dfa.

\begin{definition}[Transition monoids and syntactic monoids]
  \label{def:trans_mon_syn_mon}
  \item
  \begin{enumerate}
    \item
    Given any set $\Sigma$ we have the \emph{free $\Sigma$-generated monoid} $\FMon{\Sigma} := (\Sigma^*, \cdot, \epsilon)$ where multiplication is concatenation.

    \item
    Given a dfa $\delta = (z_0, Z, \delta_a, F)$, its \emph{transition monoid} is defined $\TM{\delta} := (\{ \delta_w : w \in \Sigma^* \}, \circ, id_Z)$ where $\circ$ is functional composition and $\delta_\epsilon = id_Z$ (see Definition \ref{def:dfas}).
    It admits a natural dfa structure accepting $L$:
    \[
      \TMDfa{\delta} := (
        id_Z,
        \{ \delta_w : w \in \Sigma^* \},
        \lambda f. \delta_a \circ f,
        \{ f : f(z_0) \in F \}
      ).
    \]
    Finally we have $\sem{-}_{\TM{\delta}} : \FMon{\Sigma} \epito \TM{\delta}$ where $\sem{w}_{\TM{\delta}} := \delta_w$.

    \item
    The \emph{syntactic monoid} of a regular language $L \subseteq \Sigma^*$ is the quotient $\SynMon{L} := \FMon{\Sigma} / \rS_L$ by the \emph{syntactic congruence}
    $
      \rS_L
      := \{ (u, v) \in \Sigma^* \times \Sigma^* : \forall x, y \in \Sigma^*.[ xuy \in L \iff xvy \in L ]  \}
    $. It admits a natural dfa structure accepting $L$:
    \[
      \SynMonDfa{L} := (
        \sem{\epsilon}_{\SynCong{{L}}}, \Sigma^* / \SynCong{L}, \lambda x. x \cdot \sem{a}_{\SynCong{L}}, \{ \sem{w}_{\SynCong{L}} : w \in L \}
      ).
    \]
    We also denote the underlying set by $Syn(L) := \Sigma^* \setminus \SynCong{L}$.
    \endbox
  \end{enumerate}
\end{definition}

\begin{lemma}[The syntactic/transition monoid are well-defined]
  \label{lem:syn_mon_well_def}
  \item
  \begin{enumerate}
    \item 
    $\TM{\delta}$ is a well-defined finite monoid and $L(\TMDfa{\delta}) = L(\delta)$.
    \item 
    $\SynMon{L}$ is a well-defined monoid and $L(\SynMonDfa{L}) = L$.
  \end{enumerate}
\end{lemma}

\begin{proof}
  \item
  \begin{enumerate}
    \item 
    Fix a dfa $\delta = (z_0, Z, \delta_a, F)$. The set of all endofunctions on a set equipped with functional composition define a finite monoid; $\TM{\delta}$ defines a submonoid. Finally:
    \[
        w \in L(\TMDfa{\delta})
        \iff \delta_w \in F_{\TMDfa{\delta}}
        \iff \delta_w(z_0) \in F
        \iff w \in L(\delta).
    \]
    
    \item
    To see $\rS_L \subseteq \Sigma^* \times \Sigma^*$ is a congruence for $(\Sigma^*, \cdot, \epsilon)$, given $\rS_L(u_1, u_2)$ and $\rS_L(v_1, v_2)$,
    \[
      x (u_1 v_1) y \in L
      \iff x (u_1) v_1 y \in L
      \iff x (u_2) v_1 y \in L
      \iff x u_2 (v_1) y \in L
      \iff x u_2 v_2 y \in L.
    \]
    Thus $\SynMon{L}$ is a well-defined monoid. It is finite because the equivalence classes are precisesly the atoms of the set-theoretic boolean algebra generated by the finite set $\{ x^{-1} L y^{-1} : x, \, y \in \Sigma^* \}$. Finally:
    \[
      w \in L(\SynMonDfa{L})
      \iff \sem{w} \in F_{\SynMonDfa{L}}
      \iff w \in L.
    \]
    The final equivalence follows because if $w \in L$ then $\sem{w}_{\rS_L} \subseteq L$. Indeed if $u \in \sem{w}_{\rS_L}$ then choosing $x = y = \epsilon$ we have $xwy \in L \iff xuy \in L$ i.e.\ $u \in L$.

  \end{enumerate}
\end{proof}

As is well-known, $L$'s syntactic monoid is isomorphic to the transition monoid of $L$'s state-minimal dfa.

\begin{theorem}[$\SynMon{L} \cong \TM{\minDfa{L}}$]
  \label{thm:tm_min_dfa_iso_syn}
  We have the monoid isomorphism:
  \[
    \lambda \sem{w}_{\rS_L} . \lambda X. w^{-1} X: \SynMon{L} \to \TM{\minDfa{L}}.
  \]
\end{theorem}

\begin{proof}
  The function is well-defined and injective because:
  \[
    \begin{tabular}{lll}
      $\sem{u_1}_{\rS_L} = \sem{u_2}_{\rS_L}$
      &
      $\iff \forall x, y \in \Sigma^*.[ xu_1y \in L \iff xu_2y \in L ]$
      \\ &
      $\iff \forall X \in \LW{L}, y \in \Sigma^*.[ u_1y \in X \iff u_2y \in X ]$
      \\ &
      $\iff \forall X \in \LW{L}, y \in \Sigma^*.[ y \in u_1^{-1} X \iff y \in u_2^{-1} X ]$
      \\ &
      $\iff \forall X \in \LW{L}.[ u_1^{-1} X = u_2^{-1} X ]$
      \\ &
      $\iff \lambda X \in \LW{L}. u_1^{-1} X = \lambda X \in \LW{L}. u_2^{-1} X$.
    \end{tabular}
  \]
  It is surjective because $\minDfa{L}$'s transition monoid consists of the functions $\{ \lambda X \in \LW{L}. w^{-1} X : w \in \Sigma^* \}$. Finally it is a monoid morphism because $\lambda X. \epsilon^{-1} X = id_{\LW{L}}$ and  $(u v)^{-1} X = v^{-1}(u^{-1} X)$.
\end{proof}

We can now introduce another canonical $\JSL$-dfa and its equivalent dependency automaton.

\begin{definition}[\emph{$L$'s minimal boolean syntactic $\JSL$-dfa}]
  \label{def:canon_bool_syn_jsl_dfa}
  Let $\LRW{L} := \{ u^{-1} L v^{-1} : u,\, v \in \Sigma^* \}$ be the left-right-word-quotients and $\LRP{L}$ the closure of $\LRW{L}$ under the set-theoretic boolean operations. Then:
  \[
    \jslDfaSynBoolMin{L} := (L, \jslLRP{L}, \lambda X. a^{-1} X, \{ K : \epsilon \in K \})
  \]
  is the \emph{canonical boolean syntactic $\JSL$-dfa} over the join-semilattice $\jslLRP{L} := (\LRP{L}, \cup, \emptyset)$.
  \endbox
\end{definition}

\begin{lemma}[$J(\jslLRP{L}) = Syn(L)$]
  $\jslLRP{L}$'s atoms are the equivalence classes of the syntactic congruence $\SynCong{L}$.
\end{lemma}

\begin{proof}
  An equivalence class amounts to $\bigcap_i u_i^{-1} L v_i^{-1} \cap \bigcap_j \overline{u_j^{-1} L v_j^{-1}}$ involving every left-right-word-quotient $u^{-1} L v^{-1}$.
\end{proof}

Next we describe the minimal boolean syntactic $\JSL$-dfa as a dependency automaton.

\begin{theorem}[Canonical boolean syntactic dependency automaton]
  \label{thm:canon_bool_syn_dep_aut}
  We have the $\AutDep$-isomorphism:
  \[
    \lambda X. \overline{X^r} : \dep{\rev{\SynMonDfa{L^r}}} \to \Airr(\jslDfaSynBoolMin{L}),
  \]
  whose inverse has action $\lambda X. X^r$.
\end{theorem}

\begin{proof}
  We have the bijection $\lambda X. X^r : Syn(L^r) \to Syn(L)$ because $\forall x, y \in \Sigma^*. [xuy \in L \iff xvy \in L]$ is equivalent to $\forall x, y \in \Sigma^*. [ xu^ry \in L^r \iff xv^ry \in L^r ]$. Then we have the $\Dep$-isomorphism $f := \lambda X. \overline{X^r}$,
  \[
    \xymatrix@=15pt{
      Syn(L^r) \ar[rr]^-{\lambda X. \overline{X}^r} && \{ \overline{X} : X \in Syn(L)  \}
      \\
      Syn(L^r) \ar[u]^{\Delta_{Syn(L^r)}} \ar[rr]_-{\lambda X. X^r}
      && Syn(L) \ar[u]_{\neg_{\jslLRP{L}}}
    }
  \]
  where $\neg_{\jslLRP{L}}$ constructs the relative complement in $\Sigma^*$. It is a $\Dep$-isomorphism because the witnesses are bijections. It remains to verify the other constraints. Denote the transitions of the left (resp.\ right) dependency automaton's lower (resp.\ upper) nfa by $\rN$ (resp.\ $\rM$). Then:
  \[
    \rN_a(\sem{u_1}_{\SynCong{L^r}}, \sem{u_2}_{\SynCong{L^r}})
    \iff
    \sem{u_2 a}_{\SynCong{L^r}} = \sem{u_1}_{\SynCong{L^r}}
  \]
  \[
    \begin{tabular}{lll}
      $\rM_a(\overline{\sem{u_1}_{\SynCong{L}}}, \overline{\sem{u_2}_{\SynCong{L}}})$
      &
      $\iff (\gamma_a)_* (\overline{\sem{u_1}_{\SynCong{L}}}) \subseteq \overline{\sem{u_2}_{\SynCong{L}}}$
      \\ &
      $\iff \bigcup \{ \sem{u}_{\SynCong{L}} : a^{-1} \sem{u}_{\SynCong{L}} \subseteq \overline{\sem{u_1}_{\SynCong{L}}} \} \subseteq \overline{\sem{u_2}_{\SynCong{L}}}$
      \\ &
      $\iff \bigcup \{ \sem{u}_{\SynCong{L}} : u_1 \nin a^{-1} \sem{u}_{\SynCong{L}} \} \subseteq \overline{\sem{u_2}_{\SynCong{L}}}$
      \\ &
      $\iff \bigcup \{ \sem{u}_{\SynCong{L}} : au_1 \nin \sem{u}_{\SynCong{L}} \} \subseteq \overline{\sem{u_2}_{\SynCong{L}}}$
      \\ &
      $\iff \overline{\sem{au_1}_{\SynCong{L}}} \subseteq \overline{\sem{u_2}_{\SynCong{L}}}$
      \\ &
      $\iff \sem{u_2}_{\SynCong{L}} \subseteq \sem{au_1}_{\SynCong{L}}$
      \\ &
      $\iff \sem{u_2}_{\SynCong{L}} = \sem{au_1}_{\SynCong{L}}$.
    \end{tabular}
  \]
  where $\gamma_a := \lambda X. a^{-1} X : \jslDfaSynBoolMin{L} \to \jslDfaSynBoolMin{L}$. We now verify the condition concerning transitions:
  \[
    \begin{tabular}{lll}
      $\rN_a ; f (\sem{u_1}_{\SynCong{L}}, \overline{\sem{u_2}_{\SynCong{L}}})$
      &
      $\iff \exists u \in \Sigma^*. [
        \sem{u a}_{\SynCong{L^r}} = \sem{u_1}_{\SynCong{L^r}}
        \,\land\,
        \overline{\sem{u}_{\SynCong{L^r}}^r} = \overline{\sem{u_2}_{\SynCong{L}}}
      ]$
      \\ &
      $\iff \exists u \in \Sigma^*.[
        \sem{u a}_{\SynCong{L^r}} = \sem{u_1}_{\SynCong{L^r}}
        \,\land\,
        \sem{u}_{\SynCong{L^r}} = \sem{u_2^r}_{\SynCong{L^r}}
      ]$
      \\ &
      $\iff \sem{u_2^r a}_{\SynCong{L^r}} = \sem{u_1}_{\SynCong{L^r}}$
      \\ &
      $\iff \sem{u_1^r}_{\SynCong{L}} = \sem{a u_2}_{\SynCong{L}}$
      \\ &
      $\iff \exists u \in \Sigma^*.[
        \sem{u_1^r}_{\SynCong{L}} = \sem{u}_{\SynCong{L}}
        \,\land\,
        \sem{u}_{\SynCong{L}} = \sem{a u_2}_{\SynCong{L}}
      ]$
      \\ &
      $\iff \exists u \in \Sigma^*. [
        f(\sem{u_1}_{\SynCong{L}}) = \overline{\sem{u}_{\SynCong{L}}}
        \,\land\,
        \rM_a\spbreve(\overline{\sem{u}_{\SynCong{L}}}, \overline{\sem{u_2}_{\SynCong{L}}})
      ]$
      \\ &
      $\iff f ; \rM_a(\sem{u_1}_{\SynCong{L}}, \overline{\sem{u_2}_{\SynCong{L}}})$.
    \end{tabular}
  \]
  Finally we calculate:
  \[
    \begin{tabular}{lll}
      $f[I_{\rev{\SynMonDfa{L^r}}}]$
      &
      $= f[F_{\SynMonDfa{L^r}}]$
      \\ &
      $= f[\{ \sem{w}_{\SynCong{L^r}} : w \in L^r \}]$
      \\ &
      $= \{ \overline{\sem{w^r}_{\SynCong{L}}} : w \in L^r \}$
      \\ &
      $= \{ \overline{\sem{w}_{\SynCong{L}}} : L^r \nsubseteq \overline{\sem{w}_{\SynCong{L}}} \}$
      \\ &
      $= F_\rM$.
    \end{tabular}
  \qquad
    \begin{tabular}{lll}
      $\breve{f}[I_\rM]$
      &
      $= \breve{f}[\{ \overline{\sem{u}_{\SynCong{L}}} : \bigcup \{ \sem{w}_{\SynCong{L}} : \epsilon \nin L \} \subseteq \overline{\sem{u}_{\SynCong{L}}}  \}]$
      \\ &
      $= \breve{f}[\{ \overline{\sem{u}_{\SynCong{L}}} : \overline{\sem{\epsilon}_{\SynCong{L}}} \subseteq \overline{\sem{u}_{\SynCong{L}}}  \}]$
      \\ &
      $= \breve{f}[\{ \overline{\sem{\epsilon}_{\SynCong{L}}} \}]$
      \\ &
      $= \sem{\epsilon}_{\SynCong{L^r}}$
      \\ &
      $= F_\rN$.
    \end{tabular}
  \]
\end{proof}


\begin{note}[Canonical distributive syntactic $\JSL$-dfa]
  \[
    \rR : (\rev{\SynMonSatDfa{L^r}}, \subseteq, \SynMonSatDfa{L^r}) \to \Airr(\jslDfaSynDistMin{L})
  \]
  \endbox
\end{note}

\subsection{Transition semirings of $\JSL$-dfas}

Whilst classical dfas induce monoids, $\JSL$-dfas induce \emph{idempotent semirings}.

\begin{definition}[Transition semiring of a $\JSL$-dfa]
  \item
  \begin{enumerate}
    \item
    $\FPow \bSigma^* := ((\FPow \Sigma, \cup, \emptyset), \cdot, \{ \epsilon \})$ is the \emph{free $\Sigma$-generated idempotent semiring} where $\FPow \Sigma^*$ is the set of finite languages, its multiplication being sequential composition of  languages.

    \item
    Fix a $\JSL$-dfa $\gamma = (s_0, \aS, \gamma_a, F)$ and recall the composites $(\gamma_w : \aS \to \aS)_{w \in \Sigma^*}$ from Definition \ref{def:dfa_jsl}. More generally for any $K \subseteq \Sigma^*$ we can construct the pointwise-join of $\{ \gamma_w : w \in K \}$,
    \[
      \gamma_K := \lambda s. \Lor_\aS \{ \gamma_w(s) : w \in K \} : \aS \to \aS.
    \]
    Then $\gamma$'s \emph{transition semiring} is the idempotent semiring $\TS{\gamma} := (\aS_\gamma, \circ, id_\aS)$ where:
    \[
      \aS_\gamma := (S_\gamma, \lor_{\aS_\gamma}, \lambda X. \bot_\aS)
      \qquad
      S_\gamma := \{ \gamma_K : K \subseteq \Sigma^* \}
      \qquad
      \gamma_U \lor_{\aS_\gamma} \gamma_V := \gamma_{U \cup V}.
    \]

    \item
    Since $\TS{\gamma}$ is $\Sigma$-generated by $\{ \gamma_a : a \in \Sigma \}$ we have the unique extension $\sem{-}_{\TS{\gamma}} : \FPow \bSigma^* \epito \TS{\gamma}$ i.e.\ a surjective idempotent semiring morphism.

    \item
    Finally the semiring $\TS{\gamma}$ has a natural associated $\JSL$-dfa structure:
    \[
      \jslDfaTs{\gamma} :=
      (
        id_{S_\gamma},
        \aS_\gamma,
        \lambda f. \gamma_a \circ f,
        \{ f : f \nleq_{\aS_\gamma} \gamma_{\overline{L}} \}
      )
    \]
    accepting $L := L(\gamma)$.
    \endbox
  \end{enumerate}
\end{definition}

\begin{lemma}[$\TS{\gamma}$ and $\jslDfaTs{\gamma}$ well-defined]
  \label{lem:ts_well_defined}
  \item
  \begin{enumerate}
    \item
    $\TS{\gamma}$ is a well-defined idempotent semiring.
    \item
    $\jslDfaTs{\gamma}$ is a $\JSL$-reachable $\JSL$-dfa accepting $L(\gamma)$.
  \end{enumerate}
\end{lemma}

\begin{proof}
  Let $\gamma = (s_0, \aS, \gamma_a, F)$ be a $\JSL$-dfa.

  \begin{enumerate}
    \item
    $\aS_\gamma$ defines an `additive' idempotent commutative monoid; $(S_\gamma, \circ, id_\aS)$ defines a `multiplicative' monoid. Multiplication left/right distributes over addition and $\bot_{\aS_\gamma}$ annihilates multiplication because composition of join-semilattice morphisms is bilinear w.r.t.\ pointwise-joins.

    \item
    We first establish $\jslDfaTs{\gamma}$ is a well-defined $\JSL$-dfa.  The transition endomorphisms are well-defined functions, and preserve the join by bilinearity. The final states are well-defined by construction since $\gamma_{\overline{L}} \in S_\gamma$. This $\JSL$-dfa accepts $L$ because $\gamma_w \nleq_{\aS_\gamma} \gamma_{\overline{L}} \iff w \in L$, as we now show.
    \begin{itemize}
      \item
      $(\To)$: contrapositive follows because if $w \in \overline{L}$ then $\gamma_{\overline{L}}$ is a join of morphisms including $\gamma_w$.
      \item
      $(\oT)$: $w \in L$ implies $\gamma_w (s_0) \in F$ whereas $\gamma_{\overline{L}} (s_0) \nin F$.
    \end{itemize}
    
    Finally it is $\JSL$-reachable because (i) each $\gamma_w$ is classically reachable from the identity function $id_\aS$, (ii) each $\gamma_K$ is the join of $\gamma_w$'s.
  \end{enumerate}
\end{proof}

\begin{lemma}
  \label{lem:dfa_tmon_as_reach_ts}
  $\TMDfa{\rsc{\rN}} \cong \reach{\jslDfaTs{\jslDfaReach{\jslDfaSc{\rN}}}}$ for any nfa $\rN$.
\end{lemma}

\begin{proof}
  Let $\aS := \jslReach{\jslDfaSc{\rN}}$ i.e.\ the closure of the reachable subsets $\rs{\rN}$ under unions. Then we need to establish the dfa isomorphism $\lambda \gamma_w. \delta_w : \gamma \to \delta$ where:
  \[
    \begin{tabular}{lll}
      $\gamma = (id_{\rs{\rN}}, \{ \gamma_w : \rs{\rN} \to \rs{\rN}, w \in \Sigma^* \}, \lambda f. \gamma_a \circ f, \{ f : f(I) \cap F \neq \emptyset \})$
      \\[1ex]
      $\delta = (id_\aS, \{ \delta_w : \aS \to \aS, \, w \in \Sigma^* \}, \lambda f. \delta_a \circ f, \{ f : f \nleq \delta_{\overline{L}} \})$
    \end{tabular}
  \]
  and both $\gamma_w$ and $\delta_w$ have action $\lambda X. \rN_w[X]$. The candidate isomorphism is a well-defined bijection because $\delta_w$ is uniquely determined by the domain-codomain restriction $\gamma_w$. It clearly preserves the initial state and preserves/reflects the transitions. Finally,
  \[
    \begin{tabular}{lll}
      $\delta_w \nleq \delta_{\overline{L}}$
      &
      $\iff \exists u \in \Sigma^*. \delta_w(\rN_u[I]) \nsubseteq \delta_{\overline{L}}(\rN_u[I])$
      \\ &
      $\iff \exists u \in \Sigma^*, z \in Z_\rN. ( z \in \rN_w[\rN_u[I]] \,\land\, z \nin \rN_{\overline{L}}[\rN_u[I]] )$
      \\ &
      $\iff w \in L$
      & (A)
      \\ &
      $\iff \gamma_w(I) \cap F \neq \emptyset$.
    \end{tabular}
  \]
  Concerning (A), $(\To)$ is immediate whereas $(\oT)$ follows by choosing $u := \epsilon$.
\end{proof}

\begin{definition}[Power semiring and syntactic semiring]
  \item
  \begin{enumerate}
    \item
    The \emph{finitary power semiring} of a monoid ${\bf M} := (M, \cdot_{\bf M}, 1_{\bf M})$ is the idempotent semiring:
    \[
      \FPow {\bf M} := ((\FPow M, \cup, \emptyset), \cdot, \{ 1_{\bf M} \})
      \qquad
      S_1 \cdot S_2 := \{ m_1 \cdot_{\bf M} m_2 : m_1 \in S_1, \, m_2 \in S_2 \}
    \]
    where $\FPow M$ is the set of finite subsets of $M$. If ${\bf M}$ is a finite monoid we may instead write $\Pow {\bf M}$.

    \item
    Given any set $\Sigma$ then $\FPow \FMon{\Sigma}$ is the \emph{free $\Sigma$-generated idempotent semiring}.

    \item
    The \emph{syntactic semiring} $\SynSr{L} := \FPow {\bf \Sigma}^* / \SynSrCong{L}$ of a regular language $L \subseteq \Sigma^*$ is the quotient of the free $\Sigma$-generated idempotent semiring by $L$'s \emph{syntactic semiring congruence} $\SynSrCong{L} \subseteq \FPow \Sigma^* \times \FPow \Sigma^*$ \cite{Polak2001}:
    \[
      \SynSrCong{L}(U, V) :\iff \forall x, y \in \Sigma^*. [
        \{ x \} \cdot U \cdot \{ y \} \subseteq \overline{L} \iff \{ x \} \cdot V \cdot \{ y \}  \subseteq \overline{L}].
    \]
    It admits a natural $\JSL$-dfa structure accepting $L$,
    \[
      \jslDfaSyn{L}
      := (
        \sem{\{ \epsilon \}}_{\SynSrCong{L}},
        (\FPow \Sigma^* / \SynSrCong{L}, \lor_{\SynSr{L}}, \sem{\emptyset}_{\SynSrCong{L}}),
        \lambda X. X \cdot_{\SynSr{L}} \sem{\{ a \}}_{\SynSrCong{L}},
        \{ \sem{U}_{\SynSrCong{L}} : U \cap L \neq \emptyset \}).
    \]
    \endbox
  \end{enumerate}
\end{definition}

\begin{lemma}[Power/syntactic semirings are well-defined]
  \label{lem:power_syntactic_well_def}
  \item
  \begin{enumerate}
    \item 
    $\FPow {\bf M}$ is a well-defined idempotent semiring.
    \item
    $\SynCong{L}(u, v) \iff \SynSrCong{L}(\{ u \}, \{ v \})$ for all $u, v \in \Sigma^*$.
    \item
    $\SynSr{L}$ is a well-defined finite idempotent semiring.
    \item
    $\jslDfaSyn{L}$ is a well-defined $\JSL$-dfa accepting $L$.
  \end{enumerate}
\end{lemma}

\begin{proof}
  \item
  \begin{enumerate}
    \item 
    Let ${\bf M} = (M, \cdot_{\bf M}, 1_{\bf M})$ be a monoid. Firstly, $(\FPow M, \cup, \emptyset)$ is the free join-semilattice on $M$. Secondly, the multiplication $\cdot$ is respectively bilinear by construction.

    \item
    We calculate:
    \[
      \begin{tabular}{lll}
        $\SynSrCong{L}(\{ u \}, \{ v \})$
        &
        $\iff \forall x, y \in \Sigma^*. [ \{ x \} \cdot \{ u \} \cdot \{ y \} \subseteq \overline{L} \iff \{ x \} \cdot \{ v \} \cdot \{ y \} \subseteq \overline{L} ]$
        \\ &
        $\iff \forall x, y \in \Sigma^*. [ xuy \in \overline{L} \iff xvy \in \overline{L} ]$
        \\ &
        $\iff \forall x, y \in \Sigma^*. [ xuy \in L \iff xvy \in L ]$
        \\ &
        $\iff \SynCong{L}(u, v)$.
      \end{tabular}
    \]

    \item 
    We'll show $\SynSrCong{L}$ is a congruence for the free idempotent semiring $\FPow \FMon{\Sigma}$. First observe:
    \[
      \tag{$\star$}
      \SynSrCong{L}(U, V) \iff \forall X, Y \in \FPow \Sigma^*.[ X \cdot U \cdot Y \subseteq \overline{L} \iff X \cdot V \cdot Y \subseteq \overline{L} ].
    \]
    Indeed: $(\oT)$ follows by restriction to words, $(\To)$ follows via $X \cdot U \cdot Y \subseteq \overline{L} \iff \forall x, y \in \Sigma^*. [ \{ x \} \cdot U \cdot \{ y \} \subseteq \overline{L} ]$. Fixing $\SynSrCong{L}(U_i, V_i)$ for $i = 1,2$, it is a congruence for binary joins and multiplication:
    \[
      \begin{tabular}{lll}
        $\{ x \} \cdot (U_1 \cup U_2) \cdot \{ y \} \subseteq \overline{L}$
        &
        $\iff \forall i \in \{ 1, 2 \}. \{ x \} \cdot U_i \cdot \{ y \} \subseteq \overline{L}$
        \\ &
        $\iff \forall i \in \{ 1, 2 \}. \{ x \} \cdot V_i \cdot \{ y \} \subseteq \overline{L}$
        \\ &
        $\iff \{ x \} \cdot (V_1 \cup V_2) \cdot \{ y \} \subseteq \overline{L}$
        \\
        \\
        $\{ x \} \cdot (U_1 \cdot U_2) \cdot \{ y \} \subseteq \overline{L}$
        &
        $\iff \{ x \} \cdot U_1 \cdot ( U_2 \cdot \{ y \}) \subseteq \overline{L}$
        \\ &
        $\iff \{ x \} \cdot V_1 \cdot ( U_2 \cdot \{ y \}) \subseteq \overline{L}$
        & (via $\star$)
        \\ &
        $\iff ( \{ x \} \cdot V_1 ) \cdot  U_2 \cdot \{ y \} \subseteq \overline{L}$
        \\ &
        $\iff ( \{ x \} \cdot V_1 ) \cdot  V_2 \cdot \{ y \} \subseteq \overline{L}$
        & (via $\star$)
        \\ &
        $\iff \{ x \} \cdot (V_1 \cdot V_2) \cdot \{ y \} \subseteq \overline{L}$.
      \end{tabular}
    \]
    To see $\SynSr{L}$ is finite, recall the syntactic monoid is finite by Lemma \ref{lem:syn_mon_well_def} and consider the mapping:
    \[
      q := \lambda \{ \sem{u}_{\SynCong{L}} : u \in U \in \FPow \Sigma^* \}. \sem{U}_{\SynSrCong{L}} : \Pow \SynMon{L} \epito \SynSr{L}.
    \]
    Well-definedness follows via (2) and it is clearly surjective, hence $\SynSr{L}$ is finite.

    \item 
    We show $\jslDfaSyn{L}$ is a well-defined $\JSL$-dfa. It is finite because the syntactic semiring is finite -- see (3). Its join-semilattice structure is well-defined because $\SynSrCong{L}$ is a well-defined congruence. Its deterministic transitions are well-defined because multiplication in $\SynSr{L}$ is bilinear. It remains to show the final states are well-defined. First observe if $\sem{U}_{\SynSrCong{L}} = \sem{V}_{\SynSrCong{L}}$ and $U \nsubseteq \overline{L}$ then $V \nsubseteq \overline{L}$ by choosing $x = y = \epsilon$. Secondly, the non-finals $\{ \sem{U}_{\SynSrCong{L}} : U \subseteq \overline{L} \}$ are closed under joins because given (finitely many) $U_i \subseteq \overline{L}$ then $\bigcup U_i \subseteq \overline{L}$ too. This well-defined $\JSL$-dfa accepts $L$ because its classically reachable part is isomorphic to the syntactic monoid $\SynMon{L}$ endowed with its dfa structure.
  \end{enumerate}
\end{proof}

Analogous to Theorem \ref{thm:tm_min_dfa_iso_syn}, $\jslDfaMin{L}$'s transition semiring is isomorphic to $L$'s syntactic semiring.

\begin{theorem}[$\SynSr{L} \cong \TS{\jslDfaMin{L}}$]
  \label{thm:syn_semiring_trans_semiring}
  We have the idempotent semiring isomorphism:
  \[
    \alpha := \lambda \sem{U}_{\rS_L^\lor}. \lambda X. U^{-1} X : \SynSr{L} \to \TS{\jslDfaMin{L}}.
  \]
  It also defines a $\JSL$-dfa isomorphism $\jslDfaSyn{L} \to \jslDfaTs{\jslDfaMin{L}}$.
\end{theorem}

\begin{proof}
  It is well-defined and injective because:
  \[
    \begin{tabular}{lll}
      $\sem{U_1}_{\SynSrCong{L}} = \sem{U_2}_{\SynSrCong{L}}$
      &
      $\iff \forall x, y \in \Sigma^*.[ x U_1 y \subseteq \overline{L} \iff x U_2 y \subseteq \overline{L}]$
      \\ &
      $\iff \forall x, y \in \Sigma^*.[ x U_1 y \nsubseteq \overline{L} \iff x U_2 y \nsubseteq \overline{L} ]$
      \\ &
      $\iff \forall x, y \in \Sigma^*.[ x U_1 y \cap L \neq \emptyset \iff x U_2 y \cap L  \neq \emptyset ]$
      \\ &
      $\iff \forall x, y \in \Sigma^*.[ U_1 y \cap x^{-1} L \neq \emptyset \iff U_2 y \cap x^{-1} L  \neq \emptyset ]$
      \\ &
      $\iff \forall X \in \LW{L}, y \in \Sigma^*.[ U_1 y \cap X \neq \emptyset \iff U_2 y \cap X  \neq \emptyset ]$
      \\ &
      $\iff \forall X \in \LW{L}, y \in \Sigma^*.[ y \in [U_1]^{-1} X \iff y \in [U_2]^{-1} X  ]$
      \\ &
      $\iff \forall X \in \LW{L}[ [U_1]^{-1} X = [U_2]^{-1} X  ]$
      \\ &
      $\iff \lambda X \in \LQ{L}. [U_1]^{-1} X = \lambda X \in \LQ{L} .[U_2]^{-1} X$.
    \end{tabular}
  \]
  Concerning the final equivalence, $\LW{L}$ join-generates $\jslLQ{L}$ and each $U^{-1}(-)$ preserves unions. Next, $\alpha$ is surjective because $\jslDfaMin{L}$'s transition semiring consists of the endomorphisms $\lambda X. U^{-1} X$ for $U \subseteq \Sigma^*$, or equivalently where $U \in \FPow \Sigma^*$ since $\LW{L}$ is finite. Next, $\alpha$ is a monoid morphism because $\lambda X. \epsilon^{-1} X$ is the identity function and $(UV)^{-1} X = V^{-1}(U^{-1} X)$. Finally it preserves the join structure because $\emptyset^{-1} X = \emptyset$ and $(U \cup V)^{-1} X = U^{-1} X \cup V^{-1} X$.

  Finally we establish the claimed $\JSL$-dfa isomorphism. The transitions follow because $(Ua)^{-1}(X) = a^{-1}(U^{-1} X)$. Concerning final states, $\alpha$ is a join-semilattice isomorphism hence an order isomorphism, so it suffices to show $\alpha$ preserves the largest non-final state. Then we must prove the marked equality below:
  \[
    \lambda X. [U]^{-1} X
    \;\overset{!}{=}\; \lambda X. [\overline{L}]^{-1} X
    \qquad
    \text{where $\sem{U}_{\SynSrCong{L}} := \Lor_{\SynSr{L}} \{ \sem{\{ u \}}_{\SynSrCong{L}} : u \nin L \}$}.
  \]
  Firstly, $U \subseteq \overline{L}$ by well-definedness. Conversely each $u_0 \in \overline{L}$ has some $u \in U$ s.t.\  $\sem{\{ u_0 \}}_{\SynSrCong{L}} = \sem{\{ u \}}_{\SynSrCong{L}}$. Then by an earlier calculation we know $\lambda X. u_0^{-1} X = \lambda X. u^{-1} X$, so the marked equality follows.
\end{proof}


\begin{corollary}
  \label{cor:syn_trans_semiring_jsl_dfa}
  $\jslDfaSyn{L} \cong \jslDfaTs{\jslDfaMin{L}}$.
\end{corollary}

\begin{proof}
  The join-semilattice isomorphism and transitions follows via Theorem 
  \ref{thm:syn_semiring_trans_semiring}. The initial state is preserved i.e.\ $\sem{\{\epsilon\}}_{\SynSrCong{L}} \mapsto id_{\jslDfaMin{L}}$. Lastly the final states are preserved/reflected:
  \[
    \begin{tabular}{lll}
      $\lambda X. U^{-1} X \nleq \gamma_{\overline{L}}$
      &
      $\iff \exists w \in \Sigma^*. U^{-1} (w^{-1} L) \nsubseteq \bigcup_{x \nin L} x^{-1} (w^{-1} L)$
      \\ &
      $\iff \exists w, v \in \Sigma^*, u \in U. (wuv \in L \,\land\, w \overline{L} v \cap L = \emptyset )$
      \\ &
      $\iff U \cap L \neq \emptyset$.
    \end{tabular}
  \]
\end{proof}

\begin{corollary}
  \label{cor:reach_syn_sr_syn_mon}
  $\SynMonDfa{L} \cong \reach{\jslDfaSyn{L}}$.
\end{corollary}

\begin{proof}
  By Corollary \ref{cor:syn_trans_semiring_jsl_dfa} we know $\jslDfaSyn{L} \cong \jslDfaTs{\jslDfaMin{L}}$. Recall that $\minDfa{L} = (L, \LW{L}, \gamma_a, F_\gamma)$ and $\jslDfaMin{L} = (L, \LQ{L}, \delta_a, F_\delta)$ where both $\gamma$ and $\delta$ have action $\lambda X . a^{-1} X$. Observe that:
  \[
    \reach{\jslDfaTs{\jslDfaMin{L}}}
    = (id_{\jslLQ{L}}, \{ \delta_w : w \in \Sigma^* \}, \lambda f. \delta_a \circ f, \{ f : f \nleq \delta_{\overline{L}} \}).
  \]
  By Theorem \ref{thm:tm_min_dfa_iso_syn} it suffices to establish the dfa isomorphism $\lambda \gamma_w . \delta_w  : \TMDfa{\minDfa{L}} \to \reach{\jslDfaTs{\jslDfaMin{L}}}$. It is a well-defined bijection because $\delta_w : \LQ{L} \to \LQ{L}$ is completely determined by its domain-codomain restriction $\gamma_w$. The initial state and transitions of the two dfas are defined in the same way. Finally,
  \[
    \begin{tabular}{lll}
      $\gamma_w(L) \in F_\gamma$
      &
      $\iff \epsilon \in w^{-1} L$
      \\ &
      $\iff w \in L$
      \\ &
      $\iff \exists u \in \Sigma^*. [ \delta_w(u^{-1} L) \nsubseteq \delta_{\overline{L}}(u^{-1} L) ]$
      & (A)
      \\ &
      $\iff \delta_w \nleq \delta_{\overline{L}}$.
    \end{tabular}
  \]
  Concerning (A), $(\oT)$ follows by contradiction whereas $(\To)$ holds by choosing $u := \epsilon$ and observing $\epsilon \nin [\overline{L}]^{-1} L$.
\end{proof}

In order to \emph{dualise} the above constructions one needs the notion of right-quotient closure (see Definition \ref{def:lang_quos}).

\begin{definition}[Right-quotient closure]
  \item
  \begin{enumerate}
    \item
    A $\JSL$-dfa $\delta$ is \emph{right-quotient closed} if $K \in \jslLangs{\delta}$ and $V \subseteq \Sigma^*$ implies $K V^{-1} \in \jslLangs{\delta}$.

    \item
    The \emph{right-quotient closure} of a $\JSL$-dfa $\gamma$ is the simplified $\JSL$-dfa:
    \[
      \jslDfaRqc{\gamma} := (L(\gamma), (T, \cup, \emptyset), \lambda X. a^{-1} X, \{ K \in T : \epsilon \in K \})
    \]
    where $T$ is the closure of $\{ j v^{-1} : j \in J(\jslLangs{\gamma}), \, v \in \Sigma^* \}$ under unions.
    \endbox
  \end{enumerate}
\end{definition}

\begin{lemma}[The right-quotient closure is well-defined]
  Fix any $\JSL$-dfa $\gamma$.
  \begin{enumerate}
    \item
    $\jslDfaRqc{\gamma}$ is a simplified $\JSL$-dfa accepting $L(\gamma)$.
    \item
    $\jslDfaRqc{\gamma}$ is the smallest right-quotient closed $\JSL$-dfa $\delta$ such that $\jslLangs{\gamma} \subseteq \jslLangs{\delta}$.
    \endbox
  \end{enumerate}
\end{lemma}

\begin{proof}
  \item
  \begin{enumerate}
    \item
    $T$ contains $L(\gamma)$ and is closed under unions. It also closed under left-letter-quotients:
    \[
      a^{-1}(\bigcup_{i \in I} j_i v_i^{-1})
      = \bigcup_{i \in I} (a^{-1} j_i) v_i^{-1}
      = \bigcup_{i \in I} (\bigcup_{k \in K_i} j_{i,k}) v_i^{-1}
      = \bigcup_{i \in I, k \in K_i} j_{i,k} v_i^{-1}.
    \]
    Then $\jslDfaRqc{\gamma}$ is a well-defined simplified $\JSL$-dfa accepting $L(\gamma)$ by Lemma \ref{lem:simplified_jsl_dfa_char}.

    \item
    We'll show $\jslDfaRqc{\gamma}$ is right-quotient closed by showing $T$ is right-word-quotient closed (recall $T$ is union-closed):
    \[
      (\bigcup_{i \in I} j_i v_i^{-1}) v^{-1}
      = \bigcup_{i \in I} (j_i v_i^{-1}) v^{-1}
      = \bigcup_{i \in I} j_i (vv_i)^{-1}.
    \]
    Since $\jslDfaRqc{\gamma}$ is simplified by (1) we deduce $\jslLangs{\gamma} \subseteq \jslLangs{\jslDfaRqc{\gamma}}$. Finally is it the smallest such $\JSL$-dfa because every state is the union of right-quotients of languages in $J(\jslLangs{\gamma})$.

  \end{enumerate}
\end{proof}

\begin{theorem}[Transition-semiring dualises right-quotient closure]
  \label{thm:trans_semi_dual_rqc}
  If $\delta$ is a $\JSL$-reachable $\JSL$-dfa then:
  \[
    acc_{(\jslDfaTs{\delta})^\pentagram}
    : (\jslDfaTs{\delta})^\pentagram \to \jslDfaRqc{\delta^\pentagram}
    \qquad
    \text{is a $\JSL$-dfa isomorphism.}
  \]
\end{theorem}

\begin{proof}
  Firstly $\gamma := \jslDfaTs{\delta}$ is $\JSL$-reachable by Lemma \ref{lem:ts_well_defined}, so its dual $\gamma^\pentagram$ is simple by Theorem \ref{thm:jsl_reach_dual_simple}. Then $acc_{\gamma^\pentagram}$ defines a $\JSL$-dfa isomorphism to its simplification $\jslDfaSimple{\gamma^\pentagram}$. We'll show the latter is precisely $\jslDfaRqc{\delta^\pentagram}$. Fix $\delta = (t_0, \aT, \delta_a, F_\delta)$ and $L := L(\delta^\pentagram)$. Then by definition $\gamma = (id_{\aT}, \aS_\delta, \gamma_a, F_\gamma)$ where:
  \[
    S_\delta := \{ \delta_K : \aT \to \aT : K \subseteq \Sigma^* \}
    \qquad
    \aS_\delta := (S_\delta, \cup, \emptyset)
    \qquad
    \gamma_a := \lambda f . \delta_a \circ f.
  \]
  Let us break the argument down into steps.

  \item
  \begin{enumerate}
    \item
    We'll show $\jslLangs{\delta^\pentagram} \subseteq \jslDfaSimple{\gamma^\pentagram}$.
    Fixing any element of $\gamma^\pentagram$ we can rewrite acceptance as follows:
    \[
      \begin{tabular}{lll}
        $u \in acc_{\gamma^\pentagram}(\delta_K)$
        &
        $\iff id_\aT \nleq (\gamma_{u^r})_*(\delta_K)$
        & (by definition)
        \\ &
        $\iff \gamma_{u^r} (id_\aT) \nleq \delta_K$
        & (adjoints)
        \\ &
        $\iff \delta_{u^r} \nleq \delta_K$
        \\ &
        $\iff \exists v \in \Sigma^*. [ \delta_{u^r}(\delta_v(t_0)) \nleq_\aT \delta_{K} (\delta_v(t_0)) ]$
        & ($\delta$ is $\JSL$-reachable)
        \\ &
        $\iff \exists v \in \Sigma^*. [ \delta_{vu^r}(t_0) \nleq_\aT \delta_{v \cdot K}(t_0) ]$
        \\ &
        $\iff \exists v \in \Sigma^*, m \in M(\aT). [
          \delta_{v \cdot K}(t_0) \leq_\aT m \,\land\, \delta_{vu^r}(t_0) \nleq_\aT m
        ]$
        \\ &
        $\iff \exists v \in \Sigma^*,m \in M(\aT). [
          t_0 \leq_\aT (\delta_{v \cdot K})_*(m) \,\land\, t_0 \nleq_\aT (\delta_{vu^r})_*(m)
        ]$
        \\ &
        $\iff  \exists v \in \Sigma^*, j \in J(\aT^\pOp). [
          (\delta_{v \cdot K})_*(j) \leq_{\aT^\pOp} t_0
          \,\land\,
          (\delta_{vu^r})_*(j) \nleq_{\aT^\pOp} t_0
        ]$
        \\ &
        $\iff \exists v \in \Sigma^*, j \in J(\jslLangs{\delta^\pentagram}) .[ uv^r \in j \,\land\, K^r v^r \cap j = \emptyset ]$
        &
        \\ &
        $\iff \exists v \in \Sigma^*, j \in J(\jslLangs{\delta^\pentagram}).[
          u \in j (v^r)^{-1} \,\land\, v^r \nin [K^r]^{-1} j
        ]$
        \\ &
        $\iff \exists v \in \Sigma^*, j \in J(\jslLangs{\delta^\pentagram}).[
          u \in j v^{-1} \,\land\, v \nin [K^r]^{-1} j
        ]$
        & (A).
      \end{tabular}
    \]
    Recalling Corollary \ref{cor:meet_gen_jsl_dfas}.2, for each $j \in J(\jslLangs{\delta^\pentagram})$ we'll show $\delta_{\overline{j^r}}$ accepts $j$. Fixing $j_0$, first observe:
    \[
      \begin{tabular}{lll}
        $v \nin [\overline{j_0}]^{-1} j$
        &
        $\iff \forall x \in \Sigma^*. [ x \in \overline{j_0} \To xv \nin j ]$
        \\ &
        $\iff \forall x \in \Sigma^*. [ x \in \overline{j_0} \To x \nin j v^{-1} ]$
        \\ &
        $\iff \forall x \in \Sigma^*. [ x \in \overline{j_0} \To x \in \overline{j} v^{-1} ]$
        \\ &
        $\iff \overline{j_0}  \subseteq \overline{j} v^{-1}$
        \\ &
        $\iff j v^{-1} \subseteq j_0$.
        & (B).
      \end{tabular}
    \]
    \[
      \begin{tabular}{lll}
        $u \in acc_{\gamma^\pentagram}(\delta_{\overline{j_0^r}})$
        &
        $\iff \exists v \in \Sigma^*, j \in J(\jslLangs{\delta^\pentagram}).[
          u \in j v^{-1} \,\land\, v \nin [\overline{j_0}]^{-1} j
          ]$
        & (by A)
        \\ &
        $\iff \exists v \in \Sigma^*, z \in Z .[
          u \in j v^{-1} \,\land\, j v^{-1} \subseteq j_0
        ]$
        & (by B)
        \\ &
        $\iff u \in j_0$.
      \end{tabular}
    \]
    Thus $\gamma^\pentagram$ accepts every language in $\jslLangs{\delta^\pentagram}$ via closure under joins.

    \item
    Next we show $\gamma^\pentagram$ is right-quotient closed. Aside from the composite endomorphisms $\gamma_w := \lambda f . \delta_w \circ f$ we also have $\phi_w := \lambda f. f \circ \delta_w : \aS_\delta \to \aS_\delta$. They are well-defined because the composition of join-semilattice morphisms is bilinear. Their adjoints witness right-word-quotient closure:
    \[
      \begin{tabular}{lll}
        $u \in acc_{\gamma^\pentagram}((\phi_w)_*(\delta_K))$
        &
        $\iff id_{\aT} \nleq_{\aS_\delta} (\gamma_{u^r})_*((\phi_w)_*(\delta_K))$
        & (by definition)
        \\ &
        $\iff id_{\aT} \nleq_{\aS_\delta} (\phi_w \circ \gamma_{u^r})_*(\delta_K)$
        \\ &
        $\iff \phi_w \circ \gamma_{u^r}(id_{\aT}) \nleq_{\aS_\delta} \delta_K$
        & (adjoints)
        \\ &
        $\iff \delta_{wu^r} \nleq_{\aS_\delta} \delta_K$
        \\ &
        $\iff \gamma_{w u^r}(id_{\aT}) \nleq_{\aS_\delta} \delta_K$
        \\ &
        $\iff id_{\aT} \nleq_{\aS_\delta} (\gamma_{w u^r})_*(\delta_K)$
        & (adjoints)
        \\ &
        $\iff uw^r \in acc_{\gamma^\pentagram}((\phi_w)_*(\delta_K))$
        \\ &
        $\iff u \in acc_{\gamma^\pentagram}((\phi_w)_*(\delta_K)) (w^r)^{-1}$.
      \end{tabular}
    \]
    Closure under right-quotients follows by closure under unions.
    
    \item
    Combining (1) with (2) we deduce $\jslDfaRqc{\delta^\pentagram} \subseteq \jslDfaSimple{\gamma^\pentagram}$. Finally, the reverse inclusion follows by (A) i.e.\ each $acc_{\gamma^\pentagram}(\delta_K)$ is a union of $jv^{-1}$'s.
  \end{enumerate}
\end{proof}

\begin{corollary}[Right-quotient closed vs.\ finite $\Sigma$-generated idempotent semirings]
  \item
  \begin{enumerate}
    \item
    If $\delta$ is a simple right-quotient closed $\JSL$-dfa, $\delta^\pentagram \cong \jslDfaTs{\delta^\pentagram}$ is a $\Sigma$-generated idempotent semiring acting on itself.
    \item
    If $\bS = (\aS, \cdot_\bS, 1_\bS)$ is a finite $\Sigma$-generated idempotent semiring and $s_1 \in S$ then $(1_\bS, \aS, \lambda s. s \cdot_\bS a, \{ s \in S : s \nleq_\aS s_1 \})^\pentagram$ is a simple right-quotient closed $\JSL$-dfa.
  \end{enumerate}
\end{corollary}

\begin{proof}
  \item
  \begin{enumerate}
    \item
    Modulo isomorphism $\delta$ is simplified. Then $\delta = \jslDfaRqc{\delta} \cong (\jslDfaTs{\delta^\pentagram})^\pentagram$ by Theorem \ref{thm:trans_semi_dual_rqc}, so that $\delta^\pentagram \cong \jslDfaTs{\delta^\pentagram}$.

    \item
    First, $\gamma := (1_\bS, \aS, \lambda s. s \cdot_\bS a, \{ s \in S : s \nleq_\aS s_1 \})$ is a well-defined $\JSL$-dfa because right-multiplication preserves joins. It is $\JSL$-reachable because $\bS$ is $\Sigma$-generated. Next we'll show $\lambda f. f(1_\bS) : \jslDfaTs{\gamma} \to \gamma$ is a $\JSL$-dfa isomorphism. It is a well-defined function by construction and surjective because $\bS$ is $\Sigma$-generated and $\gamma_U(1_\bS) = \sem{U}_\bS$. It is injective because each $\gamma_U = \lambda s. s \cdot_\bS \sem{U}_\bS$ acts as right-multiplication by $\sem{U}_\bS$. Finally it preserves joins and multiplication. Then $\gamma^\pentagram \cong (\jslDfaTs{\gamma})^\pentagram \cong \jslDfaRqc{\gamma^\pentagram}$ is simple and right-quotient closed by Theorem \ref{thm:trans_semi_dual_rqc}.
  \end{enumerate}
\end{proof}

\begin{corollary}[Quotients of finite idempotent semirings]
  \item
  \begin{enumerate}
    \item
    Given a $\JSL$-dfa inclusion morphism $\iota : \gamma \hookto \delta$ between simplified right-quotient closed $\JSL$-dfas,
    \[
      \lambda \sem{U}_{\TS{\delta^\pentagram}}.  \sem{U}_{\TS{\gamma^\pentagram}} : \TS{\delta^\pentagram} \epito \TS{\gamma^\pentagram}
    \]
    is a well-defined surjective semiring morphism.
    
    \item
    Let $f : (\aS, \cdot_{\bS}, 1_{\bS}) \epito (\aT, \cdot_{\bT}, 1_{\bT})$ be a surjective semiring morphism where $\bS$ is a finite $\Sigma$-generated idempotent semiring. Given any $s_0 \in S$ we have the $\JSL$-dfa embedding:
    \[
      f_* :
      (id_T, \aT, \lambda t. t \cdot_\bT f(a), \{ t : t \nleq_\aT f(s_0) \})^\pentagram \to (id_S, \aS, \lambda s. s \cdot_\bS a, \{ s : s \nleq_\aS s_0 \} )^\pentagram
    \]
     between simple right-quotient closed $\JSL$-dfas.
  \end{enumerate}
\end{corollary}

\begin{proof}
  \item
  \begin{enumerate}
    \item
    Firstly $\iota_* : \delta^\pentagram \to \gamma^\pentagram$ is a surjective $\JSL$-dfa morphism by Theorem \ref{thm:dfa_jsl_self_dual}. Since $\gamma$ and $\delta$ are right-quotient closed,
    \[
      \jslDfaTs{\delta^\pentagram}
      \cong (\jslDfaRqc{\delta})^\pentagram
      = (\delta)^\pentagram
      \longepi{\iota_*}
      (\gamma)^\pentagram
      = (\jslDfaRqc{\gamma})^\pentagram
      \cong \jslDfaTs{\gamma^\pentagram}
    \]
    by applying Theorem \ref{thm:trans_semi_dual_rqc}. Then we have the surjective $\JSL$-dfa morphism $f : \jslDfaTs{\delta^\pentagram} \epito \jslDfaTs{\gamma^\pentagram}$. It is a join-semilattice morphism preserving the unit (initial state) and right-multiplication by generators. Then $f((\delta_a)_*) := (\gamma_a)_*$ and thus $f((\delta_U)_*) = (\gamma_U)_*$ by induction over words and joins. Then $f$ preserves the multiplication too:
    \[
      \begin{tabular}{lll}
        $f((\delta_V)_* \circ (\delta_U)_*)$
        &
        $= f((\delta_U \circ \delta_V)_*)$
        \\ &
        $= f((\delta_{V \cdot U})_*)$
        \\ &
        $= (\gamma_{V \cdot U})_*$
        \\ &
        $= (\gamma_U \circ \gamma_V)_*$
        \\ &
        $= (\gamma_V)_* \circ (\gamma_U)_*$
      \end{tabular}
    \]
    so it is a surjective semiring morphism. Finally it preserves the generators so has the claimed description.
    
    \item
    The surjective semiring morphism also defines a $\JSL$-dfa morphism:
    \[
      f : (id_S, \aS, \lambda s. s \cdot_\bS a, \{ s : s \nleq_\aS s_0 \} )
      \to
      (id_T, \aT, \lambda t. t \cdot_\bT f(a), \{ t : t \nleq_\aT f(s_0) \})
    \]
    because right-multiplication preserves joins. Both $\JSL$-dfas are $\JSL$-reachable because $f$ is surjective, so that $f[\Sigma]$ generates $\bT$. Then its adjoint defines an injective $\JSL$-dfa morphism between simple $\JSL$-dfas. Finally each $\JSL$-dfa is right-quotient closed via closure under left multiplication on the dual side.
  \end{enumerate}
\end{proof}

Next we dualise the syntactic semiring.

\begin{definition}[$L$'s minimal syntactic $\JSL$-dfa]
    The closure of $\LRW{L} := \{ u^{-1} L v^{-1} : u,\, v \in \Sigma^* \}$ under unions defines the \emph{minimal syntactic $\JSL$-dfa} $\jslDfaSynMin{L}$ i.e.\ the smallest right-quotient closed $\JSL$-dfa accepting $L$.
    \endbox
\end{definition}

\begin{note}
  The minimal syntactic $\JSL$-dfas satisfies $\jslDfaSynMin{L} = \jslDfaRqc{\jslDfaMin{L}}$.
  \endbox
\end{note}

\begin{corollary}[Dualising the syntactic semiring]
  \label{cor:dual_syn_semiring}
  We have the $\JSL$-dfa isomorphism:
  \[
    acc_{(\jslDfaSyn{L^r})^\pentagram} : (\jslDfaSyn{L^r})^\pentagram \to \jslDfaSynMin{L}
  \]
\end{corollary}

\begin{proof}
  First let $\delta := \jslDfaMin{L^r}$. By Theorem \ref{thm:syn_semiring_trans_semiring} we know $\jslDfaSyn{L^r} \cong \jslDfaTs{\delta}$ so that $(\jslDfaSyn{L^r})^\pentagram \cong (\jslDfaTs{\delta})^\pentagram$. By Theorem \ref{thm:trans_semi_dual_rqc} we know $(\jslDfaTs{\delta})^\pentagram \cong \jslDfaRqc{\delta^\pentagram}$ because $\delta$ is $\JSL$-reachable (see Corollary \ref{cor:min_jsl_dfa_characterisation}). Finally by Theorem \ref{thm:dr_L} we have $\delta^\pentagram \cong \jslDfaMin{L}$, so that $(\jslDfaSyn{L})^\pentagram \cong \jslDfaRqc{\jslDfaMin{L}}$. Since $\jslDfaRqc{\jslDfaMin{L}}$ is simplified this isomorphism must be the acceptance map.
\end{proof}

Finally we describe the dual of the power semiring of the syntactic monoid. It is essentially the canonical boolean syntactic $\JSL$-dfa from Definition \ref{def:canon_bool_syn_jsl_dfa}.

\begin{corollary}[Dualising $\Pow \SynMon{L}$]
  \label{cor:dual_pow_semiring_syn}
  Let $\delta := \jslDfaSc{\SynMonDfa{L}}$.
  \begin{enumerate}
    \item
    We have the semiring isomorphism
    $
      \lambda \delta_U. \{ \sem{u}_{\SynCong{L}} : u \in U \} : \TS{\delta} \to \Pow \SynMon{L}
    $.
    \item
    We have the $\JSL$-dfa isomorphism $acc_{(\jslDfaTs{\delta})^\pentagram} : (\jslDfaTs{\delta})^\pentagram \to \jslDfaSynBoolMin{L}$.
  \end{enumerate}
\end{corollary}

\begin{proof}
  \item
  \begin{enumerate}
    \item
    Denote the candidate isomorphism by $\alpha$.
    Given $\gamma := \SynMonDfa{L}$ then $\delta_U = \lambda S. \bigcup_{u \in U} \gamma_u[S]$. Given $\delta_{U_1} = \delta_{U_2}$ then applying them to $\{ \sem{\epsilon}_{\SynCong{L}} \}$ we see $\alpha$ is a well-defined injective function. It is clearly surjective and also preserves joins i.e.\ $\alpha(\delta_{U_1 \cup U_2}) = \alpha(\delta_{U_1}) \cup \alpha(\delta_{U_2})$. Finally $\alpha(\delta_\emptyset) = \emptyset$ and the multiplication is also preserved:
    \[
      \begin{tabular}{lll}
        $\alpha(\delta_V \circ \delta_U)$
        &
        $= \alpha(\delta_{U \cdot V})$
        \\ &
        $= \{ \sem{x}_{\SynCong{L}} : x \in U \cdot V \}$
        \\ &
        $= \{ \sem{u}_{\SynCong{L}} : u \in U \} \cdot \{ \sem{v}_{\SynCong{L}} : v \in V \}$
        \\ &
        $= \alpha(\delta_U) \cdot \alpha(\delta_V)$.
      \end{tabular}
    \]

    \item
    By Example \ref{ex:dual_full_subset_construction} $\delta^\pentagram \cong \jslDfaSc{\rev{\SynMonDfa{L^r}}}$ hence $\delta^\pentagram \cong \jslDfaSynBoolMin{L}$ by Theorem \ref{thm:canon_bool_syn_dep_aut}. Applying Theorem \ref{thm:trans_semi_dual_rqc} yields:
    \[
      (\jslDfaTs{\delta})^\pentagram
      \cong \jslDfaRqc{\delta^\pentagram}
      \cong \jslDfaRqc{\jslDfaSynBoolMin{L}}
      = \jslDfaSynBoolMin{L},
    \]
    since the latter is simplified and right-quotient closed by construction.
  \end{enumerate}
\end{proof}

\section{The Kameda-Weiner Algorithm and Beyond}

\subsection{$L$-coverings}

An $L$-covering is an edge-covering of the dependency relation $\rDR{L}$ (Definition \ref{def:canon_dep_aut}) by left-maximal bicliques. That is, each biclique $A \times B \subseteq \rDR{L}$ is inclusion-maximal on the left. Importantly, they can be defined as certain $\Dep$-morphisms.

\begin{definition}[$L$-coverings]
  \label{def:l_covering}
  Fix any regular language $L \subseteq \Sigma^*$.
  \begin{enumerate}
    \item 
    An \emph{$L$-covering} is a $\Dep$-morphism $\rDR{L} : \rDR{L} \to \rH$ such that $\rH_t = \LW{L^r}$.
    
    The $\Dep-$morphism is determined by $\rH$, so it may be denoted $\rLCov{L}{\rH}$.
    We may also refer to the $L$-covering via the relation $\rH \subseteq \rH_s \times \LW{L^r}$ alone.

    \item
    Given an $L$-covering $\rH$, Definition \ref{def:maximum_witnesses} provides $\rLCov{L}{\rH}_- \subseteq \LW{L} \times \rH_s$. But we may also directly define:
    \[
      \rLCov{L}{\rH}_-(u^{-1} L, h_s)
      :\iff
      \forall Y \in \LW{L^r}. [
          \rH(h_s, Y) \To \rDR{L}(u^{-1} L , Y)
        ]
    \]
    without knowing $\rH \subseteq \rH_s \times \LW{L^r}$ is an $L$-covering.
    \endbox
  \end{enumerate}
\end{definition}


\begin{definition}[$L$-covering constructions]
  \label{def:l_covering_constructs}
  Fix an $L$-covering $\rDR{L} : \rDR{L} \to \rH$.

  \begin{enumerate}
    \item 
    $\rH$'s \emph{biclique-form} is the $L$-covering $\rH^\flat \subseteq \rH_s^\flat \times \LW{L^r}$ where:
    \[
      \rH^\flat_s :=
      \{ \rLCov{L}{\rH}_-\spbreve [h_s] \times \rH[h_s] : h_s \in \rH_s  \}
      \qquad
      \rH^\flat(A \times B, Y) :\iff Y \in B.
    \]
    It turns out that $\rLCov{L}{\rH^\flat}_-(X, A \times B) \iff X \in A$.
    Finally, we say $\rH$ is \emph{in biclique-form} if $\rH = \rH^\flat$.

    \item
    $\rH$'s \emph{induced nfa} $\rN_\rH$ has states $\rH_s$ and is defined:
    \[
      I_{\rN_\rH} := \rLCov{L}{\rH}_-[L]
      \qquad
      F_{\rN_\rH} := \{ h \in \rH_s : \epsilon \in \bigcap \rLCov{L}{\rH}_-\spbreve[h] \}
    \]
    \[
      \rN_{\rH,a} (h_1, h_2)
      :\iff \forall X \in \LW{L}. (
        \rLCov{L}{\rH}_-(X, h_1) \To \rLCov{L}{\rH}_-(a^{-1} X, h_2)
      ).
    \]
    Just as $\rLCov{L}{\rH}$ is completely determined by $\rH$, the induced nfa $\rN_\rH$ is completely determined by $\rLCov{L}{\rH}$.

    \item 
    An $L$-covering $\rH'$ \emph{extends} another $L$-covering $\rH$ if $\rH_s = \rH'_s$, $\rH \subseteq \rH'$ and $\rLCov{L}{\rH}_- = \rLCov{L}{\rH'}_-$. We say $\rH$ is \emph{maximal} if its only extension is itself.

    \item
    $\rH$ is \emph{legitimate} if $L(\rN_\rH) = L$, see \cite[Definition 16]{KamedaWeiner1970}.

    \item
    $\rH$'s \emph{dual} is the $L^r$-covering $\rH^\diamond := \rLCov{L}{\rH}_-\spbreve \subseteq \rH_s \times \LW{L}$.
    \endbox

  \end{enumerate}
\end{definition}

\begin{note}[Concerning extensions of $L$-coverings]
  Given any two $L$-extensions satisfying $\rH_s = \rH'_s$ and $\rH \subseteq \rH'$ we necessarily have $\rLCov{L}{\rH'}_- \subseteq \rLCov{L}{\rH}_-$.
  Then Definition \ref{def:l_covering_constructs}.3 could equivalently require $\rLCov{L}{\rH}_- \subseteq \rLCov{L}{\rH'}_-$, which is more in-keeping with `maximality'.
  \endbox
\end{note}

We now prove various basic facts concerning $L$-coverings.

\begin{lemma}[$L$-coverings]
  \label{lem:l_coverings}
  
  \item
  \begin{enumerate}
    \item
    $\rH \subseteq \rH_s \times \LW{L^r}$ is an $L$-covering iff $\rDR{L} = \rI ; \rH$ for some $\rI \subseteq \LW{L} \times \rH_s$.
    
    \item 
    Each $L$-covering is a $\Dep$-monomorphism via the witnesses:
    \[
      \xymatrix@=15pt{
        \LW{L^r} \ar[rr]^-{\Delta_{\LW{L^r}}} && \LW{L^r}
        \\
        \LW{L} \ar[u]^{\rDR{L}} \ar[rr]_-{\rLCov{L}{\rH}_-} && \rH_s \ar[u]_{\rH}
      }
    \]

    \item
    If $\rH$ is an $L$-covering then so is its biclique-form $\rH^\flat$; moreover $\rLCov{L}{\rH^\flat}_- (X, A \times B) \iff X \in A$.

    \item
    If $\rH$ is an $L$-covering in biclique-form then its induced nfa satisfies:
    \[
      A \times B \in I_{\rN_\rH} \iff L \in A
      \qquad
      A \times B \in F_{\rN_\rH} \iff \epsilon \in \bigcap A
    \]
    \[
      \rN_{\rH,a} (A_1 \times B_1, A_2 \times B_2) \iff \gamma_a [A_1] \subseteq A_2
      \qquad
      \text{where $\gamma_a := \lambda X. a^{-1} X$}.
    \]

    \item
    $L(\rN_{\rH}) = L(\rN_{\rH^\flat})$ because
    $
      q := \lambda h_s. \rLCov{L}{\rH}_-\spbreve[h_s] \times \rH[h_s] : \rN_{\rH} \epito \rN_{\rH^\flat}
    $
    is a surjection which preserves/reflects initial states, final states and transitions. If $\rN_\rH$ is state-minimal then $q$ is an nfa isomorphism.\footnote{However, even when $\rN_\rH$ is not state-minimal $q$ is almost the same thing as an isomorphism.}

    \item
    $L(\rN_{\rH}) \subseteq L$ for each $L$-covering $\rH$.

    \item
    Let $\rH$ be an $L$-covering.
    \begin{enumerate}[a.]
      \item 
      $\rH^\diamond$ is a well-defined maximal $L^r$-covering.
      
      \item
      $\rH \subseteq \rH^{\diamond\diamond}$ and $\rLCov{L}{\rH}_- = (\rH^\diamond)\spbreve = \rLCov{L}{\rH^{\diamond\diamond}}_-$, so $\rH^{\diamond\diamond}$ extends $\rH$.

      \item
      $\rH$ is maximal iff $\rH = \rH^{\diamond\diamond}$.
      
      \item
      $\rH^\flat$ is maximal if $\rH$ is.
      
      \item
      $\rH^{\diamond\diamond}$ is legitimate if $\rH$ is.

      \item
      $A \times B \in (\rH^{\diamond\diamond})_s^\flat \iff B \times A \in (\rH^\diamond)_s^\flat$.
    \end{enumerate}

    \item
    If $\rH$ is an $L$-covering in biclique-form and $A \times B \in (\rN_\rH)_u[I_{\rN_\rH}]$ then $u^{-1} L \in A$.


  \end{enumerate}
\end{lemma}

\begin{proof}
  \item
  \begin{enumerate}
    \item 
    If $\rDR{L} : \rDR{L} \to \rH$ is an $L$-covering it is a $\Dep$-morphism, so $\rLCov{L}{\rH}_- ; \rH = \rDR{L}$ via the maximum witnesses (Lemma \ref{lem:mor_char_max_witness}). Conversely if $\rDR{L} = \rI ; \rH$ we know $\rI ; \rH = \rDR{L} = \rDR{L} ; \Delta_{\LW{L^r}}\spbreve$ hence $\rDR{L}$ is a $\Dep$-morphism.

    \item
    The commuting diagram follows via maximum witnesses. It defines a $\Dep$-mono because the upper witness is a bijective function, recalling $\Dep$-composition from Definition \ref{def:the_cat_dep}.

    \item
    Since $\rH$ is an $L$-covering we know $\rLCov{L}{\rH}_- ; \rH = \rDR{L}$. For completely general reasons $\bigcup \rH^\flat_s = \rDR{L}$ i.e.\ the union of the cartesian products $\rLCov{L}{\rH}_-\spbreve[h_s] \times \rH[h_s]$ is $L$'s dependency relation (see Note \ref{note:bipartitioned_as_binary}). If we define $\rI \subseteq \LW{L} \times \rH^\flat_s$ as $\rI(X, A \times B) :\iff X \in A$ then:
    \[
      \begin{tabular}{lll}
        $\rI; \rH^\flat (X, Y)$
        &
        $\iff \exists A \times B \in \rH^\flat_s. [ X \in A \,\land\, Y \in B ]$
        \\ &
        $\iff \exists h_s \in \rH_s. [ (X,Y) \in \rLCov{L}{\rH}_-\spbreve[h_s] \times \rH[h_s] ]$
        \\ &
        $\iff (X, Y) \in \bigcup \rH^\flat_s$
        \\ &
        $\iff \rDR{L}(X, Y)$.
      \end{tabular}
    \]
    Then by (1) $\rH^\flat$ is a well-defined $L$-covering. It remains to establish $\rLCov{L}{\rH^\flat}_- = \rI$. To this end, let $A \times B = \rS_-\spbreve[h_s] \times \rH[h_s]$ and consider:
    \[
      \begin{tabular}{lll}
        $\rLCov{L}{\rH^\flat}_-(X, A \times B)$
        &
        $:\iff \rH^\flat[A \times B] \subseteq \rDR{L}[X]$
        & (definition \ref{def:maximum_witnesses})
        \\ &
        $\iff B \subseteq \rDR{L}[X]$
        \\ &
        $\iff \rH[h_s] \subseteq \rDR{L}[X]$
        & (by def.\ of $B$).
        \\ &
        $\iff \rLCov{L}{\rH}_-(h_s, X)$
        & (definition \ref{def:maximum_witnesses})
        \\ &
        $\iff X \in \rLCov{L}{\rH}_-[h_s]$
        \\ &
        $\iff X \in A$
        & (by def.\ of $A$).
      \end{tabular}
    \]

    \item
    Concerning the transitions:
    \[
      \begin{tabular}{lll}
        $\rN_{\rH,a} (A_1 \times B_1, A_2 \times B_2)$
        \\[0.5ex]
        $:\iff \forall X \in \LW{L}. [ \rLCov{L}{\rH}_-(X, A_1 \times B_1) \To \rLCov{L}{\rH}_-(a^{-1} X, A_2 \times B_2) ]$
        \\
        $\iff \forall X \in \LW{L}. [ X \in A_1 \To a^{-1} X \in A_2 ]$
        & (by (3))
        \\
        $\iff \gamma_a[A_1] \subseteq A_2$.
      \end{tabular}
    \]
    The characterisations of the initial/final states follow easily.

    \item
    Consider the well-defined surjection $q : \rH_s \epito \rH^\flat_s$ with action $\lambda h_s. \rLCov{L}{\rH}_-\spbreve[h_s] \times \rH[h_s]$. It preserves and reflects the initial states and also the final states:
    \[
      h_s \in I_{\rN_\rH}
      \iff h_s \in \rLCov{L}{\rH}_-[L]
      \iff L \in \rLCov{L}{\rH}_-\spbreve[h_s]
      \iff q(h_s) \in I_{\rN_{\rH^\flat}}
    \]
    \[
      h_s \in F_{\rN_\rH}
      \iff \epsilon \in \bigcap \rLCov{L}{\rH}\spbreve[h_s]
      \iff q(h_s) \in F_{\rN_{\rH^\flat}}.
    \]
    The transitions are also preserved and reflected:
    \[
      \begin{tabular}{lll}
        $\rN_{\rH^\flat}(q(h_1), q(h_2))$
        &
        $\iff \forall X \in \LW{L}.[ X \in \rLCov{L}{\rH}_-\spbreve[h_1] \To a^{-1} X \in \rLCov{L}{\rH}_-\spbreve[h_2] ]$
        \\ &
        $\iff \forall X \in \LW{L}. [ \rLCov{L}{\rH}_-(X, h_1) \To \rLCov{L}{\rH}_-(a^{-1} X, h_2)]$
        \\ &
        $\iff \rN_\rH(h_1, h_2)$.
      \end{tabular}
    \]
    Thus $L(\rN_{\rH}) = L(\rN_{\rH^\flat})$ because the nfas simulate one another. If $\rN_\rH$ is state-minimal $q$ must be an isomorphism.

    \item
    By (5) we may assume the $L$-covering is in biclique-form. The induced nfa $\rN_{\rH}$ is described in (4). If $w \in L(\rN_\rH)$ then by induction we have $\gamma_w[A_1] \subseteq A_n$ where $L \in A_1$, $A_n \times B_n \subseteq \rDR{L}$ and $\epsilon \in \bigcap A_n$. Thus $\epsilon \in w^{-1} L$ so $w \in L$.

    \item
    \begin{enumerate}[a.]
      \item 
      By Lemma \ref{lem:l_coverings}.1 we know $\rLCov{L}{\rH}_- ; \rH = \rDR{L}$ hence $\breve{\rH} ; \rLCov{L}{\rH}_-\spbreve = (\rDR{L})\spbreve = \rDR{L^r}$. Thus $\rH^\diamond := \rLCov{L}{\rH}_-\spbreve$ is an $L^r$-covering by applying Lemma \ref{lem:l_coverings}.1 again. Finally $\rH^\diamond$ is \emph{maximal} because its converse is a maximal lower witness.
      
      \item
      Below on the left we've depicted $\rH$ together with the respective lower component $\rLCov{L}{\rH}_-$.
      \[
        \xymatrix@=15pt{
          \LW{L^r} \ar[rr]^-{\Delta_{\LW{L^r}}} && \LW{L^r}
          \\
          \LW{L} \ar[rr]_-{\rLCov{L}{\rH}_-} \ar[u]^-{\rDR{L}} && \rH_s \ar[u]_-{\rH}
        }
        \qquad
        \xymatrix@=15pt{
          \LW{L} \ar[rr]^-{\Delta_{\LW{L}}} && \LW{L}
          \\
          \LW{L^r} \ar[rr]_-{\rLCov{L}{\rH^\diamond}_-} \ar[u]^-{\rDR{L^r}} && \rH_s \ar[u]_-{\rH^{\diamond}}
        }
        \qquad
        \xymatrix@=15pt{
          \LW{L^r} \ar[rr]^-{\Delta_{\LW{L^r}}} && \LW{L^r}
          \\
          \LW{L} \ar[rr]_-{(\rH^\diamond)\spbreve} \ar[u]^-{\rDR{L}} && \rH_s \ar[u]_-{\rH^{\diamond\diamond}}
        }
      \]
      The central diagram shows $\rH$'s dual $L^r$-covering $\rH^\diamond := \rLCov{L}{\rH}_-\spbreve$ and the respective lower component $\rLCov{L}{\rH^\diamond}_-$. Then the central diagram arises by dualising $\rH$ and the right-most diagram arises by dualising $\rH^\diamond$. Notice that since $\rH^\diamond$ is already maximal, the right-most square swaps and reverses \emph{both} relations. In particular:
      \begin{itemize}
        \item[--]
        from left to right $\rLCov{L}{\rH}_- = (\rH^\diamond)\spbreve = \rLCov{L}{\rH^{\diamond\diamond}}_-$.
        \item[--]
        $\rH^\diamond = \rLCov{L}{\rH}_-\spbreve$ implies $\breve{\rH} \subseteq \rLCov{L^r}{\rH^\diamond}_-$ by maximality, hence $\rH \subseteq \rLCov{L^r}{\rH^\diamond}_-\spbreve = \rH^{\diamond\diamond}$.
      \end{itemize}

      \item
      If $\rH$ is maximal then $\rH = \rH^{\diamond\diamond}$ by (b).
      Conversely if $\rH = (\rH^\diamond)^\diamond$ then it is maximal by (a).

      \item
      If $\rH$ is maximal then $\rH^\flat$ amounts to a bijective relabelling of $\rH_s$, so it is also maximal.

      \item
      Since $\rLCov{L}{\rH}_- = \rLCov{L}{\rH^{\diamond\diamond}}_-$ we know $\rN_{\rH} = \rN_{\rH^{\diamond\diamond}}$, hence $\rH^{\diamond\diamond}$ is also legitimate.

      \item
      We have $\rH^{\diamond\diamond} = \rLCov{L^r}{\rH^\diamond}_-\spbreve$ and also $\rLCov{L}{\rH^{\diamond\diamond}}_- = (\rH^\diamond)\spbreve$ by (b). Then constructing the biclique-form of $\rH^\diamond$ amounts to constructing bicliques:
      \[
        \rLCov{L^r}{\rH^\diamond}_-\spbreve [h_s] \times \rH^\diamond[h_s]
        = \rH^{\diamond\diamond} [h_s] \times \rLCov{L}{\rH^{\diamond\diamond}}_-\spbreve [h_s]
      \]
      and the claim the follows.

    \end{enumerate}

    \item
    Given $A \times B \in (\rN_\rH)_u(I_{\rN_\rH})$ we'll prove $u^{-1} L \in A$ by induction on $u$. If $u = \epsilon$ this holds by definition of $I_{\rN_\rH}$.
    If $u = u_0 a$ we have $A_0 \times B_0 \in (\rN_\rH)_{u_0}[I_{\rN_\rH}]$ and $a^{-1}[A_0] \subseteq A_1$ by Lemma \ref{lem:l_coverings}.4. Then by induction $u_0^{-1} L \in A_0$ and hence $(u_0 a)^{-1} L \in A$, so we are done.





  \end{enumerate}
\end{proof}

\subsection{Saturated machines}

There are various ways an nfa can be have many initial/final states and transitions.

\begin{definition}[Locally/intersection-saturated and transition-maximality]
  \label{def:local_sat_trans_max}
  Let $\rN = (I, Z, \rN_a, F)$ be an nfa.
  \begin{enumerate}
    \item 
    $\rN$ is \emph{locally-saturated} if for all $a \in \Sigma$ and $z, z_1, z_2 \in Z$,
    \[
      z \in I \iff L(\rN_{@z}) \subseteq L(\rN)
      \qquad
      \rN_a(z_1, z_2) :\iff L(\rN_{@z_2}) \subseteq a^{-1} L(\rN_{@z_1}).
    \]

    \item
    $\rN$ is \emph{intersection-saturated} if for all $z, z_1, z_2 \in Z$.
    \[
      \begin{tabular}{c}
        $\rN_a(z_1, z_2)
        \iff \forall u \in \Sigma^*. ( z_1 \in \rN_u [I_\rN] \To z_2 \in \rN_{ua} [I_\rN] )$
        \\[1ex]
        $z \in F_\rN \iff \forall u \in \Sigma^*.( z \in \rN_u[I_\rN] \To \rN_u[I_\rN] 
        \cap F_\rN \neq \emptyset)$.
      \end{tabular}
    \]
    These are the conditions for transitions and final states from Kameda and Weiner's intersection rule \cite{KamedaWeiner1970}.

    \item
    $\rN$ is \emph{transition-maximal} if adding transitions or colouring additional initial/final states changes the accepted language. More formally, \emph{an nfa $\rM$ extends $\rN$} if $\rM = (I_\rM, Z, \rM_a, F_\rM)$ where $I \subseteq I_\rM$, each $\rN_a \subseteq \rM_a$, $F \subseteq F_\rM$ and finally $L(\rM) = L(\rN)$. Then $\rN$ is \emph{transition-maximal} if its only extension is itself.
    \endbox
    
  \end{enumerate}
\end{definition}

The concept of being \emph{locally-saturated} arises naturally from canonical constructions, as we'll see. It is `local' because one can enforce it without changing the languages accepted by the individual states. It is worth clarifying the second concept straight away.

\begin{note}
  \label{note:intersection_saturated_basic_char}
  An nfa $\rN$ is intersection-saturated iff the following hold:
  \begin{enumerate}
    \item[--]
    {\em whenever for every $u$-path to $z_1$ there exists a $ua$-path to $z_2$ then $\rN_a(z_1, z_2)$.}
    \item[--]
    {\em whenever for every $u$-path to $z$ we have $u \in L(\rN)$ then $z$ is final.}
  \end{enumerate}
  
  Then each transition relation $\rN_a$ can be reconstructed from the deterministic transitions $\rN_u[I] \to_a \rN_{ua}[I]$ of the reachable subset construction. Soon we'll prove an nfa $\rN$ is intersection-saturated iff $\rev{\rN}$ is locally-saturated.
  \endbox
\end{note}

Perhaps unsurprisingly, \emph{transition-maximal} machines are both locally-saturated and intersection-saturated. We now provide various examples of nfas in different classes.

\begin{example}[Comparing notions of saturation]
  \item
  \begin{enumerate}
    \item
    \emph{Locally but not intersection-saturated (via final states)}. The nfa below accepts $a + aa$ and is locally-saturated e.g.\ there is no transition from the left-most state to the right-most because $\{ \epsilon \} \nsubseteq a^{-1} \{ aa \}$. However it is not intersection-saturated because the central state should be final by Note \ref{note:intersection_saturated_basic_char}.
    \[
      \xymatrix@=5pt{
        {\tt i} \ar[rr]^a  && {\tt i} \ar[rr]^a && {\tt o} 
      }
    \]

    \item
    \emph{Locally but not intersection-saturated (via transitions)}. This locally-saturated nfa accepts $a(bb^* + cc^*)$:
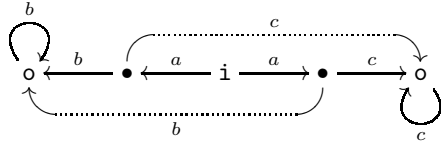
    \[
      \xymatrix@=10pt{
        {\tt o} \ar@(ul,ur)^b && \bullet \ar[ll]_b \ar@{..<}`u[r]`[rrrrrr]^-c[rrrrrr] && {\tt i} \ar[rr]^a \ar[ll]_a && \bullet \ar@{..>}`d`[llllll]^-b[llllll] \ar[rr]^c && {\tt o} \ar@(dr,dl)^c
        \\
      }
    \]
    It is not intersection-saturated because by Note \ref{note:intersection_saturated_basic_char} it should have the dashed transitions too.

    \item
    \emph{Intersection-saturated but not locally-saturated}. Take the reverse nfa of either (1) or (2). This follows by Theorem \ref{thm:intersection_rule_char} further below.

    \item
    \emph{Locally-saturated, not transition-maximal}. Example (2) is locally-saturated but not transition-maximal.

    \item
    \emph{Locally and intersection-saturated, not transition-maximal}.
    This nfa accepts $L := a + b + (a^+ + b^+)c^+$ and is locally-saturated e.g.\ there is no dashed $c$-transition because $c^* \nsubseteq c^{-1} \{ \epsilon \}$. It is also intersection-saturated.
    
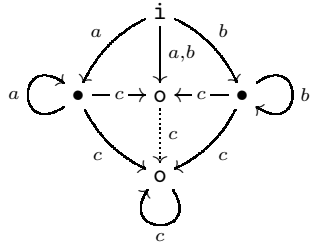
\[
      \xymatrix@=20pt{
        & {\tt i}  \ar@/_5pt/[dl]_a \ar[d]^{a,b} \ar@/^5pt/[dr]^b
        \\
        \bullet \ar@(dl,ul)^a \ar@/_5pt/[dr]_c \ar[r]|{\;c\;} & {\tt o} \ar@{..>}[d]^c &  \bullet \ar@(ur,dr)^b \ar@/^5pt/[dl]^c \ar[l]|{\;c\;}
        \\
        & {\tt o} \ar@(dr,dl)^c
      }
    \]
    However, it is not transition-maximal -- adding the dashed $c$-transition preserves the accepted language.

    \item
    The nfa $\minSatDfa{L}$ from Example \ref{ex:saturated_min_dfa} is always transition-maximal, as the reader may verify.
    \endbox
  \end{enumerate}
\end{example}

There is a canonical way to locally saturate an nfa.

\begin{definition}[Irreducible simplification]
  \label{def:nfa_irr_simplification}
  We define the \emph{irreducible simplification} of an nfa $\rN = (I, Z, \rR_a, F)$ as:
  \[
    \simpleIrr{\rN} := (
      \{ X \in J(\aS) : X \subseteq L(\rN) \},
      J(\aS),
      \lambda X. a^{-1} X,
      \{ X \in J(\aS) : \epsilon \in X \}
    )
  \]
  where $\aS := \jslLangs{\dep{\rN}}$ is the join-semilattice of languages accepted by $\rN$.
  \endbox
\end{definition}

\begin{note}[Irreducible simplification is canonical]
  $\simpleIrr{\rN}$ is the lower nfa of $\Airr(\jslDfaSimple{\dep{\rN}})$.
  \endbox
\end{note}

\begin{lemma}[Concerning irreducible simplifications]
  \label{lem:irr_simplifying_nfas}
  \item
  \begin{enumerate}
    \item 
    $\simpleIrr{\rN}$ accepts $L(\rN)$.
    \item
    $L((\simpleIrr{\rN})_{@Y}) = Y$ for each state $Y \in J(\jslLangs{\rN})$.
    \item
    $\simpleIrr{-}$ preserves reachability.
    \item
    $\simpleIrr{-}$ is idempotent.
  \end{enumerate}
\end{lemma}
\begin{proof}
  \item
  \begin{enumerate}
    \item 
    $\Det(-)$ and $\Airr(-)$ preserve the accepted language by Note \ref{note:det_airr_preserve_acceptance}, the latter defined in terms of the lower nfa. Finally $\jslDfaSimple{-}$ preserves the accepted language since, ignoring the join-semilattice structure, it is a sub-dfa.

    \item
    Each state $X$ in $\delta := \jslDfaSimple{\Det(\rN, \Delta_Z, \rev{\rN})}$ accepts $X$ by Lemma \ref{lem:jsl_reach_simple_constructions}.3. The lower nfa of $\Airr\delta$ accepts $L(\rN)$ and is $\simpleIrr{\rN}$. For $Y \in J(\jslLangs{\rN})$, the lower nfa of $\Airr (\delta_{@Y})$ accepts $Y$ and is $(\simpleIrr{\rN})_{@I_Y}$ where $I_Y$ is the principal downset generated by $Y$. Thus $(\simpleIrr{\rN})_{@\{ Y \}}$ accepts $Y$ since $acc_\delta$ is monotonic.
    
    \item
    Let $\rN = (I, Z, \rN_a, F)$ be reachable and $\delta = \Det(\rN, \Delta_Z, \rev{\rN})$. By surjectivity, given state $X$ in $\simpleIrr{\rN}$ there exists $z \in Z$ such that $X = acc_\delta(\{ z \})$. By reachability we have a path in $\rN$:
    \[
      I \ni z_1 \to_{a_1} \cdots \to_{a_n} z_n = z \ni F
    \]
    so that $L(\rN_{@z_{i + 1}}) \subseteq a^{-1} L(\rN_{@z_i})$ for each $0 \leq i < n$. This implies:
    \[
      I_{\simpleIrr{\rN}} \ni acc_\delta(\{z_1 \}) \to_{a_1} \cdots \to_{a_n} acc_\delta(\{ z_n \}) = X \in F_{\simpleIrr{\rN}}  
    \]
    in the nfa $\simpleIrr{\rN}$.

    \item
    Follows by (2).
  \end{enumerate}
\end{proof}

\begin{lemma}
  \label{lem:irr_simple_is_locally_saturated}
  $\simpleIrr{\rN}$ is locally-saturated with no more states than $\rN$.
\end{lemma}

\begin{proof}
  Recall Definition \ref{def:nfa_irr_simplification} and let $\delta := \Det(\dep{\rN})$. Then $\Airr(\jslDfaSimple{\delta})$'s lower nfa is locally-saturated via their initial states and transition structure (Definition \ref{def:airr_and_det}) because each state $X$ accepts $X$ by Lemma \ref{lem:jsl_reach_simple_constructions}. Finally,
  \[
    |J(\jslLangs{\delta})| \leq |J(\Open\Delta_Z)| = |Z|
  \]
  via the surjective join-semilattice morphism $acc_{\delta} : \Open\Delta_Z \epito \jslLangs{\delta}$ and Note \ref{note:jsl_extras}.3.
\end{proof}

Actually, irreducible simplifications are precisely those nfas which are both locally-saturated and `union-free'.

\begin{theorem}[Characterizing irreducible simplifications]
  \label{thm:char_irr_simple}
  The following statements are equivalent:
  \begin{enumerate}
    \item
    $\rN \cong \simpleIrr{\rN}$.
    \item
    $\lambda z.L(\rN_{@z}) : \rN \to \simpleIrr{\rN}$ defines an nfa isomorphism.
    \item
    $\rN$ is locally-saturated and satisfies:
    \[
      \tag{union-free}
      \forall z \in Z. \forall S \subseteq Z. [(L(\rN_{@z}) = L(\rN_{@S})) \To z \in S]
    \]
  \end{enumerate}
\end{theorem}

\begin{proof}
  \item
  \begin{enumerate}
    \item
    $(1 \iff 2)$: given (1) then each state accepts a distinct language, so there is only one possible nfa isomorphism.
    \item
    $(2 \implies 3)$:
    Suppose $\lambda z.L(\rN_{@z})$ defines an nfa isomorphism. Then $\rN$ is locally-saturated because $\simpleIrr{\rN}$ is locally-saturated by (1), and this property is preserved by the nfa isomorphism. Recall the join-semilattice of accepted languages $\jslLangs{\rN}$ and also the relationship $L(\rN_{@S}) = \bigcup_{z \in S} L(\rN_{@z})$ from Definition \ref{def:nfa_basics}. Then (union-free) holds via Lemma \ref{lem:irr_simplifying_nfas}.2 because it asserts each $z \in Z$ accepts $L(\rN_{@z}) \in J(\jslLangs{\rN})$.
    \item 
    $(3 \implies 2)$:
    Suppose $\rN$ is locally-saturated and satisfies (union-free).
    Firstly, $f := \lambda z.L(\rN_{@z}) : Z \to J(\jslLangs{\rN})$ is a well-defined function because by (union-free) we know each $L(\rN_{@z}) \in J(\jslLangs{\rN})$. Furthermore $f$ is injective for otherwise (union-free) would fail, and surjective because $J(\jslLangs{\rN})$ is the minimal join-generating subset of $\jslLangs{\rN}$ (see Note \ref{note:jsl_extras}.4). Concerning the nfa isomorphism, $z$ is final iff $\epsilon \in L(\rN_{@z})$ hence $f$ preserves and reflects final states. The initial states and transitions are preserved and reflected because $\rN$ is locally-saturated.
  \end{enumerate}
\end{proof}

We now characterize the intersection-saturated nfas.

\begin{theorem}
  \label{thm:intersection_rule_char}
  An nfa $\rN$ is intersection-saturated iff $\rev{\rN}$ is locally-saturated.
\end{theorem}

\begin{proof}
  Let $\rN = (I, Z, \rN_a, F)$ and fix any $z \in Z$. For completely general reasons:
  \[
    L(\rev{\rN}_{@z})
    = L(\{ z \}, Z, \rN_a\spbreve, I)
    = (L(I, Z, \rN_a, \{ z \}))^r
    = (\{ u \in \Sigma^* : z \in \rN_u[I] \})^r.
  \]
  Assuming $\rev{\rN}$ is locally-saturated we prove the condition concerning transitions:
  \[
    \begin{tabular}{lll}
      $\rN_a(z_1, z_2)$
      & $\iff \rN_a\spbreve(z_2, z_1)$
      \\ & $\iff L((\rev{\rN})_{@z_1}) \subseteq a^{-1} L((\rev{\rN})_{@z_2})$
      & ($\rev{\rN}$ locally-saturated)
      \\ & $\iff \{ u \in \Sigma^* : z_1 \in \rN_u[I] \}^r \subseteq a^{-1} (\{ u \in \Sigma^* : z_2 \in \rN_u[I] \}^r)$
      & (see above)
      \\ & $\iff \{ u \in \Sigma^* : z_1 \in \rN_u[I] \} \subseteq (\{ u \in \Sigma^* : z_2 \in \rN_u[I] \}) a^{-1}$
      & (since $(a^{-1} X)^r = X^r a^{-1}$)
      \\ & $\iff \forall u \in \Sigma^*. ( z_1 \in \rN_u[I] \To z_2 \in \rN_{ua}[I] )$.
    \end{tabular}
  \]
  Finally $\forall u \in \Sigma^*.( z \in \rN_u\spbreve[F_\rN] \To \rN_u\spbreve[F_\rN] \cap I_\rN \neq \emptyset )$ is equivalent to requiring $L(\rN_{@z}) \subseteq L$, which follows by local saturation. Conversely if $\rN$ is intersection-saturated it is locally-saturated by reversing the above arguments.
\end{proof}

Then there is also a canonical way to intersection saturate an nfa.

\begin{corollary}
  $\rev{\simpleIrr{\rev{\rN}}}$ is an intersection-saturated nfa accepting $L(\rN)$, no larger than $\rN$.
\end{corollary}

\begin{proof}
  $\rev{\simpleIrr{\rev{\rN}}}$ accepts the same language because $\rev{-}$ reverses it and $\simpleIrr{-}$ preserves it (Lemma \ref{lem:irr_simplifying_nfas}.1). Moreover $\rev{-}$ preserves the number of states and $\simpleIrr{-}$ never increases it by Lemma \ref{lem:irr_simple_is_locally_saturated}. By the same Lemma we know $\simpleIrr{\rev{\rN}}$ is locally-saturated, hence its reverse satisfies the intersection rule by Theorem \ref{thm:intersection_rule_char}.
\end{proof}

Finally we collect a few results concerning transition-maximal nfas. Given any nfa, there is a \emph{non-canonical way} to construct a transition-maximal extension: keep adding initial/final states and transitions whenever doing so preserves the accepted language. Let us formally state this basic fact, an instantiation of Zorn's Lemma in the finite seatting.

\begin{lemma}
  \label{lem:nfa_has_trans_max_ext}
  Every nfa $\rN$ has a transition-maximal extension (see Definition \ref{def:local_sat_trans_max}.2).
\end{lemma}

\begin{lemma}
  \label{lem:rev_preserves_trans_max}
  $\rev{-}$ preserves transition-maximality.
\end{lemma}

\begin{proof}
  Holds because an nfa $\rM$ extends $\rev{\rN}$ iff $\rev{\rM}$ extends $\rN$.
\end{proof}

Transition-maximal transitions are determined by the order-structure of $\jslLangs{\rN} \supseteq \LW{L(\rN)}$.

\begin{lemma}[Transition-maximal transitions and finality]
  \label{lem:trans_max_nfa_property}
  If an nfa $\rN = (I, Z, \rN_a, F)$ is transition-maximal,
  \[
    \tag{T}
    \rN_a(z_1, z_2)
    \iff \forall X \in \LW{L(\rN)}.[
      L(\rN_{@z_1}) \subseteq X \To L(\rN_{@z_2}) \subseteq a^{-1} X
    ]
  \]
  \[
    \tag{F}
    z \in F
    \iff \forall X \in \LW{L(\rN)}. ( L(\rN_{@z}) \subseteq X \To \epsilon \in X ).
  \]
\end{lemma}

\begin{proof}
  Let $L := L(\rN)$. 
  \begin{enumerate}
    \item
    We'll prove (T). Given an nfa $\rN$ where $\rN_a(z_1, z_2)$ then $L(\rN_{@z_2}) \subseteq a^{-1} L(\rN_{@z_1})$, so ($\To$) holds generally because $a^{-1}(-)$ is monotonic w.r.t.\ inclusions. We'll refer to (T)'s right hand side by (RHS).
    
    Suppose $\rN$ is transition-maximal and (RHS) holds for specific $z_1$, $z_2$. For a contradiction assume $(z_1, z_2) \nin \rN_a$, letting $\rM$ be $\rN$ with the new transition. We know $L \subseteq L(\rM)$ and we'll show the converse, contradicting transition-maximality. Consider:
      \[
        I \ni i \xto{u}_\rN z_1 \xto{a} z_2 \xto{v}_\rM f \in F
      \]
    where the $v$-path uses the new transition $n \geq 0$ times. We know $L(\rN_{@z_1}) \subseteq u^{-1} L$ and may write $v = (\prod_{1 \leq i \leq n} v_i a) w$ where $\rN_{v_i}(z_2, z_1)$ for $1 \leq i \leq n$ and $w \in L(\rN_{@z_2})$. Then it suffices to establish $L(\rN_{@z_2}) \subseteq (ua (\prod_{1 \leq i \leq n} v_i a))^{-1} L$ by induction. For $n = 0$ we just apply (RHS). For the inductive case $n + 1$ we combine:
    \[
      L(\rN_{@z_2}) \subseteq (ua (\prod_{1 \leq i \leq n} v_i a))^{-1} L
      \qquad
      L(\rN_{@z_1}) \subseteq (v_{n + 1})^{-1} L(\rN_{@z_2})
    \]
    to infer $L(\rN_{@z_1}) \subseteq (ua (\prod_{1 \leq i \leq n} v_i a)v_{n+1})^{-1} L$ and finally $L(\rN_{@z_2}) \subseteq (ua (\prod_{1 \leq i \leq n} v_i a)v_{n+1}a)^{-1} L$ via (RHS).

    \item
    We'll prove (F). The implication ($\To$) is trivial because $z \in F$ implies $\epsilon \in L(\rN_{@z})$. Conversely we'll use transition-maximality. Assuming (RHS), and given any $u$-path $I \ni z_1 \to_u z_n = z$ through $\rN$, since $L(\rN_{@z}) \subseteq u^{-1} L$ we infer $\epsilon \in u^{-1} L$ i.e.\ $u \in L$, so by transition-maximality $z \in F$.
  \end{enumerate}
\end{proof}

\begin{corollary}
  \label{cor:trans_max_implies_locally_saturated}
  If $\rN$ is transition-maximal it is locally-saturated and intersection-saturated.
\end{corollary}

\begin{proof}
  Given transition-maximal $\rN = (I, Z, \rN_a, I)$ we first we show $\rN$ is locally-saturated. Given $L(\rN_{@z}) \subseteq L(\rN)$ then $z \in I$ by transition-maximality; the converse is trivial. Concerning transitions, $\rN_a(z_1, z_2)$ certainly implies $L(\rN_{@z_2}) \subseteq a^{-1} L(\rN_{@z_1})$. Conversely if the latter holds, then whenever $L(\rN_{@z_1}) \subseteq X \in \LW{L}$ we infer $L(\rN_{@z_2}) \subseteq a^{-1} L(\rN_{@z_1}) \subseteq a^{-1} X$ because $a^{-1}(-)$ is monotonic w.r.t.\ inclusions. Thus $\rN_a(z_1, z_2)$ by Lemma \ref{lem:trans_max_nfa_property}, so $\rN$ is locally-saturated. Finally, $\rev{\rN}$ is transition-maximal by Lemma \ref{lem:rev_preserves_trans_max}  hence locally-saturated, so $\rN$ is intersection-saturated by Theorem \ref{thm:intersection_rule_char}.
\end{proof}

\begin{corollary}
  \label{cor:tmax_and_ufree_simple}
  If $\rN$ is transition-maximal and union-free then $\rN \cong \simpleIrr{\rN}$.
\end{corollary}

\begin{proof}
  By Corollary \ref{cor:trans_max_implies_locally_saturated} $\rN$ is locally-saturated, so by union-freeness $\rN \cong \simpleIrr{\rN}$ via Theorem \ref{thm:char_irr_simple}.
\end{proof}

\begin{lemma}
  \label{lem:irr_simple_preserves_trans_max}
  $\simpleIrr{-}$ preserves transition-maximality.
\end{lemma}

\begin{proof}
  Given $\rN = (I, Z, \rN_a, F)$ we have the full subset construction $\delta := \Det(\dep{\rN})$, the quotient $\JSL$-dfa $acc_\delta : \delta \epito \jslDfaSimple{\delta}$ and also the irreducible simplification $\simpleIrr{\rN}$. If $\rN_a(z_1, z_2)$ then for completely general reasons $L(\rN_{@z_2}) \subseteq a^{-1} L(\rN_{@z_1})$, or equivalently $acc_\delta(\{ z_2 \}) \subseteq a^{-1} acc_\delta(\{ z_1 \})$ i.e.\ $acc_\delta(\{ z_1 \}) \to_{a} acc_\delta(\{ z_2 \})$ in $\simpleIrr{\rN}$.
  
  \begin{enumerate}
    \item 
    One cannot add an initial state to $\simpleIrr{\rN}$ whilst preserving acceptance because, by local saturation (Lemma \ref{lem:irr_simple_is_locally_saturated}), any additional state accepts $K \nsubseteq L(\rN)$.

    \item 
    For a contradiction suppose adding a final state $K$  to $\simpleIrr{\rN}$ preserves acceptance. Then $\rN' := (I, Z, \rN_a, F \cup acc_\delta^{-1}(\{ K\}))$ accepts $L(\rN)$ (which is a contradiction) because any additional accepting $\rN'$-path $I \ni z_0 \to_{a_0} \cdots \to_{a_n} z_{n+1} \in acc_\delta^{-1}(\{ K\})$ directly induces $I_{\simpleIrr{\rN}} \ni acc_\delta(\{ z_0 \}) \to_{a_0} \cdots \to_{a_n} acc_\delta(\{ z_{n+1} \}) = K$ in $\simpleIrr{\rN}$'s extension. 

    \item 
    It remains to show no additional transitions can be added. For a contradiction, assume $\rN'$ obtained by adding a single new transition $X_1 \to_{a_0} X_2$ to $\simpleIrr{\rN}$ satisfies $L(\rN') = L(\simpleIrr{\rN}) = L(\rN)$.  Consider the nfa:
    \[
      \rM := (I, Z, \rM_a, F)
      \qquad
      \rM_a :=
        \begin{cases}
          \rN_a \cup acc_\delta^{-1}(\{ X_1 \}) \times acc_\delta^{-1}(\{ X_2 \})
          & \text{if $a = a_0$}
          \\
          \rN_a & \text{otherwise.}
        \end{cases}
    \]
    Let us show $\rM$ has strictly more transitions than $\rN$. Firstly $Y := acc_\delta^{-1}(\{ X_1 \}) \times acc_\delta^{-1}(\{ X_2 \})$ is non-empty because $acc_\delta$ is surjective. Secondly $\rN_{a_0} \cap Y = \emptyset$ for otherwise $X_1 \to_{a_0} X_2$ would already be in $\simpleIrr{\rN}$.
  
    \smallskip
    Certainly $L(\rN) \subseteq L(\rM)$. For a contradiction we establish the converse. Given an accepting $\rM$-path shown below left:
    \[
      I \ni z_0 \xto{a_1} \cdots \xto{a_n} z_n \in F
      \qquad
      I_{\simpleIrr{\rN}} \ni acc_\delta(\{ z_0 \}) \xto{a_1} \cdots \xto{a_n} acc_\delta(\{ z_n \}) \in F_{\simpleIrr{\rN}}
    \]
    there is a respective accepting $\rN'$-path shown above right. Indeed, if $\rN_{a_{i+1}}(z_i, z_{i+1})$ then $acc_\delta(\{ z_i \}) \to_{a_{i+1}} acc_\delta(\{ z_{i+1} \})$ in $\simpleIrr{\rN}$ and hence $\rN'$. Otherwise $(z_i, z_{i+1}) \in acc_\delta^{-1}(\{ X_1 \}) \times acc_\delta^{-1}(\{ X_2 \})$ is covered by the single extra transition in $\rN'$.
  \end{enumerate}
\end{proof}

\subsection{$L$-extensions}

\begin{definition}[$L$-extension]
  Recall the transitions of the state-minimal $\JSL$-dfa $\jslDfaMin{L}$ i.e.\ $\gamma_a = \lambda X.a^{-1} X : \jslLQ{L} \to \jslLQ{L}$ from Definition \ref{def:state_min_jsl_dfa}.
  An \emph{$L$-extension} $e : \jslLQ{L} \monoto (\aT, \delta_a)$ is an injective $\JSL_f$-morphism $e : \jslLQ{L} \monoto \aT$ together with $\aT$-endomorphisms $\delta_a$ such that $e \circ \gamma_a = \delta_a \circ e$ for each $a \in \Sigma$.
  \endbox
\end{definition}

Then an $L$-extension is a join-preserving order-embedding of $\jslLQ{L}$ into $\aT$. Additionally each endomorphism $\lambda X. a^{-1} X$ of the former is extended by the endomorphism $\delta_a$ of the latter.

\begin{note}[Representation theory]
  By Theorem \ref{thm:syn_semiring_trans_semiring} one can view an $L$-extension as a \emph{representation of $L$'s syntactic semiring} \cite{Polak2001}. Then we are considering the `representation theory' of finite idempotent semirings.
  \endbox
\end{note}

\begin{example}[$L$-extensions]
  \label{ex:l_extensions}
  \item
  \begin{enumerate}
    \item
    Given $\gamma_a := \lambda X. a^{-1}X$ we have two bijective $L$-extensions:
    \[
      id_{\jslLQ{L}} : \jslLQ{L} \to (\jslLQ{L}, \gamma_a)
      \qquad
      dr_L^{\bf-1}: \jslLQ{L} \to ((\jslLQ{L^r})^\pOp, (\gamma_a)_*)
    \]
    The second one follows by Theorem \ref{thm:dr_L}. They are essentially the same extension i.e.\ they are isomorphic when viewed as algebras with $|\Sigma|$-many unary operations.
    
    \item
    $\jslLQ{\Sigma^*} := (\{ \emptyset, \Sigma^* \}, \cup, \emptyset)$ where each $\lambda X.a^{-1} X = id_{\jslLQ{\Sigma^*}}$. Any $\aS \in \JSL_f$ has endomorphism:
    \[
      c_{\top_\aS} := \lambda s.
      \begin{cases}
        \bot_\aS & s = \bot_\aS
        \\
        \top_\aS & \text{otherwise.}
      \end{cases}
    \]
    If $\Sigma \neq \emptyset$, the number of injective $e : \jslLQ{\Sigma^*} \monoto \aS$ is $|S| - 1$. Each $e$ defines an $L$-extension $\jslLQ{\Sigma^*} \monoto (\aS, c_{\top_\aS})$.

    \item
    Let $T$ be a finite union-closed set of languages such that (a) $L \in T$ and (b) $X \in T \implies a^{-1} X \in T$ for all $a \in \Sigma$. Then the inclusion $\iota : \jslLQ{L} \hookto ((T, \cup, \emptyset), \lambda X. a^{-1} X)$ is an $L$-extension.
    \endbox

  \end{enumerate}
\end{example}

There is a direct translation from an $L$-extension to a $\JSL$-dfa: inherit the initial state and extend the final states of $\jslDfaMin{L}$ (see below). Conversely each $\JSL$-dfa induces an $L$-extension by first simplifying and then forgetting the initial state and final states.

\begin{definition}[Translation between $L$-extensions and $\JSL$-dfas]
  \label{def:l_ext_jsl_dfa_correspondence}
  \item
  \begin{enumerate}
    \item
    The \emph{induced $\JSL$-dfa} of an $L$-extension $e : \jslLQ{L} \monoto (\aT, \delta_a)$ is:
    \[
      \toJdfa{e} := (e(L), \aT, \delta_a, \overline{\down_\aT e(dr_L(L^r))} )
      \qquad
      \text{and accepts $L$.}
    \]

    \item
    Conversely given any $\JSL$-dfa $\delta = (s_0, \aS, \delta_a, F)$ then:
    \[
      \toLext{\delta} := \iota : \jslLQ{L} \hookto (\jslLangs{\delta}, \lambda X. a^{-1} X).
    \]
    is its \emph{induced $L(\delta)$-extension}.
    \endbox
  \end{enumerate}
\end{definition}

\begin{note}
  \label{note:to_lext_jsdfa_pres_lang}
  \item
  \begin{enumerate}
    \item 
    Concerning $\toJdfa{-}$, the largest non-final state in $\jslLQ{L}$ is $\bigcup \{ X \in \LW{L} : \epsilon \nin L \} = dr_L(L^r)$. Then by Definition \ref{def:dfa_jsl} it is well-defined $\JSL$-dfa. It accepts $L$ because the embedding $e$ restricts to a dfa-isomorphism from $\minDfa{L}$ i.e.\ the classical state-minimal dfa which is a sub dfa of $\jslDfaMin{L}$.

    \item
    Concerning $\toLext{-}$, the join-semilattice of accepted language $\jslLangs{\delta}$ is from Definition \ref{def:jsl_reach_simple}. It was used to define the simplification $\jslDfaSimple{\delta}$ of the $\JSL$-dfa $\delta$. We remarked in Example \ref{ex:l_extensions}.3 that such structures are well-defined $L$-extensions.
    \endbox
  \end{enumerate}
\end{note}

\begin{definition}[Simplicity, reachability, transition-maximality, state-minimality]
  \label{def:l_ext_simp_reach_tmax}
  Fix an $L$-extension $e$.
  \begin{enumerate}
    \item
    $e$ is \emph{simple} if $\toJdfa{e}$ is simple (Definition \ref{def:jsl_reach_simple}.2). Then $e$'s \emph{simplification} is $\jslDfaSimple{e} := \toLext{\toJdfa{e}}$.
    \item
    $e$ is \emph{reachable} if the lower nfa of $\Airr(\toJdfa{e})$ is reachable (Definition \ref{def:nfa_basics}).
    \item
    $e$ is \emph{transition-maximal} if the lower nfa of $\Airr(\toJdfa{e})$ is transition-maximal (Definition \ref{def:local_sat_trans_max}).
    \item
    $e$ is \emph{state-minimal} if the lower nfa of $\Airr(\toJdfa{e})$ is state-minimal.
    \endbox
  \end{enumerate}
\end{definition}

To simplify an $L$-extension one views it as a $\JSL$-dfa, simplifies it, and finally forgets the initial state and final states. Well-definedness follows because $\toJdfa{e}$ accepts $L$, so that $\toLext{\toJdfa{e}}$ is an $L$-extension. The notions of simplicity and simplification are inherited from $\JSL$-dfas, whereas the notions of reachability, transition-maximality and state-minimality are inherited from nfas.

\subsubsection{Transition-maximal $L$-extensions}

\begin{note}
  {\bf The results in this section are currently not being used elsewhere.}
  \endbox
\end{note}

\begin{lemma}[Reachability degeneracy]
  \label{lem:l_ext_reach_degeneracy}
  Let $e$ be a simple transition-maximal $L$-extension and $\Airr(\toJdfa{e}) = (\rM, \rG, \rM')$. Then $\rM$ has at most one unreachable state, accepting $\Sigma^*$ if it exists.
\end{lemma}

\begin{proof}
  By assumption the lower nfa $\rM$ is transition-maximal.
  Then those states not reachable from an initial state are all final and have transitions to every other state by transition-maximality. Thus they all accept $\Sigma^*$, so by simplicity there is at most one of them.
\end{proof}

We now come to another important notion of `maximality' definable purely in terms of an $L$-extension's structure.

\begin{definition}[Meet-maximality]
  An $L$-extension $e: \jslLQ{L} \monoto (\aT, \delta_a)$ is \emph{meet-maximal} if:
    \[
      j = \Land_\aT \{ e(X) : X \in \LW{L}, \; j \leq_\aT e(X) \}
      \qquad
      \delta_a(j) = \Land_\aT \{ e(a^{-1} X) :  X \in \LW{L}, \; j \leq_\aT e(X) \}
    \]
  for all $j \in J(\aT)$ and $a \in \Sigma$. \endbox
\end{definition}
  
Then in meet-maximal $L$-extensions each $j \in J(\aT)$ is the meet of those embedded left word quotients of $L$ above it. Moreover, the endomorphism extensions $\delta_a : \aT \to \aT$ preserve these special meets. Importantly, each transition-maximal nfa induces a meet-maximal $L$-extension.

\begin{lemma}
  If $\rN$ is a transition-maximal nfa, $\toLext{\Det(dep(\reach{\rN}))}$ is a simple, reachable and transition-maximal $L(\rN)$-extension.
\end{lemma}

\begin{proof}
  Setting $L := L(\rN)$ then the specified $e$ is an $L$-extension because each operation preserves acceptance. It is simple because $\toLext{-}$ first simplifies and then forgets the initial state and final states. Concerning reachability, if $\rM := \reach{\rN}$ then the lower nfa of $\Airr(\toJdfa{e})$ is precisely the irreducible simplification $\simpleIrr{\rM}$ (Definition \ref{def:nfa_irr_simplification}) and the latter is reachable by Lemma \ref{lem:irr_simplifying_nfas}.3. Concerning transition-maximality, $\rM = \reach{\rN}$ is transition-maximal for otherwise $\rN$ wouldn't be, hence $\simpleIrr{\rM}$ is transition-maximal by Lemma \ref{lem:irr_simple_preserves_trans_max}.
\end{proof}

\begin{theorem}[Meet-maximality]
  \label{thm:meet_maximality}
  If an $L$-extension is simple and transition-maximal it is meet-maximal.
\end{theorem}

\begin{proof}
    We may assume $e$ is simplified i.e.\ $e = \jslDfaSimple{e}$. Then it is an inclusion $e :\jslLQ{L} \hookto (\aT, \lambda X.a^{-1} X)$ where $\aT = (T, \cup, \emptyset)$ and $T$ is the set of languages accepted by the individual states of $\delta:= \toJdfa{e}$. Let $\rM$ be the lower nfa of $\Airr \delta$ which is transition-maximal by assumption, hence locally-saturated by Lemma \ref{lem:irr_simple_is_locally_saturated}. Since each $j \in J(\aT)$ accepts $j$ by Lemma \ref{lem:jsl_reach_simple_constructions}.3, invoking Lemma \ref{lem:trans_max_nfa_property} yields:
    \[
      \tag{T}
      \rM_a(j_1, j_2)
      \iff \forall X \in \LW{L}.[
        j_1 \subseteq X \To j_2 \subseteq a^{-1} X
      ].
    \]
    Furthermore by Lemma \ref{lem:trans_max_nfa_property} $\rM$'s final states are:
    \[
      \tag{F}
      F_\rM = \{ j \in J(\aT) : \epsilon \in \bigcap \{ X \in \LW{L} : j \subseteq X  \} \}.
    \]
    We're ready to prove meet-maximality, so fix any $j \in J(\aT)$ and let $M_j := \Land_\aT \{ X \in \LW{L} : j \subseteq X \}$. Certainly $j \subseteq M_j$. For the reverse inclusion, first observe $M_j = \bigcup J$ for some non-empty $J \subseteq J(\aT)$ and fix any $j_0 \in J$.
    \begin{itemize}
      \item[--] 
      If $\epsilon \in j_0 \subseteq M_j$ then necessarily $\epsilon \in j$, for otherwise by (F) we'd have $j \subseteq X \in \LW{L}$ with $\epsilon \nin X$ and hence the contradiction $\epsilon \nin M_j$.
      \item[--] 
      Concerning transitions, $\forall j_2 \in J(\aT). [\rM_a(j_0, j_2) \To \rM_a(j, j_2)]$ via (T). In particular, given $j \subseteq X \in \LW{L}$ then $j_0 \subseteq M_j \subseteq X$ so we deduce $j_2 \subseteq a^{-1} X$ by assumption. 
    \end{itemize}
    
    \noindent
    So every word accepted by $j_0 \in J$ is accepted by $j$ i.e.\ $j_0 \subseteq j$; moreover  $M_j \subseteq Y$ because $j_0 \subseteq M$ was arbitrary. Then we've established:
    \[
      j=
      \Land_\aT \{ X \in \LW{L} : j \subseteq X \}
      \qquad
      \text{for each $j \in J(\aT)$}.
    \]
    Fixing any $a \in \Sigma$ and $j \in J(\aT)$ it remains to establish:
    \[
      a^{-1} j = \Land_\aT \{ a^{-1} X : j \subseteq X \in \LW{L} \}.
    \]
    Indeed if $j' \in J(\aT)$ lies below the (RHS) i.e.\ $\forall X \in \LW{L}. [ j \subseteq X \To j' \subseteq a^{-1} X ]$ then by (T) we infer $\rM_a(j, j')$ and hence $j' \subseteq a^{-1} j$ because $\rM$ is locally-saturated. Finally $a^{-1} j$ is itself a lower bound for (RHS) because $a^{-1}(-)$ is monotone w.r.t.\ inclusion.
\end{proof}

Are these special meets of left word quotients $u^{-1} L$ always their intersection?
The answer is \emph{no}.

\begin{example}[Meet-maximal meets needn't be intersections]
\label{ex:meet_max_meets_not_intersections}
In \cite[Theorem 7]{TheoryOfAtomataBrzTamm2014} a language $L$ is implicitly provided s.t.\ if an nfa $\rN$ accepts $L$ and each $L(\rN_{@\{z\}})$ is a set-theoretic boolean combination of $\LW{L}$ then $\rN$ is not state-minimal. Given a transition-maximal extension of a state-minimal nfa we obtain a meet-maximal $L$-extension  by Theorem \ref{thm:meet_maximality}. If the special meets $\Land_\aT \{ j : j \subseteq X \in \LW{L}  \}$ were intersections we'd obtain a contradiction via the lower nfa of $\Airr(\toJdfa{e})$ -- which is also a  state-minimal nfa accepting $L$. \endbox
\end{example}

We finally mention some related properties.

\begin{lemma}
  \label{lem:tmax_irr_lwq}
  If $e : \jslLQ{L} \hookto (\aT, \delta_a)$ is transition-maximal and simplified,
  \[
    j \subseteq u^{-1} L
    \iff 
    \bigcap \{ X \in \LW{L} : j \subseteq X \} \subseteq u^{-1} L
    \qquad
    \text{for any $j \in J(\aT)$ and $u \in \Sigma^*$.}
  \]
\end{lemma}

\begin{proof}
  The implication $(\To)$ is immediate. Conversely we know $e$ is meet-maximal by Theorem \ref{thm:meet_maximality}, so that $j = \Land_\aT \{ X \in \LW{L} : j \subseteq X \}
  \subseteq \bigcap \{ X \in \LW{L} : j \subseteq X \}
  \subseteq u^{-1}L$.
\end{proof}

\begin{lemma}
  If $e : \jslLQ{L} \hookto (\aT, \delta_a)$ is transition-maximal and simplified then for any $j \in J(\aT)$,
  \[
    j_2 \subseteq a^{-1} j_1
    \iff
    \bigcap \{ X \in \LW{L} : j_2 \subseteq X \}
    \subseteq 
    \bigcap \{ a^{-1} X \in \LW{L} : j_1 \subseteq X \}.
  \]
\end{lemma}

\begin{proof}
  We calculate:
  \[
    \begin{tabular}{lll}
      $j_2 \subseteq a^{-1} j_1$
      & $\iff \forall X \in \LW{L}.[ j_1 \subseteq X \To j_2 \subseteq a^{-1} X ]$
      & (by Theorem \ref{thm:meet_maximality})
      \\ &
      $\iff \forall X \in \LW{L}.[ j_1 \subseteq X \To \bigcap \{ X \in \LW{L} : j_2 \subseteq X \} \subseteq a^{-1} X ]$
      & (by Lemma \ref{lem:tmax_irr_lwq})
      \\ &
      $\iff \bigcap \{ X \in \LW{L} : j_2 \subseteq X \}
      \subseteq 
      \bigcap \{ a^{-1} X \in \LW{L} : j_1 \subseteq X \}.$
    \end{tabular}
  \]
\end{proof}

\subsubsection{Reversing $L$-extensions}

\begin{note}
  {\bf The results in this section are currently not being used elsewhere.}
  \endbox
\end{note}

\begin{definition}[Reversal of an $L$-extension]
  Given an $L$-extension $e$ let $\rN$ be the lower nfa of $\Airr(\toJdfa{e})$.
  Then $e$'s \emph{reversal} is the $L^r$-extension:
  \[
    \rev{e} := \toLext{\Det(dep(\rev{\rN}))}.
  \]
  It is union-generated by the languages
  $\revEl{e}{j} := L(\rev{\rN}_{@j}) = \{ w \in \Sigma^* : j \in \rN_{w^r} [I_\rN] \}$ where $j \in J(\aT)$.
  \endbox
\end{definition}

\begin{note}[Alternative descriptions of $\rev{e}$]
  \item
  \begin{enumerate}
    \item
    It is $\jslDfaSimple{\Det(\dep{\rev{\rN}})}$ without the initial state or final states.
    \item
    By Corollary \ref{cor:reach_de_morgan_dual_simple}, it is isomorphic to $(\jslDfaReach{\Det(\dep{\rN})})^\pentagram$ without the initial state or final states.
    \endbox
  \end{enumerate}
\end{note}

\begin{lemma}[Reversing $L$-extensions]
  \label{lem:rev_l_ext}
  Fix any $L$-extension $e : \jslLQ{L} \monoto (\aT, \delta_a)$.

  \begin{enumerate}
    \item
    $\rev{e}$ is a well-defined simplified $L^r$-extension.
    \item
    If $e$ is simple then $\rev{e}$ is reachable.
    \item
    If $e$ is transition-maximal then so is $\rev{e}$.
    Similarly if $e$ is state-minimal then so is $\rev{e}$.
    \item
    If $e$ is simple, transition-maximal and $\revEl{e}{j_0} = \bigcup \{ \revEl{e}{j} : j \in J \}$ for some $j_0 \in J(\aT)$, $J \subseteq J(\aT)$ then $j_0 = \Land_\aT J$.

    \item
    If $e$ is simplified, transition-maximal and state-minimal then $\rev{\rev{e}} = e$.

  \end{enumerate}
\end{lemma}

\begin{proof}
  \item
  \begin{enumerate}
    \item
    Well-definedness follows by construction, noting that $dep(-)$, $\Det(-)$ and $\toLext{-}$ preserve the accepted language $L^r$. Likewise $\rev{e}$ is simplified by construction.

    \item
    If $e$ is simple then $\rN$ is coreachable because each $j \in J(\aT)$ accepts a non-empty language.
    Then $\rev{\rN}$ is reachable and hence so is $\simpleIrr{\rev{\rN}}$ by Lemma \ref{lem:irr_simplifying_nfas}.

    \item
    If $\rN$ is transition-maximal then $\rev{\rN}$ is too by Lemma \ref{lem:rev_preserves_trans_max}, hence so is $\simpleIrr{\rev{\rN}}$ by Lemma \ref{lem:irr_simple_preserves_trans_max}. If $e$ is state-minimal then the lower nfa of $\Airr(\toJdfa{e})$ is a state-minimal nfa $\rM$ accepting $L$. Since the lower nfa of $\Airr(\toJdfa{\rev{e}})$ accepts $L^r$ and has no more states than $\rM$ it is also state-minimal, so $\rev{e}$ is state-minimal.

    \item
    We may assume $e = \jslDfaSimple{e}$ is simplified.
    Let $\delta := \toJdfa{e}$ and $\rN$ be the lower nfa of $\Airr\delta$.
    Suppose $\revEl{e}{j_0} = \bigcup \{ \revEl{e}{j} : j \in J \}$.
    Then $\forall u \in \Sigma^*. \forall j \in J. (j \in \rN_u[I_\rN] \To j_0 \in \rN_u[I_\rN])$. Now, since $\rN$ is transition-maximal the intersection rule holds by Corollary \ref{cor:trans_max_implies_locally_saturated}, so that:
    \[
      \tag{A}
      \forall j \in J. \forall a \in \Sigma. \rN_a[j_0] \subseteq \rN_a[j]
    \]
    because whenever $\rN_a(j_0, j')$ and there is a $u$-path to $j$ there is a $ua$-path to $j'$. Furthermore:
    \[
      \tag{B}
      \text{$j_0$ is final iff every $j \in J$ is final}.
    \]
    Indeed, if $j_0$ is final and $j \in J$ then for every $u$-path to $j$ we have a $u$-path to $j_0$ and hence $u \in L$, so by transition-maximality $j$ is final. Similarly, if every $j \in J$ is final then every $u$-path to $j_0$ satisfies $u \in L$ so $j_0$ is final by transition-maximality. Then by (A) and (B) we deduce $j_0 \subseteq \bigcap J$.

    \smallskip
    To establish $j_0 = \Land_\aT J$ we fix any $j' \subseteq \Land_\aT J$ and prove $j' \subseteq j_0$. Certainly $\forall j \in J. j' \subseteq j$ hence:
    \[
      \tag{C}
      \forall j \in J. \rN_a[j'] \subseteq \rN_a[j]  
    \]
    because $j'' \subseteq a^{-1} j'$ implies $j'' \subseteq a^{-1} j$. We now aim to prove $\rN_a[j'] \subseteq \rN_a[j_0]$. Given $\rN_a(j', j'')$ we certainly know $\forall j \in J. \rN_a(j, j'')$ by (C). Equivalently $\forall j \in J. \rN_a\spbreve(j'', j)$ in $\rev{\rN}$ and thus $\forall j \in J. \revEl{e}{j} \subseteq a^{-1} \revEl{e}{j'}$, where the latter uses a general property of nfas. Then:
    \[
      \revEl{e}{j_0}
      = \bigcup_{j \in J} \revEl{e}{j} \subseteq a^{-1} \revEl{e}{j'}.
    \]
    In other words, for every $u^r$-path to $j_0$ in $\rN$ there exists an $(au)^r$-path to $j''$. Applying the intersection-rule we deduce $\rN_a(j_0, j'')$ as desired i.e.\ we have established $\rN_a[j'] \subseteq \rN_a[j_0]$. Furthermore if $j'$ is final then every $j \in J$ is final, so that $j_0$ is final by (B). Then we've proved that $j' \subseteq j$ and we're done.

    \item
      An nfa is state-minimal iff its reverse is state-minimal. Then since $\rev{e} : \jslLQ{L} \hookto (\aU, \phi_a)$ accepts $L^r$ we deduce $\alpha := \lambda j. \revEl{e}{j} : J(\aT) \to J(\aU)$ is bijective, for otherwise we'd contradict state-minimality. Let $\rN$ be the lower nfa of $\Airr(\toJdfa{e})$ and $\rM$ be the lower nfa of $\Airr(\toJdfa{\rev{e}})$. Then we can bijectively relabel $\rM$ to obtain the nfa:
      \[
        \rM' := (\alpha^{\bf-1}[I_\rM], J(\aT), \rM'_a, \alpha^{\bf-1}[F_\rM])
        \qquad
        \rM'_a(j_1, j_2) :\iff \rM_a(\revEl{e}{j_1}, \revEl{e}{j_2}).
      \]
      We're going to show that $\rev{\rM'}$ extends $\rN$. By (2) we know $\rM$ is transition-maximal, hence:
      \[
        \begin{tabular}{lll}
          $(\rev{\rM'})_a(j_1, j_2)$
          & $\iff \rM'_a(j_2, j_1)$
          \\ & $\iff \rM_a(\revEl{e}{j_1}, \revEl{e}{j_2})$
          \\ & $\iff \revEl{e}{j_1} \subseteq a^{-1} \revEl{e}{j_2}$
          & (by Corollary \ref{cor:trans_max_implies_locally_saturated})
          \\
          & $\iff \forall u \in \Sigma^*.[ j_1 \in \rN_u[I_\rN] \To j_2 \in \rN_{ua}[I_\rN] ]$
          & (by definition).
        \end{tabular}
      \]
      Thus $\rN(j_1, j_2) \To (\rev{\rM'})_a(j_1, j_2)$ because whenever there is a $u$-path to $j_1$ we obtain a $ua$-path to $j_2$. Next, if $j$ is initial in $\rN$ then $j \in I_\rN = \rN_\epsilon[I_\rN]$ and hence $\epsilon \in \revEl{e}{j}$, so that $j$ is final in $\rM'$, thus initial in $\rev{\rM'}$. Finally if $j$ is final in $\rN$ then $\forall u \in \Sigma^*.[ j \in \rN_u[I_\rN] \To u \in L ]$ or equivalently $\revEl{e}{j} \subseteq L^r$, so that $j$ is initial in $\rM'$, thus final in $\rev{\rM'}$. Having established that $\rev{\rM'}$ extends $\rN$ we immediately deduce $\rev{\rM'} = \rN$ by transition-maximality. It follows that:
      \[
        \begin{tabular}{lll}
          $\revEl{\rev{e}}{\revEl{e}{j}}$
          &
          $= L(\rev{\rM}_{@\revEl{e}{j}})$
          \\ &
          $= L(\rev{\rM'}_{@j})$
          & (via $\rM \cong \rM'$)
          \\ &
          $= L(\rN_{@j})$
          & (via $\rev{\rM'} = \rN$)
          \\ &
          $ = j$.
        \end{tabular}
      \]
      Since $\alpha$ is bijective we know every $j \in J(\aT)$ has a unique corresponding $\revEl{e}{j} \in J(\aU)$. Combining this with the above equality we deduce $\rev{\rev{e}} = e$.
      
  \end{enumerate}
\end{proof}


\subsection{The Atomizer}

This section is based on recent work of Tamm \cite{Tamm2016NewIA}.
Recall the minimal boolean and distributive $\JSL$-dfa from Definition \ref{def:canon_bool_jsl_dfa}. Fixing $L$, the left predicates $\LP{L}$ are those finitely many languages arising as a set-theoretic boolean combination of the left word quotients $\LW{L}$. Importantly, any language can be transformed into a left predicate via a closure operator.

\begin{definition}[Atomic languages, $\cl_L$ and $\rE_L$]
  \label{def:atomic_cl_langs}
  \item 
  \begin{enumerate}
    \item 
    The \emph{atomic closure operator} $\cl_L : \Pow \Sigma^* \to \Pow \Sigma^*$ is defined:
    \[
      \begin{tabular}{lrl}
        $\cl_L (X)$
        &
        $:=$ & $\bigcup \{ \alpha \in J(\jslLP{L}) : \alpha \cap X \neq \emptyset \}$
        \\ &
        $=$ & $\bigcap \{ Y \in \LP{L} : X \subseteq Y \}$.
      \end{tabular}
    \]
    Moreover the equivalence relation $\rE_L \subseteq \Sigma^* \times \Sigma^*$ is defined:
    \[
      \rE_L(u, v) :\iff \forall X \in \LW{L}.[ u \in X \iff v \in X]
    \]
    with equivalence classes $\sem{w}_{\rE_L} \subseteq \Sigma^*$.
    
    \item
    A language is \emph{atomic} w.r.t $L$ \cite{TheoryOfAtomataBrzTamm2014} if it is a fixpoint of $\cl_L$. They are precisely the languages in $\LP{L}$

    \item
    A language is \emph{positively atomic} w.r.t $L$ if it lies in $\jslLD{L} \subseteq \jslLP{L}$. 
    
    \item
    A language is \emph{subatomic} w.r.t\ $L$ if it lies in $\jslLRP{L}$.
    \endbox
  \end{enumerate}
\end{definition}

\begin{note}[Compatible definitions of $\cl_L$]
  The distinct definitions of $\cl_L(X)$ are consistent: each element of $\LP{L}$ is (i) the join (union) of join-irreducibles (atoms) below it, (ii) the meet (intersection) of elements above it.
  \endbox
\end{note}

\begin{lemma}[Concerning atomic closure]
  \label{lem:cl_L_basics}
  Let $\alpha \in J(\jslLP{L})$ and $X, Y \subseteq \Sigma^*$.
  \begin{enumerate}
    \item
    $\cl_L$ is a well-defined closure operator.
    \item
    $\cl_L(X)$ is the smallest left predicate containing $X$.
    \item
    $\alpha \subseteq \cl_L(X)$ iff $\alpha \cap X \neq \emptyset$.
    \item
    $\cl_L(X \cup Y) = \cl_L(X) \cup \cl_L(Y)$.
    \item 
    $\cl_L(w^{-1} X) \subseteq w^{-1} \cl_L(X)$ for all $w \in \Sigma^*$.
    \item
    $J(\jslLP{L}) = \{ \sem{w}_{\rE_L} : w \in \Sigma^* \}$.
    \item
    $\cl_L(X) = \bigcup_{w \in X} \sem{w}_{\rE_L}$.
  \end{enumerate}
\end{lemma}

\begin{proof}
  \item
  \begin{enumerate}
    \item 
    $\cl_L$ is monotone: if $X \subseteq Y$ then $\forall \alpha \in J(\jslLP{L}). (\alpha \cap X \neq \emptyset \To \alpha \cap Y \neq \emptyset)$ hence $\cl_L(X) \subseteq \cl_L(Y)$. Next, $X \subseteq \cl_L(X)$ because the latter is an intersection of supersets of $X$. Finally $\cl_L \circ \cl_L (X) = \cl_L(X)$ because $\cl_L(X) \in \LP{L}$ is the union of the atoms it includes.

    \item
    Follows by alternate definition.

    \item
    If $\alpha \cap X \neq \emptyset$ then $\alpha \subseteq \cl_L(X)$ by definition. Conversely if $\alpha \subseteq \cl_L(X)$ then some $u \in \alpha$ satisfies $u \in X$ for otherwise we'd know $X \subseteq \overline{\alpha}$ (the latter being a coatom), so that $\cl_L(X) \subseteq \overline{\alpha}$ by the alternate definition (a contradiction).

    \item
    The inclusion ($\supseteq$) follows by monotonicity. Conversely, given an atom $\alpha \subseteq \cl_L(X \cup Y)$ then by (3) there exists $u \in \alpha \cap (X \cup Y)$ and hence w.l.o.g.\ $u \in \alpha \cap X$ and thus $\alpha \subseteq \cl_L(X)$.
    
    \item
    Since $X \subseteq \cl_L(X)$ we deduce $w^{-1} X \subseteq w^{-1} \cl_L(X)$ and hence $\cl_L(w^{-1} X) \subseteq w^{-1} \cl_L(X)$ because (a) the former is the least atomic language above $w^{-1} X$, (b) the latter is in $\LP{L}$ because $w^{-1}(-)$ preserves all set-theoretic boolean operations.

    \item
    An atom amounts to specifying $X$ or $\overline{X}$ for each $X \in \LW{L}$ i.e.\ an $\rE_L$ equivalence-class.

    \item
    Follows by definition via (6).
  \end{enumerate}
\end{proof}

\begin{note}[$\cl_L$ preserves unions]
  The fixpoints (closed sets) of every closure operator are closed under intersections.
  By Lemma \ref{lem:cl_L_basics} $\LP{L}$ is also closed under \emph{unions}, which is not a general property of closure operators. Then the closed sets form a distributive lattice, in fact a boolean lattice because $\LP{L}$ is closed under relative complement.
  \endbox
\end{note}




We've now arrived at the main definition of this section.

\begin{definition}[Atomizer]
  \label{def:atomizer}
  Each $L$-extension $e : \jslLQ{L} \monoto (\aT, \delta_a)$ has associated join-semilattice morphism:
  \[
    \lambda Y. \cl_L(acc_{\toJdfa{e}}(Y)) : \aT \to \jslLP{L}.  
  \]
  Restricting to the image yields the \emph{atomizer} $\at{e} : \aT \epito \jslAtz{e}$ where $\jslAtz{e} := (\Atz{e}, \cup, \emptyset)$ is the \emph{atomized semilattice}.
  \endbox
\end{definition}

\begin{note}[Atomizer's action]
  The atomizer constructs the closure of the accepted language. We often construct $L$-extensions by simplifying a $\JSL$-dfa, in which case the atomizer is a domain/codomain restriction of $\cl_L$. \endbox
\end{note}

\begin{lemma}
  The atomizer is a well-defined join-semilattice morphism.
\end{lemma}

\begin{proof}
  Fixing an $L$-extension $e : \jslLQ{L} \monoto (\aT, \delta_a)$, $\at{e}$ is a well-defined function because $\cl_L(X) \in \LP{L}$. 
  Given $\delta := \toJdfa{e}$ then $acc_\delta : \aT \epito \jslLangs{\delta}$ is a well-defined surjective $\JSL_f$-morphism by Definition \ref{def:jsl_reach_simple}. Finally $\cl_L$ restricts to a morphism $\cl_L : \jslLangs{\delta} \to \jslLP{L}$ because $\emptyset = \cl_L(\emptyset)$ is atomic and $\cl_L$ preserves binary unions by Lemma \ref{lem:cl_L_basics}.4.
\end{proof}

Recall the canonical quotient-atom bijection $\kappa_L$ from Theorem \ref{thm:atom_quotient_bijection}. We now use it to represent each atomized semilattice inside $\Dep$.

\begin{definition}[Atomizer relation $\rH_e$]
  Each $L$-extension $e : \jslLQ{L} \monoto (\aT, \delta_a)$ has an associated \emph{atomizer relation},
  \[
    \rH_e \subseteq J(\aT) \times \LW{L^r}
    \qquad
    \rH_e(j, v^{-1} L^r) :\iff v^r \in \at{e}(j).
  \]
  Furthermore if $e$ is simplified this becomes $\rH_e(j, v^{-1} L^r)  \iff v^r \in \cl_L(j)$. \endbox
\end{definition}

This important concept is preserved under simplification of the $L$-extension.

\begin{lemma}[$\rH_e \cong \rH_{\jslDfaSimple{e}}$]
  \label{lem:h_e_assume_simplified}
  We have the $\Dep$-isomorphism:
  \[
    \xymatrix@=15pt{
      \LW{L^r} \ar[rr]^{\Delta_{\LW{L^r}}} && \LW{L^r}
      \\
      J(\aT) \ar[u]^-{\rH_e} \ar[rr]_{acc_{\toJdfa{e}}} && J(\jslLangs{e}) \ar[u]_-{\rH_{\jslDfaSimple{e}}}
    }
  \]
\end{lemma}

\begin{proof}
  The diagram commutes by unwinding the definitions, recalling each state $Y$ in $\jslDfaSimple{\toJdfa{e}}$ accepts $Y$. Since $\Open \rH_e = \Open \rH_{\jslDfaSimple{e}}$, applying $\Open$ yields an identity morphism (see Note \ref{note:open_morphism_alt}), so this $\Dep$-morphism is actually an isomorphism.
\end{proof}

\begin{theorem}[$\jslAtz{e} \cong \Open\rH_e$]
  \label{thm:rep_atomized_in_dep}
  For any $L$-extension $e$ we have the join-semilattice isomorphism:
  \[
    \theta_e : \jslAtz{e} \to \Open\rH_e
    \qquad
    \theta_e(Y)
    := \{ w^{-1} L^r : w \in Y^r \}
    \qquad
    \theta_e^{\bf-1} (S) := \{ w \in \Sigma^* : w^{-1} L^r \in S  \}^r.
  \]
\end{theorem}

\begin{proof}
  Recall that the atomized semilattice $\jslAtz{e}$ is a sub join-semilattice of $\jslLP{L}$. Concerning the latter, we may instantiate Proposition \ref{prop:dep_generator_isoms} with $J_{\jslLP{L}} := \LP{L}$ (every element) and $M_{\jslLP{L}} := M(\jslLP{L})$ (the coatoms). We immediately obtain the $\Dep$-isomorphism:
  \[
    \rI_{\jslLP{L}}^{\bf-1} : \Pirr\jslLP{L} \to \; \rG
    \qquad\text{where}\qquad
    \rG := \; \nsubseteq |_{\LP{L} \times M(\jslLP{L})}
  \]
  and thus the composite join-semilattice isomorphism:
  \[
    \alpha := \qquad
      \jslLP{L}
      \xto{\rep_{\jslLP{L}}} \Open\Pirr\jslLP{L}
      \xto{\Open \rI_{\jslLP{L}}^{\bf-1}} \Open\rG
    \qquad
    \text{with action $Y \mapsto \rep_{\jslLP{L}}(Y)$}.
  \]
  To clarify, $\alpha$ acts as $\rep_{\jslLP{L}}$ because each $Y \in O(\Pirr\jslLP{L})$ is downwards-closed in $M(\jslLP{L})$ w.r.t.\ inclusion, and $(\rI_{\jslLP{L}}^{\bf-1})_+\spbreve[Y]$ constructs the downwards-closure (see Proposition \ref{prop:dep_generator_isoms}). It follows that $\jslAtz{e} \subseteq \jslLP{L}$ may be represented as a sub join-semilattice of $\Open\rG$. Since $\at{e}$ is surjective we know $J := \at{e}[J(\aT)]$ join-generates the atomized semilattice $\jslAtz{e}$. Then:
  \[
    \text{
      $\alpha$ restricts to the isomorphism
      $\beta : \jslAtz{e} \to \Open \rH$
      \qquad
      \text{where $\rH := \at{e} |_{J(\aT) \times J} ; \rG$}.
    }
  \]
  To explain, each $\cl_L(j) \in J(\jslAtz{e}) \subseteq J_{\jslLP{L}}$ satisfies $\rep_{\jslLP{L}} (\cl_L(j)) = \rG[\cl_L(j)]$, and all other open sets are unions of them. To construct the desired isomorphism $\theta_e$ recall the bijection $\kappa_L : \LW{L^r} \to J(\jslLP{L})$ from Theorem \ref{thm:atom_quotient_bijection}, and the bijection $J(\jslLP{L}) \to M(\jslLP{L})$ between atoms and coatoms (relative complement). Consider the relations:
  \[
    \xymatrix@=15pt{
      M(\jslLQ{L}) \ar[rr]^{f} && \LW{L^r}
      \\
      J(\aT) \ar[u]^{\rH} \ar[rr]_{\Delta_{J(\aT)}} && J(\aT) \ar[u]_{\rH_e}
    }
  \]
  where the composite bijection $f$ has action $\overline{\sem{w}_{\rE_L}} \mapsto \kappa_L^{\bf-1}(\sem{w}_{\rE_L}) = (w^r)^{-1} L^r$. If they commute we have a $\Dep$-isomorphism because the lower and upper witnesses are bijections. Then let us calculate:
  \[
    \begin{tabular}{lll}
      $\rH ; f (j, v^{-1} L^r)$
      & $\iff \exists w \in \Sigma^*. [
          \at{e}(j) \nsubseteq \overline{\sem{w}_{\rE_L}}
          \,\land\, f(\overline{\sem{w}_{\rE_L}}) = v^{-1} L^r
        ]$
      \\ &
      $\iff \at{e}(j) \nsubseteq \overline{\sem{v^r}_{\rE_L}}$
      & (A)
      \\ &
      $\iff \sem{v^r}_{\rE_L} \subseteq \at{e}(j)$
      & (atom vs.\ coatom)
      \\ &
      $\iff v^r \in \at{e}(j)$
      & (via Lemma \ref{lem:cl_L_basics}.7)
      \\ &
      $\iff \rH_e(j, v^{-1} L^r)$
      & (by definition).
    \end{tabular}
  \]
  Concerning (A), $(\To)$ follows because we know $(w^r)^{-1} L^r = v^{-1} L^r$ and thus $\sem{w}_{\rE_L} = \sem{v^r}_{\rE_L}$ because $\kappa_L$ is injective. Conversely $(\oT)$ follows by choosing $w := v^r$. So we have the isomorphism $\theta_e := \Open\rH_e \circ \beta$ with action:
  \[
    \begin{tabular}{llll}
      $Y$ & $\mapsto$ & $f[\rep_{\jslLP{L}} (Y)]$
      & (by Note \ref{note:open_morphism_alt})
      \\ & $=$ & $\{ f(\overline{\sem{w}_{\rE_L}}) : Y \nsubseteq \overline{\sem{w}_{\rE_L}} \}$
      & (def.\ of $\rep$)
      \\ & $=$ & $\{ [w^r]^{-1} L^r : Y \nsubseteq \overline{\sem{w}_{\rE_L}} \}$
      & (def.\ of $f$)
      \\ & $=$ & $\{ [w^r]^{-1} L^r : w \in Y \}$
      & ($Y$ is atomic).
    \end{tabular}
  \]
\end{proof}

\subsection{Explaining Kameda-Weiner}

Recall the notion of $L$-covering $\rH$ i.e.\ Definition \ref{def:l_covering}.
They amount to biclique edge-coverings of the dependency relation $\rDR{L}$. They are \emph{legitimate} if their induced nfa $\rN_\rH$ (defined over the bicliques) accepts $L$. Crucially $\rH_e$ is a legitimate $L$-covering for any $L$-extension $e$.

\begin{theorem}
  \label{thm:l_ext_induce_l_cover}
  $\rH_e$ is a legitimate $L$-covering for any $L$-extension $e$,
  \[
    \rLCov{L}{\rH_e}_- ; \rH_e = \rDR{L}
    \qquad\text{where}\qquad
    \rLCov{L}{\rH_e}_- (u^{-1} L, j) \iff \acc_{\toJdfa{e}}(j) \subseteq u^{-1} L.
  \]
\end{theorem}

\begin{proof}
  Denote the acceptance map $\alpha := \acc_{\toJdfa{e}}$ for brevity.
  Observe $\alpha(j) \subseteq u^{-1} L \iff \at{e}(j) \subseteq u^{-1} L$ because $u^{-1} L$ is atomic w.r.t.\ $L$. We first compute $\rLCov{L}{\rH_e}_-$ without knowing $\rH_e$ is an $L$-covering. Afterwards we'll verify the claimed equality.
  \[
    \begin{tabular}{lll}
      $\rLCov{L}{\rH_e}_- (u^{-1} L, j)$
      &
      $:\iff \rH_e[j] \subseteq \rDR{L}[u^{-1} L]$
      \\ &
      $\iff \forall v \in \Sigma^*.[ v^r \in \at{e}(j) \To uv^r \in L]$
      \\ &
      $\iff \forall v \in \Sigma^*. [ v \in \at{e}(j) \To v \in u^{-1} L ]$
      \\ &
      $\iff \at{e}(j) \subseteq u^{-1} L$
      \\ &
      $\iff \alpha(j) \subseteq u^{-1} L$
      & (see above).
      \\
      \\
      $\rLCov{L}{\rH_e}_- ; \rH_e(u^{-1}L, v^{-1} L^r)$
      &
      $\iff \exists j \in J(\aT).[
        \alpha(j) \subseteq u^{-1} L \,\land\,
        v^r \in \at{e}(j)
      ]$
      & (by def.)
      \\ &
      $\iff \exists j \in J(\aT).[
        \at{e}(j) \subseteq u^{-1} L \,\land\,
        v^r \in \at{e}(j)
      ]$
      & (see above)
      \\ &
      $\iff v^r \in u^{-1} L$
      & (A)
      \\ &
      $\iff \rDR{L}(u^{-1}L, v^{-1}L^r)$.
    \end{tabular}
  \]
  Concerning (A), $(\To)$ follows immediately. As for $(\oT)$, since $u^{-1} L = \bigcup \alpha[S]$ for some $S \subseteq J(\aT)$ we deduce $u^{-1} L = \bigcup \at{e}[S]$, so that $v^r \in u^{-1} L$ implies $v^r \in \at{e}(j) \subseteq u^{-1} L$ for some $j \in J(\aT)$. Then $\rH_e$ is an $L$-covering by Lemma \ref{lem:l_coverings}.1.
  
  It remains to establish the legitimacy of $\rH_e$. By Lemma \ref{lem:l_coverings}.6 we at least know $L(\rN_{\rH_e}) \subseteq L$, and it remains to prove the reverse inclusion. First let $\rM$ be the lower nfa of $\Airr(\toJdfa{e})$, which accepts $L$ by Note \ref{note:det_airr_preserve_acceptance}. These two nfas have the same states $J(\jslLangs{e})$; concerning their transitions:
  \[
    \begin{tabular}{lll}
      $\rM_a(\alpha(j_1), \alpha(j_2))$
      &
      $\iff \alpha(j_2) \subseteq a^{-1} \alpha(j_1)$
      & (by definition)
      \\ &
      $\implies \forall X \in \LW{L}.[ \alpha(j_1) \subseteq X \To \alpha(j_2) \subseteq a^{-1} X ]$
      \\ &
      $\iff \rN_{\rH_e, a}(\alpha(j_1), \alpha(j_2))$
      & (see $\rLCov{L}{\rH_e}_-$ above).
    \end{tabular}
  \]
  Moreover (a) $I_\rM = I_{\rN_{\rH_e}}$ since $\alpha(j) \subseteq L \iff \rLCov{L}{\rH_e}_-\spbreve (L, \alpha(j))$ and (b) $F_\rM \subseteq F_{\rN_{\rH_e}}$ because $\epsilon \in \alpha(j) \implies \epsilon \in \at{e}(j) \iff \rH_e(j, L^r)$. It follows that $\rN_{\rH_e}$ simulates $\rM$ i.e.\ $L \subseteq L(\rN_{\rH_e})$ and we are done.
\end{proof}

\begin{corollary}[Maximal legitimate $L$-coverings]
  $\rH_e^{\diamond\diamond}$ is a maximal legitimate $L$-covering.
\end{corollary}

\begin{proof}
  By Theorem \ref{thm:l_ext_induce_l_cover} we know $\rH_e$ is a legitimate $L$-covering. Then $\rH_e^{\diamond\diamond}$ is a maximal $L$-covering by Lemma \ref{lem:l_coverings}.7.c and legitimate by Lemma \ref{lem:l_coverings}.7.e.
\end{proof}

\begin{corollary}
  \label{cor:trans_max_cover_nfa_canonical}
  If $e$ is simplified and transition-maximal then $\rN_{\rH_e}$ is the lower nfa of $\Airr(\toJdfa{e})$.
\end{corollary}

\begin{proof}
  Let $\rM$ be the lower nfa of $\Airr(\toJdfa{e})$. In the proof of Theorem \ref{thm:l_ext_induce_l_cover} we showed $\rN_{\rH_e}$ is an extension of $\rM$. Then by transition-maximality $\rM = \rN_{\rH_e}$.
\end{proof}


Each nfa canonically induces a legitimate $L$-covering -- a pattern which the Kameda-Weiner algorithm can recognise. Moreover every transition-maximal union-free nfa (see Theorem \ref{thm:char_irr_simple}) arises as an induced nfa.\footnote{However, induced nfas needn't be transition-maximal nor union-free.}

\begin{corollary}
  Fix any nfa $\rN$ accepting $L$.
  \begin{enumerate}
    \item 
    $\rH_{\toLext{\jslDfaSc{\rN}}}$ is a legitimate $L$-covering.
    \item
    If $\rN$ is transition-maximal and union-free then $\rN \cong \rN_{\rH_{\toLext{\jslDfaSc{\rN}}}}$.
  \end{enumerate}
\end{corollary}

\begin{proof}
  \item
  \begin{enumerate}
    \item
    The full subset construction $\jslDfaSc{\rN} = \Det(\dep{\rN})$ accepts $L$ by Note \ref{note:det_airr_preserve_acceptance}, so the well-defined $L$-extension $\toLext{\jslDfaSc{\rN}}$ accepts $L$ by Note \ref{note:to_lext_jsdfa_pres_lang}. Then the claim follows by Theorem \ref{thm:l_ext_induce_l_cover}.


    \item
    Since $\rN$ is transition-maximal and union-free it is isomorphic to $\simpleIrr{\rN}$ by Corollary \ref{cor:tmax_and_ufree_simple}. Then $\rN$ is the isomorphic to the lower nfa of $\Airr(\jslDfaSimple{\jslDfaSc{\rN}})$, so the claim follows by Corollary \ref{cor:trans_max_cover_nfa_canonical}.
  \end{enumerate}
\end{proof}

\subsection{Atomic nfas and $L$-extensions}

Fixing any regular language $L$, there are finitely languages arising as a union of the atoms $J(\jslLP{L})$. Recall that these languages are called \emph{atomic} (Definition \ref{def:atomic_cl_langs}). Likewise there are \emph{positively atomic} languages (a subclass of the atomic ones) and also the \emph{subatomic} languages (a superclass of the atomic ones).

\begin{definition}[Atomic, positively atomic and subatomic nfas and $L$-extensions]
  Fix an nfa $\rN$ accepting $L$.
  \begin{enumerate}
    \item 
    $\rN$ is \emph{atomic} if each state accepts an atomic language (equiv.\ $\jslLangs{\rN} \subseteq \jslLP{L}$).
    \item
    $\rN$ is \emph{positively atomic} if each state accepts a positively atomic language (equiv.\ $\jslLangs{\rN} \subseteq \jslLD{L}$).
    \item
    $\rN$ is \emph{subatomic} if each individual state accepts a subatomic language (equiv.\ $\jslLangs{\rN} \subseteq \jslLRP{L}$).
  \end{enumerate}
  
  \noindent
  Finally, an $L$-extension $e$ is \emph{atomic} (resp.\ \emph{positively atomic}, \emph{subatomic}) if the lower nfa of $\Airr(\toJdfa{e})$ is atomic (resp.\ \emph{positively atomic}, \emph{subatomic}).
  \endbox
\end{definition}
  
\begin{lemma}[Atomic $L$-extensions]
  \label{lem:atomic_l_ext_char}
  The following statements concerning $L$-extensions are equivalent.
  \begin{enumerate}
    \item 
    $e$ is atomic.
    \item
    $\jslLangs{\toJdfa{e}} \subseteq \jslLP{L}$.
    \item
    $\jslDfaSimple{\toJdfa{e}}$ is a sub $\JSL$-dfa of $\jslDfaBoolMin{L}$.
    \item 
    $\at{e}$ defines a surjective $\JSL$-dfa morphism to a sub $\JSL$-dfa of $\jslDfaBoolMin{L}$.
  \end{enumerate}
\end{lemma}

\begin{proof}
  \item
  \begin{itemize}
    \item[--]
    $(1) \iff (2)$: By Corollary \ref{cor:dfa_nfa_lang_correspondence} the languages accepted by the lower nfa (varying over subsets) are precisely those accepted by $\toJdfa{e}$ (varying over individual states).

    \item[--]
    $(2) \iff (3)$: Follows because the transition structure of the two $\JSL$-dfas is defined in the same way.
    
    \item[--]
    $(2) \implies (4)$: We know each state of $\toJdfa{e}$ accepts an atomic language, so $\at{e}$ acts in the same way as the $\JSL$-dfa morphism $acc_{\toJdfa{e}}$. Then $\at{e}$ defines a $\JSL$-dfa morphism to a sub $\JSL$-dfa of $\jslDfaBoolMin{L}$.

    \item[--]
    $(4) \implies (2)$: The dfa morphism informs us that each state accepts an atomic language.
  \end{itemize}
\end{proof}

\begin{definition}[Pseudo-atomicity]
  An $L$-extension $e : \jslLQ{L} \monoto (\aT, \delta_a)$ is \emph{pseudo-atomic} if the kernel of the atomizer $ker(\at{e}) \subseteq T \times T$ is closed under $\lambda (x, y). (\delta_a(x), \delta_a(y))$ for each $a \in \Sigma$.
  \endbox
\end{definition}

Then $e$ is pseudo-atomic if the join-semilattice congruence $ker(\at{e}) \subseteq T \times T$ is also a congruence for each unary operation $\delta_a : T \to T$. We're going to show that atomicity and pseudo-atomicity are equivalent concepts.

\begin{lemma}[Pseudo-atomic $L$-extensions]
  \label{lem:pseudo_atomic_basics}
  \item
  \begin{enumerate}
    \item 
    Every atomic $L$-extension is pseudo-atomic.
    \item
    $\jslDfaSimple{-}$ preserves pseudo-atomicity.
    \item
    $e : \jslLQ{L} \monoto (\aT, \delta_a)$ is pseudo-atomic iff the atomized semilattice admits the $L$-extension structure:
    \[
      \iota_e : \jslLQ{L} \hookto (\jslAtz{e}, \phi_a)
      \qquad
      \phi_a(\at{e}(t)) := \at{e}(\delta_a(t)).
    \]
  \end{enumerate}
\end{lemma}

\begin{proof}
  \item
  \begin{enumerate}
    \item 
    By Lemma \ref{lem:atomic_l_ext_char}.4 the dfa morphism $\at{e}$ satisfies $\at{e}(\delta_a(t)) = \delta_a(\at{e}(t))$ for each $a \in \Sigma$. The latter implies $e$ is pseudo-atomic.

    \item
    Let $\delta := \toJdfa{e}$ and $e_0 := \jslDfaSimple{e}$ recalling Definition \ref{def:l_ext_simp_reach_tmax}. Recalling Lemma \ref{lem:jsl_reach_simple_constructions}.3, $\at{e_0}(acc_\delta(t)) = \cl_L(acc_\delta(t)) = \at{e}(t)$ for every $t \in T$. Then given any $Y_i = acc_\delta(t_i)$,
    \[
    \begin{tabular}{lll}
      $\at{e_0}(Y_1) = \at{e_0}(Y_2)$
      &
      $\implies \at{e}(t_1) = \at{e}(t_2)$
      & (see above)
      \\&
      $\implies \at{e}(\delta_a(t_1)) = \at{e}(\delta_a(t_2))$
      & (by assumption)
      \\&
      $\iff \cl_L(acc_\delta(\delta_a(t_1))) = \cl_L(acc_\delta(\delta_a(t_2)))$
      & (by definition)
      \\&
      $\iff \cl_L(a^{-1}(acc_\delta(t_1))) = \cl_L(a^{-1}(acc_\delta(t_2)))$
      & ($acc_\delta$ a $\JSL$-dfa morphism)
      \\&
      $\iff \cl_L(a^{-1} Y_1) = \cl_L(a^{-1} Y_2)$
      \\&
      $\iff \at{e_0}(a^{-1} Y_1) = \at{e_0}(\cl_L(a^{-1} Y_2))$.
    \end{tabular}
    \]
    Hence the simpification of $e$ is also pseudo-atomic.

    \item
    If $\iota_e$ is a well-defined $L$-extension then whenever $\at{e}(t_1) = \at{e}(t_2)$ we deduce $\at{e}(\delta_a(t_1)) = \phi_a(\at{e}(t_1)) = \phi_a(\at{e}(t_2)) = \at{e}(\delta_a(t_2))$, so that $e$ is pseudo-atomic. Conversely, $\at{e}$ is stable under each $\delta_a$ so the endomorphisms $\phi_a : \jslAtz{e} \to \jslAtz{e}$ are well-defined.
    Since $\delta := \toJdfa{e}$ accepts $L$, by varying the initial state it accepts every $Y \in \LQ{L} \subseteq \LP{L}$, so the inclusion $\iota_e : \jslLQ{L} \hookto \jslAtz{e}$ is a well-defined join-semilattice morphism. Observe that each $Y \in \LQ{L}$ has some $t_Y \in T$ with $\at{e}(t_Y) = Y$. Then the calculation:
    \[
      \begin{tabular}{lll}
        $\phi_a(\iota_e (Y))$
        &
        $= \phi_a(Y)$
        \\ &
        $= \phi_a(\at{e}(t_Y))$
        \\ &
        $= \at{e}(\delta_a(t_Y))$
        & (def.\ of $\phi_a$)
        \\ &
        $= \cl_L(acc_{\toJdfa{e}}(\delta_a(t_Y)))$
        & (def.\ of $\at{e}$)
        \\ &
        $= \cl_L(a^{-1}(acc_{\toJdfa{e}}(t_Y)))$
        & ($acc_\delta$ a $\JSL$-dfa morphism)
        \\ &
        $= a^{-1} Y$
        & ($\LP{L}$ closed under $a^{-1}(-)$)
        \\ &
        $= \iota_e(a^{-1} Y)$
      \end{tabular}
    \]
    establishes that $\iota_e$ is an $L$-extension.
  \end{enumerate}
\end{proof}

\begin{theorem}
  An $L$-extension is atomic iff it is pseudo-atomic.
\end{theorem}

\begin{proof}
  If an $L$-extension is atomic it is pseudo-atomic by Lemma \ref{lem:pseudo_atomic_basics}. Conversely given a pseudo-atomic $L$-extension $e : \jslLQ{L} \monoto (\aT, \delta_a)$ then $e_0 := \jslDfaSimple{e}$ is pseudo-atomic by Lemma \ref{lem:pseudo_atomic_basics}.2. By Lemma \ref{lem:pseudo_atomic_basics}.3 we have the $L$-extension $\iota_{e_0} : \jslLQ{L} \hookto (\jslAtz{e}, \phi_a)$ where $\phi_a := \lambda X. a^{-1} X : \jslAtz{e} \to \jslAtz{e}$, since $\cl_L(a^{-1} X) = a^{-1} X$. Then $e_0$ is simplified and each state $X$ accepts $X \in \Atz{e}$, so $e$ is atomic.
\end{proof}

Tamm and Brzozowski proved an nfa $\rN$ is atomic iff $\rev{\rN}$'s reachable subset construction is state-minimal \cite{TheoryOfAtomataBrzTamm2014}. We reprove their result using our terminology and then:
\begin{itemize}
  \item[--]
  refine their result i.e.\ $\rN$ is positively atomic (see Definition \ref{def:atomic_cl_langs}.3) iff the dfa isomorphism is also an order isomorphism.
  
  \item[--]
  generalise their result i.e.\ $\rN$ is subatomic (see Definition \ref{def:atomic_cl_langs}.4) iff $\rsc{\rev{\rN}}$'s transition monoid is suitably isomorphic to $L^r$'s syntactic monoid. 

\end{itemize}

\begin{theorem}[Atomicity and $\rsc{\rev{\rN}}$]
  Let $\rN$ be an nfa accepting $L$.
  \begin{enumerate}
    \item 
    $\rN$ is atomic iff $\rsc{\rev{\rN}} \cong \minDfa{L^r}$ \cite{TheoryOfAtomataBrzTamm2014}.
    \item 
    $\rN$ is positively atomic iff the dfa isomorphism from (1) is also an order isomorphism w.r.t.\ inclusion.
    \item
    $\rN$ is subatomic iff $\rsc{\rev{\rN}}$'s transition monoid is isomorphic to $L^r$'s syntactic monoid via:
    \[
      \lambda \sem{w}_{\TM{\rsc{\rev{\rN}}}}. \sem{w}_{\SynCong{L^r}}
      : \TM{\rsc{\rev{\rN}}} \to \SynMon{L^r}.
    \]
  \end{enumerate}
\end{theorem}

\begin{proof}
  \item
  
  \begin{enumerate}
    \item
    Let $\delta := \jslDfaSc{\minDfa{L^r}}$ be the dual of $\jslDfaBoolMin{L}$ -- see Corollary \ref{cor:dual_jsl_dfa_of_left_preds}.

    \smallskip
    Assuming $\rN$ is an atomic nfa, we have the composite $\JSL$-dfa morphism:
    \[
      \underbrace{\jslDfaReach{\delta}
      \hookto
      \delta
      \overset{q}{\epito}
      \jslDfaReach{\jslDfaSc{\rev{\rN}}}}_\phi.
    \]
    The surjection $q$ arises by dualising $\iota : \jslDfaSimple{\jslDfaSc{\rN}} \hookto \jslDfaBoolMin{L}$ (see Lemma \ref{lem:atomic_l_ext_char}) and applying Corollary \ref{cor:reach_de_morgan_dual_simple} and Corollary \ref{cor:dual_jsl_dfa_of_left_preds}. Viewing $\jslDfaReach{\delta}$ as its underlying dfa, consider its classically reachable part:
    \[
      \begin{tabular}{lll}
        $\reach{\jslDfaReach{\delta}}$
        &
        $= \reach{\delta}$
        & (by definition of $\jslDfaReach{-}$)
        \\ &
        $= \reach{\jslDfaSc{\minDfa{L^r}}}$
        \\ &
        $= \rsc{\rev{\rev{\minDfa{L^r}}}}$
        & (by Note \ref{note:det_airr_preserve_acceptance})
        \\ &
        $= \rsc{\minDfa{L^r}}$
        \\ &
        $\cong \minDfa{L^r}$
        & (holds for any dfa).
      \end{tabular}
    \]
    The above observation provides the injective dfa morphism $\psi$ below:
    \[
      \overbrace{
      \minDfa{L^r}
      \overset{\psi}{\monoto} \jslDfaReach{\delta}
      \overset{\phi}{\to} \jslDfaReach{\jslDfaSc{\rev{\rN}}}}^\chi
    \]
    Since $\minDfa{L^r}$ is state-minimal, the composite dfa morphism $\chi$ is injective. Since $\minDfa{L^r}$ is reachable we obtain the dfa isomorphism $\minDfa{L^r}
    \cong
    \reach{\jslDfaReach{\jslDfaSc{\rev{\rN}})}}
    = \rsc{\rev{\rN}}$ recalling Note \ref{note:det_airr_preserve_acceptance}.

    \smallskip
    Conversely fix an nfa $\rN$ such that $\minDfa{L^r} \cong \rsc{\rev{\rN}}$. By definition of $\jslDfaReach{-}$ and Note \ref{note:det_airr_preserve_acceptance},
    \[
      \reach{\jslDfaReach{\jslDfaSc{\rev{\rN}})}}
      = \reach{\jslDfaSc{\rev{\rN}}}
      = \rsc{\rev{\rN}}.
    \]
    Then by definition of $\jslDfaReach{-}$ we have an injective dfa morphism $\chi : \minDfa{L^r} \monoto \jslDfaReach{\jslDfaSc{\rev{\rN}})}$. By taking the free $\JSL$-dfa on a dfa (Theorem \ref{thm:free_jsl_dfa}) this extends to a $\JSL$-dfa morphism $\hat{\chi} : \delta \to \jslDfaReach{\jslDfaSc{\rev{\rN}}}$. Applying duality, Corollary \ref{cor:dual_jsl_dfa_of_left_preds} and Corollary \ref{cor:reach_de_morgan_dual_simple} we obtain a $\JSL$-dfa morphism $\jslDfaSimple{\jslDfaSc{\rN}} \to \jslDfaBoolMin{L}$. Then every language accepted by $\rN$ is atomic, so that $\rN$ is itself atomic.

    \item
    Let $\delta := \Det (\minSatDfa{L^r}, \subseteq, \rev{\minSatDfa{L^r}})$ be the dual of $\jslDfaDistMin{L}$ -- see Corollary \ref{cor:dual_jsl_dfa_of_left_dist_preds}.
    
    Assuming $\rN$ is positively atomic, we have the composite $\JSL$-dfa morphism:
    \[
      \underbrace{\jslDfaReach{\delta}
      \hookto
      \delta
      \overset{q}{\epito}
      \jslDfaReach{\jslDfaSc{\rev{\rN}}}}_\phi.
    \]
    The surjection $q$ arises by dualising $\iota : \jslDfaSimple{\jslDfaSc{\rN}} \hookto \jslDfaDistMin{L}$ and applying Corollary \ref{cor:reach_de_morgan_dual_simple} and Corollary \ref{cor:dual_jsl_dfa_of_left_dist_preds}. Repeating the argument from (1) we obtain the injective dfa morphism $\psi$ below:
    \[
      \minDfa{L^r}
      \overset{\psi}{\monoto} \jslDfaReach{\delta}
      \overset{\phi}{\to} \jslDfaReach{\jslDfaSc{\rev{\rN}}}.
    \]
    Again repeating the argument in (1), we obtain the dfa isomorphism $\minDfa{L^r} \cong \rsc{\rev{\rN}}$. To see it is an order-isomorphism w.r.t.\ inclusion, first observe $\psi$ has action $\lambda u^{-1} L^r . \{ Y \in \LW{L^r} : Y \subseteq u^{-1} L^r \}$ so it preserves/reflects the inclusion ordering. Finally, $\phi$ certainly preserves inclusions since it is join-semilattice morphism. It reflects inclusions when restricted to $\psi[\minDfa{L^r}]$ because simplicity forbids additional inclusions.

    \smallskip
    Conversely, fix an nfa $\rN$ such that $\minDfa{L^r} \cong \rsc{\rev{\rN}}$ where this isomorphism also preserves and reflects inclusions. Repeating (1) yields the injective dfa morphism $\chi : \minDfa{L^r} \monoto \jslDfaReach{\jslDfaSc{\rev{\rN}}}$, additionally preserving inclusions. In fact, $\chi$ is an ordered dfa morphism (see Definition \ref{def:dfas}) so applying the respective free construction (Theorem \ref{thm:free_jsl_dfa_ordered}) provides $\hat{\chi} : \delta \to \jslDfaReach{\jslDfaSc{\rev{\rN}}}$. Applying duality, Corollary \ref{cor:dual_jsl_dfa_of_left_dist_preds} and Corollary \ref{cor:reach_de_morgan_dual_simple} yields $\jslDfaSimple{\jslDfaSc{\rN}} \to \jslDfaDistMin{L}$, so $\rN$ is positively atomic.

    \item
    Assume $\rN$ is a subatomic nfa. Then we have $\iota : \gamma \hookto \jslDfaSynBoolMin{L}$ where $\gamma := \jslDfaSimple{\jslDfaSc{\rN}}$. Since $\iota$'s codomain is right-quotient closed we also have the $\JSL$-dfa inclusion morphism $\iota_1 : \jslDfaRqc{\gamma} \hookto \jslDfaSynBoolMin{L}$. Dualising, and applying Theorem \ref{thm:trans_semi_dual_rqc} and Theorem \ref{thm:canon_bool_syn_dep_aut}, we obtain a surjective morphism $q_1 : \jslDfaSc{\SynMonDfa{L^r}} \epito \jslDfaTs{\gamma^\pentagram}$.
    Furthermore applying right-quotient closure to the inclusion $\jslDfaMin{L} \hookto \gamma$ yields $\jslDfaSynMin{L} = \jslDfaRqc{\jslDfaMin{L}} \hookto \jslDfaRqc{\gamma}$. Dualising the latter and applying Corollary \ref{cor:dual_syn_semiring} we obtain a surjective morphism $q_2 : \jslDfaTs{\gamma^\pentagram} \epito \jslDfaSyn{L^r}$ onto $L^r$'s syntactic semiring (viewed as a $\JSL$-dfa). Then consider the composite $\JSL$-dfa morphism:
    \[
      \jslDfaSc{\SynMonDfa{L^r}}
      \longepi{q_1} \jslDfaTs{\gamma^\pentagram}
      \longepi{q_2} \jslDfaSyn{L^r}.
    \]
    The classically reachable part of the domain $\JSL$-dfa consists of singleton sets and is isomorphic to $\SynMonDfa{L^r}$. Likewise by Corollary \ref{cor:reach_syn_sr_syn_mon} the reachable part of the codomain is isomorphic to $\SynMonDfa{L^r}$. The image of a reachable dfa under a dfa morphism is reachable, so the composite morphism restricts to:
    \[
      \xymatrix@=10pt{
        \reach{\jslDfaSc{\SynMonDfa{L^r}}} \ar@{->>}[rr]^-{q_1}
        &&
        \reach{\jslDfaTs{\gamma^\pentagram}} \ar@{->>}[rr]^-{q_2}
        &&
        \reach{\jslDfaSyn{L^r}}
        \\
        \SynMonDfa{L^r} \ar@{-}[u]^\cong \ar@{=}[rrrr] &&&& \SynMonDfa{L^r} \ar@{-}[u]_\cong
      }
    \]
    Then $q_1$ is bijective and hence a dfa isomorphism, so that $\reach{\jslDfaTs{\gamma^\pentagram}} \cong \SynMonDfa{L^r}$ too. Importantly,
    \[
      \reach{\jslDfaTs{\gamma^\pentagram}}
      \cong \TMDfa{\rsc{\rev{\rN}}}
    \]
    because $\gamma^\pentagram \cong \jslDfaReach{\jslDfaSc{\rev{\rN}}}$ by Corollary \ref{cor:reach_de_morgan_dual_simple} and Example \ref{ex:dual_full_subset_construction}, so we can apply Lemma \ref{lem:dfa_tmon_as_reach_ts}. Finally, the action of the dfa isomorphism $\TMDfa{\rsc{\rev{\rN}}} \cong \SynMonDfa{L^r}$ defines the desired monoid isomorphism.

    \smallskip
    Conversely suppose $\TM{\rsc{\rev{\rN}}} \cong \SynMon{L^r}$ via the generator-preserving mapping $\sem{w}_{\TM{\rsc{\rev{\rN}}}}$ $\mapsto$ $\sem{w}_{\SynMon{L^r}}$. Its action defines a dfa isomorphism $\TMDfa{\rsc{\rev{\rN}}} \cong \SynMonDfa{L^r}$, where the conditions concerning the initial state and transitions are obvious. The final states are preserved/reflected because:
    \[
      \begin{tabular}{lll}
        $\sem{w}_{\TM{\rsc{\rev{\rN}}}} \in F_{\TMDfa{\rsc{\rev{\rN}}}}$
        &
        $\iff \breve{\rN}_w [I_{\rev{\rN}}] \cap F_{\rev{\rN}} \neq \emptyset$
        \\ &
        $\iff \breve{\rN}_w [F_\rN] \cap I_\rN \neq \emptyset$
        \\ &
        $\iff w^r \in L$
        \\ &
        $\iff w \in L^r$
        \\ &
        $\iff \sem{w}_{\SynCong{L^r}} \in F_{\SynMonDfa{L^r}}$.
      \end{tabular}
    \]
    Applying Lemma \ref{lem:dfa_tmon_as_reach_ts} we deduce
     $\SynMonDfa{L^r} \cong \TMDfa{\rsc{\rev{\rN}}} \cong \reach{\jslDfaTs{\jslDfaSc{\rev{\rN}}}}$. Then we have a dfa morphism $f : \SynMonDfa{L^r} \to \jslDfaTs{\jslDfaSc{\rev{\rN}}}$. Applying the free construction (Theorem \ref{thm:free_jsl_dfa}) we obtain $\hat{f} : \jslDfaSc{\SynMonDfa{L^r}} \to \jslDfaTs{\jslDfaSc{\rev{\rN}}}$. It is actually surjective because $\jslDfaTs{-}$ constructs $\JSL$-reachable machines. Dualising this free-extension yields:
     \[
       \jslDfaRqc{\jslDfaSc{\rN}}
       \cong
       (\jslDfaTs{\jslDfaSc{\rev{\rN}}})^\pentagram
       \xto{\hat{f}^\pentagram}
       (\jslDfaSc{\SynMonDfa{L^r}})^\pentagram
       \cong
       \jslDfaSynBoolMin{L}.
     \]
     The left isomorphism follows by Theorem \ref{thm:trans_semi_dual_rqc} and Example \ref{ex:dualising_reachable_subset_construction}, whereas the right one follows by Theorem \ref{thm:canon_bool_syn_dep_aut}. Finally since $\jslDfaSimple{\jslDfaSc{\rN}} \hookto \jslDfaRqc{\jslDfaSc{\rN}}$ we deduce $\rN$ is subatomic.

  \end{enumerate}

\end{proof}

\bibliographystyle{alpha}
\bibliography{bib-2019}

\end{document}